\documentclass[smallcondensed]{svjour3}

\newif\ifdescdefetc

\newif\iftospringer
\tospringerfalse

\newif\ifnoappendix
\usepackage{currfile}
\ifcurrfilename{trpuc-paper.tex}{
\noappendixfalse
\tospringerfalse
\descdefetctrue
}{
\noappendixtrue
\descdefetctrue
}

\newif\ifforprint
\forprintfalse

\ifnoappendix
\else

\fi

\smartqed

\usepackage[style=alphabetic,backend=biber,backref,hyperref,maxbibnames=9999,minalphanames=3,maxalphanames=4]{biblatex} 
\usepackage{adjustbox}
\usepackage[linesnumbered,vlined]{algorithm2e}
\usepackage{amsmath}
\usepackage{color}
\usepackage{complexity}
\usepackage{dsfont}
\usepackage{fix-cm}
\usepackage{fixltx2e}
\usepackage{latexsym}
\usepackage{mathptmx}
\usepackage{multirow}
\ifnoappendix
\usepackage[numbers]{natbib}
\fi
\usepackage{paralist}
\usepackage{pgf}
\usepackage{pifont}
\usepackage{placeins}
\usepackage[figuresright]{rotating}
\usepackage{tikz}\usetikzlibrary{arrows,automata,shapes}
\iftospringer
\else
\usepackage{xr-hyper}
\fi
\usepackage{xspace}

\definecolor{darkgreen}{rgb}{0.0,0.33,0.0}

\iftospringer
\else
\fi

\iftospringer
\else
\def\svnversion{1316}
\fi

\def\ourtitle{%
Extracting Unsatisfiable Cores for LTL via Temporal Resolution%
}

\ifnoappendix
\iftospringer

\else

\fi
\else

\fi

\bibliography{trpuc-paper-url}
\renewcommand*{\labelalphaothers}{\textsuperscript{+}}
\renewcommand*{\sortalphaothers}{+}
\renewcommand{\multicitedelim}{\addcomma\linebreak[0]}

\newcommand{\mydefinition}[1]{\ifdescdefetc\begin{definition}[{#1}]\else\begin{definition}\fi}

\newcommand{\mytheorem}[4]{%
\newcommand{#1}[1]{%
\ifthenelse{\equal{##1}{true}}{\newcounter{cnt:#2}\setcounter{cnt:#2}{\value{theorem}}}{\newcounter{cnt:#2:backup}\setcounter{cnt:#2:backup}{\value{theorem}}\setcounter{theorem}{\value{cnt:#2}}}%
\ifdescdefetc\begin{theorem}[#3]\else\begin{theorem}\fi%
\ifthenelse{\equal{##1}{true}}{\label{#2}}{}%
{#4}%
\end{theorem}%
\ifthenelse{\equal{##1}{true}}{}{\setcounter{theorem}{\value{cnt:#2:backup}}}%
}%
}

\newcommand{\myproposition}[4]{%
\newcommand{#1}[1]{%
\ifthenelse{\equal{##1}{true}}{\newcounter{cnt:#2}\setcounter{cnt:#2}{\value{proposition}}}{\newcounter{cnt:#2:backup}\setcounter{cnt:#2:backup}{\value{proposition}}\setcounter{proposition}{\value{cnt:#2}}}%
\ifdescdefetc\begin{proposition}[#3]\else\begin{proposition}\fi%
\ifthenelse{\equal{##1}{true}}{\label{#2}}{}%
{#4}%
\end{proposition}%
\ifthenelse{\equal{##1}{true}}{}{\setcounter{proposition}{\value{cnt:#2:backup}}}%
}%
}

\newcommand{\mylemma}[4]{%
\newcommand{#1}[1]{%
\ifthenelse{\equal{##1}{true}}{\newcounter{cnt:#2}\setcounter{cnt:#2}{\value{lemma}}}{\newcounter{cnt:#2:backup}\setcounter{cnt:#2:backup}{\value{lemma}}\setcounter{lemma}{\value{cnt:#2}}}%
\ifdescdefetc\begin{lemma}[#3]\else\begin{lemma}\fi%
\ifthenelse{\equal{##1}{true}}{\label{#2}}{}%
{#4}%
\end{lemma}%
\ifthenelse{\equal{##1}{true}}{}{\setcounter{lemma}{\value{cnt:#2:backup}}}%
}%
}

\newcommand{\myremark}[4]{%
\newcommand{#1}[1]{%
\ifthenelse{\equal{##1}{true}}{\newcounter{cnt:#2}\setcounter{cnt:#2}{\value{remark}}}{\newcounter{cnt:#2:backup}\setcounter{cnt:#2:backup}{\value{remark}}\setcounter{remark}{\value{cnt:#2}}}%
\ifdescdefetc\begin{remark}[#3]\else\begin{remark}\fi%
\ifthenelse{\equal{##1}{true}}{\label{#2}}{}%
{#4}%
\end{remark}%
\ifthenelse{\equal{##1}{true}}{}{\setcounter{remark}{\value{cnt:#2:backup}}}%
}%
}

\newcommand{\mycorollary}[4]{%
\newcommand{#1}[1]{%
\ifthenelse{\equal{##1}{true}}{\newcounter{cnt:#2}\setcounter{cnt:#2}{\value{corollary}}}{\newcounter{cnt:#2:backup}\setcounter{cnt:#2:backup}{\value{corollary}}\setcounter{corollary}{\value{cnt:#2}}}%
\ifdescdefetc\begin{corollary}[#3]\else\begin{corollary}\fi%
\ifthenelse{\equal{##1}{true}}{\label{#2}}{}%
{#4}%
\end{corollary}%
\ifthenelse{\equal{##1}{true}}{}{\setcounter{corollary}{\value{cnt:#2:backup}}}%
}%
}

\newcommand{\vstodo}[1]%
{{\color{red} \bf \tt<VS todo>}{#1}{\color{red} \bf \tt</VS todo>}}

\newcommand{\uc}{UC\xspace}
\newcommand{\ucs}{UCs\xspace}
\newcommand{\UC}{UC\xspace}
\newcommand{\UCs}{UCs\xspace}

\newcommand{\tr}{TR\xspace}
\newcommand{\TR}{TR\xspace}

\newcommand{\result}[1]{{\emph{#1}}\xspace}
\newcommand{\sat}{\result{sat}}

\newcommand{\unsat}{\result{unsat}}
\newcommand{\Unsat}{\result{Unsat}}

\newcommand{\proofimplies}{\ensuremath{\Rightarrow}}
\newcommand{\proofbiimplies}{\ensuremath{\Leftrightarrow}}
\newcommand{\definedas}{\ensuremath{\equiv}}

\newcommand{\mkpair}[2]{\ensuremath{({#1},{#2})}}
\newcommand{\mkpowerset}[1]{\ensuremath{2^{#1}}}

\newcommand{\fcardinality}[1]{\ensuremath{|{#1}|}}

\newcommand{\bigo}{\ensuremath{\mathcal{O}}}
\newcommand{\fbigo}[1]{\ensuremath{\bigo({#1})}}

\newcommand{\allaps}{\ensuremath{\mathit{AP}}}
\newcommand{\ap}{\ensuremath{\mathit{p}}}
\newcommand{\app}{\ensuremath{\mathit{q}}}
\newcommand{\appp}{\ensuremath{\mathit{r}}}
\newcommand{\apppp}{\ensuremath{\mathit{s}}}
\newcommand{\apres}{\ensuremath{\mathit{l}}}
\newcommand{\apresp}{\ensuremath{\mathit{l'}}}

\newcommand{\apreszero}{\ensuremath{\mathit{l_0}}}
\newcommand{\apresone}{\ensuremath{\mathit{l_1}}}
\newcommand{\aprestwo}{\ensuremath{\mathit{l_2}}}
\newcommand{\apresthree}{\ensuremath{\mathit{l_3}}}
\newcommand{\apresfour}{\ensuremath{\mathit{l_4}}}
\newcommand{\apresfive}{\ensuremath{\mathit{l_5}}}
\newcommand{\dap}{\ensuremath{\mathit{P}}}
\newcommand{\dapp}{\ensuremath{\mathit{Q}}}
\newcommand{\dappp}{\ensuremath{\mathit{R}}}
\newcommand{\dapppp}{\ensuremath{\mathit{S}}}

\newcommand{\allbools}{\ensuremath{\mathds{B}}}
\newcommand{\false}{\ensuremath{\textsc{false}}}
\newcommand{\true}{\ensuremath{\textsc{true}}}

\newcommand{\iword}{\ensuremath{\pi}}
\newcommand{\fiwordgetpos}[2]{\ensuremath{{#1}[{#2}]}}

\newcommand{\allnats}{\ensuremath{\mathds{N}}}

\newcommand{\pos}{\ensuremath{{i}}}
\newcommand{\posp}{\ensuremath{{i'}}}
\newcommand{\pospp}{\ensuremath{{i''}}}
\newcommand{\posppp}{\ensuremath{{i'''}}}
\newcommand{\maxind}{\ensuremath{{n}}}
\newcommand{\maxindp}{\ensuremath{{n'}}}
\newcommand{\maxindpp}{\ensuremath{{n''}}}
\newcommand{\maxindppp}{\ensuremath{{n'''}}}

\newcommand{\inp}{\ensuremath{\mathit{\phi}}}
\newcommand{\inpp}{\ensuremath{\mathit{\phi'}}}
\newcommand{\inppp}{\ensuremath{\mathit{\phi''}}}

\newcommand{\inpuc}{\ensuremath{\mathit{\phi^{uc}}}}

\newcommand{\prt}{\ensuremath{\mathit{\psi}}}
\newcommand{\prtp}{\ensuremath{\mathit{\psi'}}}

\newcommand{\positivepolarity}{\ensuremath{+}}
\newcommand{\negativepolarity}{\ensuremath{-}}

\newcommand{\fbnotname}{\ensuremath{\neg}}
\newcommand{\fbnot}[1]{\ensuremath{\fbnotname{#1}}}
\newcommand{\fborname}{\ensuremath{\vee}}
\newcommand{\fbor}[2]{\ensuremath{{#1}\fborname{#2}}}
\newcommand{\fbandname}{\ensuremath{\wedge}}
\newcommand{\fband}[2]{\ensuremath{{#1}\fbandname{#2}}}
\newcommand{\fbimpliesname}{\ensuremath{\rightarrow}}
\newcommand{\fbimplies}[2]{\ensuremath{{#1}\fbimpliesname{#2}}}

\newcommand{\fnextname}{\ensuremath{{\bf X}}}
\newcommand{\fnext}[1]{\ensuremath{\fnextname{#1}}}
\newcommand{\ffinallyname}{\ensuremath{{\bf F}}}
\newcommand{\ffinally}[1]{\ensuremath{\ffinallyname{#1}}}
\newcommand{\fgloballyname}{\ensuremath{{\bf G}}}
\newcommand{\fglobally}[1]{\ensuremath{\fgloballyname{#1}}}
\newcommand{\funtilname}{\ensuremath{{\bf U}}}
\newcommand{\funtil}[2]{\ensuremath{{#1}{\funtilname}{#2}}}
\newcommand{\freleasesname}{\ensuremath{{\bf R}}}
\newcommand{\freleases}[2]{\ensuremath{{#1}{\freleasesname}{#2}}}

\newcommand{\allclauses}{\ensuremath{\mathds{C}}}

\newcommand{\ficlause}[1]{\ensuremath{({#1})}}
\newcommand{\fgnxclause}[2]{\ensuremath{(\fglobally{(\fbor{{#1}}{\fnext{({#2})}})})}}
\newcommand{\fgnclause}[1]{\ensuremath{(\fglobally{({#1})})}}
\newcommand{\fgxclause}[1]{\ensuremath{(\fglobally{\fnext{({#1})}})}}
\newcommand{\fgneclause}[2]{\ensuremath{(\fglobally{(\fbor{{#1}}{\ffinally{{#2}}})})}}
\newcommand{\fgeclause}[1]{\ensuremath{(\fglobally{\ffinally{({#1})}})}}
\newcommand{\emptyclause}{\ensuremath{\Box}}
\newcommand{\clause}{\ensuremath{{c}}}
\newcommand{\clausep}{\ensuremath{{c'}}}
\newcommand{\clausepp}{\ensuremath{{c''}}}

\newcommand{\clausezero}{\ensuremath{{c_0}}}
\newcommand{\clauseone}{\ensuremath{{c_1}}}
\newcommand{\clausetwo}{\ensuremath{{c_2}}}
\newcommand{\clausethree}{\ensuremath{{c_3}}}
\newcommand{\clausefour}{\ensuremath{{c_4}}}
\newcommand{\clausefive}{\ensuremath{{c_5}}}
\newcommand{\clausesix}{\ensuremath{{c_6}}}

\newcommand{\setclauses}{\ensuremath{{C}}}
\newcommand{\setclausesp}{\ensuremath{{C'}}}
\newcommand{\setclausespp}{\ensuremath{{C''}}}
\newcommand{\setclausesuc}{\ensuremath{{C^{uc}}}}
\newcommand{\setclausesucp}{\ensuremath{{C^{uc'}}}}

\newcommand{\mainpartition}{\ensuremath{M}}
\newcommand{\mainpartitionp}{\ensuremath{M'}}
\newcommand{\looppartition}{\ensuremath{L}}
\newcommand{\looppartitionp}{\ensuremath{L'}}

\newcommand{\mainorlooppartition}{\ensuremath{M\!L}}

\newcommand{\refinferencerule}[1]{{\scriptsize\setlength\fboxsep{2pt}\fbox{{#1}}}\xspace}
\newcommand{\refrule}{\refinferencerule{rule}}
\newcommand{\refinitii}{\refinferencerule{init-ii}}
\newcommand{\refinitin}{\refinferencerule{init-in}}
\newcommand{\refstepnn}{\refinferencerule{step-nn}}
\newcommand{\refstepnx}{\refinferencerule{step-nx}}
\newcommand{\refstepxx}{\refinferencerule{step-xx}}
\newcommand{\refloopitinitx}{\refinferencerule{BFS-loop-it-init-x}}
\newcommand{\refloopitinitn}{\refinferencerule{BFS-loop-it-init-n}}
\newcommand{\refloopitinitc}{\refinferencerule{BFS-loop-it-init-c}}
\newcommand{\refloopitsub}{\refinferencerule{BFS-loop-it-sub}}
\newcommand{\refloopconclusionone}{\refinferencerule{BFS-loop-conclusion1}}
\newcommand{\refloopconclusiontwo}{\refinferencerule{BFS-loop-conclusion2}}
\newcommand{\refaugone}{\refinferencerule{aug1}}
\newcommand{\refaugtwo}{\refinferencerule{aug2}}
\newcommand{\premise}{\ensuremath{p}\xspace}
\newcommand{\conclusion}{\ensuremath{c}\xspace}

\newcommand{\fapwaitfor}[1]{\ensuremath{w{#1}}}

\newcommand{\tool}[1]{{\tt #1}\xspace}
\newcommand{\lwb}{\tool{LWB}}
\newcommand{\nusmv}{\tool{NuSMV}}
\newcommand{\pltl}{\tool{pltl}}
\newcommand{\trp}{\tool{TRP++}}
\newcommand{\tspass}{\tool{TSPASS}}
\newcommand{\pltlmup}{\tool{PLTL-MUP}}
\newcommand{\procmine}{\tool{procmine}}
\newcommand{\benchmark}[1]{{\bf #1}}

\newcommand{\algassignname}{\ensuremath{\leftarrow}}
\newcommand{\algassign}[2]{\ensuremath{{#1} \algassignname {#2}}}

\newcommand{\fdisjunion}[2]{\ensuremath{{{#1}}\uplus{{#2}}}}

\newcommand{\graph}{\ensuremath{{G}}}
\newcommand{\graphp}{\ensuremath{{G'}}}
\newcommand{\setvertices}{\ensuremath{V}}

\newcommand{\vertex}{\ensuremath{v}}
\newcommand{\vertexp}{\ensuremath{v'}}
\newcommand{\vertexpp}{\ensuremath{v''}}

\newcommand{\vertexzero}{\ensuremath{v_0}}
\newcommand{\vertexone}{\ensuremath{v_1}}
\newcommand{\vertextwo}{\ensuremath{v_2}}
\newcommand{\vertexthree}{\ensuremath{v_3}}
\newcommand{\vertexfour}{\ensuremath{v_4}}
\newcommand{\vertexfive}{\ensuremath{v_5}}
\newcommand{\vertexsix}{\ensuremath{v_6}}
\newcommand{\vertexseven}{\ensuremath{v_7}}
\newcommand{\setedges}{\ensuremath{E}}

\newcommand{\vertexlabelingname}{\ensuremath{L_V}}
\newcommand{\fvertexlabeling}[1]{\ensuremath{\vertexlabelingname({#1})}}

\newcommand{\partitioningv}{\ensuremath{\mathcal{P}^\setvertices}}
\newcommand{\mainpartitionv}{\ensuremath{M^\setvertices}}
\newcommand{\looppartitionv}{\ensuremath{L^\setvertices}}
\newcommand{\edge}{\ensuremath{e}}

\newcommand{\ppath}{\ensuremath{\pi}}
\newcommand{\ppathp}{\ensuremath{\pi'}}

\newcommand{\fdCNF}[1]{\ensuremath{\mathit{SNF}({#1})}}
\newcommand{\fdCNFaux}[1]{\ensuremath{\mathit{SNF}_\mathit{aux}({#1})}}
\newcommand{\alldCNFvars}{\ensuremath{X}}
\newcommand{\fdCNFvar}[1]{\ensuremath{{x}_{#1}}}
\newcommand{\fdCNFvarm}[1]{\ensuremath{\color{blue}\setlength\fboxsep{1pt}\fbox{${x}_{#1}$}}}
\newcommand{\dCNFvar}{\ensuremath{x}}
\newcommand{\dCNFvarp}{\ensuremath{x'}}
\newcommand{\dCNFconj}{\ensuremath{c}}

\newcommand{\fgroupname}{\ensuremath{group}}
\newcommand{\fgroup}[2]{\ensuremath{\fgroupname(\mkpair{{#1}}{{#2}})}}
\newcommand{\fgroupltlname}{\ensuremath{group_{LTL}}}
\newcommand{\fgroupltl}[2]{\ensuremath{\fgroupltlname(\mkpair{{#1}}{{#2}})}}
\newcommand{\alloccurrences}{\ensuremath{\mathcal{O}}}
\newcommand{\occurrence}{\ensuremath{o}}
\newcommand{\occurrencep}{\ensuremath{o'}}
\newcommand{\occurrencepp}{\ensuremath{o''}}

\newcommand{\occurrencezero}{\ensuremath{o_0}}
\newcommand{\occurrenceone}{\ensuremath{o_1}}
\newcommand{\occurrencetwo}{\ensuremath{o_2}}
\newcommand{\occurrencethree}{\ensuremath{o_3}}
\newcommand{\occurrencefour}{\ensuremath{o_4}}
\newcommand{\occurrencefive}{\ensuremath{o_5}}
\newcommand{\occurrencesix}{\ensuremath{o_6}}
\newcommand{\occurrenceseven}{\ensuremath{o_7}}
\newcommand{\graphgrouped}{\ensuremath{\graph_{\fgroupname}}}
\newcommand{\fmapname}{\ensuremath{f}}
\newcommand{\fmap}[1]{\ensuremath{\fmapname({#1})}}
\newcommand{\fmapltlname}{\ensuremath{f_{LTL}}}
\newcommand{\fmapltl}[1]{\ensuremath{\fmapltlname({#1})}}
\newcommand{\fsubstitutionname}{\ensuremath{subst}}
\newcommand{\fsubstitution}[1]{\ensuremath{\fsubstitutionname({#1})}}
\newcommand{\fsubstitutionltlname}{\ensuremath{subst_{LTL}}}
\newcommand{\fsubstitutionltl}[1]{\ensuremath{\fsubstitutionltlname({#1})}}
\newcommand{\fmapltlocctosnfoccsname}{\ensuremath{g}}

\newcommand{\setclausesucgrouped}{\ensuremath{{C^{uc}_{\fgroupname}}}}
\newcommand{\setclausesucgroupedltl}{\ensuremath{{C^{uc}_{\fgroupltlname}}}}
\newcommand{\inpucgrouped}{\ensuremath{\mathit{\phi^{uc}_{\fgroupname}}}}

\newcommand{\system}{\ensuremath{\zeta}}
\newcommand{\spec}{\ensuremath{\inp}}
\newcommand{\specp}{\ensuremath{\inpp}}
\newcommand{\setoccurrences}{\ensuremath{O}}
\newcommand{\setoccurrencesp}{\ensuremath{O'}}
\newcommand{\sysandnotspec}{\ensuremath{\mu}}
\newcommand{\sysandnotspecp}{\ensuremath{\mu'}}

\newclass{\FPSPACE}{FPSPACE}

\ifnoappendix
\iftospringer
\newcommand{\refformalcoreextraction}{App.~A of \cite{fullversion}\xspace}
\newcommand{\refmoreplots}{App.~B of \cite{fullversion}\xspace}
\else
\newcommand{\refformalcoreextraction}{App.~\ref{full-sec:formal-coreextraction} of \cite{fullversion}\xspace}
\newcommand{\refmoreplots}{App.~\ref{full-sec:moreplots} of \cite{fullversion}\xspace}
\fi
\else
\newcommand{\refformalcoreextraction}{App.~\ref{sec:formal-coreextraction}\xspace}
\newcommand{\refmoreplots}{App.~\ref{sec:moreplots}\xspace}
\fi
\newcommand{\refalgltlsattrp}{the algorithm in Fig.~\ref{fig:ltlsattrp}\xspace}

\ifnoappendix
\iftospringer
\else
\fi
\fi

\journalname{Acta Informatica}

\begin{document}
\allowdisplaybreaks[4]

\title{\ourtitle\thanks{A preliminary version of this paper appeared in \cite{VSchuppan-TIME-2013}.}}

\ifnoappendix
\else
\subtitle{(full version; r\svnversion, May 23, 2015)}
\fi

\author{Viktor Schuppan}
\institute{\email{Viktor.Schuppan@gmx.de}\\
           URL: \url{http://www.schuppan.de/viktor/}}

\date{Received: / Accepted: }

\maketitle

\begin{abstract}
Unsatisfiable cores (\ucs) are a well established means for debugging in a declarative setting.
Still, there are few tools that perform automated extraction of \ucs for LTL.
Existing tools compute a \uc as an unsatisfiable subset of the set of top-level conjuncts of an LTL formula.
Using resolution graphs to extract \ucs is common in other domains such as SAT.
In this article we construct and optimize resolution graphs for temporal resolution as implemented in the temporal resolution-based solver \trp, and we use them to extract \ucs for propositional LTL.
The resulting \ucs are more fine-grained than the \ucs obtained from existing tools because \uc extraction also simplifies top-level conjuncts instead of treating them as atomic entities.
For example, given an unsatisfiable LTL formula of the form $\inp \definedas \fband{(\fglobally{\prt})}{\ffinally{\prtp}}$ existing tools return $\inp$ as a \uc irrespective of the complexity of $\prt$ and $\prtp$, whereas the approach presented in this article continues to remove parts not required for unsatisfiability inside $\prt$ and $\prtp$.
Our approach also identifies groups of occurrences of a proposition that do not interact in a proof of unsatisfiability.
We implement our approach in \trp.
Our experimental evaluation demonstrates that our approach
\begin{inparaenum}[(i)]
\item extracts \ucs that are often significantly smaller than the input formula with an acceptable overhead and
\item produces more fine-grained \ucs than competing tools while remaining at least competitive in terms of run time and memory usage.
\end{inparaenum}
The source code of our tool is publicly available.

\keywords{LTL \and unsatisfiable cores \and resolution graphs \and vacuity \and temporal resolution}
\end{abstract}

\section{Introduction}

\subsection{Motivation}

Debugging is an activity that many hardware and software developers spend a fair amount of time on.
When faced with some input that induces an undesired behavior it is typically suggested to minimize that failure-inducing input in order to simplify identification of the problem (e.g., \cite{AZellerRHildebrandt-TrSE-2002}).
Corresponding research has been performed, e.g., in linear programming (e.g., \cite{JChinneckEDravnieks-ORSAJournalOnComputing-1991}), constraint satisfaction (e.g., \cite{RBakkerFDikkerFTempelmanPWognum-IJCAI-1993}), compilers (e.g., \cite{DWhalley-TOPLAS-1994}), SAT (e.g., \cite{RBruniASassano-SAT-2001}), description logics (e.g., \cite{SSchlobachRCornet-IJCAI-2003}), declarative specifications (e.g., \cite{IShlyakhterRSeaterDJacksonMSridharanMTaghdiri-ASE-2003}), and LTL satisfiability (e.g., \cite{VSchuppan-SCP-2012}) and realizability (e.g., \cite{ACimattiMRoveriVSchuppanATchaltsev-VMCAI-2008}).

LTL \ifnoappendix\cite{EEmerson-HandbookOfTheoreticalComputerScience-1990} \else\cite{APnueli-FOCS-1977,EEmerson-HandbookOfTheoreticalComputerScience-1990} \fi and its relatives (e.g., \cite{CEisnerDFisman-2006,HKressGazitGFainekosGPappas-AdvancedRobotics-2008,MPesicWVanDerAalst-BPM-2006}) are important specification languages for reactive systems (e.g., \cite{CEisnerDFisman-2006,HKressGazitGFainekosGPappas-AdvancedRobotics-2008}) and for business processes (e.g., \cite{MPesicWVanDerAalst-BPM-2006}).
Experience in verification as well as in synthesis has lead to specifications themselves becoming objects of analysis.
Beer et al.~report \cite{IBeerSBenDavidCEisnerYRodeh-FMSD-2001} that in their experience ``[...] during the first formal verification runs of a new hardware design, typically 20 \% of formulas are found to be trivially valid, and that trivial validity always points to a real problem in either the design or its specification or environment.''.
In a work on LTL synthesis \cite{RBloemSGallerBJobstmannNPitermanAPnueliMWeiglhofer-COCV-2007} Bloem et al.~state that ``[...] writing a complete formal specification [...] was not trivial.'' and ``Although this approach removes the need for verification [...] the specification itself still needs to be validated.''.

Typically, a specification is expected to be satisfiable.
If it turns out to be unsatisfiable, finding a reason for unsatisfiability can help with the ensuing debugging.
Frequently, such a reason for unsatisfiability is taken to be a part of the unsatisfiable specification that is by itself unsatisfiable (e.g., \cite{VSchuppan-SCP-2012,RBakkerFDikkerFTempelmanPWognum-IJCAI-1993,JChinneckEDravnieks-ORSAJournalOnComputing-1991}); this is called an unsatisfiable core (\uc) (e.g., \cite{VSchuppan-SCP-2012,EGoldbergYNovikov-DATE-2003,LZhangSMalik-DATE-2003,HHoos-UBritishColumbiaTR-1999}).

Less simplistic ways to examine an LTL specification $\inp$ exist \cite{IPillSSempriniRCavadaMRoveriRBloemACimatti-DAC-2006}, and understanding their results also benefits from availability of \ucs.
First, one can ask whether a certain scenario $\inpp$, given as an LTL formula, is permitted by $\inp$, i.e., whether in a situation, in which the specification $\inp$ holds, the scenario $\inpp$ can occur.
That is the case iff $\fband{\inp}{\inpp}$ is satisfiable.
Second, one can check whether $\inp$ ensures a certain LTL property $\inppp$. $\inppp$ holds in $\inp$ iff $\fband{\inp}{\fbnot{\inppp}}$ is unsatisfiable.
In the first case, if the scenario turns out not to be permitted by the specification, a \uc can help to understand which parts of the specification and the scenario are responsible for that.
In the second case a \uc can show which parts of the specification imply the property.
Moreover, if there are parts of the property that are not part of the \uc, then those parts of the property could be strengthened without falsifying the property in the specification; i.e., the property is vacuously satisfied (e.g., \cite{IBeerSBenDavidCEisnerYRodeh-FMSD-2001,OKupfermanMVardi-STTT-2003,RArmoniLFixAFlaisherOGrumbergNPitermanATiemeyerMVardi-CAV-2003,AGurfinkelMChechik-TACAS-2004,JSimmondsJDaviesAGurfinkelMChechik-STTT-2010,DFismanOKupfermanSSheinvaldFaragyMVardi-HVC-2008,OKupferman-CONCUR-2006}).

\UCs are therefore an important part of design methods for embedded systems (e.g., \cite{IPillSSempriniRCavadaMRoveriRBloemACimatti-DAC-2006}) as well as for business processes (e.g., \cite{AAwadRGoreZHouJThomsonMWeidlich-InformationSystems-2012}).
Note that specifications of real world systems may be 100s of pages long (e.g., \cite{AChiappiniACimattiLMacchiORebolloMRoveriASusiSTonettaBVittorini-ICSE-2010}).
Hence, providing automated support for obtaining a \uc in case such a specification turns out to be unsatisfiable is crucial.

\subsection{Contributions}

\begin{itemize}
\item[{\bf Fine-Grained \UCs for LTL}]

\ifnoappendix
\else
\begin{sloppypar}
\fi
Despite the relevance of \ucs for LTL as outlined above interest in them has been somewhat limited (e.g., \cite{ACimattiMRoveriVSchuppanSTonetta-CAV-2007,VSchuppan-SCP-2012,AAwadRGoreZHouJThomsonMWeidlich-InformationSystems-2012,VSchuppan-QAPL-2013,RGoreJHuangTSergeantJThomson-Report-2013,FHantryMHacid-FLACOS-2011,FHantryLSaisMHacid-ElectronicColloquiumOnComputationalComplexity-2012}).
In particular, publicly available tools that automatically extract \ucs for propositional LTL are scarce.
We are aware of two such tools: \pltlmup\footnote{\url{http://www.timsergeant.com/pltl-mup/}} \cite{RGoreJHuangTSergeantJThomson-Report-2013} and \procmine\footnote{\url{http://users.cecs.anu.edu.au/~rpg/BusinessProcessModelling/procmine.zip}} \cite{AAwadRGoreZHouJThomsonMWeidlich-InformationSystems-2012}.
\ifnoappendix
\else
\end{sloppypar}
\fi

The \ucs produced by \pltlmup and \procmine are somewhat coarse-grained in the following sense.
Both tools take as input a set of LTL formulas $\inp$.\footnote{A set of LTL formulas $\inp$ is interpreted as the conjunction of all formulas in $\inp$.} 
If that set of LTL formulas is unsatisfiable, then they produce a subset $\inpuc \subseteq \inp$ that is still unsatisfiable.
However, they treat the LTL formulas that are the elements of $\inp$ as atomic entities; i.e., they do not analyze whether all subformulas of the LTL formulas that make up $\inpuc$ are required for unsatisfiability.

In this article we propose an approach that takes an LTL formula as input and determines for each node in the syntax tree whether the subformula rooted at that node is necessary for unsatisfiability.
Hence, if $\inp \equiv \{\fglobally{(\fband{\ap}{\prt})}, \ffinally{(\fband{(\fbnot{\ap})}{\prtp})}\}$, then \pltlmup and \procmine will return $\inp$ as \uc irrespective of the complexity of $\prt$ and $\prtp$.
Our approach takes $\fband{(\fglobally{(\fband{\ap}{\prt})})}{\ffinally{(\fband{(\fbnot{\ap})}{\prtp})}}$ as input and returns $\fband{(\fglobally{(\fband{\ap}{\true})})}{\ffinally{(\fband{(\fbnot{\ap})}{\true})}}$ as \uc.

\item[{\bf \UCs for LTL via Temporal Resolution}]

Extracting \ucs is often possible using any solver for the logic under consideration by weakening subformulas one by one and using the solver to test whether the weakened formula is still unsatisfiable (e.g., \cite{JMarquesSilva-JournalOnMultipleValuedLogicAndSoftComputing-2012}).
Although that is simple to implement, repeated testing for preservation of unsatisfiability may impose a significant run time burden.
A potential alternative are methods that extract \ucs by analyzing a single run of a solver.
It is interesting to investigate such methods because they might not only reduce the run time burden but could also reveal additional information on why a formula is unsatisfiable (see, e.g., Sec.~\ref{sec:groupedpropositionsucs} and \cite{VSchuppan-QAPL-2013}).
Extracting \ucs from resolution graphs is common in SAT (e.g., \ifnoappendix\cite{LZhang-PhDThesis-2003}\else\cite{LZhangSMalik-SAT-2003,LZhang-PhDThesis-2003}\fi).
A resolution method (e.g., \ifnoappendix\cite{LBachmairHGanzinger-HandbookOfAutomatedReasoning-2001}\else\cite{JRobinson-JACM-1965,LBachmairHGanzinger-HandbookOfAutomatedReasoning-2001}\fi) for LTL, temporal resolution (\tr), was suggested by Fisher \cite{MFisher-IJCAI-1991,MFisherCDixonMPeim-ACMTrComputationalLogic-2001} and implemented in \trp \cite{UHustadtBKonev-CollegiumLogicum-2004,UHustadtBKonev-CADE-2003}.
\trp is available as source code\footnote{\url{http://www.csc.liv.ac.uk/~konev/software/trp++/}}.

\TR lends itself as a basis for extracting \ucs for LTL for two reasons.
First, the \tr-based solver \trp proved to be competitive in a recent evaluation of solvers for LTL satisfiability, in particular on unsatisfiable instances (see pp.~51--55 of the full version of \cite{VSchuppanLDarmawan-ATVA-2011}).
Second, a \tr proof naturally induces a resolution graph, which provides a clean framework for extracting a \uc.
Among the other solvers evaluated in \cite{VSchuppanLDarmawan-ATVA-2011} we mention the BDD-based solver \nusmv \cite{ACimattiEClarkeEGiunchigliaFGiunchigliaMPistoreMRoveriRSebastianiATacchella-CAV-2002} and the tableau-based solvers \lwb \cite{AHeuerdingGJaegerSSchwendimannMSeyfried-TABLEAUX-1995} and \pltl\footnote{\url{http://users.cecs.anu.edu.au/~rpg/PLTLProvers/}}.
Although \nusmv also performed well on unsatisfiable instances in \cite{VSchuppanLDarmawan-ATVA-2011}, the BDD layer makes extraction of a \uc more involved than a \tr proof.
On the other hand, \lwb and \pltl provide access to a proof of unsatisfiability comparable to \tr, yet tended to perform worse than \trp on unsatisfiable instances in \cite{VSchuppanLDarmawan-ATVA-2011}.

In this article we show how to obtain a \uc from an execution of the \tr algorithm as implemented in \trp.
At the heart of our method is the construction of a resolution graph for \tr for propositional LTL; note that \tr is significantly more complex than propositional resolution.
We also show how to use the specifics of \tr in \trp to optimize the construction of the resolution graph.

\item[{\bf Interaction of Occurrences of Propositions in a \UC}]

\begin{sloppypar}
A resolution graph contains not only information on which parts of the input formula were used to derive unsatisfiability but also how these parts were used.
We therefore suggest to exploit the resolution graph to provide more detailed information on unsatisfiability.
In this article we use the resolution graph to point out which occurrences of atomic propositions interact in a \uc.
In a companion paper \cite{VSchuppan-QAPL-2013}, which this article provides the basis for, we use it to show which subformulas are relevant for unsatisfiability at which points in time.
\end{sloppypar}

\item[{\bf Mapping a \UC in Separated Normal Form to a \UC in LTL}]

A potential disadvantage of using \tr for extracting \ucs for LTL is the fact that \tr does not work directly on LTL but on a clausal normal form called Separated Normal Form (SNF) \cite{MFisher-IJCAI-1991,MFisherPNoel-UManchesterTR-1992,MFisherCDixonMPeim-ACMTrComputationalLogic-2001}.
Translations from an LTL formula into an equisatisfiable formula in a clausal normal form are well known both in temporal resolution (e.g., \cite{MFisher-IJCAI-1991,MFisherPNoel-UManchesterTR-1992,MFisherCDixonMPeim-ACMTrComputationalLogic-2001}) and in (symbolic) model checking (e.g., \cite{OLichtensteinAPnueli-POPL-1985,JBurchEClarkeKMcMillanDDillLHwang-IaC-1992,EClarkeOGrumbergKHamaguchi-FMSD-1997,YKestenAPnueliLRaviv-ICALP-1998,ABiereKHeljankoTJunttilaTLatvalaVSchuppan-LMCS-2006}).
However, for optimal support of a user who tries to track down the source of the unsatisfiability of a formula it is likely to be helpful if the \uc that is presented to her is ``syntactically close'' to the formula that she provided as an input to the solver.
Hence, it is necessary to map a \uc obtained in SNF back to LTL.
We show how to translate a \uc from SNF back to the LTL formula that the user provided as an input.

\item[{\bf Reductions Between \UCs for LTL and Mutual Vacuity}]

We discuss the relation between \ucs for LTL and mutual vacuity \cite{AGurfinkelMChechik-TACAS-2004}.
Specifically, we prove that the problem of mutual vacuity for a system and a specification described by an LTL formula is reducible to the problem of finding a \uc for an LTL formula and vice versa.

\item[{\bf Publicly Available Implementation}]

We implement our method in \trp.
We make the source code of our solver publicly available.

\item[{\bf Experimental Evaluation}]

Our experimental evaluation demonstrates that our approach
\begin{inparaenum}[(i)]
\item is viable in terms of the run time and memory overhead that it induces,
\item can significantly reduce the size of an unsatisfiable input formula, and
\item is competitive to alternative approaches, while it produces more fine-grained results.
\end{inparaenum}
\end{itemize}

\UCs also have applications in avoiding the exploration of parts of a search space that can be known not to contain a solution for reasons ``equivalent'' to the reasons for previous failures (e.g., \cite{EClarkeMTalupurHVeithDWang-SAT-2003,ACimattiMRoveriVSchuppanSTonetta-CAV-2007}) and in certifying the correctness of a result of unsatisfiability (e.g., \cite{AVanGelder-AMAI-2002,EGoldbergYNovikov-DATE-2003,LZhangSMalik-DATE-2003}).
These applications also benefit from our results.
Nevertheless, we only focus on debugging.

\subsection{Related Work}

As discussed above we are aware of two tools that compute \ucs for LTL: \pltlmup and \procmine; both produce \ucs that are less fine-grained than the ones obtained with our approach.
They also rely on different algorithms to determine the satisfiability of an LTL formula and to subsequently extract a \uc from an unsatisfiable formula.
\pltlmup, which appeared after an early version of this work \cite{VSchuppan-NFM-2013-full-arXiv} was made public, applies an approach to determine minimal \ucs by Huang for SAT \cite{JHuang-ASPDAC-2005} to a BDD-based solver for LTL.
\procmine extracts \ucs as part of a tool set for synthesizing business process templates.
It uses a tableau-based solver to obtain an initial subset of an unsatisfiable set of LTL formulas and then applies deletion-based minimization to that subset.
In principle also tools to determine mutual vacuity for LTL can be used to obtain \ucs for LTL.
However, none of the tools we considered turned out to be suitable for that task in practice.
For details see Sec.~\ref{sec:vacuity}.

In \cite{ACimattiMRoveriVSchuppanSTonetta-CAV-2007} Cimatti et al.~perform extraction of \ucs for PSL to accelerate a PSL satisfiability solver by performing Boolean abstraction. Their notion of \ucs is coarser than ours and their solver is based on BDDs and on SAT.
An investigation of notions of \ucs for LTL including the relation between \ucs and vacuity is performed by the author in \cite{VSchuppan-SCP-2012}. Various notions based on syntax trees, on conjunctive normal forms, on tableaux, and on SAT-based bounded model checking (e.g., \ifnoappendix\cite{ABiere-HandbookOfSatisfiability-2009}\else\cite{ABiereACimattiEClarkeYZhu-TACAS-1999,ABiere-HandbookOfSatisfiability-2009}\fi) are discussed. One of the translations from LTL into a conjunctive normal form is very similar to the one used in this article. It is accompanied by a translation back from the conjunctive normal form to an LTL formula, but it assumes a minimal \uc as its input. In some cases it provides more information than the translation from SNF back to LTL used in this article. For example, it points out that a positive polarity occurrence of an until formula can be replaced with a weak until formula. That feature is left as future work in this article. No implementation or experimental results are reported, and \tr is not considered.
Hantry et al.~suggest a method to extract \ucs for LTL in a tableau-based solver \cite{FHantryMHacid-FLACOS-2011}. No implementation or experiments are reported.
In \cite{FHantryLSaisMHacid-ElectronicColloquiumOnComputationalComplexity-2012} the decision and search problems for minimal \ucs for LTL are shown to be \PSPACE- and \FPSPACE-complete, respectively.
In \cite{ACimattiSMoverSTonetta-FMCAD-2011} Cimatti et al.~show how to prove and explain unfeasibility of message sequence charts for networks of hybrid automata. They consider a different specification language and use an SMT-based (e.g., \cite{CBarrettRSebastianiSSeshiaCTinelli-HandbookOfSatisfiability-2009}) algorithm.

Some work deals with unrealizable rather than unsatisfiable cores. \cite{ACimattiMRoveriVSchuppanATchaltsev-VMCAI-2008} handles specifications in GR(1), which is a proper subset of LTL. K\"{o}nighofer et al.~present methods to help debugging unrealizable specifications by extracting unrealizable cores and simulating counterstrategies \cite{RKoenighoferGHofferekRBloem-FMCAD-2009} as well as performing error localization using model-based diagnosis \cite{RKoenighoferGHofferekRBloem-HVC-2010}. Raman and Kress-Gazit \cite{VRamanHKressGazit-CAV-2011} present a tool that points out unrealizable cores in the context of robot control. \cite{VSchuppan-SCP-2012} explores more fine-grained notions of unrealizable cores than \cite{ACimattiMRoveriVSchuppanATchaltsev-VMCAI-2008,RKoenighoferGHofferekRBloem-FMCAD-2009}.

In vacuity Simmonds et al.~\cite{JSimmondsJDaviesAGurfinkelMChechik-STTT-2010} use SAT-based bounded model checking for vacuity detection.
They only consider $k$-step vacuity, i.e., taking into account bounded model checking runs up to a bound $k$, and leave the problem of removing the bound $k$ open.
They use a similar scheme as described in Sec.~\ref{sec:groupedpropositionsucs} on proofs of unsatisfiability of propositional formulas to determine whether some proposition is vacuous.
They only distinguish whether in a proof of unsatisfiability of a bounded model checking instance any occurrence of a given proposition in the specification interacts with any occurrence of that proposition in the system or not.
Armoni et al.~\cite{RArmoniLFixAFlaisherOGrumbergNPitermanATiemeyerMVardi-CAV-2003} discuss vacuity of an LTL specification in different polarity occurrences of subformulas.
For a more extensive discussion on the relation between vacuity and \ucs for LTL we refer to Sec.~\ref{sec:vacuity} and to \cite{VSchuppan-SCP-2012}.

Many algorithmic considerations are shared between \ucs for LTL and \ucs for other logics.
Here we mention some work that can serve as starting point for studying \ucs for other logics as well as work that contains ideas that might be applicable also to \ucs for LTL.
Extensive research has been performed on \ucs for SAT.
For a brief overview of early algorithms see, e.g., \cite{ANadel-PhDThesis-2009}, pp.~68--70.
Surveys of different aspects are, e.g., \cite{JMarquesSilva-JournalOnMultipleValuedLogicAndSoftComputing-2012, JMarquesSilvaMJanota-arXiv_1402_3011-2014, HKleineBueningOKullmann-HandbookOfSatisfiability-2009}.
Early work on \ucs for SMT was performed by Cimatti et al.~in \cite{ACimattiAGriggioRSebastiani-JournalOfArtificialIntelligenceResearch-2011}.
Pointers to work on modifying proofs in order to obtain smaller proofs and/or different interpolants (for some references see, e.g., \cite{VDSilvaDKroeningMPurandareGWeissenbacher-VMCAI-2010}) in the propositional or SMT cases can be found in \cite{SRolliniRBruttomessoNSharyginaATsitovich-FMSD-2014}.
Some of this work traces which occurrences of literals were resolved with each other (e.g., \cite{HAmjad-AVoCS-2006}).
Belov and Marques-Silva \cite{ABelovJMarquesSilva-SAT-2011} and Shlyakhter \cite{IShlyakhter-PhDThesis-2005} investigate \ucs for circuits, i.e., formulas with potential sharing of subformulas rather than sets of clauses.
\UCs obtained from a SAT solver are also used to support debugging in the declarative modeling language Alloy \cite{IShlyakhterRSeaterDJacksonMSridharanMTaghdiri-ASE-2003,ETorlakFShengHoChangDJackson-FM-2008}.
This requires translation from Alloy into an input for a SAT solver and back \cite{IShlyakhterRSeaterDJacksonMSridharanMTaghdiri-ASE-2003}; while the description in \cite{IShlyakhterRSeaterDJacksonMSridharanMTaghdiri-ASE-2003} is somewhat abstract, some basic ideas regarding the translation from Alloy into the input for the SAT solver and back are similar to ideas used in our translation.
Also in their case a minimal \uc obtained from the SAT solver does not guarantee a minimal \uc in Alloy.
\cite{ETorlakFShengHoChangDJackson-FM-2008} evaluates schemes to increase the efficiency of the \uc extraction method.
In description logics (e.g., \cite{FBaaderDCalvaneseDMcGuinnessDNardiPPatelSchneider-2007}) \ucs are frequently used to support knowledge engineers in performing their tasks.
In particular, understanding, debugging, and possibly repairing entailments in an ontology can be aided by providing the user with subsets of axioms of an ontology that justify a given entailment.
Some work also considers parts of axioms, leading to more fine-grained results.
For an overview see, e.g., \cite{MHorridge-PhDThesis-2011}.

\subsection{Structure of the Article}

We start in Sec.~\ref{sec:examples} by providing examples that illustrate how \ucs are useful for debugging.
In Sec.~\ref{sec:preliminaries} the more formal exposition begins with preliminaries.
In Sec.~\ref{sec:ucextraction} we describe the construction and optimization of a resolution graph and its use to obtain a \uc.
In Sec.~\ref{sec:postprocessingucs} we adapt two methods to post-process the \ucs obtained to make them more useful.
The relation between \ucs for LTL and mutual vacuity is discussed in Sec.~\ref{sec:vacuity}.
We present our implementation and experimental evaluation in Sec.~\ref{sec:experimentalevaluation}.
In Sec.~\ref{sec:conclusions} we draw conclusions.

\ifnoappendix
Due to space constraints a few more involved parts of a proof as well as some more detailed data from our experimental evaluation are omitted. Both are included in the full version \cite{fullversion} of this article, which can be obtained from \url{http://www.schuppan.de/viktor/actainformatica15/} along with implementation, examples, and log files.
\else
Due to space constraints a few more involved parts of a proof as well as some more detailed data from our experimental evaluation are omitted in the main part; these can be found in the appendices. Our implementation, examples, and log files can be obtained from \url{http://www.schuppan.de/viktor/actainformatica15/}.
\fi

\section{Motivating Examples}
\label{sec:examples}

In this section we present examples of using \ucs for LTL to help understand why a specification given in LTL is unsatisfiable.
As we formally introduce LTL only in Sec.~\ref{sec:ltl}, those unfamiliar with LTL are asked to jump ahead to Sec.~\ref{sec:ltl} first.
We start with a toy example and then proceed to a more realistic one.
Except for minor rewriting, all \ucs in this section were obtained with our implementation.

\subsection{Toy Example}

The first example \eqref{ex:demov1:1}--\eqref{ex:demov1:3} is based on \cite{BJobstmannRBloem-FMCAD-2006}.
\eqref{ex:demov1:1} requires a $req$ (request) to be followed by three $gnt$s (grant).
In contrast, \eqref{ex:demov1:2} forbids two subsequent $gnt$s.
\eqref{ex:demov1:3} states that from the time point after a $cancel$ no $gnt$ may be issued until a $go$ is received.
We would like to see whether a $req$ can eventually be issued \eqref{ex:demov1:4}.
\begin{subequations}\label{ex:demov1}
\begin{align}
&{(\fglobally{(\fbimplies{req}{(\fband{(\fnext{gnt})}{\fband{(\fnext{\fnext{gnt}})}{\fnext{\fnext{\fnext{gnt}}}}})})})} \label{ex:demov1:1} \\
\fbandname\; &{(\fglobally{(\fbimplies{gnt}{\fnext{\fbnot{gnt}}})})} \label{ex:demov1:2} \\
\fbandname\; &{(\fglobally{(\fbimplies{cancel}{\fnext{(\funtil{(\fbnot{gnt})}{go})}})})} \label{ex:demov1:3} \\
\fbandname\; &{\ffinally{req}} \label{ex:demov1:4}
\end{align}

Clearly, \eqref{ex:demov1} is unsatisfiable.
If a $req$ were issued in \eqref{ex:demov1:4}, \eqref{ex:demov1:1} would trigger three subsequent $gnt$s.
However, already the second of those $gnt$s would be forbidden by \eqref{ex:demov1:2}.
Note that in this reasoning neither the third $gnt$ in \eqref{ex:demov1:1} nor \eqref{ex:demov1:3} play a role.
Hence, \eqref{ex:demov1} would be unsatisfiable even if $\fnext{\fnext{\fnext{gnt}}}$ and $(\fglobally{(\fbimplies{cancel}{\fnext{(\funtil{(\fbnot{gnt})}{go})}})})$ were ``removed''.
The \uc in \eqref{ex:demov1:c} does just that by replacing these two subformulas with $\true$ and simplifying.
Note that in \eqref{ex:demov1:c} not only a whole top-level conjunct has been removed from \eqref{ex:demov1} but also a proper subformula inside a top-level conjunct.
\end{subequations}
\begin{equation}\label{ex:demov1:c}
\fband{{(\fglobally{(\fbimplies{req}{(\fband{(\fnext{gnt})}{\fnext{\fnext{gnt}}})})})}}{\fband{{(\fglobally{(\fbimplies{gnt}{\fnext{\fbnot{gnt}}})})}}{{\ffinally{req}}}}
\end{equation}

\subsection{Lift Specification}

The second example \eqref{ex:lift} in Fig.~\ref{fig:ex:lift} is adapted from a lift specification in \cite{AHarding-PhDThesis-2005}.
The lift has two floors, indicated by $f_0$ and $f_1$.
On each floor there is a button to call the lift ($b_0$, $b_1$). $sb$ is $\true$ if some button is pressed.
If the lift moves up, then $up$ must be $\true$; if it moves down, then $up$ must be $\false$.
$u$ switches turns between actions by users of the lift ($u$ is $\true$) and actions by the lift ($u$ is $\false$).
For more details we refer to \cite{AHarding-PhDThesis-2005}.
\begin{figure}
\centering
\begin{subequations}\label{ex:lift}
\begin{align}
&\fband{(\fbnot{u})}{\fband{f_0}{\fband{(\fbnot{b_0})}{\fband{(\fbnot{b_1})}{(\fbnot{up})}}}} \\
\fbandname\; &(\fglobally{(\fband{(\fbimplies{u}{\fbnot{\fnext{u}}})}{(\fbimplies{(\fbnot{\fnext{u}})}{u})})}) \\
\fbandname\; &(\fglobally{(\fbimplies{f_0}{\fbnot{f_1}})}) \\
\fbandname\; &(\fglobally{(\fband{(\fbimplies{f_0}{\fnext{(\fbor{f_0}{f_1})}})}{(\fbimplies{f_1}{\fnext{(\fbor{f_0}{f_1})}})})}) \\
\fbandname\; &(\fglobally{(\fbimplies{u}{(\fband{\fband{(\fbimplies{f_0}{\fnext{f_0}})}{(\fbimplies{(\fnext{f_0})}{f_0})}}{\fband{(\fbimplies{f_1}{\fnext{f_1}})}{(\fbimplies{(\fnext{f_1})}{f_1})}})})}) \\
\fbandname\; &(\fglobally{(\fbimplies{(\fbnot{u})}{(\fband{\fband{(\fbimplies{b_0}{\fnext{b_0}})}{(\fbimplies{(\fnext{b_0})}{b_0})}}{\fband{(\fbimplies{b_1}{\fnext{b_1}})}{(\fbimplies{(\fnext{b_1})}{b_1})}})})})\label{ex:lift:6} \\
\fbandname\; &(\fglobally{(\fband{(\fbimplies{(\fband{b_0}{\fbnot{f_0}})}{\fnext{b_0}})}{(\fbimplies{(\fband{b_1}{\fbnot{f_1}})}{\fnext{b_1}})})}) \\
\fbandname\; &(\fglobally{(\fbimplies{(\fband{f_0}{\fnext{f_0}})}{(\fband{(\fbimplies{up}{\fnext{up}})}{(\fbimplies{(\fnext{up})}{up})})})}) \\
\fbandname\; &(\fglobally{(\fbimplies{(\fband{f_1}{\fnext{f_1}})}{(\fband{(\fbimplies{up}{\fnext{up}})}{(\fbimplies{(\fnext{up})}{up})})})}) \\
\fbandname\; &(\fglobally{(\fband{(\fbimplies{(\fband{f_0}{\fnext{f_1}})}{up})}{(\fbimplies{(\fband{f_1}{\fnext{f_0}})}{\fbnot{up}})})}) \\
\fbandname\; &(\fglobally{(\fband{(\fbimplies{sb}{(\fbor{b_0}{b_1})})}{(\fbimplies{(\fbor{b_0}{b_1})}{sb})})}) \\
\fbandname\; &(\fglobally{(\fbimplies{(\fband{f_0}{\fbnot{sb}})}{(\funtil{f_0}{(\freleases{sb}{(\fband{(\ffinally{f_0})}{\fbnot{up}})})})})}) \\
\fbandname\; &(\fglobally{(\fbimplies{(\fband{f_1}{\fbnot{sb}})}{(\funtil{f_1}{(\freleases{sb}{(\fband{(\ffinally{f_0})}{\fbnot{up}})})})})}) \\
\fbandname\; &(\fglobally{(\fband{(\fbimplies{b_0}{\ffinally{f_0}})}{(\fbimplies{b_1}{\ffinally{f_1}})})})
\end{align}
\end{subequations}
\begin{minipage}{0.24\linewidth}\caption{\label{fig:ex:lift} A lift specification}\end{minipage}
\end{figure}

We first assume that an engineer is interested in seeing whether it is possible that $b_1$ is pressed at time point 0 \eqref{ex:lift:s1}.
As the \uc \eqref{ex:lift:c1} shows, this is impossible because $b_1$ must be $\false$ at the beginning.
\eqref{ex:lift:c1} was obtained by conjoining \eqref{ex:lift} with \eqref{ex:lift:s1}, determining unsatisfiability, replacing all top-level conjuncts except for $\fbnot{b_1}$ and $b_1$ with $\true$, and simplifying.
\begin{equation}\label{ex:lift:s1}
b_1
\end{equation}
\begin{equation}\label{ex:lift:c1}
\fband{(\fbnot{b_1})}{b_1}
\end{equation}

Now the engineer modifies her query such that $b_1$ is pressed at time point 1 \eqref{ex:lift:s2}.
As shown by the \uc in \eqref{ex:lift:c2} that turns out to be impossible, too.
Note that while in the previous scenario \eqref{ex:lift:s1} unsatisfiability was not too hard to see even without \uc extraction, in the current scenario \eqref{ex:lift:s2} \uc extraction already is quite helpful in localizing which parts of \eqref{ex:lift} and \eqref{ex:lift:s2} are responsible for unsatisfiability.
The \uc in \eqref{ex:lift:c2} was obtained by conjoining \eqref{ex:lift} with \eqref{ex:lift:s2}, determining unsatisfiability, replacing all top-level conjuncts except for $\fbnot{u}$, $\fbnot{b_1}$, \eqref{ex:lift:6}, and $\fnext{b_1}$ with $\true$, replacing $\fbimplies{b_0}{\fnext{b_0}}$, $\fbimplies{(\fnext{b_0})}{b_0}$, $\fbimplies{b_1}{\fnext{b_1}}$ in \eqref{ex:lift:6} with $\true$, and simplifying.
As in the toy example in the previous subsection not only top-level conjuncts but also proper subformulas inside a top-level conjunct have been simplified.
\begin{equation}\label{ex:lift:s2}
\fnext{b_1}
\end{equation}
\begin{equation}\label{ex:lift:c2}
\fband{(\fbnot{u})}{\fband{(\fbnot{b_1})}{\fband{(\fglobally{(\fbimplies{(\fbnot{u})}{(\fbimplies{(\fnext{b_1})}{b_1})})})}{\fnext{b_1}}}}
\end{equation}

The engineer now tries to have $b_1$ pressed at time point 2 and, again, obtains a \uc.
She becomes suspicious and checks whether $b_1$ can be pressed at all \eqref{ex:lift:s3}.
The unsatisfiability of the conjunction of \eqref{ex:lift} and \eqref{ex:lift:s3} tell her that $b_1$ cannot be pressed at all and, therefore, this specification of a lift must contain a bug.
She can now use the \uc in \eqref{ex:lift:c3} to track down the problem.
This example clearly illustrates the use of \ucs for debugging as \eqref{ex:lift:c3:1}--\eqref{ex:lift:c3:6} is significantly smaller than \eqref{ex:lift}.
\begin{equation}\label{ex:lift:s3}
\ffinally{b_1}
\end{equation}
\begin{samepage}
\begin{subequations}\label{ex:lift:c3}
\begin{align}
&\fband{f_0}{\fband{(\fbnot{b_1})}{(\fbnot{up})}}\label{ex:lift:c3:1} \\
\fbandname\; &(\fglobally{(\fbimplies{f_0}{\fbnot{f_1}})}) \\
\fbandname\; &(\fglobally{(\fbimplies{f_0}{\fnext{(\fbor{f_0}{f_1})}})}) \\
\fbandname\; &(\fglobally{(\fbimplies{(\fband{f_0}{\fnext{f_0}})}{(\fbimplies{(\fnext{up})}{up})})}) \\
\fbandname\; &(\fglobally{(\fbimplies{(\fband{f_0}{\fnext{f_1}})}{up})})\\
\fbandname\; &(\fglobally{(\fbimplies{b_1}{\ffinally{f_1}})})\label{ex:lift:c3:6} \\
\fbandname\; &\ffinally{b_1}
\end{align}
\end{subequations}
\end{samepage}

\section{Preliminaries}
\label{sec:preliminaries}

In this section we first present LTL (Sec.~\ref{sec:ltl}).
We continue with SNF, the clausal normal form for LTL that the temporal resolution algorithm works on, (Sec.~\ref{sec:separatednormalform}) and with a translation from LTL into SNF (Sec.~\ref{sec:translatingltlintosnf}).
Finally, in Sec.~\ref{sec:tr} we discuss the temporal resolution algorithm that our method for \uc extraction in Sec.~\ref{sec:ucextraction} is based on.

\subsection{LTL}
\label{sec:ltl}

We use a standard version of LTL, see, e.g., \cite{EEmerson-HandbookOfTheoreticalComputerScience-1990}. Let $\allbools$ be the set of Booleans, and let $\allaps$ be a finite set of atomic propositions. The set of \emph{LTL formulas} is constructed inductively as follows. The Boolean constants $\false$, $\true$ $\in \allbools$ and any atomic proposition $\ap \in \allaps$ are LTL formulas. If $\prt$, $\prtp$ are LTL formulas, so are $\fbnot{\prt}$ (not), $\fbor{\prt}{\prtp}$ (or), $\fband{\prt}{\prtp}$ (and), $\fnext{\prt}$ (next time), $\funtil{\prt}{\prtp}$ (until), $\freleases{\prt}{\prtp}$ (releases), $\ffinally{\prt}$ (finally), and $\fglobally{\prt}$ (globally). We use $\fbimplies{\prt}{\prtp}$ (implies) as an abbreviation for $\fbor{(\fbnot{\prt})}{\prtp}$.
An occurrence of a subformula $\prt$ of an LTL formula $\inp$ has \emph{positive polarity} ($\positivepolarity$) if it appears under an even number of negations in $\inp$ and \emph{negative polarity} ($\negativepolarity$) otherwise.
The \emph{size} of an LTL formula $\inp$ is measured as the sum of the numbers of occurrences of atomic propositions, Boolean operators, and temporal operators in $\inp$.

LTL is interpreted over words in $(\mkpowerset{\allaps})^\omega$.
For the semantics of LTL see Fig.~\ref{fig:ltlsemantics}. 
A word $\iword \in (\mkpowerset{\allaps})^\omega$ \emph{satisfies} an LTL formula $\inp$ iff $\mkpair{\iword}{0} \models \inp$.
A word $\iword$ that satisfies $\inp$ is also called a \emph{satisfying assignment} for $\inp$.
An LTL formula $\inp$ is \emph{satisfiable} if there exists a word $\iword \in (\mkpowerset{\allaps})^\omega$ that satisfies $\inp$; otherwise, it is \emph{unsatisfiable}.
The problem of determining the satisfiability of an LTL formula is \PSPACE-complete \cite{ASistlaEClarke-JACM-1985,JHalpernJReif-TheoreticalComputerScience-1983}, even if the set of atomic propositions $\allaps$ only contains one element \cite{SDemriPSchnoebelen-InformationAndComputation-2002}.
\begin{figure}
\centering
{
\begin{tabular}{lcl}
$\mkpair{\iword}{\pos} \models \true$ & &
\\
$\mkpair{\iword}{\pos} \not \models \false$ & &
\\
$\mkpair{\iword}{\pos} \models \ap$ & $\proofbiimplies$ & $\ap \in \fiwordgetpos{\iword}{\pos}$
\\
$\mkpair{\iword}{\pos} \models \fbnot{\prt}$ & $\proofbiimplies$ & $\mkpair{\iword}{\pos} \not \models \prt$
\\
$\mkpair{\iword}{\pos} \models \fbor{\prt}{\prtp}$ & $\proofbiimplies$ & $\mkpair{\iword}{\pos} \models \prt \mbox{ or } \mkpair{\iword}{\pos} \models \prtp$
\\
$\mkpair{\iword}{\pos} \models \fband{\prt}{\prtp}$ & $\proofbiimplies$ & $\mkpair{\iword}{\pos} \models \prt \mbox{ and } \mkpair{\iword}{\pos} \models \prtp$
\\
$\mkpair{\iword}{\pos} \models \fnext{\prt}$ & $\proofbiimplies$ & $\mkpair{\iword}{\pos + 1} \models \prt$
\\
$\mkpair{\iword}{\pos} \models \funtil{\prt}{\prtp}$ & $\proofbiimplies$ & $\exists \posp \ge \pos \;.\; (\fband{\mkpair{\iword}{\posp} \models \prtp}{\forall \pos \le \pospp < \posp \;.\; \mkpair{\iword}{\pospp} \models \prt})$
\\
$\mkpair{\iword}{\pos} \models \freleases{\prt}{\prtp}$ & $\proofbiimplies$ & $\forall \posp \ge \pos \;.\; (\fbor{\mkpair{\iword}{\posp} \models \prtp}{\exists \pos \le \pospp < \posp \;.\; \mkpair{\iword}{\pospp} \models \prt})$
\\
$\mkpair{\iword}{\pos} \models \ffinally{\prt}$ & $\proofbiimplies$ & $\exists \posp \ge \pos \;.\; \mkpair{\iword}{\posp} \models \prt$
\\
$\mkpair{\iword}{\pos} \models \fglobally{\prt}$ & $\proofbiimplies$ & $\forall \posp \ge \pos \;.\; \mkpair{\iword}{\posp} \models \prt$
\\
\end{tabular}
}
\begin{minipage}{0.66\linewidth}\caption{\label{fig:ltlsemantics} Semantics of LTL. $\iword$ is a word in $(\mkpowerset{\allaps})^\omega$, $\pos$ is a time point in $\allnats$.}\end{minipage}
\end{figure}

\subsection{Separated Normal Form}
\label{sec:separatednormalform}

Temporal resolution works on formulas in a clausal normal form called Separated Normal Form (SNF) \cite{MFisher-IJCAI-1991,MFisherPNoel-UManchesterTR-1992,MFisherCDixonMPeim-ACMTrComputationalLogic-2001}.
For any atomic proposition $\ap \in \allaps$ $\ap$ and $\fbnot{\ap}$ are \emph{literals}.
Let $\ap_1, \ldots, \ap_\maxind$, $\app_1, \ldots, \app_\maxindp$, $\apres$ with $0 \le \maxind, \maxindp$ be literals such that $\forall 1 \le \pos < \posp \le \maxind \;.\; \ap_\pos \ne \ap_\posp$ and $\forall 1 \le \pos < \posp \le \maxindp \;.\; \app_\pos \ne \app_\posp$. Then
\begin{inparaenum}[(i)]
\item $\ficlause{\fbor{\ap_1}{\fbor{\ldots}{\ap_\maxind}}}$ is an \emph{initial clause};
\item $(\fgloballyname(\fbor{\ap_1}{\fbor{\ldots}{\ap_\maxind}} \fborname \fnext{(\fbor{\app_1}{\fbor{\ldots}{\app_\maxindp}})}))$ is a \emph{global clause}; and
\item $\fgneclause{\fbor{\ap_1}{\fbor{\ldots}{\ap_\maxind}}}{\apres}$ is an \emph{eventuality clause}.
\end{inparaenum}
$\apres$ is called an \emph{eventuality literal}.
As usual an empty disjunction (resp.~conjunction) stands for $\false$ (resp.~\true).
${\ficlause{}}$ or ${\fgnclause{}}$, denoted $\emptyclause$, stand for $\false$ or $\fglobally{(\false)}$ and are called \emph{empty clause}.
The set of all SNF clauses is denoted $\allclauses$.
Let $\clause_1, \ldots, \clause_\maxind$ with $0 \le \maxind$ be SNF clauses. Then $\bigwedge_{1 \le \pos \le \maxind} \clause_\pos$ is an LTL formula in \emph{SNF}.
Every LTL formula $\inp$ can be transformed into an equisatisfiable formula $\inpp$ in SNF \cite{MFisherCDixonMPeim-ACMTrComputationalLogic-2001}.

\subsection{Translating LTL into SNF}
\label{sec:translatingltlintosnf}

We use a structure-preserving translation (e.g., \cite{DPlaistedSGreenbaum-JSymbComput-1986}) to translate an LTL formula into a set of SNF clauses.
Our translation is based on the tableau construction for LTL that is often used in (symbolic) model checking (see, e.g., \cite{OLichtensteinAPnueli-POPL-1985,JBurchEClarkeKMcMillanDDillLHwang-IaC-1992,EClarkeOGrumbergKHamaguchi-FMSD-1997,YKestenAPnueliLRaviv-ICALP-1998,ABiereKHeljankoTJunttilaTLatvalaVSchuppan-LMCS-2006}) rather than on \cite{MFisherCDixonMPeim-ACMTrComputationalLogic-2001} as we find the former to be more straightforward.

\mydefinition{Translation from LTL into SNF}\label{def:ltltosnf}
Let $\inp$ be an LTL formula over atomic propositions $\allaps$, and let $\alldCNFvars = \{\dCNFvar, \dCNFvarp, \ldots\}$ be a set of fresh atomic propositions that don't occur in $\inp$. Assign to each occurrence of a subformula $\prt$ in $\inp$ a Boolean value or a proposition according to column 2 of Tab.~\ref{tab:ltltosnf}, which is used to reference $\prt$ in the SNF clauses for its superformula. Moreover, assign to each occurrence of $\prt$ a set of SNF clauses according to columns 3 and 4 of Tab.~\ref{tab:ltltosnf}. Let $\fdCNFaux{\inp}$ be the set of all SNF clauses obtained from $\inp$ that way. Then the \emph{SNF of $\inp$} is defined as $\fdCNF{\inp} \definedas \fband{\fdCNFvar{\inp}}{\bigwedge_{\dCNFconj \in \fdCNFaux{\inp}} \dCNFconj}$.

\begin{table}
\centering
\begin{minipage}{0.37\linewidth}\caption{\label{tab:ltltosnf}Translation from LTL into SNF}\end{minipage}
{
\begin{tabular}{|l|l|l|l|}
\hline
\hline
Subformula & Proposition & Polarity & SNF Clauses \\
\hline
\hline
$\true$/$\false$/$\ap$ & $\true$/$\false$/$\ap$ & \positivepolarity/\negativepolarity & \mbox{none} \\
\hline
$\fbnot{\prt}$ & $\fdCNFvar{\fbnot{\prt}}$ & \positivepolarity & $\fgnclause{\fbimplies{\fdCNFvar{\fbnot{\prt}}}{\fbnot{\fdCNFvarm{\prt}}}}$ \\
\cline{3-4}
& & \negativepolarity & $\fgnclause{\fbimplies{(\fbnot{\fdCNFvar{\fbnot{\prt}}})}{\fdCNFvarm{\prt}}}$ \\
\hline
$\fbor{\prt}{\prtp}$ & $\fdCNFvar{\fbor{\prt}{\prtp}}$ & \positivepolarity & $\fgnclause{\fbimplies{\fdCNFvar{\fbor{\prt}{\prtp}}}{(\fbor{\fdCNFvarm{\prt}}{\fdCNFvarm{\prtp}})}}$ \\
\cline{3-4}
& & \negativepolarity & $\fgnclause{\fbimplies{(\fbnot{\fdCNFvar{\fbor{\prt}{\prtp}}})}{\fbnot{\fdCNFvarm{\prt}}}}$, \\
& & & $\fgnclause{\fbimplies{(\fbnot{\fdCNFvar{\fbor{\prt}{\prtp}}})}{\fbnot{\fdCNFvarm{\prtp}}}}$ \\
\hline
$\fband{\prt}{\prtp}$  & $\fdCNFvar{\fband{\prt}{\prtp}}$ & \positivepolarity & $\fgnclause{\fbimplies{\fdCNFvar{\fband{\prt}{\prtp}}}{\fdCNFvarm{\prt}}}$, \\
& & & $\fgnclause{\fbimplies{\fdCNFvar{\fband{\prt}{\prtp}}}{\fdCNFvarm{\prtp}}}$ \\
\cline{3-4}
& & \negativepolarity & $(\fglobally{(\fbimplies{(\fbnot{\fdCNFvar{\fband{\prt}{\prtp}}})}{(\fbor{(\fbnot{\fdCNFvarm{\prt}})}{\fbnot{\fdCNFvarm{\prtp}}})})})$ \\
\hline
$\fnext{\prt}$  & $\fdCNFvar{\fnext{\prt}}$ & \positivepolarity & $\fgnclause{\fbimplies{\fdCNFvar{\fnext{\prt}}}{\fnext{\fdCNFvarm{\prt}}}}$ \\
\cline{3-4}
& & \negativepolarity & $\fgnclause{\fbimplies{(\fbnot{\fdCNFvar{\fnext{\prt}}})}{\fnext{\fbnot{\fdCNFvarm{\prt}}}}}$ \\
\hline
$\funtil{\prt}{\prtp}$ & $\fdCNFvar{\funtil{\prt}{\prtp}}$ & \positivepolarity & $\fgnclause{\fbimplies{\fdCNFvar{\funtil{\prt}{\prtp}}}{(\fbor{\fdCNFvarm{\prtp}}{\fdCNFvarm{\prt}})}}$,\\
& & & $(\fglobally{(\fbimplies{\fdCNFvar{\funtil{\prt}{\prtp}}}{(\fbor{\fdCNFvarm{\prtp}}{\fnext{\fdCNFvar{\funtil{\prt}{\prtp}}}})})})$,\\
& & & $\fgnclause{\fbimplies{\fdCNFvar{\funtil{\prt}{\prtp}}}{\ffinally{\fdCNFvarm{\prtp}}}}$\\
\cline{3-4}
& & \negativepolarity & $\fgnclause{\fbimplies{(\fbnot{\fdCNFvar{\funtil{\prt}{\prtp}}})}{\fbnot{\fdCNFvarm{\prtp}}}}$,\\
& & & $(\fglobally{(\fbimplies{(\fbnot{\fdCNFvar{\funtil{\prt}{\prtp}}})}{(\fbor{(\fbnot{\fdCNFvarm{\prt}})}{\fnext{\fbnot{\fdCNFvar{\funtil{\prt}{\prtp}}}}})})})$\\
\hline
$\freleases{\prt}{\prtp}$ & $\fdCNFvar{\freleases{\prt}{\prtp}}$ & \positivepolarity & $\fgnclause{\fbimplies{\fdCNFvar{\freleases{\prt}{\prtp}}}{\fdCNFvarm{\prtp}}}$,\\
& & & $(\fglobally{(\fbimplies{\fdCNFvar{\freleases{\prt}{\prtp}}}{(\fbor{\fdCNFvarm{\prt}}{\fnext{\fdCNFvar{\freleases{\prt}{\prtp}}}})})})$\\
\cline{3-4}
& & \negativepolarity & $(\fglobally{(\fbimplies{(\fbnot{\fdCNFvar{\freleases{\prt}{\prtp}}})}{(\fbor{(\fbnot{\fdCNFvarm{\prtp}})}{\fbnot{\fdCNFvarm{\prt}}})})})$,\\
& & & $(\fglobally{(\fbimplies{(\fbnot{\fdCNFvar{\freleases{\prt}{\prtp}}})}{(\fbor{(\fbnot{\fdCNFvarm{\prtp}})}{\fnext{\fbnot{\fdCNFvar{\freleases{\prt}{\prtp}}}}})})})$,\\
& & & $\fgnclause{\fbimplies{(\fbnot{\fdCNFvar{\freleases{\prt}{\prtp}}})}{\ffinally{\fbnot{\fdCNFvarm{\prtp}}}}}$\\
\hline
$\ffinally{\prt}$ & $\fdCNFvar{\ffinally{\prt}}$ & \positivepolarity & $\fgnclause{\fbimplies{\fdCNFvar{\ffinally{\prt}}}{\ffinally{\fdCNFvarm{\prt}}}}$ \\
\cline{3-4}
& & \negativepolarity & $\fgnclause{\fbimplies{(\fbnot{\fdCNFvar{\ffinally{\prt}}})}{\fnext{\fbnot{\fdCNFvar{\ffinally{\prt}}}}}}$,\\
& & & $\fgnclause{\fbimplies{(\fbnot{\fdCNFvar{\ffinally{\prt}}})}{\fbnot{\fdCNFvarm{\prt}}}}$ \\
\hline
$\fglobally{\prt}$  & $\fdCNFvar{\fglobally{\prt}}$ & \positivepolarity & $\fgnclause{\fbimplies{\fdCNFvar{\fglobally{\prt}}}{\fnext{\fdCNFvar{\fglobally{\prt}}}}}$,\\
& & & $\fgnclause{\fbimplies{\fdCNFvar{\fglobally{\prt}}}{\fdCNFvarm{\prt}}}$\\
\cline{3-4}
& & \negativepolarity & $\fgnclause{\fbimplies{(\fbnot{\fdCNFvar{\fglobally{\prt}}})}{\ffinally{\fbnot{\fdCNFvarm{\prt}}}}}$ \\
\hline
\hline
\end{tabular}
}
\end{table}
\end{definition}

Note that to make the SNF clauses in column 4 of Tab.~\ref{tab:ltltosnf} and elsewhere in this article easier to understand we often use implication to formulate them.
However, in \trp SNF clauses cannot contain implications and, therefore, in our implementation of Def.~\ref{def:ltltosnf} implication is expanded using its definition.
The fact that some propositions are marked {\color{blue}\setlength\fboxsep{1pt}\fbox{blue boxed}} in Tab.~\ref{tab:ltltosnf} will be used later in Sec.~\ref{sec:extractingaucinltl} when translating a \uc back from SNF to LTL.
It is well known that $\inp$ and $\fdCNF{\inp}$ are equisatisfiable and that a satisfying assignment for $\inp$ (resp.~$\fdCNF{\inp}$) can be extended (resp.~restricted) to a satisfying assignment for $\fdCNF{\inp}$ (resp.~$\inp$).
Below we sometimes identify the SNF of $\inp$, $\fdCNF{\inp}$, with the set of SNF clauses $\{\fdCNFvar{\inp}\} \cup \fdCNFaux{\inp}$ that $\fdCNF{\inp}$ is constructed from.

\myremark%
{\thmltltosnfcomplexity}%
{thm:ltltosnfcomplexity}%
{Complexity Considerations Regarding the Translation from LTL into SNF}%
{Let $\inp$ be an LTL formula over atomic propositions $\allaps$, and let $\fdCNF{\inp}$ be the SNF of $\inp$. It is easy to see that
\begin{inparaenum}[(i)]
\item the number of clauses in $\fdCNF{\inp}$ is linear in the size of $\inp$,
\item for each occurrence of a Boolean or temporal operator in $\inp$ one fresh atomic proposition $\dCNFvar, \dCNFvarp, \ldots$ is introduced in $\fdCNF{\inp}$ by the translation, and
\item the size of $\fdCNF{\inp}$ is linear in the size of $\inp$.
\end{inparaenum}
}%
\thmltltosnfcomplexity{true}

As an example we translate the formula $\inp$ shown in \eqref{ex:ltltosnf} into SNF.
\begin{equation}\label{ex:ltltosnf}
\inp \definedas \fband{(\fglobally{(\fband{p}{q})})}{\ffinally{\fbnot{p}}}
\end{equation}
The SNF of $\inp$, $\fdCNF{\inp}$, is given in \eqref{ex:ltltosnf:snf}.
$\fdCNFvar{\fband{(\fglobally{(\fband{p}{q})})}{\ffinally{\fbnot{p}}}}$ represents $\fband{(\fglobally{(\fband{p}{q})})}{\ffinally{\fbnot{p}}}$ in the sense that if in a satisfying assignment for $\fdCNF{\inp}$ $\fdCNFvar{\fband{(\fglobally{(\fband{p}{q})})}{\ffinally{\fbnot{p}}}}$ is $\true$ at time point $\pos$, then $\fband{(\fglobally{(\fband{p}{q})})}{\ffinally{\fbnot{p}}}$ is $\true$ at time point $\pos$ on that assignment.
The same holds for the other fresh atomic propositions $\fdCNFvar{\fglobally{(\fband{p}{q})}}$, $\fdCNFvar{\fband{p}{q}}$, $\fdCNFvar{\ffinally{\fbnot{p}}}$, and $\fdCNFvar{\fbnot{p}}$ in $\fdCNF{\inp}$.\footnote{Note that in our example fresh atomic propositions $\dCNFvar, \dCNFvarp, \ldots$ in $\fdCNF{\inp}$ only represent positive polarity occurrences of subformulas. If there were a negative polarity occurrence of some subformula $\prt$ represented by some fresh atomic proposition $\fdCNFvar{\prt}$, then $\fdCNFvar{\prt}$ being $\false$ at time point $\pos$ in a satisfying assignment for $\inp$ would imply $\prt$ being $\false$ at time point $\pos$ on that assignment.}
Together $(\fdCNFvar{\fband{(\fglobally{(\fband{p}{q})})}{\ffinally{\fbnot{p}}}})$ and $(\fglobally{(\fbimplies{\fdCNFvar{\fband{(\fglobally{(\fband{p}{q})})}{\ffinally{\fbnot{p}}}}}{\fdCNFvar{\fglobally{(\fband{p}{q})}}})})$ force $\fdCNFvar{\fglobally{(\fband{p}{q})}}$ to be $\true$ at time point 0.
Similarly, $\fdCNFvar{\ffinally{\fbnot{p}}}$ is forced to be $\true$ at time point 0 via $(\fglobally{(\fbimplies{\fdCNFvar{\fband{(\fglobally{(\fband{p}{q})})}{\ffinally{\fbnot{p}}}}}{\fdCNFvar{\ffinally{\fbnot{p}}}})})$.
$\fglobally{(\fband{p}{q})}$ is translated into four clauses $(\fglobally{(\fbimplies{\fdCNFvar{\fglobally{(\fband{p}{q})}}}{\fnext{\fdCNFvar{\fglobally{(\fband{p}{q})}}}})})$, $(\fglobally{(\fbimplies{\fdCNFvar{\fglobally{(\fband{p}{q})}}}{\fdCNFvar{\fband{p}{q}}})})$, $(\fglobally{(\fbimplies{\fdCNFvar{\fband{p}{q}}}{p})})$, and $(\fglobally{(\fbimplies{\fdCNFvar{\fband{p}{q}}}{q})})$.
With $\fdCNFvar{\fglobally{(\fband{p}{q})}}$ being $\true$ at time point 0 the first of these four clauses makes $\fdCNFvar{\fglobally{(\fband{p}{q})}}$ $\true$ at all time points.
The remaining three clauses then make $p$ and $q$ $\true$ continuously.
Using the truth of $\fdCNFvar{\ffinally{\fbnot{p}}}$ at time point 0 the two last clauses $(\fglobally{(\fbimplies{\fdCNFvar{\ffinally{\fbnot{p}}}}{\ffinally{\fdCNFvar{\fbnot{p}}}})})$ and $(\fglobally{(\fbimplies{\fdCNFvar{\fbnot{p}}}{\fbnot{p}})})$ ensure that $\fbnot{p}$ becomes $\true$ eventually.
Last but not least notice that --- as \eqref{ex:ltltosnf} is obviously unsatisfiable --- there can be no such satisfying assignment for $\fdCNF{\inp}$.
\begin{equation}\label{ex:ltltosnf:snf}
\fdCNF{\inp} =
\begin{array}{l}
\hspace{-0.5em}\{(\fdCNFvar{\fband{(\fglobally{(\fband{p}{q})})}{\ffinally{\fbnot{p}}}}),\\
(\fglobally{(\fbimplies{\fdCNFvar{\fband{(\fglobally{(\fband{p}{q})})}{\ffinally{\fbnot{p}}}}}{\fdCNFvar{\fglobally{(\fband{p}{q})}}})}),\\
(\fglobally{(\fbimplies{\fdCNFvar{\fglobally{(\fband{p}{q})}}}{\fnext{\fdCNFvar{\fglobally{(\fband{p}{q})}}}})}),\\
(\fglobally{(\fbimplies{\fdCNFvar{\fglobally{(\fband{p}{q})}}}{\fdCNFvar{\fband{p}{q}}})}),\\
(\fglobally{(\fbimplies{\fdCNFvar{\fband{p}{q}}}{p})}),\\
(\fglobally{(\fbimplies{\fdCNFvar{\fband{p}{q}}}{q})}),\\
(\fglobally{(\fbimplies{\fdCNFvar{\fband{(\fglobally{(\fband{p}{q})})}{\ffinally{\fbnot{p}}}}}{\fdCNFvar{\ffinally{\fbnot{p}}}})}),\\
(\fglobally{(\fbimplies{\fdCNFvar{\ffinally{\fbnot{p}}}}{\ffinally{\fdCNFvar{\fbnot{p}}}})}),\\
(\fglobally{(\fbimplies{\fdCNFvar{\fbnot{p}}}{\fbnot{p}})})\}.
\end{array}
\end{equation}

\subsection{Temporal Resolution in \trp}
\label{sec:tr}

In this subsection we describe temporal resolution (\tr) \cite{MFisherCDixonMPeim-ACMTrComputationalLogic-2001} as implemented in \trp \cite{UHustadtBKonev-CADE-2003,UHustadtBKonev-CollegiumLogicum-2004}.
We provide a concise description of \tr in \trp as required for the purposes of this article (Sec.~\ref{sec:thetemporalresolutionalgorithmintrp}), followed by an example (Sec.~\ref{sec:tralgorithmexample}).
Temporal resolution has been developed since the early 1990s \cite{MFisher-IJCAI-1991}, and an extensive body of literature exists.
It is out of the scope of this article to provide a detailed introduction or a tutorial on the subject.
The following references are among the most suitable as an introduction to \tr as needed for this article: \cite{MFisherCDixonMPeim-ACMTrComputationalLogic-2001} provides a general overview of the method and is a good starting point, \cite{CDixon-ICTL-1997,CDixon-AnnalsOfMathematicsAndArtificialIntelligence-1998} explains BFS loop search as used in \trp, and \cite{UHustadtBKonev-CollegiumLogicum-2004} covers the implementation of \tr in \trp.
In \cite{VSchuppan-NFM-2013-full-arXiv} we provide some intuition on temporal resolution with a slant towards BDD-based symbolic model checking (e.g., \ifnoappendix\cite{EClarkeOGrumbergDPeled-2001}\else\cite{JBurchEClarkeKMcMillanDDillLHwang-IaC-1992,EClarkeOGrumbergDPeled-2001}\fi).

\subsubsection{The Temporal Resolution Algorithm in \trp}
\label{sec:thetemporalresolutionalgorithmintrp}

The production rules of \trp are shown in Tab.~\ref{tab:productionrules}. The first column assigns a name to a production rule. The second and fourth columns list the premises. The sixth column gives the conclusion. Columns 3, 5, and 7 are described below.

\begin{table}
\centering
\caption{\label{tab:productionrules}Production rules used in \trp. Let $\dap \definedas \fbor{\ap_1}{\fbor{\ldots}{\ap_\maxind}}$, $\dapp \definedas \fbor{\app_1}{\fbor{\ldots}{\app_\maxindp}}$, $\dappp \definedas \fbor{\appp_1}{\fbor{\ldots}{\appp_\maxindpp}}$, and $\dapppp \definedas \fbor{\apppp_1}{\fbor{\ldots}{\apppp_\maxindppp}}$.}
{
\begin{tabular}{||@{\hspace{0.29em}}l@{\hspace{0.29em}}||@{\hspace{0.29em}}l@{\hspace{0.29em}}|@{\hspace{0.29em}}c@{\hspace{0.29em}}|@{\hspace{0.29em}}l@{\hspace{0.29em}}|@{\hspace{0.29em}}c@{\hspace{0.29em}}|@{\hspace{0.29em}}l@{\hspace{0.29em}}|@{\hspace{0.29em}}c@{\hspace{0.29em}}||}
\hline
\hline
rule &
premise 1 &
part. &
premise 2 &
part. &
conclusion &
part.
\\
\hline
\hline
\multicolumn{7}{||c||}{saturation}
\\
\hline
init-ii &
\ficlause{\fbor{\dap}{\apres}} &
$\mainpartition$ &
\ficlause{\fbor{(\fbnot{\apres})}{\dapp}} &
$\mainpartition$ &
\ficlause{\fbor{\dap}{\dapp}} &
$\mainpartition$
\\
\hline
init-in &
\ficlause{\fbor{\dap}{\apres}} &
$\mainpartition$ &
\fgnclause{\fbor{(\fbnot{\apres})}{\dapp}} &
$\mainpartition$ &
\ficlause{\fbor{\dap}{\dapp}} &
$\mainpartition$
\\
\hline
step-nn &
\fgnclause{\fbor{\dap}{\apres}} &
$\mainpartition$ &
\fgnclause{\fbor{(\fbnot{\apres})}{\dapp}} &
$\mainpartition$ &
\fgnclause{\fbor{\dap}{\dapp}} &
$\mainpartition$
\\
\hline
step-nx &
\fgnclause{\fbor{\dap}{\apres}} &
$\mainpartition$ &
\fgnxclause{\dapp}{\fbor{(\fbnot{\apres})}{\dappp}} &
$\mainpartition$ &
\fgnxclause{\dapp}{\fbor{\dap}{\dappp}} &
$\mainpartition$
\\
\hline
step-xx &
\fgnxclause{\dap}{\fbor{\dapp}{\apres}} &
$\mainorlooppartition$ &
\fgnxclause{\dappp}{\fbor{(\fbnot{\apres})}{\dapppp}} &
$\mainorlooppartition$ &
\fgnxclause{\fbor{\dap}{\dappp}}{\fbor{\dapp}{\dapppp}} &
$\mainorlooppartition$
\\
\hline
\hline
\multicolumn{7}{||c||}{augmentation}
\\
\hline
aug1 &
\multicolumn{3}{@{\hspace{0.29em}}l@{\hspace{0.29em}}|@{\hspace{0.29em}}}{\fgneclause{\dap}{\apres}} &
$\mainpartition$ &
\fgnclause{\fbor{\dap}{\fbor{\apres}{\fapwaitfor{\apres}}}} &
$\mainpartition$
\\
\hline
aug2 &
\multicolumn{3}{@{\hspace{0.29em}}l@{\hspace{0.29em}}|@{\hspace{0.29em}}}{\fgneclause{\dap}{\apres}} &
$\mainpartition$ &
\fgnxclause{(\fbnot{\fapwaitfor{\apres}})}{\fbor{\apres}{\fapwaitfor{\apres}}} &
$\mainpartition$
\\
\hline
\hline
\multicolumn{7}{||c||}{BFS loop search}
\\
\hline
BFS-loop-it-init-x &
\multicolumn{3}{@{\hspace{0.29em}}l@{\hspace{0.29em}}|@{\hspace{0.29em}}}{$\clause \definedas \fgnxclause{\dap}{\fbor{\app_1}{\fbor{\ldots}{\app_\maxindp}}}$ with $\maxindp > 0$} &
$\mainpartition$ &
\clause &
$\looppartition$
\\
\hline
BFS-loop-it-init-n &
\multicolumn{3}{@{\hspace{0.29em}}l@{\hspace{0.29em}}|@{\hspace{0.29em}}}{(\fglobally{\;\dap})} &
$\mainpartition$ &
(\fglobally{\,\fnext{\,\dap}}) &
$\looppartition$
\\
\hline
BFS-loop-it-init-c &
(\fglobally{\;\dap}) &
$\looppartitionp$ &
\fgneclause{\dapp}{\apres} &
$\mainpartition$ &
\fgxclause{\fbor{\dap}{\apres}} &
$\looppartition$
\\
\hline
BFS-loop-it-sub &
\multicolumn{3}{@{\hspace{0.29em}}l@{\hspace{0.29em}}|@{\hspace{0.29em}}}{\clause \definedas (\fglobally{\;\dap}) \mbox{ with } \fbimplies{\clause}{(\fglobally{\;\dapp})}} &
$\looppartition$ &
\begin{minipage}{0.22\linewidth}
\fgxclause{\fbor{\dapp}{\apres}} generated by BFS-loop-it-init-c
\end{minipage}
&
$\looppartition$
\\
\hline
&&&&&&\\[-2.25ex]
\begin{minipage}{0.17\linewidth}
BFS-loop-conclusion1
\end{minipage}
&
(\fglobally{\;\dap}) &
$\looppartition$ &
\fgneclause{\dapp}{\apres} &
$\mainpartition$ &
\fgnclause{\fbor{\dap}{\fbor{\dapp}{\apres}}} &
$\mainpartition$
\\[1.25ex]
\hline
&&&&&&\\[-2.25ex]
\begin{minipage}{0.17\linewidth}
BFS-loop-conclusion2
\end{minipage}
&
(\fglobally{\;\dap}) &
$\looppartition$ &
\fgneclause{\dapp}{\apres} &
$\mainpartition$ &
\fgnxclause{(\fbnot{\fapwaitfor{\apres}})}{\fbor{\dap}{\apres}} &
$\mainpartition$
\\[1.25ex]
\hline
\hline
\end{tabular}
}
\end{table}

The algorithm in Fig.~\ref{fig:ltlsattrp} provides a high level view of \tr in \trp \cite{UHustadtBKonev-CollegiumLogicum-2004}. The algorithm takes a set of starting clauses $\setclauses$ in SNF as input. It returns $\unsat$ if $\setclauses$ is found to be unsatisfiable (by deriving $\emptyclause$) and $\sat$ otherwise.
Resolution between two initial or two global clauses or between an initial and a global clause is performed by a straightforward extension of propositional resolution (e.g., \ifnoappendix\cite{LBachmairHGanzinger-HandbookOfAutomatedReasoning-2001}\else\cite{JRobinson-JACM-1965,JFrancoJMartin-HandbookOfSatisfiability-2009,LBachmairHGanzinger-HandbookOfAutomatedReasoning-2001}\fi). This is expressed in the five production rules listed under \emph{saturation} in Tab.~\ref{tab:productionrules}. Given a set of SNF clauses $\setclauses$ we say that one \emph{saturates} $\setclauses$, if one applies these production rules to clauses in $\setclauses$ until the empty clause $\emptyclause$ has been derived, or until no new clauses are generated.
Resolution between a set of initial and global clauses and an eventuality clause with eventuality literal $\apres$ requires finding a set of global clauses that allows one to infer conditions under which $\fnext{\fglobally{\fbnot{\apres}}}$ holds. Such a set of clauses is called a \emph{loop} in $\fbnot{\apres}$.
\trp implements the BFS approach to loop search \cite{CDixon-ICTL-1997,CDixon-AnnalsOfMathematicsAndArtificialIntelligence-1998,CDixon-PhDThesis-1995}.
Loop search involves all production rules in Tab.~\ref{tab:productionrules} except \refinitii, \refinitin, \refstepnn, and \refstepnx.

{
\def\checkterm{\lIf{$\emptyclause \in \mainpartition$}{\Return{\unsat}}\;}
\begin{figure}[t]
\centering
\begin{minipage}{0.95\linewidth}
\begin{algorithm}[H]
\SetKw{Or}{or}
\KwIn{A set of SNF clauses $\setclauses$.}
\KwOut{$\Unsat$ if $\setclauses$ is unsatisfiable; $\sat$ otherwise.}
\BlankLine
\algassign{\mainpartition}{\setclauses}; \checkterm\nllabel{line:fig:ltlsattrp:init1}
saturate(\mainpartition); \checkterm\nllabel{line:fig:ltlsattrp:sat1}
augment(\mainpartition)\;\nllabel{line:fig:ltlsattrp:aug}
saturate(\mainpartition); \checkterm\nllabel{line:fig:ltlsattrp:sat2}
\algassign{\mainpartitionp}{\emptyset}\;
\While{$\mainpartitionp \ne \mainpartition$}{
  \algassign{\mainpartitionp}{\mainpartition}\;
  \For{$\clause \in \setclauses \;.\; \clause \mbox{ is an eventuality clause}$}{
    \algassign{\setclausesp}{\{\emptyclause\}}\;\nllabel{line:fig:ltlsattrp:loopsearchinit}
    \Repeat{found \Or $\setclausesp = \emptyset$}{
      initialize-BFS-loop-search-iteration(\mainpartition, \clause, \setclausesp, \looppartition)\;\nllabel{line:fig:ltlsattrp:loopitinit}
      saturate-step-xx(\looppartition)\;\nllabel{line:fig:ltlsattrp:loopitsat}
      \algassign{\setclausesp}{\{\clausep \in \looppartition \mid \clausep \mbox{ has empty $\fnextname$ part}\}}\;\nllabel{line:fig:ltlsattrp:loopitchecksub1}
      $\setclausespp \leftarrow \{(\fglobally{\;\dapp}) \mid \fgxclause{\fbor{\dapp}{\apres}} \in \mbox{  }\looppartition \mbox{ generated by {\tiny{\setlength{\fboxsep}{1pt}\fbox{BFS-loop-it-init-c}}}}\}$\;\nllabel{line:fig:ltlsattrp:loopitchecksub2}
      {\itshape found} $\leftarrow$ subsumes(\setclausesp, \setclausespp)\;\nllabel{line:fig:ltlsattrp:loopitchecksub3}\nllabel{line:fig:ltlsattrp:loopitend}
    }\nllabel{line:fig:ltlsattrp:loopsearchend1}
    \If{found}{
      derive-BFS-loop-search-conclusions(\clause, \setclausesp, \mainpartition)\;\nllabel{line:fig:ltlsattrp:loopsearchconclusions}
      saturate(\mainpartition); \checkterm\nllabel{line:fig:ltlsattrp:sat3}
    }\nllabel{line:fig:ltlsattrp:loopsearchend2}
  }
}
\Return{\sat}\;\nllabel{line:fig:ltlsattrp:returnsat}
\end{algorithm}
\end{minipage}
\begin{minipage}{0.47\linewidth}\caption{\label{fig:ltlsattrp}LTL satisfiability checking via \tr in \trp}\end{minipage}
\end{figure}
}

In line \ref{line:fig:ltlsattrp:init1} \refalgltlsattrp initializes $\mainpartition$ with the set of starting clauses and terminates iff one of these is the empty clause. Then, in line \ref{line:fig:ltlsattrp:sat1}, it saturates $\mainpartition$ (terminating iff the empty clause is generated). In line \ref{line:fig:ltlsattrp:aug} it \emph{augments} $\mainpartition$ by applying production rule \refaugone to each eventuality clause in $\mainpartition$ and \refaugtwo once per eventuality literal in $\mainpartition$, where $\fapwaitfor{\apres}$ is a fresh proposition. This is followed by another round of saturation in line \ref{line:fig:ltlsattrp:sat2}.\footnote{Here we report the algorithm as implemented in the version of \trp that we obtained. There saturation is performed directly before and directly after augmentation.}
From now on \refalgltlsattrp alternates between searching for a loop for some eventuality clause $\clause$ (lines \ref{line:fig:ltlsattrp:loopsearchinit}--\ref{line:fig:ltlsattrp:loopsearchconclusions}) and saturating $\mainpartition$ if loop search has generated new clauses (line \ref{line:fig:ltlsattrp:sat3}).
It terminates if either the empty clause was derived (line \ref{line:fig:ltlsattrp:sat3}) or if no new clauses were generated (line \ref{line:fig:ltlsattrp:returnsat}).

Loop search for some eventuality clause $\clause$ may take several \emph{iterations} (lines \ref{line:fig:ltlsattrp:loopitinit}--\ref{line:fig:ltlsattrp:loopitend}). Each loop search iteration uses saturation restricted to \refstepxx as a subroutine (line \ref{line:fig:ltlsattrp:loopitsat}). Correctness of BFS loop search requires that each BFS loop search iteration has its own set of clauses $\looppartition$ in which it works. We call $\mainpartition$ and $\looppartition$ \emph{partitions}. The set of partitions is a set of sets; i.e., a clause may appear in several partitions but not more than once in any one partition. Columns 3, 5, and 7 in Tab.~\ref{tab:productionrules} indicate whether a premise (resp.~conclusion) of a production rule is taken from (resp.~put into) the main partition ($\mainpartition$), the loop partition of the current loop search iteration ($\looppartition$), the loop partition of the previous loop search iteration ($\looppartitionp$), or either of $\mainpartition$ or $\looppartition$ as long as premises and conclusion are in the same partition ($\mainorlooppartition$). In line \ref{line:fig:ltlsattrp:loopitinit} partition $\looppartition$ of a loop search iteration is initialized by applying production rule \refloopitinitx once to each global clause with non-empty $\fnextname$ part in $\mainpartition$, rule \refloopitinitn once to each global clause with empty $\fnextname$ part in $\mainpartition$, and rule \refloopitinitc once to each global clause with empty $\fnextname$ part in the partition of the previous loop search iteration $\looppartitionp$. Notice that by construction at this point $\looppartition$ contains only global clauses with non-empty $\fnextname$ part. Then $\looppartition$ is saturated using only rule \refstepxx (line \ref{line:fig:ltlsattrp:loopitsat}). A loop has been found iff each global clause with empty $\fnextname$ part that was derived in the previous loop search iteration is subsumed by at least one global clause with empty $\fnextname$ part that was derived in the current loop search iteration (lines \ref{line:fig:ltlsattrp:loopitchecksub1}--\ref{line:fig:ltlsattrp:loopitchecksub3}). Subsumption between a pair of clauses corresponds to an instance of production rule \refloopitsub; note, though, that this rule does not produce a new clause but records a relation between two clauses to be used later for extraction of a \uc. Loop search for $\clause$ terminates if either a loop has been found or no clauses with empty $\fnextname$ part were derived (line \ref{line:fig:ltlsattrp:loopsearchend1}). If a loop has been found, rules \refloopconclusionone and \refloopconclusiontwo are applied once to each global clause with empty $\fnextname$ part that was derived in the current loop search iteration (line \ref{line:fig:ltlsattrp:loopsearchconclusions}) to obtain the loop search conclusions for the main partition.

The \tr method with a BFS algorithm for loop search, which is implemented in \trp, is a sound and complete decision procedure for the satisfiability of a set of SNF clauses \cite{UHustadtBKonev-CADE-2003,UHustadtBKonev-CollegiumLogicum-2004,MFisherCDixonMPeim-ACMTrComputationalLogic-2001,CDixon-ICTL-1997,CDixon-AnnalsOfMathematicsAndArtificialIntelligence-1998,CDixon-PhDThesis-1995}.
We are not aware of a detailed complexity analysis of \tr as implemented in \trp; for complexity analyses of parts of the \tr method relevant to the implementation in \trp see \cite{MFisherCDixonMPeim-ACMTrComputationalLogic-2001,CDixon-AnnalsOfMathematicsAndArtificialIntelligence-1998}.
To understand the complexity of our extension of \refalgltlsattrp for the extraction of \ucs proposed in this article we are mainly interested in the size of the resolution graph that we will construct from an execution of \refalgltlsattrp.
This will be analyzed in Lemma \ref{thm:sizeofresolutiongraph}.

\subsubsection{Example}
\label{sec:tralgorithmexample}


We now continue the example from Sec.~\ref{sec:translatingltlintosnf}.
We would like to execute \refalgltlsattrp on the SNF of $\inp \definedas \fband{(\fglobally{(\fband{p}{q})})}{\ffinally{\fbnot{p}}}$.
For technical reasons we make some minor modifications.
First, when translating $\inp$ into a set of SNF clauses $\setclauses$ our implementation treats top level conjuncts as separate formulas.
We therefore separately translate $\fglobally{(\fband{p}{q})}$ and $\ffinally{\fbnot{p}}$ according to Def.~\ref{def:ltltosnf}.
Second, while the use of implication makes the SNF clauses in column 4 of Tab.~\ref{tab:ltltosnf} easier to understand, implication is not an operator that is available for SNF clauses in \trp.
Hence, we syntactically expand implication using its definition.
Third, due to space constraints in Fig.~\ref{fig:completeexresgraphucex} we replace the formula indices in the fresh atomic propositions $\dCNFvar, \dCNFvarp, \ldots$ with numerical indices.
Fourth, we remove some parentheses and let the negation and next time operators bind stronger than the or operator.
The modified SNF of $\inp$, $\setclauses$, is shown in \eqref{ex:complete:snf}.
$x_{1}$ corresponds to $\fdCNFvar{\fglobally{(\fband{p}{q})}}$, $x_2$ to $\fdCNFvar{\fband{p}{q}}$, $x_3$ to $\fdCNFvar{\ffinally{\fbnot{p}}}$, and $x_4$ to $\fdCNFvar{\fbnot{p}}$.
\begin{equation}\label{ex:complete:snf}
\setclauses \definedas
\begin{array}{l}
\hspace{-0.5em}\{x_{1},\\
\fglobally{(\fbor{\fbnot{x_{1}}}{\fnext{x_{1}}})},\\
\fglobally{(\fbor{\fbnot{x_{1}}}{x_{2}})},\\
\fglobally{(\fbor{\fbnot{x_{2}}}{p})},\\
\fglobally{(\fbor{\fbnot{x_{2}}}{q})},\\
x_{3},\\
\fglobally{(\fbor{\fbnot{x_{3}}}{\ffinally{x_{4}}})},\\
\fglobally{(\fbor{\fbnot{p}}{\fbnot{x_{4}}})}\}.
\end{array}
\end{equation}

In Fig.~\ref{fig:completeexresgraphucex} we show an execution of \refalgltlsattrp on $\setclauses$.
In Fig.~\ref{fig:completeexresgraphucex} \tr generally proceeds from bottom to top.
At the bottom in the rectangle shaded in light red are the clauses in $\setclauses$.
The leftmost clause in the top row is the empty clause $\emptyclause$, indicating unsatisfiability.
Clauses are connected with directed edges from premises to conclusions.
Edges are labeled with production rules, where ``BFS-loop'' is abbreviated to ``loop'', ``init'' to ``i'', and ``conclusion'' to ``conc''.
Please ignore the different colors and styles of the clauses and edges for now.
These will be explained when discussing extraction of a \uc later in Sec.~\ref{sec:extractingaucinsnf}.

Saturation in line \ref{line:fig:ltlsattrp:sat1} of \refalgltlsattrp produces no new clauses.\footnote{While it may seem that some clauses are not considered for saturation, this is due to either subsumption of one clause by another (e.g., $\fglobally{(\fbor{\fbnot{\fapwaitfor{x_{4}}}}{\fbor{\fnext{\fbnot{x_{1}}}}{\fnext{x_{4}}}})}$ obtained from $\fglobally{(\fbor{\fbnot{\fapwaitfor{x_{4}}}}{\fbor{\fnext{x_{4}}}{\fnext{\fapwaitfor{x_{4}}}}})}$ and $\fglobally{(\fbor{\fbnot{x_{1}}}{\fbnot{\fapwaitfor{x_{4}}}})}$ is subsumed by $\fglobally{(\fbor{\fbnot{\fapwaitfor{x_{4}}}}{\fnext{\fbnot{x_{1}}}})}$) or the fact that \trp uses \emph{ordered} resolution (e.g., $x_{1}$ with $\fglobally{(\fbor{\fbnot{x_{1}}}{x_{2}})}$ --- the order here is $x_{1} < x_{2} < p < q < x_{3} < x_{4}$; \cite{UHustadtBKonev-CADE-2003,LBachmairHGanzinger-HandbookOfAutomatedReasoning-2001}). Both are issues of completeness of \tr and, therefore, not discussed in this article.}
The two clauses in row 2 are generated by augmentation (line \ref{line:fig:ltlsattrp:aug}).
The following saturation (line \ref{line:fig:ltlsattrp:sat2}) produces no new clauses.
The dark green shaded rectangle is the loop partition for the first loop search iteration.
Row 3 contains the clauses obtained by initialization of the BFS loop search iteration (line \ref{line:fig:ltlsattrp:loopitinit}).
Row 4 then contains the clauses generated from those in row 3 by saturation restricted to \refstepxx (line \ref{line:fig:ltlsattrp:loopitsat}).
The subsumption test fails in this iteration, as $\fbnot{x_{1}}$ (from $\fglobally{(\fbnot{x_{1}})}$) does not subsume $\emptyclause$ (from $\fglobally{(\fnext{x_{4}})}$) (lines \ref{line:fig:ltlsattrp:loopitchecksub1}--\ref{line:fig:ltlsattrp:loopitend}).
The light green shaded rectangle is the loop partition for the second loop search iteration.
Row 5 contains the clauses obtained by initialization and row 6 those obtained from them by restricted saturation.
This time the subsumption test succeeds, and the loop search conclusions are shown in row 7 (line \ref{line:fig:ltlsattrp:loopsearchconclusions}).
Finally, while row 8 contains a ``blind alley'', row 9 has the derivation of the empty clause $\emptyclause$ via saturation (line \ref{line:fig:ltlsattrp:sat3}).

\newcommand{\refshortinferencerule}[1]{{\large{#1}}\xspace}
\newcommand{\refshortinitii}{\refshortinferencerule{init-ii}}
\newcommand{\refshortinitin}{\refshortinferencerule{init-in}}
\newcommand{\refshortstepnn}{\refshortinferencerule{step-nn}}
\newcommand{\refshortstepnx}{\refshortinferencerule{step-nx}}
\newcommand{\refshortstepxx}{\refshortinferencerule{step-xx}}
\newcommand{\refshortloopitinitx}{\refshortinferencerule{loop-it-i-x}}
\newcommand{\refshortloopitinitn}{\refshortinferencerule{loop-it-i-n}}
\newcommand{\refshortloopitinitc}{\refshortinferencerule{loop-it-i-c}}
\newcommand{\refshortloopitsub}{\refshortinferencerule{loop-it-sub}}
\newcommand{\refshortloopconclusionone}{\refshortinferencerule{loop-conc1}}
\newcommand{\refshortloopconclusiontwo}{\refshortinferencerule{loop-conc2}}
\newcommand{\refshortaugone}{\refshortinferencerule{aug1}}
\newcommand{\refshortaugtwo}{\refshortinferencerule{aug2}}
{
\newlength\mylinewidth
\setlength\mylinewidth{\the\linewidth}

\begin{sidewaysfigure}
\vspace{98ex}
\centering
\resizebox{!}{0.955\mylinewidth}{
\begin{tikzpicture}[auto,
    =<triangle 45,
    every state/.style={draw=none},
    corev/.style={draw=blue,rectangle,rounded corners=1mm,very thick,densely dashed,text=blue},
    coreea/.style={blue,very thick,densely dashed},
    coreel/.style={draw=none,very thick,text=blue},
    corenov/.style={draw=blue,rectangle,rounded corners=1mm,very thick,densely dashed,text=blue},
    corenoea/.style={blue,very thick,densely dashed},
    corenoel/.style={draw=none,very thick,text=blue},
    noncoreea/.style={},
    noncoreel/.style={},
    corenononrgea/.style={blue,very thick,densely dashed},
    corenononrgel/.style={draw=none,very thick,text=blue},
    noncorenonrgea/.style={},
    noncorenonrgel/.style={},
    xscale=0.875,yscale=0.75]

  \input{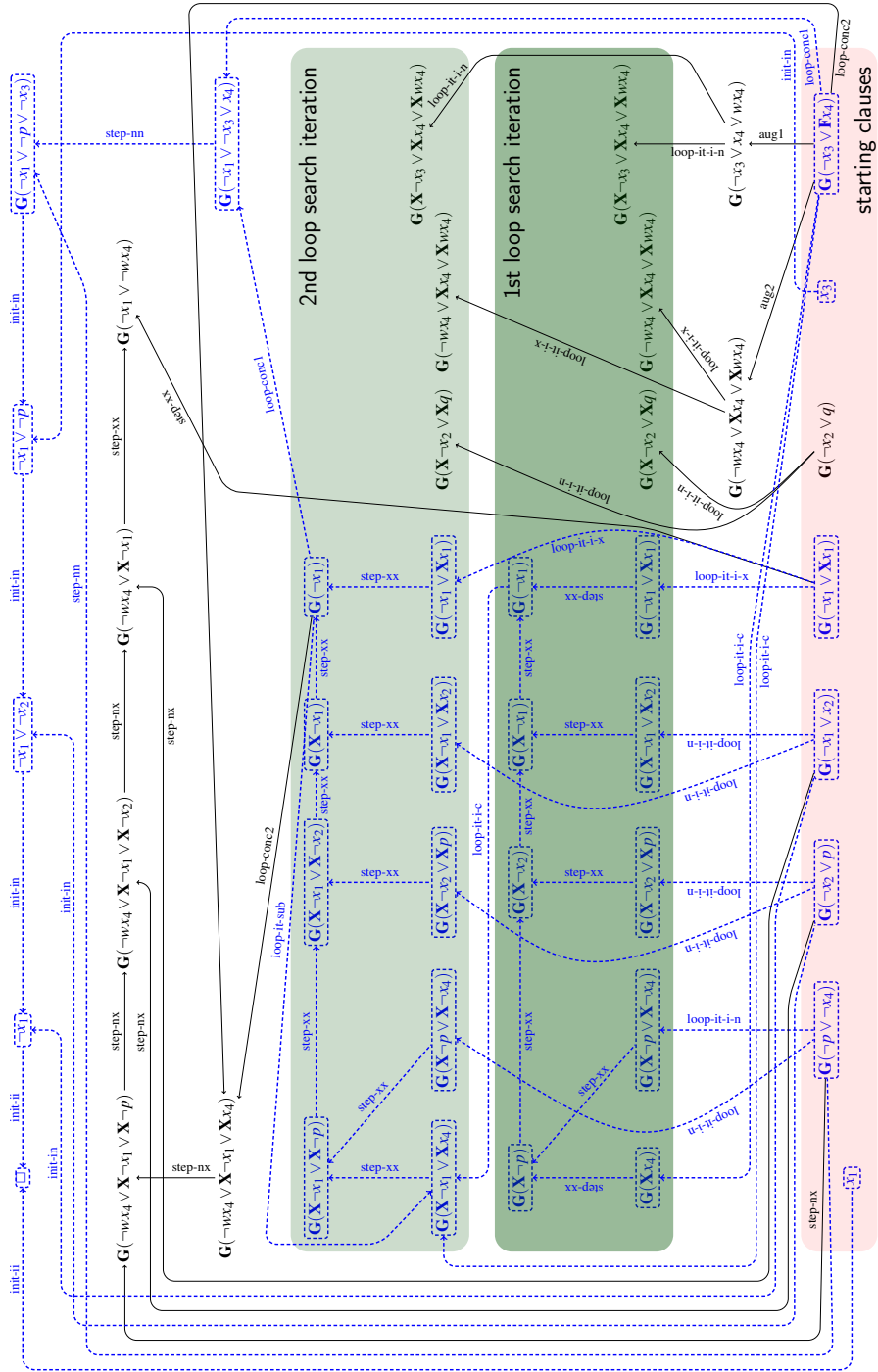}

\end{tikzpicture}
}
\begin{minipage}{0.66\linewidth}\caption{\label{fig:completeexresgraphucex} Example of an execution of the \tr algorithm with corresponding resolution graph and \uc extraction in SNF}\end{minipage}
\end{sidewaysfigure}
}

\section{\UC Extraction}
\label{sec:ucextraction}

In this section we present our method to extract a \uc from an execution of \refalgltlsattrp.
We first show how to extract a \uc in SNF (Sec.~\ref{sec:extractingaucinsnf}).
Then we map that \uc back to LTL (Sec.~\ref{sec:extractingaucinltl}).

\subsection{Extracting a \UC in SNF}
\label{sec:extractingaucinsnf}

In this subsection we describe, given an unsatisfiable set of SNF clauses $\setclauses$, how to obtain a subset of $\setclauses$, $\setclausesuc$, that is by itself unsatisfiable.
During the execution of \refalgltlsattrp a resolution graph is built that records which clauses were used to generate other clauses.
Then the resolution graph is traversed backwards from the empty clause to find the subset of $\setclauses$ that was actually used to prove unsatisfiability (Sec.~\ref{sec:extractingaucinsnffromaresolutiongraph}).
The specifics of the \tr algorithm are then used to construct an optimized version of the resolution graph (Sec.~\ref{sec:extractingaucinsnffromanoptimizedresolutiongraph}).

While we are not aware of corresponding previous work in the domain of temporal logic, the general idea of the construction is unsurprising.
In fact, applying a resolution method to a set of input clauses in a clausal normal form, constructing a resolution graph, and determining an unsatisfiable subset of the set of input clauses by backward traversal of the resolution graph from the empty clause is well known in SAT (e.g., \ifnoappendix\cite{LZhang-PhDThesis-2003}\else\cite{LZhangSMalik-SAT-2003,LZhang-PhDThesis-2003}\fi).
The complexity analysis and the optimization of the resolution graph are original and specific to \refalgltlsattrp.

Note that in the preliminary version of this paper \cite{VSchuppan-TIME-2013} we presented the optimized version of the resolution graph right away.

\subsubsection{Extracting a \UC in SNF from a Resolution Graph}
\label{sec:extractingaucinsnffromaresolutiongraph}

\paragraph{Resolution Graph --- Construction and Added Complexity}

In Def.~\ref{def:resolutiongraphtype}--\ref{def:resolutiongraphupdate} we state the construction of a resolution graph during an execution of \refalgltlsattrp.
Then we discuss the complexity that the construction of a resolution graph adds to an execution of \refalgltlsattrp (Rem.~\ref{thm:vertexatmostthreeincomingedges} and \ref{thm:rgconstructioncomplexity}).
Finally, in Lemma \ref{thm:sizeofresolutiongraph} we establish a bound on the size of a resolution graph.

The definition of a resolution graph is split into three parts.
Definition \ref{def:resolutiongraphtype} gives the ``type'' of a resolution graph.
Definition \ref{def:resolutiongraphinitialization} describes the initialization of a resolution graph at the beginning of an execution of \refalgltlsattrp.
Finally, Def.~\ref{def:resolutiongraphupdate} describes how a resolution graph is updated during an execution of \refalgltlsattrp.
Splitting the definition of a resolution graph allows to limit the scope of the changes that are required for the definition of an optimized resolution graph in Sec.~\ref{sec:extractingaucinsnffromanoptimizedresolutiongraph} to how a resolution graph is updated.

\mydefinition{Resolution Graph: Type}\label{def:resolutiongraphtype}
Let $\setclauses$ be a set of SNF clauses, and assume an execution of \refalgltlsattrp on $\setclauses$. A \emph{resolution graph} $\graph$ is a directed graph consisting of
\begin{inparaenum}[(i)]
\item a set of vertices $\setvertices$,
\item a set of directed edges $\setedges \subseteq \setvertices \times \setvertices$,
\item a labeling of vertices with SNF clauses $\vertexlabelingname : \setvertices \rightarrow \allclauses$, and
\item a partitioning $\partitioningv$ of the set of vertices $\setvertices$ into one main partition $\mainpartitionv$ and one partition $\looppartitionv_\pos$ for each BFS loop search iteration in the execution of \refalgltlsattrp on the set of SNF clauses $\setclauses$: $\;\;\partitioningv\!\! : \setvertices = \fdisjunion{\mainpartitionv}{\fdisjunion{\looppartitionv_0}{\fdisjunion{\ldots}{\looppartitionv_\maxind}}}$.
\end{inparaenum}
\end{definition}

\mydefinition{Resolution Graph: Initialization}\label{def:resolutiongraphinitialization}
Let $\setclauses$ be a set of SNF clauses, and assume an execution of \refalgltlsattrp on $\setclauses$. The \emph{resolution graph} $\graph$ \emph{is initialized} in line \ref{line:fig:ltlsattrp:init1} of \refalgltlsattrp as follows:
\begin{inparaenum}[(i)]
\item $\setvertices$ contains one vertex $\vertex$ per clause $\clause$ in $\setclauses$: $\setvertices = \{\vertex_\clause \mid \clause \in \setclauses\}$,
\item $\setedges$ is empty: $\setedges = \emptyset$,
\item each vertex is labeled with the corresponding clause: $\vertexlabelingname : \setvertices \rightarrow \setclauses, \fvertexlabeling{\vertex_\clause} = \clause$, and
\item the partitioning $\partitioningv$ contains only the main partition $\mainpartitionv$, which contains all vertices: $\partitioningv : \mainpartitionv = \setvertices$.
\end{inparaenum}
\end{definition}

\mydefinition{Resolution Graph: Update}\label{def:resolutiongraphupdate}
Let $\setclauses$ be a set of SNF clauses, and assume an execution of \refalgltlsattrp on $\setclauses$. The \emph{resolution graph} $\graph$ \emph{is updated} as follows.
Whenever a new BFS loop search iteration is entered (line \ref{line:fig:ltlsattrp:loopitinit}), a new partition $\looppartitionv_\pos$ is created and added to $\partitioningv$.
For each application of a production rule from Tab.~\ref{tab:productionrules} that either generates a new clause in partition $\mainpartitionv$ or $\looppartitionv_\pos$ or is the first application of rule \refloopitsub to clause $\clausepp$ in $\setclausespp$ in line \ref{line:fig:ltlsattrp:loopitchecksub3}:
\begin{inparaenum}[(i)]
\item if the applied production rule is not rule \refloopitsub, then a new vertex $\vertex$ is created for the conclusion $\clause$ (which is a new clause), labeled with $\clause$, and put into partition $\mainpartitionv$ or $\looppartitionv_\pos$;
\item an edge is created from the vertex labeled with premise 1 (resp.~premise 2) in partition $\mainpartitionv$, $\looppartitionv_\pos$, or $\looppartitionv_{\pos-1}$ to the vertex labeled with the conclusion in partition $\mainpartitionv$ or $\looppartitionv_\pos$.
\end{inparaenum}
\end{definition}

\myremark%
{\thmresolutiongraphdet}%
{thm:resolutiongraphdet}%
{Determinants of Resolution Graph}%
{Note that a resolution graph constructed according to Def.~\ref{def:resolutiongraphtype}--\ref{def:resolutiongraphupdate} is determined not only by the set of input clauses $\setclauses$ but also by the execution of \refalgltlsattrp on $\setclauses$.
For example, consider $\setclauses \definedas \{\ficlause{\ap}, \ficlause{\fbnot{\ap}}, \ficlause{\app}, \ficlause{\fbnot{\app}}\}$.
Depending on the order in which clauses are considered by \refalgltlsattrp the empty clause $\emptyclause$ will be derived either from $\ficlause{\ap}$ and $\ficlause{\fbnot{\ap}}$ or from $\ficlause{\app}$ and $\ficlause{\fbnot{\app}}$.
As saturation halts as soon as $\emptyclause$ has been derived, two different resolution graphs for $\setclauses$ will be obtained.
Hence, strictly speaking, a resolution graph is parameterized by a set of SNF clauses $\setclauses$ and by an execution of \refalgltlsattrp on $\setclauses$.
In this article both parameters are usually provided by the context.
Therefore, in the remainder of this article when we say ``resolution graph'', then we refer to the object that is obtained from Def.~\ref{def:resolutiongraphtype}--\ref{def:resolutiongraphupdate} after an execution of \refalgltlsattrp on a set of clauses $\setclauses$.
}
\thmresolutiongraphdet{true}

Remember that BFS loop search requires that clauses in different BFS loop search iterations are kept apart from each other.
Hence, in any partition there can be at most one vertex labeled with a given clause, but there may well exist two or more vertices in different partitions labeled with the same clause.
In fact, an application of production rule \refloopitinitx will lead to such a situation as it copies a clause from the main partition to the partition of the current BFS loop search iteration.

In Rem.~\ref{thm:vertexatmostthreeincomingedges} and \ref{thm:rgconstructioncomplexity} we establish the complexity of the construction of a resolution graph in terms of its own size.
In Lemma \ref{thm:sizeofresolutiongraph} we then obtain a limit on the size of the resolution graph by bounding the number of 
\begin{inparaenum}[(i)]
\item different clauses in each partition,
\item iterations in each loop search by the length of the longest monotonically increasing sequence of Boolean formulas over $\allaps$, and
\item loop searches by the number of different loop search conclusions.
\end{inparaenum}

\myremark%
{\thmvertexatmostthreeincomingedges}%
{thm:vertexatmostthreeincomingedges}%
{Vertex in Resolution Graph has At Most Three Incoming Edges}%
{Let $\graph$ be a resolution graph. Inspection of Tab.~\ref{tab:productionrules} shows that each vertex in the resolution graph $\graph$ has at most three incoming edges. Let $\vertex_\clause$ be a vertex in $\graph$ labeled with clause $\clause$. If the clause $\clause$ was generated by an application of any production rule except for \refloopitinitc, then the vertex $\vertex_\clause$ only has incoming edges originating at the premises of $\clause$, and each production rule in Tab.~\ref{tab:productionrules} has at most two premises. Otherwise, if the clause $\clause$ was generated by an application of production rule \refloopitinitc, then the vertex $\vertex_\clause$ has at most two incoming edges originating at the premises of $\clause$ and at most one incoming edge from an application of rule \refloopitsub.}
\thmvertexatmostthreeincomingedges{true}

\myremark%
{\thmrgconstructioncomplexity}%
{thm:rgconstructioncomplexity}%
{Added Complexity of Construction of Resolution Graph}%
{Let $\graph$ be a resolution graph with set of vertices $\setvertices$ and set of edges $\setedges$. Assume that the data structure used to represent clauses has an entry for a pointer to the vertex of $\graph$ that it labels, and similarly the data structure used to represent vertices has an entry for a pointer to the clause that it is labeled with. Moreover, the data structure used to represent vertices has three slots with pointers to its incoming edges (remember that by Rem.~\ref{thm:vertexatmostthreeincomingedges} each vertex in $\graph$ has at most three incoming edges) and a Boolean flag to mark vertices in the main partition that are labeled with an initial clause. It is now easy to see that initializing the resolution graph with respect to the set of SNF clauses $\setclauses$ and updating the resolution graph $\graph$ whenever a new clause is generated during the execution of \refalgltlsattrp can be performed using constant time for each clause. In other words, the construction of the resolution graph takes time $\fbigo{\fcardinality{\setvertices} + \fcardinality{\setedges}}$ overall in addition to the time required to run \refalgltlsattrp.}
\thmrgconstructioncomplexity{true}

\mylemma%
{\thmsizeofresolutiongraph}%
{thm:sizeofresolutiongraph}%
{Size of Resolution Graph}%
{Let $\setclauses$ be a set of SNF clauses, and let $\graph$ be a resolution graph with set of vertices $\setvertices$ and set of edges $\setedges$. Then $\fcardinality{\setvertices}$ and $\fcardinality{\setedges}$ are at most exponential in $\fcardinality{\allaps} + log(\fcardinality{\setclauses})$.}
\thmsizeofresolutiongraph{true}

\begin{sloppypar}
\begin{proof}
The following reasoning shows that $\fcardinality{\setvertices}$ is at most exponential in $\fcardinality{\allaps} + log(\fcardinality{\setclauses})$:
\begin{enumerate}
\item In an initial clause a proposition can be not present, present non-negated, or present negated. Hence, the number of different initial clauses is $\fbigo{3^{\fcardinality{\allaps}}}$.
\item In a global clause a proposition can be one of not present, present non-negated, or present negated; and prefixed by $\fnextname$ not present, present non-negated, or present negated. Hence, the number of different global clauses is $\fbigo{9^{\fcardinality{\allaps}}}$.
\item The number of clauses in the main partition is bounded by $\fcardinality{\setclauses} + \fbigo{3^{\fcardinality{\allaps}}} + \fbigo{9^{\fcardinality{\allaps}}} = \fbigo{\fcardinality{\setclauses} + 9^{\fcardinality{\allaps}}}$.
\item The number of clauses in a partition for a BFS loop search iteration is bounded by $\fbigo{9^{\fcardinality{\allaps}}}$.
\item The number of partitions is bounded by 1 plus the number of BFS loop search iterations.
\item The number of iterations in a BFS loop search is bounded by the length of the longest monotonically increasing sequence of Boolean formulas over $\allaps$, which is $\fbigo{2^{\fcardinality{\allaps}}}$. (In \cite{CDixon-AnnalsOfMathematicsAndArtificialIntelligence-1998} that result is obtained by reasoning on a behavior graph.)
\item The number of BFS loop searches is bounded by the number of different clauses that can be the result of a BFS loop search. The number of different clauses that can be the consequence of BFS loop search conclusion 1 \refloopconclusionone is bounded by the number of different global clauses with empty next part, which is $\fbigo{3^{\fcardinality{\allaps}}}$. The number of different clauses that can be the consequence of BFS loop search conclusion 2 \refloopconclusiontwo is bounded by the number of different eventuality literals times the number of different global clauses with empty next part, which is $\fbigo{\fcardinality{\setclauses} \cdot 3^{\fcardinality{\allaps}}}$. Hence, the number of BFS loop searches is bounded by $\fbigo{\fcardinality{\setclauses} \cdot 3^{\fcardinality{\allaps}}}$.
\item Taking all of the above into account, the number of clauses is bounded by $\fbigo{\fcardinality{\setclauses} + 9^{\fcardinality{\allaps}} + \fcardinality{\setclauses} \cdot 3^{\fcardinality{\allaps}} \cdot 2^{\fcardinality{\allaps}} \cdot 9^{\fcardinality{\allaps}}} = \fbigo{\fcardinality{\setclauses} \cdot 54^{\fcardinality{\allaps}}}$.
\end{enumerate}
To see the result for $\fcardinality{\setedges}$ notice that by Rem.~\ref{thm:vertexatmostthreeincomingedges} each vertex in $\graph$ has at most three incoming edges.
This concludes the proof.
\qed
\end{proof}
\end{sloppypar}

\paragraph{\UC in SNF --- Construction, Correctness, and Added Complexity}

We are now ready to describe the extraction of a \uc in SNF from a resolution graph in Def.~\ref{def:coreextraction}.
Theorem \ref{thm:coreisunsat}, which establishes correctness of the construction, is then straightforward to obtain.

\mydefinition{\UC in SNF}\label{def:ucinsnf}
Let $\setclauses$ be an unsatisfiable set of SNF clauses.
Let $\setclausesuc$ be an unsatisfiable subset of $\setclauses$.
Then $\setclausesuc$ is a \emph{\uc of $\setclauses$ in SNF}.
\end{definition}

\mydefinition{\UC in SNF via \TR}\label{def:coreextraction}
Let $\setclauses$ be an unsatisfiable set of SNF clauses, let $\graph$ be a resolution graph, and let $\vertex_\emptyclause$ be the (unique) vertex in the main partition $\mainpartitionv$ of the resolution graph $\graph$ labeled with the empty clause $\emptyclause$.
Let $\graphp$ be the smallest subgraph of $\graph$ that contains $\vertex_\emptyclause$ and all vertices in $\graph$ (and the corresponding edges) that are backward reachable from $\vertex_\emptyclause$.
The \emph{\uc of $\setclauses$ in SNF via \TR}, $\setclausesuc$, is the subset of $\setclauses$ such that there exists a vertex $\vertex$ in the subgraph $\graphp$, labeled with $\clause \in \setclauses$, and contained in the main partition $\mainpartitionv$ of $\graph$: $\setclausesuc = \{\clause \in \setclauses \mid \exists \vertex \in \setvertices_\graphp \;.\; \fband{\fvertexlabeling{\vertex} = \clause}{\vertex \in \mainpartitionv}\}$.
\end{definition}

\mytheorem%
{\thmcoreisunsat}%
{thm:coreisunsat}%
{Unsatisfiability of \UC in SNF via \TR}%
{Let $\setclauses$ be an unsatisfiable set of SNF clauses, and let $\setclausesuc$ be a \uc of $\setclauses$ in SNF via \tr. Then $\setclausesuc$ is unsatisfiable.}
\thmcoreisunsat{true}

\begin{proof}
Notice that in the resolution graph each conclusion is connected by an edge to all of its premises.
Therefore, the \uc in SNF according to Def.~\ref{def:coreextraction} contains all clauses of the set of starting clauses $\setclauses$ that contributed to deriving the empty clause and, hence, to establishing unsatisfiability of $\setclauses$.
It now follows directly from the correctness of \tr that $\setclausesuc$ is unsatisfiable.
This concludes the proof.
\qed
\end{proof}

As $\setclausesuc$ is a subset of $\setclauses$, we have the following corollary.

\mycorollary%
{\thmcoreisunsatcor}%
{thm:coreisunsatcor}%
{\UC in SNF via \TR is \UC in SNF}%
{Let $\setclauses$ be an unsatisfiable set of SNF clauses, and let $\setclausesuc$ be a \uc of $\setclauses$ in SNF via \tr. Then $\setclausesuc$ is a \uc of $\setclauses$ in SNF.}
\thmcoreisunsatcor{true}

\myremark%
{\thmcoredet}%
{thm:coredet}%
{Determinants of \UC in SNF via \TR}%
{Note that an unsatisfiable set of SNF clauses $\setclauses$ may have several \ucs in SNF.
Moreover, not all of them might be obtained using Def.~\ref{def:coreextraction}.
For example, for $\setclauses \definedas \{\ficlause{\ap}, \ficlause{\fbnot{\ap}}, (\fglobally{\;\app}), (\ffinally{\;\fbnot{\app}})\}$ \refalgltlsattrp will derive the empty clause $\emptyclause$ from $\ficlause{\ap}$ and $\ficlause{\fbnot{\ap}}$ during the first round of saturation, leading to $\{\ficlause{\ap}, \ficlause{\fbnot{\ap}}\}$ as a \uc of $\setclauses$ in SNF via \tr.
The alternative \uc of $\setclauses$ in SNF, $\{(\fglobally{\;\app}), (\ffinally{\;\fbnot{\app}})\}$, requires loop search, which is only performed later in \refalgltlsattrp and, therefore, will not be produced by Def.~\ref{def:coreextraction}.
Finally, as Def.~\ref{def:coreextraction} relies on a resolution graph, the determinants of a resolution graph in Rem.~\ref{thm:resolutiongraphdet} also become determinants of a \uc in SNF via \tr. 
}
\thmcoredet{true}

Next we discuss the complexity that our method for \uc extraction adds to an execution of \refalgltlsattrp.
Using Rem.~\ref{thm:rgconstructioncomplexity} and Lemma \ref{thm:sizeofresolutiongraph} the desired result is easily obtained in Prop.~\ref{thm:corecomplexity}.
It turns out that the effort that is induced by our method for \uc extraction in addition to the effort required by an execution of \refalgltlsattrp is linearly bounded by the effort required by \refalgltlsattrp:
For all vertices present in the resolution graph a corresponding clause must have been previously generated by \refalgltlsattrp, and each vertex in the resolution graph requires constant time for construction and traversal.

\myproposition%
{\thmcorecomplexity}%
{thm:corecomplexity}%
{Added Complexity of \UC Extraction}%
{Let $\setclauses$ be an unsatisfiable set of SNF clauses. Construction of $\setclausesuc$ according to Def.~\ref{def:coreextraction} can be performed in time exponential in $\fcardinality{\allaps} + log(\fcardinality{\setclauses})$ in addition to the time required to run \refalgltlsattrp.}
\thmcorecomplexity{true}

\begin{proof}
Let $\graph$ be the resolution graph with set of vertices $\setvertices$ and set of edges $\setedges$.
Assume data structures for clauses and vertices in $\graph$ are as in Rem.~\ref{thm:rgconstructioncomplexity}.
By Rem.~\ref{thm:rgconstructioncomplexity} construction of $\graph$ takes time $\fbigo{\fcardinality{\setvertices} + \fcardinality{\setedges}}$ overall.
Once the empty clause $\emptyclause$ has been derived in the main partition, backward traversal of $\graph$ from the unique vertex in the main partition labeled with the empty clause, $\vertex_\emptyclause$, can be performed (e.g., using breadth first search) in time $\fbigo{\fcardinality{\setvertices} + \fcardinality{\setedges}}$.
Whenever a vertex labeled with an initial clause $\clause$ is encountered during the backward traversal of $\graph$, then $\clause$ is signaled to be part of $\setclausesuc$, requiring time $\fbigo{\fcardinality{\setclauses}}$ overall.
Using Lemma \ref{thm:sizeofresolutiongraph} the result follows.
This concludes the proof.
\qed
\end{proof}

\myremark%
{\thmrgpruning}%
{thm:rgpruning}%
{Pruning the Resolution Graph at Run Time}%
{The specifics of \tr in \refalgltlsattrp allow to optimize extraction of \ucs by pruning the resolution graph during the execution of \refalgltlsattrp extended with the construction in Def.~\ref{def:coreextraction} as follows.
Notice that after the completion of a (successful or unsuccessful) loop search for some eventuality clause $\clause$ in lines \ref{line:fig:ltlsattrp:loopsearchinit}--\ref{line:fig:ltlsattrp:loopsearchend2} of \refalgltlsattrp no new edges between the main partition and one of the partitions used during the just completed loop search for $\clause$ will be created.
Hence, after completion of an execution of lines \ref{line:fig:ltlsattrp:loopsearchinit}--\ref{line:fig:ltlsattrp:loopsearchend2} of \refalgltlsattrp vertices in the partitions used during the just completed loop search that are not backward reachable from the main partition can be pruned from the resolution graph.}
\thmrgpruning{true}

\paragraph{Example}

We now continue the running example from Sec.~\ref{sec:tralgorithmexample}.
We show how to extract a \uc of $\inp$ (see \eqref{ex:ltltosnf}) in SNF via \tr from the execution of \refalgltlsattrp on $\setclauses$ (see \eqref{ex:complete:snf}) in Fig.~\ref{fig:completeexresgraphucex}, which contains the resolution graph according to Def.~\ref{def:resolutiongraphtype}--\ref{def:resolutiongraphupdate}.
Vertices are given by the clauses they are labeled with and by the partition in which they appear.
We now apply Def.~\ref{def:coreextraction}.
The dashed, blue clauses and edges show the part of the resolution graph that is backward reachable from $\emptyclause$.
Clause $\fglobally{(\fbor{\fbnot{x_{2}}}{q})}$ is the only clause of $\setclauses$ that is not backward reachable from $\emptyclause$.
Hence, the \uc of $\inp$ in SNF via \tr according to Def.~\ref{def:coreextraction} is $\setclausesuc = \setclauses \setminus \{\fglobally{(\fbor{\fbnot{x_{2}}}{q})}\}$ as shown in \eqref{ex:complete:snfcore}.
\begin{equation}\label{ex:complete:snfcore}
\setclausesuc \definedas
\begin{array}{l}
\hspace{-0.5em}\{x_{1},\\
\fglobally{(\fbor{\fbnot{x_{1}}}{\fnext{x_{1}}})},\\
\fglobally{(\fbor{\fbnot{x_{1}}}{x_{2}})},\\
\fglobally{(\fbor{\fbnot{x_{2}}}{p})},\\
x_{3},\\
\fglobally{(\fbor{\fbnot{x_{3}}}{\ffinally{x_{4}}})},\\
\fglobally{(\fbor{\fbnot{p}}{\fbnot{x_{4}}})}\}.
\end{array}
\end{equation}

\subsubsection{Extracting a \UC in SNF from an Optimized Resolution Graph}
\label{sec:extractingaucinsnffromanoptimizedresolutiongraph}

In Def.~\ref{def:optimizedresolutiongraphupdate} we optimize the construction of a resolution graph by not including edges between some premises and conclusions.
Definition \ref{def:coreextractionopt} correspondingly adapts the extraction of a \uc.
The proof of correctness of the optimized construction is then significantly more complex than for the unoptimized variant.
Theorem \ref{thm:coreisunsatopt} is the main theorem and Lemmas \ref{thm:nobacklinkaugtwo}--\ref{thm:noloopitbacklinkloopinitc} contain the details.
Proposition \ref{thm:minimalitypremises} complements the definition and proof of correctness of the optimized construction by showing for the remaining edges that they are indeed required to obtain a \uc.

The ``type'' and initialization of an optimized resolution graph are as for a resolution graph (Def.~\ref{def:resolutiongraphtype}, \ref{def:resolutiongraphinitialization}); only the update changes (Def.~\ref{def:optimizedresolutiongraphupdate}).
In the remainder of this article when we say ``optimized resolution graph'', then we refer to the object that is obtained from Def.~\ref{def:resolutiongraphtype}, \ref{def:resolutiongraphinitialization}, and \ref{def:optimizedresolutiongraphupdate} after an execution of \refalgltlsattrp on a set of clauses $\setclauses$.

\mydefinition{Optimized Resolution Graph: Update}\label{def:optimizedresolutiongraphupdate}
An \emph{optimized resolution graph is updated} in the same way as a resolution graph, except that contrary to Def.~\ref{def:resolutiongraphupdate} no edge is added between a pair of vertices labeled with a premise and a conclusion in the following four cases:
\begin{inparaenum}[(i)]
\item between vertices labeled with premise 1 and the conclusion of rule \refaugtwo,
\item between vertices labeled with premise 1 and the conclusion of rule \refloopitinitc,
\item between vertices labeled with premise 2 and the conclusion of rule \refloopitinitc, and
\item between vertices labeled with premise 2 and the conclusion of rule \refloopconclusiontwo.
\end{inparaenum}
\end{definition}

\mydefinition{\UC in SNF via \TR from an Optimized Resolution Graph}\label{def:coreextractionopt}
Let $\setclauses$ be an unsatisfiable set of SNF clauses.
The \emph{\uc of $\setclauses$ in SNF via \tr from an optimized resolution graph} is obtained in the same way as the \uc of $\setclauses$ in SNF via \tr, except that an optimized resolution graph is used.
\end{definition}

\mytheorem%
{\thmcoreisunsatopt}%
{thm:coreisunsatopt}%
{Unsatisfiability of \UC in SNF via \tr from an Optimized Resolution Graph}%
{Let $\setclauses$ be an unsatisfiable set of SNF clauses, and let $\setclausesuc$ be a \uc of $\setclauses$ in SNF via \tr from an optimized resolution graph. Then $\setclausesuc$ is unsatisfiable.}
\begin{sloppypar}
\thmcoreisunsatopt{true}
\end{sloppypar}

\begin{proof}
Assume for a moment that in Def.~\ref{def:optimizedresolutiongraphupdate} no edges are excluded in the update of the optimized resolution graph.
In that case unsatisfiability of $\setclausesuc$ has been established in Thm.~\ref{thm:coreisunsat}.
In the remainder of the proof we show that for constructing a \uc of $\setclauses$ in SNF via \tr the optimized resolution graph can be used in place of the resolution graph, i.e., that none of the four exclusions of edges in Def.~\ref{def:optimizedresolutiongraphupdate} render the resulting \uc of $\setclauses$ in SNF satisfiable.

We have to show that
\begin{inparaenum}[(i)]
\item not including an edge from the vertex labeled with premise 1 to the vertex labeled with the conclusion for an application of rule \refaugtwo, \label{enum:thmunsatisfiabilityofcore:1}
\item not including an edge from the vertex labeled with premise 2 to the vertex labeled with the conclusion for an application of rule \refloopconclusiontwo, \label{enum:thmunsatisfiabilityofcore:2}
\item not including an edge from the vertex labeled with premise 2 to the vertex labeled with the conclusion for an application of rule \refloopitinitc, and \label{enum:thmunsatisfiabilityofcore:3}
\item not including an edge from the vertex labeled with premise 1 to the vertex labeled with the conclusion for an application of rule \refloopitinitc \label{enum:thmunsatisfiabilityofcore:4}
\end{inparaenum}
in the resolution graph $\graph$ maintains the fact that the resulting $\setclausesuc$ is unsatisfiable.

To see the intuition behind (\ref{enum:thmunsatisfiabilityofcore:1}) note that for a vertex $\vertex_\clause$ labeled with the conclusion $\clause$ of an application of rule \refaugtwo in the main partition $\mainpartitionv$ to be backward reachable from the (unique) vertex in the main partition $\mainpartitionv$ of the resolution graph $\graph$ labeled with the empty clause $\emptyclause$, $\vertex_\emptyclause$, the occurrence of $\fbnot{\fapwaitfor{\apres}}$ in $\clause$ must be ``resolved away'' at some point on the path from $\vertex_\clause$ to $\vertex_\emptyclause$. It turns out that this can only happen by resolution with a clause that is derived from the conclusion of rule \refaugone applied to an eventuality clause $\clausep$ with eventuality literal $\apres$. By the construction of the resolution graph $\graph$ $\vertex_\clausep$ must be backward reachable from $\vertex_\emptyclause$ and, therefore, $\clausep$ must be included in the \uc in SNF. Hence, an execution of \refalgltlsattrp with input $\setclausesuc$ will produce $\clause$ from $\clausep$. For a formal proof see Lemma \ref{thm:nobacklinkaugtwo}.

A similar reasoning as for (\ref{enum:thmunsatisfiabilityofcore:1}) applies to (\ref{enum:thmunsatisfiabilityofcore:2}), formalized in Lemma \ref{thm:noevbacklinkloopconclusiontwo}.

\begin{sloppypar}
For (\ref{enum:thmunsatisfiabilityofcore:3}) notice that a vertex labeled with the conclusion of an application of rule \refloopitinitc can only be backward reachable from $\vertex_\emptyclause$ if the corresponding BFS loop search iteration is successful and a vertex labeled with one of the resulting conclusions of rules \refloopconclusionone or \refloopconclusiontwo is backward reachable from $\vertex_\emptyclause$. The latter fact implies that an eventuality clause with the same eventuality literal as in premise 2 of rule \refloopitinitc is present in the \uc in SNF. Hence, an execution of \refalgltlsattrp with input $\setclausesuc$ will produce premise 2 of \refloopitinitc as required. This is formally proven in Lemma \ref{thm:noevbacklinkloopinitc}.
\end{sloppypar}

Finally, (\ref{enum:thmunsatisfiabilityofcore:4}) corresponds to considering only the last iteration of a successful loop search to obtain the \uc $\setclausesuc$. (\ref{enum:thmunsatisfiabilityofcore:4}) is obtained by understanding that in a BFS loop search iteration the premises 1 of rule \refloopitinitc essentially constitute a hypothetical fixed point; if the BFS loop search iteration is successful, then the hypothetical fixed point is proven to be an actual fixed point. For the correctness of a proof of the unsatisfiability of $\setclauses$ it is only relevant that this hypothetical fixed point is shown to be an actual fixed point but not how the hypothesis is obtained. This is formalized in Lemma \ref{thm:noloopitbacklinkloopinitc}. This concludes the proof.
\qed
\end{proof}

\begin{lemma}\label{thm:nobacklinkaugtwo}
Let $\setclauses$ be an unsatisfiable set of SNF clauses, let $\graph$ be an optimized resolution graph, and let $\graphp$ be the subgraph according to Def.~\ref{def:coreextractionopt}. Let $\vertexzero$ be a vertex in $\graphp$ labeled with a clause $\clausezero = \fgnxclause{(\fbnot{\fapwaitfor{\apres}})}{\fbor{\apres}{\fapwaitfor{\apres}}}$ created by augmentation \refaugtwo from some eventuality clause $\fgneclause{\fbor{\ap_1}{\fbor{\ldots}{\ap_\maxind}}}{\apres} \in \setclauses$ with eventuality literal $\apres$. Then there is a vertex $\vertexone$ in $\graphp$ labeled with an eventuality clause $\clauseone = \fgneclause{\fbor{\app_1}{\fbor{\ldots}{\app_\maxindp}}}{\apres}$ $\in \setclauses$ with eventuality literal $\apres$.
\end{lemma}

\begin{proof}
There exists a path $\ppath$ of non-zero length in $\graphp$ from $\vertexzero$ to the unique vertex $\vertex_\emptyclause$ in the main partition $\mainpartition$ labeled with the empty clause $\emptyclause$.
On the path $\ppath$ there exist two vertices $\vertextwo, \vertexthree$ such that $\vertextwo$ is labeled with a clause $\clausetwo$ that contains $\fbnot{\fapwaitfor{\apres}}$ or $\fnext{\fbnot{\fapwaitfor{\apres}}}$, while $\vertexthree$ and all of its successors on $\ppath$ are labeled with clauses that contain neither $\fbnot{\fapwaitfor{\apres}}$ nor $\fnext{\fbnot{\fapwaitfor{\apres}}}$.
Let $\clausethree$ be the clause labeling $\vertexthree$.
\begin{itemize}
\item \emph{Case 1.} $\clausethree$ is generated by initial or step resolution \refinitii, \refinitin, \refstepnn, \refstepnx, or \refstepxx from $\clausetwo$ and some other clause $\clausefour$. $\clausefour$ must contain $\fapwaitfor{\apres}$ or $\fnext{\fapwaitfor{\apres}}$. Moreover, there must be a path $\ppathp$ (possibly of zero length) that starts from a vertex $\vertexfive$ labeled with a clause $\clausefive$ and that ends in the vertex $\vertexfour$ labeled with $\clausefour$, such that each vertex on the path $\ppathp$ is labeled with a clause that contains $\fapwaitfor{\apres}$ or $\fnext{\fapwaitfor{\apres}}$. Finally, $\fapwaitfor{\apres}$ or $\fnext{\fapwaitfor{\apres}}$ must be present in $\clausefive$ either because $\clausefive$ is contained in the set of input clauses in SNF, $\setclauses$, or because $\clausefive$ is generated by some production rule that introduces $\fapwaitfor{\apres}$ or $\fnext{\fapwaitfor{\apres}}$ in the conclusion.
\begin{itemize}
\item \emph{Case 1.1.} $\clausefive$ is contained in the set of input clauses in SNF, $\setclauses$. Impossible: $\fapwaitfor{\apres}$ is a fresh proposition in \refaugone and \refaugtwo.
\item \emph{Case 1.2.} $\clausefive$ is generated by initial or step resolution \refinitii, \refinitin, \refstepnn, \refstepnx, or \refstepxx. Impossible: initial and step resolution do not generate literals that are not contained (modulo time-shifting) in at least one of the premises.
\item \emph{Case 1.3.} $\clausefive$ is generated by augmentation 1 \refaugone. By the construction of the resolution graph $\graph$ and the subgraph $\graphp$ there is an edge in $\graphp$ from a vertex $\vertexone$ in $\graphp$ labeled with an eventuality clause $\clauseone = \fgneclause{\fbor{\app_1}{\fbor{\ldots}{\app_\maxindp}}}{\apres} \in \setclauses$ with eventuality literal $\apres$ to $\vertexfive$.
\item \emph{Case 1.4.} $\clausefive$ is generated by augmentation 2 \refaugtwo, i.e., $\clausefive = \clausezero$. This introduces another occurrence of $\fbnot{\fapwaitfor{\apres}}$ to be ``resolved away''. Note that in the main partition only new clauses are generated from existing ones with edges leading from existing vertices labeled with existing clauses to new vertices labeled with new clauses. Therefore, the main partition of $\graphp$ is a finite directed acyclic graph, and this case cannot happen infinitely often.
\item \emph{Case 1.5.} $\clausefive$ is generated by BFS loop search initialization \refloopitinitx. Impossible: the production rule \refloopitinitx copies a clause verbatim. I.e., it cannot be the case that $\clausefive$ contains $\fapwaitfor{\apres}$ or $\fnext{\fapwaitfor{\apres}}$, while the premise does not.
\item \emph{Case 1.6.} $\clausefive$ is generated by BFS loop search initialization \refloopitinitn. Impossible: the production rule \refloopitinitn copies and time-shifts a clause. I.e., it cannot be the case that $\clausefive$ contains $\fnext{\fapwaitfor{\apres}}$, while the premise does not contain $\fapwaitfor{\apres}$.
\item \emph{Case 1.7.} $\clausefive$ is generated by BFS loop search initialization \refloopitinitc. Impossible: the production rule \refloopitinitc copies and time-shifts a clause from a previous BFS loop search iteration (or initializes with the empty clause $\emptyclause$) and disjoins with an eventuality literal $\fnext{\apresp}$. I.e., it cannot be the case that $\clausefive$ contains $\fnext{\fapwaitfor{\apres}}$, while the premise does not contain $\fapwaitfor{\apres}$.
\item \emph{Case 1.8.} $\vertexfive$ is linked to via BFS loop search subsumption \refloopitsub. This case can be ignored as BFS loop search subsumption \refloopitsub does not actually generate a clause but merely links existing ones.
\item \emph{Case 1.9.} $\clausefive$ is generated by BFS loop search conclusion 1 \refloopconclusionone. Impossible: production rule \refloopconclusionone copies all literals verbatim from a clause derived in loop search, copies all literals verbatim from an eventuality clause except for the eventuality literal $\apresp$ prefixed by $\ffinallyname$, and disjoins with the eventuality literal $\apresp$. I.e., it cannot be the case that $\clausefive$ contains $\fapwaitfor{\apres}$, while the premises do not.
\item \emph{Case 1.10.} $\clausefive$ is generated by BFS loop search conclusion 2 \refloopconclusiontwo. Impossible: production rule \refloopconclusiontwo copies and time-shifts all literals from a clause $\clausesix$ derived in loop search and disjoins with $\fbnot{\fapwaitfor{\apresp}}$ and $\fnext{\apresp}$ for some eventuality literal $\apresp$. I.e., it cannot be the case that $\clausefive$ contains $\fnext{\fapwaitfor{\apres}}$, while the  premise $\clausesix$ does not contain \fapwaitfor{\apres}.
\end{itemize}
\item \emph{Case 2.} $\clausethree$ is generated by augmentation \refaugone or \refaugtwo. Impossible: the premise of the production rules \refaugone and \refaugtwo cannot contain either $\fbnot{\fapwaitfor{\apres}}$ or $\fnext{\fbnot{\fapwaitfor{\apres}}}$, as $\fapwaitfor{\apres}$ is assumed to be a fresh proposition in \refaugone and \refaugtwo.
\item \emph{Case 3.} $\clausethree$ is generated by BFS loop search initialization \refloopitinitx. Impossible: the production rule \refloopitinitx copies a clause verbatim. I.e., it cannot be the case that $\clausetwo$ contains $\fbnot{\fapwaitfor{\apres}}$ or $\fnext{\fbnot{\fapwaitfor{\apres}}}$, while $\clausethree$ does not.
\item \emph{Case 4.} $\clausethree$ is generated by BFS loop search initialization \refloopitinitn. Impossible: the production rule \refloopitinitn copies and time-shifts a clause. I.e., it cannot be the case that $\clausetwo$ contains $\fbnot{\fapwaitfor{\apres}}$, while $\clausethree$ does not contain $\fnext{\fbnot{\fapwaitfor{\apres}}}$.
\item \emph{Case 5.} $\clausethree$ is generated by BFS loop search initialization \refloopitinitc. Impossible: the production rule \refloopitinitc copies and time-shifts a clause from a previous BFS loop search iteration (or initializes with the empty clause $\emptyclause$) and disjoins with an eventuality literal $\fnext{\apresp}$. I.e., it cannot be the case that $\clausetwo$ contains $\fbnot{\fapwaitfor{\apres}}$, while $\clausethree$ does not contain $\fnext{\fbnot{\fapwaitfor{\apres}}}$.
\item \emph{Case 6.} $\vertextwo$ and $\vertexthree$ are linked via BFS loop search subsumption \refloopitsub, i.e., a time-shifted version of $\clausetwo$ subsumes $\clausethree$. Impossible: \refloopitsub links from a clause with fewer literals to a clause with (modulo time-shifting) the same and more literals. I.e., it cannot be the case that $\clausetwo$ contains $\fbnot{\fapwaitfor{\apres}}$, while $\clausethree$ does not contain $\fnext{\fbnot{\fapwaitfor{\apres}}}$.
\item \emph{Case 7.} $\clausethree$ is generated by BFS loop search conclusion 1 \refloopconclusionone. Impossible: production rule \refloopconclusionone copies all literals verbatim from a clause derived in loop search, copies all literals verbatim from an eventuality clause except for the eventuality literal $\apresp$ prefixed by $\ffinallyname$, and disjoins with the eventuality literal $\apresp$. I.e., it cannot be the case that $\clausetwo$ contains $\fbnot{\fapwaitfor{\apres}}$, while $\clausethree$ does not.
\item \emph{Case 8.} $\clausethree$ is generated by BFS loop search conclusion 2 \refloopconclusiontwo. Impossible: production rule \refloopconclusiontwo copies and time-shifts all literals from a clause derived in loop search and disjoins with $\fbnot{\fapwaitfor{\apresp}}$ and $\fnext{\apresp}$ for some eventuality literal $\apresp$. I.e., it cannot be the case that $\clausetwo$ contains $\fbnot{\fapwaitfor{\apres}}$, while $\clausethree$ does not contain $\fnext{\fbnot{\fapwaitfor{\apres}}}$.
\end{itemize}
Notice that the only possible cases are case 1.3 and 1.4. Of those, case 1.4 can only happen a finite number of times and must be followed by an occurrence of case 1.3. This concludes the proof.
\qed
\end{proof}

\begin{lemma}\label{thm:noevbacklinkloopconclusiontwo}
Let $\setclauses$ be an unsatisfiable set of SNF clauses, let $\graph$ be an optimized resolution graph, and let $\graphp$ be the subgraph according to Def.~\ref{def:coreextractionopt}. Let $\vertex$ be a vertex in $\graphp$ labeled with a clause $\clause = \fgnxclause{(\fbnot{\fapwaitfor{\apres}})}{\fbor{\fbor{\app_1}{\fbor{\ldots}{\app_\maxindp}}}{\apres}}$ generated by BFS loop search conclusion 2 \refloopconclusiontwo from some eventuality clause $\fgneclause{\fbor{\ap_1}{\fbor{\ldots}{\ap_\maxind}}}{\apres} \in \setclauses$ with eventuality literal $\apres$ (and some other clause). Then there is a vertex $\vertexpp$ in $\graphp$ labeled with an eventuality clause $\clausepp = \fgneclause{\fbor{\appp_1}{\fbor{\ldots}{\appp_\maxindpp}}}{\apres} \in \setclauses$ with eventuality literal $\apres$.
\end{lemma}

\begin{proof}
Analogous to the proof of Lemma \ref{thm:nobacklinkaugtwo}.
\qed
\end{proof}

\begin{lemma}\label{thm:noevbacklinkloopinitc}
Let $\setclauses$ be an unsatisfiable set of SNF clauses, let $\graph$ be an optimized resolution graph, and let $\graphp$ be the subgraph according to Def.~\ref{def:coreextractionopt}. Let $\vertex$ be a vertex in $\graphp$ labeled with a clause $\clause = \fgxclause{\fbor{\app_1}{\fbor{\ldots}{\fbor{\app_\maxindp}{\apres}}}}$ generated by production rule \refloopitinitc from some eventuality clause $\fgneclause{\fbor{\ap_1}{\fbor{\ldots}{\ap_\maxind}}}{\apres} \in \setclauses$ with eventuality literal $\apres$ (and some other clause). Then there is a vertex $\vertexpp$ in $\graphp$ labeled with an eventuality clause $\clausepp = \fgneclause{\fbor{\appp_1}{\fbor{\ldots}{\appp_\maxindpp}}}{\apres} \in \setclauses$ with eventuality literal $\apres$.
\end{lemma}

\begin{proof}
By the construction of the optimized resolution graph $\graph$ (Def.~\ref{def:resolutiongraphtype}, \ref{def:resolutiongraphinitialization}, \ref{def:optimizedresolutiongraphupdate}) and its subgraph $\graphp$ (Def.~\ref{def:coreextractionopt}) $\vertex$ is included in $\graphp$ only if $\graphp$ also includes some vertex $\vertexp$ labeled with some clause $\clausep$ such that $\clausep$ was generated by BFS loop search conclusion \refloopconclusionone or \refloopconclusiontwo from the BFS loop search iteration of which $\clause$ is part.
\begin{itemize}
\item \emph{Case 1.} $\clausep$ is generated by BFS loop search conclusion 1 \refloopconclusionone. The claim follows from the construction of the optimized resolution graph $\graph$ and its subgraph $\graphp$. By Def.~\ref{def:optimizedresolutiongraphupdate} $\vertexp$ has an incoming edge from a vertex $\vertexpp$ labeled with an eventuality clause $\clausepp = \fgneclause{\fbor{\appp_1}{\fbor{\ldots}{\appp_\maxindpp}}}{\apres} \in \setclauses$ with eventuality literal $\apres$ and by Def.~\ref{def:coreextractionopt} $\vertexpp$ is included in $\graphp$ if $\vertexp$ is included.
\item \emph{Case 2.} $\clausep$ is generated by BFS loop search conclusion 2 \refloopconclusiontwo. In that case the claim follows directly from Lemma \ref{thm:noevbacklinkloopconclusiontwo}.
\end{itemize}
This concludes the proof.
\qed
\end{proof}

\begin{lemma}\label{thm:noloopitbacklinkloopinitc}
Let $\setclauses$ be a set of SNF clauses, assume an execution of \refalgltlsattrp on $\setclauses$, and let $\setclausesp \definedas \{\fgnclause{\fbor{\app_{\posp,1}}{\fbor{\ldots}{\app_{\posp,\maxindp_\posp}}}} \mid 1 \le \posp \le \maxindp\}$ and $\setclausespp \definedas \{\fgxclause{\fbor{\fbor{\ap_{\pos,1}}{\fbor{\ldots}{\ap_{\pos,\maxind_\pos}}}}{\apres}} \mid 1 \le \pos \le \maxind\}$ be the sets of clauses obtained in line \ref{line:fig:ltlsattrp:loopitchecksub1} and line \ref{line:fig:ltlsattrp:loopitchecksub2} of \refalgltlsattrp in the last iteration of a successful loop search for eventuality literal $\apres$.
Then, assuming $\setclauses$, $\{(\fglobally{\fbor{(\fbor{\app_{\posp,1}}{\fbor{\ldots}{\app_{\posp,\maxindp_1}}}}{\fnext{\fglobally{\fbnot{\apres}}}})}) \mid 1 \le \posp \le \maxindp\}$ is provable.
\end{lemma}

\begin{proof}
After initialization of a BFS loop search iteration in line \ref{line:fig:ltlsattrp:loopitinit} of \refalgltlsattrp there are three sets of clauses according to the three production rules for initializing a BFS loop search iteration. Clauses generated by \refloopitinitx and \refloopitinitn are (partly time-shifted) duplicates of clauses derived so far in the main partition. \refloopitinitc generates the set of clauses $\setclausespp$. From these three sets saturation restricted to rule \refstepxx in line \ref{line:fig:ltlsattrp:loopitsat} derives another set of clauses, $\setclausesp$. Taking the restriction of saturation to rule \refstepxx into account, that BFS loop search iteration has established that, assuming $\setclauses$, the following fact is provable:
{
\begin{equation}
\fglobally{(\fbimplies{(\bigwedge_{1 \le \pos \le \maxind} \fnext{(\fbor{\fbor{\ap_{\pos,1}}{\fbor{\ldots}{\ap_{\pos,\maxind_\pos}}}}{\apres})})}{\bigwedge_{1 \le \posp \le  \maxindp} (\fbor{\app_{\posp,1}}{\fbor{\ldots}{\app_{\posp,\maxindp_\posp}}}))}}\label{eqn:thmunsatisfiabilityofcore1:app}.
\end{equation}
}
Moreover, because for a successful BFS loop search iteration subsumption in line \ref{line:fig:ltlsattrp:loopitchecksub3} succeeds, the following fact is also provable:
{
\begin{equation}
\bigwedge_{1 \le \pos \le \maxind} \; \bigvee_{1 \le \posp \le \maxindp} \fglobally{(\fbimplies{(\fbor{\app_{\posp,1}}{\fbor{\ldots}{\app_{\posp,\maxindp_\posp}}})}{(\fbor{\ap_{\pos,1}}{\fbor{\ldots}{\ap_{\pos,\maxind_\pos}}})})}\label{eqn:thmunsatisfiabilityofcore2:app}.
\end{equation}
}
We rewrite \eqref{eqn:thmunsatisfiabilityofcore1:app} and \eqref{eqn:thmunsatisfiabilityofcore2:app} as follows:
{
\begin{align}
& \;\;\fglobally{(\fbimplies{(\bigwedge_{1 \le \pos \le \maxind} \fnext{(\fbor{\fbor{\ap_{\pos,1}}{\fbor{\ldots}{\ap_{\pos,\maxind_\pos}}}}{\apres})})}{\bigwedge_{1 \le \posp \le  \maxindp} (\fbor{\app_{\posp,1}}{\fbor{\ldots}{\app_{\posp,\maxindp_\posp}}})})} \nonumber \\
\proofbiimplies & \;\;\fglobally{\bigwedge_{1 \le \posp \le  \maxindp} (\fbimplies{(\bigwedge_{1 \le \pos \le \maxind} \fnext{(\fbor{\fbor{\ap_{\pos,1}}{\fbor{\ldots}{\ap_{\pos,\maxind_\pos}}}}{\apres})})}{(\fbor{\app_{\posp,1}}{\fbor{\ldots}{\app_{\posp,\maxindp_\posp}}})})} \nonumber \\
\proofbiimplies & \;\;\bigwedge_{1 \le \posp \le  \maxindp} \fglobally{(\fbimplies{(\bigwedge_{1 \le \pos \le \maxind} \fnext{(\fbor{\fbor{\ap_{\pos,1}}{\fbor{\ldots}{\ap_{\pos,\maxind_\pos}}}}{\apres})})}{(\fbor{\app_{\posp,1}}{\fbor{\ldots}{\app_{\posp,\maxindp_\posp}}})})} \nonumber \\
\proofbiimplies & \;\;\bigwedge_{1 \le \posp \le  \maxindp} \fglobally{(\fbimplies{(\fbnot{(\fbor{\app_{\posp,1}}{\fbor{\ldots}{\app_{\posp,\maxindp_\posp}}})})}{\fbnot{\bigwedge_{1 \le \pos \le \maxind}\fnext{(\fbor{\fbor{\ap_{\pos,1}}{\fbor{\ldots}{\ap_{\pos,\maxind_\pos}}}}{\apres})}}})} \nonumber \\
\proofbiimplies & \;\;\bigwedge_{1 \le \posp \le  \maxindp} \fglobally{(\fbimplies{(\fbnot{(\fbor{\app_{\posp,1}}{\fbor{\ldots}{\app_{\posp,\maxindp_\posp}}})})}{\bigvee_{1 \le \pos \le \maxind} \fnext{\fbnot{(\fbor{\ap_{\pos,1}}{\fbor{\fbor{\ldots}{\ap_{\pos,\maxind_\pos}}}{\apres}})}}})} \nonumber \\
\proofbiimplies & \;\;\bigwedge_{1 \le \posp \le  \maxindp} \fglobally{(\fbimplies{(\fbnot{(\fbor{\app_{\posp,1}}{\fbor{\ldots}{\app_{\posp,\maxindp_\posp}}})})}{\bigvee_{1 \le \pos \le \maxind} \fnext{(\fband{(\fbnot{(\fbor{\ap_{\pos,1}}{\fbor{\ldots}{\ap_{\pos,\maxind_\pos}}})})}{\fbnot{\apres}})}})} \nonumber \\
\proofbiimplies & \;\bigwedge_{1 \le \posp \le  \maxindp} \fglobally{(\fbimplies{(\fbnot{(\fbor{\app_{\posp,1}}{\fbor{\ldots}{\app_{\posp,\maxindp_\posp}}})})}{(\fband{(\fnext{\fbnot{\apres}})}{\bigvee_{1 \le \pos \le \maxind}\fnext{\fbnot{(\fbor{\ap_{\pos,1}}{\fbor{\ldots}{\ap_{\pos,\maxind_\pos}}})}}})})} \label{eqn:thmunsatisfiabilityofcore3:app},
\end{align}
}
{
\begin{eqnarray}
& & \bigwedge_{1 \le \pos \le \maxind} \bigvee_{1 \le \posp \le \maxindp} \fglobally{(\fbimplies{(\fbor{\app_{\posp,1}}{\fbor{\ldots}{\app_{\posp,\maxindp_\posp}}})}{(\fbor{\ap_{\pos,1}}{\fbor{\ldots}{\ap_{\pos,\maxind_\pos}}})})} \nonumber \\
& \proofbiimplies & \bigwedge_{1 \le \pos \le \maxind} \bigvee_{1 \le \posp \le \maxindp} \fglobally{(\fbimplies{(\fbnot{(\fbor{\ap_{\pos,1}}{\fbor{\ldots}{\ap_{\pos,\maxind_\pos}}})})}{\fbnot{(\fbor{\app_{\posp,1}}{\fbor{\ldots}{\app_{\posp,\maxindp_\posp}}})}})} \label{eqn:thmunsatisfiabilityofcore4:app}.
\end{eqnarray}
}
Putting \eqref{eqn:thmunsatisfiabilityofcore3:app} and \eqref{eqn:thmunsatisfiabilityofcore4:app} together, we obtain \eqref{eqn:thmunsatisfiabilityofcore5:app}, which is exactly the premise required to perform eventuality resolution with an eventuality clause with eventuality literal $\apres$ \cite{MFisherCDixonMPeim-ACMTrComputationalLogic-2001}:
{
\begin{equation}
\label{eqn:thmunsatisfiabilityofcore5:app}
\bigwedge_{1 \le \posp \le \maxindp} \fglobally{(\fbor{\fbor{\app_{\posp,1}}{\fbor{\ldots}{\app_{\posp,\maxindp_\posp}}}}{\fnext{\fglobally{\fbnot{\apres}}}})}.\\
\end{equation}
}
This concludes the proof.
\qed
\end{proof}

\mycorollary%
{\thmcoreisunsatoptcor}%
{thm:coreisunsatoptcor}%
{\UC in SNF via \tr from an Optimized Resolution Graph is \UC in SNF}%
{Let $\setclauses$ be an unsatisfiable set of SNF clauses, and let $\setclausesuc$ be a \uc of $\setclauses$ in SNF via \tr from an optimized resolution graph. Then $\setclausesuc$ is a \uc of $\setclauses$ in SNF.}
\thmcoreisunsatoptcor{true}

Theorem \ref{thm:coreisunsatopt} shows that not including some premises during the update of the optimized resolution graph still leads to a \uc.
It does not discuss whether the remaining premises actually need to be included to guarantee a \uc.
For all premises of all production rules not excluded in the update of the optimized resolution graph in Def.~\ref{def:optimizedresolutiongraphupdate} it turns out that they are indeed required to obtain a \uc.
In other words, for any of the premises not excluded in the update of the optimized resolution graph in Def.~\ref{def:optimizedresolutiongraphupdate} there exists a set of SNF clauses $\setclauses$ such that excluding that particular premise from the update of the optimized resolution graph for $\setclauses$ in Def.~\ref{def:optimizedresolutiongraphupdate} and from the subsequent extraction of a \uc in SNF in Def.~\ref{def:coreextractionopt} leads to a satisfiable rather than unsatisfiable set of clauses $\setclausesuc \subseteq \setclauses$.

\myproposition%
{\thmminimalitypremises}%
{thm:minimalitypremises}%
{Minimality of Set of Premises to Include in Optimized Resolution Graph}%
{The set of premises included in the update of the optimized resolution graph in Def.~\ref{def:optimizedresolutiongraphupdate} is minimal.
}%
\begin{sloppypar}
\thmminimalitypremises{true}
\end{sloppypar}

\begin{proof}
The proof is obtained by providing suitable examples. For a given premise \premise of some production rule \refinferencerule{rule} with conclusion \conclusion the need for inclusion of edges between instances of \premise and \conclusion induced by \refinferencerule{rule} in the optimized resolution graph to obtain a \uc can be established as follows.\footnote{Note that this proof assumes that there are no edges between instances of the premise of \refaugtwo, the premises of \refloopitinitc, and premise 2 of \refloopconclusiontwo and their conclusions induced by these production rules as indicated in Def.~\ref{def:optimizedresolutiongraphupdate}. I.e., different minimal sets of premises to include in an optimized resolution graph may exist.}
Let $\setclauses$ be an unsatisfiable set of SNF clauses, let $\graph$ be an optimized resolution graph, let $\vertex_\emptyclause$ be the (unique) vertex in the main partition $\mainpartitionv$ of the optimized resolution graph $\graph$ labeled with the empty clause $\emptyclause$, and let $\setclausesuc$ be the \uc of $\setclauses$ in SNF via \tr from an optimized resolution graph.
Assume that $\setclausesuc$ is a minimal \uc, i.e., for any $\clausep \in \setclausesuc$ we have that $\setclausesuc \setminus \{\clausep\}$ is satisfiable.
Now, if removing all edges between instances of \premise and \conclusion induced by \refinferencerule{rule} from $\graph$ makes a vertex $\vertex_\clausep$ in $\mainpartitionv$ labeled with $\clausep \in \setclausesuc$ backward unreachable from $\vertex_\emptyclause$, then the \uc obtained without including edges between instances of \premise and \conclusion induced by \refinferencerule{rule} clearly would not be unsatisfiable.

Hence, for all premises of all production rules not excluded in the update of the optimized resolution graph in Def.~\ref{def:optimizedresolutiongraphupdate} we need to provide triples of premises \premise of production rules \refinferencerule{rule}, minimally unsatisfiable SNFs $\setclauses \equiv \setclausesuc$, and subgraphs $\graphp$ of optimized resolution graphs with the following properties.
\begin{inparaenum}[(i)]
\item $\graphp$ is the subgraph according to Def.~\ref{def:coreextractionopt} for $\setclauses$.
\item There exists a vertex $\vertex_\clausep$ in $\graphp$ labeled with $\clausep \in \setclauses$ and an edge $\edge$ that is an instance of \premise and \conclusion induced by \refinferencerule{rule} such that removal of $\edge$ from $\graphp$ makes $\vertex_\clausep$ backward unreachable from $\vertex_\emptyclause$.
\end{inparaenum}
In the graphs below $\edge$ and $\vertex_\clausep$ are marked dashed, blue. The vertex labels use \trp syntax.

Note that the search for such triples can be supported by a corresponding modification to the temporal resolution solver. For the more complex cases candidates were obtained in that way and then optimized by hand.

Below we provide triples for \refinitii and \refloopitinitx.
The triples for \refinitin, \refstepnn, \refstepnx, \refstepxx are trivially obtained as the one for \refinitii.
The triple for \refloopitinitx also serves as triple for premises 1 and 2 of \refloopconclusionone.
The remaining triples are more complex. They can be found in \refformalcoreextraction.

{
\newcommand{\tablineone}[2]{
\begin{samepage}
\noindent
\begin{tabular}{lc}
\hline
\\[-0.5ex]
\begin{minipage}[c]{0.65\linewidth}
{#1}
\end{minipage}
&
\begin{minipage}[c]{0.25\linewidth}
\begin{center}
{#2}
\end{center}
\end{minipage}
\end{tabular}
}

\newcommand{\tablinethree}{
\\
\begin{center}
}

\newcommand{\tablinefour}{
\end{center}
\vspace{0.75ex}
\end{samepage}

}

\vspace{2.5ex}
\tablineone
{\refinitii}
{$\{\ficlause{a}, \ficlause{\fbnot{a}}\}$}
\tablinethree
\adjustbox{scale=0.67,keepaspectratio=true}{
\iftrue
\begin{tikzpicture}[>=latex,line join=bevel,scale=0.75]
\node (1) at (132bp,18bp) [draw,ellipse] {or([not a])};
  \node (0) at (33bp,18bp) [draw=blue,ellipse,ultra thick, densely dashed, text=blue] {or([a])};
  \node (2) at (82bp,106bp) [draw,ellipse] {or([])};
  \draw [->] (1) ..controls (114.79bp,48.6bp) and (104.68bp,65.988bp)  .. (2);
  \definecolor{strokecol}{rgb}{0.0,0.0,0.0};
  \pgfsetstrokecolor{strokecol}
  \draw (126.5bp,62bp) node {init-ii};
  \draw [blue,->,ultra thick,densely dashed] (0) ..controls (49.611bp,48.154bp) and (59.577bp,65.646bp)  .. (2);
  \draw (77.5bp,62bp) node[blue] {init-ii};
\end{tikzpicture}
\else
\begin{dot2tex}[options={--cache --graphstyle "scale=0.75"}]
digraph G {
graph [rankdir="BT"];
0[label="or([a])",color=blue,style="ultra thick, densely dashed, text=blue"];
1[label="or([not a])"];
2[label="or([])"];
1->2 [label="init-ii"];
0->2 [label="init-ii",color=blue,lblstyle="blue",style="ultra thick, densely dashed"];
}
\end{dot2tex}
\fi
}
\tablinefour

\tablineone
{\refloopitinitx}
{$
\begin{array}{l}
\hspace{-0.5em}\{\ficlause{a},\\
(\fglobally{(\fbor{(\fbnot{a})}{\fnext{a}})}),\\
\fgeclause{\fbnot{a}}\}
\end{array}
$}
\tablinethree
\adjustbox{scale=0.67,keepaspectratio=true}{
\iftrue
\begin{tikzpicture}[>=latex,line join=bevel,scale=0.75]
\node (11) at (140bp,370bp) [draw,ellipse] {or([])};
  \node (10) at (222bp,282bp) [draw,ellipse] {always(or([not a]))};
  \node (1) at (423bp,18bp) [draw=blue,ellipse,ultra thick, densely dashed, text=blue] {always(or([not a,next(a)]))};
  \node (0) at (93bp,282bp) [draw,ellipse] {or([a])};
  \node (2) at (118bp,194bp) [draw,ellipse] {always(or([sometime(not a)]))};
  \node (5) at (423bp,106bp) [draw,ellipse] {always(or([not a,next(a)]))};
  \node (7) at (471bp,282bp) [draw,ellipse] {always(or([next(not a)]))};
  \node (8) at (423bp,194bp) [draw,ellipse] {always(or([not a]))};
  \draw [->] (7) ..controls (509.94bp,254.77bp) and (520.19bp,242.02bp)  .. (512bp,230bp) .. controls (506.32bp,221.66bp) and (498.33bp,215.31bp)  .. (8);
  \definecolor{strokecol}{rgb}{0.0,0.0,0.0};
  \pgfsetstrokecolor{strokecol}
  \draw (534bp,238bp) node {step-xx};
  \draw [->] (5) ..controls (423bp,136.25bp) and (423bp,152.18bp)  .. (8);
  \draw (443bp,150bp) node {step-xx};
  \draw [->] (0) ..controls (108.93bp,312.15bp) and (118.49bp,329.65bp)  .. (11);
  \draw (139.5bp,326bp) node {init-ini};
  \draw [->] (8) ..controls (335bp,209.81bp) and (306.78bp,217.9bp)  .. (283bp,230bp) .. controls (269.47bp,236.88bp) and (256.31bp,247.26bp)  .. (10);
  \draw (348bp,238bp) node {BFS-loop-conclusion1g};
  \draw [blue,->,ultra thick,densely dashed] (1) ..controls (423bp,48.254bp) and (423bp,64.183bp)  .. (5);
  \draw (472.5bp,62bp) node[blue] {BFS-loop-it-init-x};
  \draw [->] (10) ..controls (193.12bp,313.29bp) and (175.37bp,331.91bp)  .. (11);
  \draw (209bp,326bp) node {init-inn};
  \draw [->] (2) ..controls (112.05bp,222.9bp) and (112bp,236.42bp)  .. (119bp,246bp) .. controls (126.35bp,256.04bp) and (136.73bp,263.24bp)  .. (10);
  \draw (183.5bp,238bp) node {BFS-loop-conclusion1e};
  \draw [->] (8) ..controls (415.51bp,222.41bp) and (414.55bp,235.64bp)  .. (420bp,246bp) .. controls (422.44bp,250.63bp) and (425.8bp,254.77bp)  .. (7);
  \draw (464bp,238bp) node {BFS-loop-it-sub};
\end{tikzpicture}
\else
\begin{dot2tex}[options={--cache --graphstyle "scale=0.75"}]
digraph G {
graph [rankdir="BT"];
0[label="or([a])"];
1[label="always(or([not a,next(a)]))",color=blue,style="ultra thick, densely dashed, text=blue"];
2[label="always(or([sometime(not a)]))"];
5[label="always(or([not a,next(a)]))"];
7[label="always(or([next(not a)]))"];
8[label="always(or([not a]))"];
10[label="always(or([not a]))"];
11[label="or([])"];
1->5 [label="BFS-loop-it-init-x",color=blue,lblstyle="blue",style="ultra thick, densely dashed"];
7->8 [label="step-xx"];
5->8 [label="step-xx"];
8->7 [label="BFS-loop-it-sub"];
8->10 [label="BFS-loop-conclusion1g"];
2->10 [label="BFS-loop-conclusion1e"];
10->11 [label="init-inn"];
0->11 [label="init-ini"];
}
\end{dot2tex}
\fi
}
\tablinefour

\noindent
\begin{tabular}{p{0.65\linewidth}p{0.25\linewidth}}
\hline
&\\
\end{tabular}
}

\noindent
This concludes the proof.
\qed
\end{proof}

Notice that using optimized resolution graphs has no impact on the added complexity of \uc extraction.
Hence, we state Prop.~\ref{thm:corecomplexityopt}.

\myproposition%
{\thmcorecomplexityopt}%
{thm:corecomplexityopt}%
{Added Complexity of \UC Extraction from an Optimized Resolution Graph}%
{Let $\setclauses$ be an unsatisfiable set of SNF clauses. Construction of $\setclausesuc$ according to Def.~\ref{def:coreextractionopt} can be performed in time exponential in $\fcardinality{\allaps} + log(\fcardinality{\setclauses})$ in addition to the time required to run \refalgltlsattrp.}
\thmcorecomplexityopt{true}

\myremark%
{\thmrgpruningopt}%
{thm:rgpruningopt}%
{Pruning the Optimized Resolution Graph at Run Time}%
{The fact that not all premises need to be included during the update of the optimized resolution graph permits the following optimization for an optimized resolution graph during the execution of \refalgltlsattrp extended with the construction in Def.~\ref{def:coreextractionopt}.
Note that, because in the optimized resolution graph there is no edge from a vertex labeled with premise 1 to a vertex labeled with the conclusion of an application of production rule \refloopitinitc, there are no outgoing edges from a failed loop search iteration (lines \ref{line:fig:ltlsattrp:loopitinit}--\ref{line:fig:ltlsattrp:loopitend} of \refalgltlsattrp).
Therefore, if a loop search iteration fails, all vertices and edges in the partition of that loop search iteration can be pruned from the optimized resolution graph right away.
Moreover, Rem.~\ref{thm:rgpruning} also applies to optimized resolution graphs.
}
\thmrgpruningopt{true}

In the next Subsec.~\ref{sec:extractingaucinltl} and in the following Sec.~\ref{sec:postprocessingucs} we only consider optimized resolution graphs and \ucs in SNF via \tr from optimized resolution graphs, and we drop the designators ``optimized'' and ``from an optimized resolution graph''.

\paragraph{Example}

In Fig.~\ref{fig:completeexresgraphucexopt} we return to our running example from Sec.~\ref{sec:tralgorithmexample} and \ref{sec:extractingaucinsnffromaresolutiongraph} to illustrate the construction of the optimized resolution graph and its use to extract a \uc in SNF via \tr.
The graph in Fig.~\ref{fig:completeexresgraphucexopt} shows the same, unoptimized resolution graph as in Fig.~\ref{fig:completeexresgraphucex}.
The edges that are excluded from the optimized resolution graph according to Def.~\ref{def:optimizedresolutiongraphupdate} are marked dotted, red.
The reader can easily verify that the dashed, blue edges and the solid, black edges in Fig.~\ref{fig:completeexresgraphucexopt} are not excluded in the update of the optimized resolution graph in Def.~\ref{def:optimizedresolutiongraphupdate}, while the dotted, red edges are in fact excluded there.
As in Fig.~\ref{fig:completeexresgraphucex} the dashed, blue clauses and edges show the part of the resolution graph that is backward reachable from the empty clause $\emptyclause$.
Notice that now no clause in the loop partition for the first loop search iteration (in the rectangle in the middle shaded in dark green), which was unsuccessful, is backward reachable from $\emptyclause$.
Still, the subset of the starting clauses (in the rectangle at the bottom shaded in light red) that is backward reachable from $\emptyclause$ is the same as when using the unoptimized resolution graph in Fig.~\ref{fig:completeexresgraphucex}.
Therefore, also the \uc in SNF via \tr from an optimized resolution graph according to Def.~\ref{def:coreextractionopt} is unchanged (see \eqref{ex:complete:snfcore}).

{
\setlength\mylinewidth{\the\linewidth}

\begin{sidewaysfigure}
\vspace{98ex}
\centering
\resizebox{!}{0.955\mylinewidth}{
\begin{tikzpicture}[auto,
    =<triangle 45,
    every state/.style={draw=none},
    corev/.style={draw=blue,rectangle,rounded corners=1mm,very thick,densely dashed,text=blue},
    coreea/.style={blue,very thick,densely dashed},
    coreel/.style={draw=none,very thick,text=blue},
    corenov/.style={},
    corenoea/.style={},
    corenoel/.style={},
    noncoreea/.style={},
    noncoreel/.style={},
    corenononrgea/.style={red,very thick,loosely dotted},
    corenononrgel/.style={draw=none,very thick,text=red},
    noncorenonrgea/.style={red,very thick,loosely dotted},
    noncorenonrgel/.style={},
    xscale=0.875,yscale=0.75]

  \input{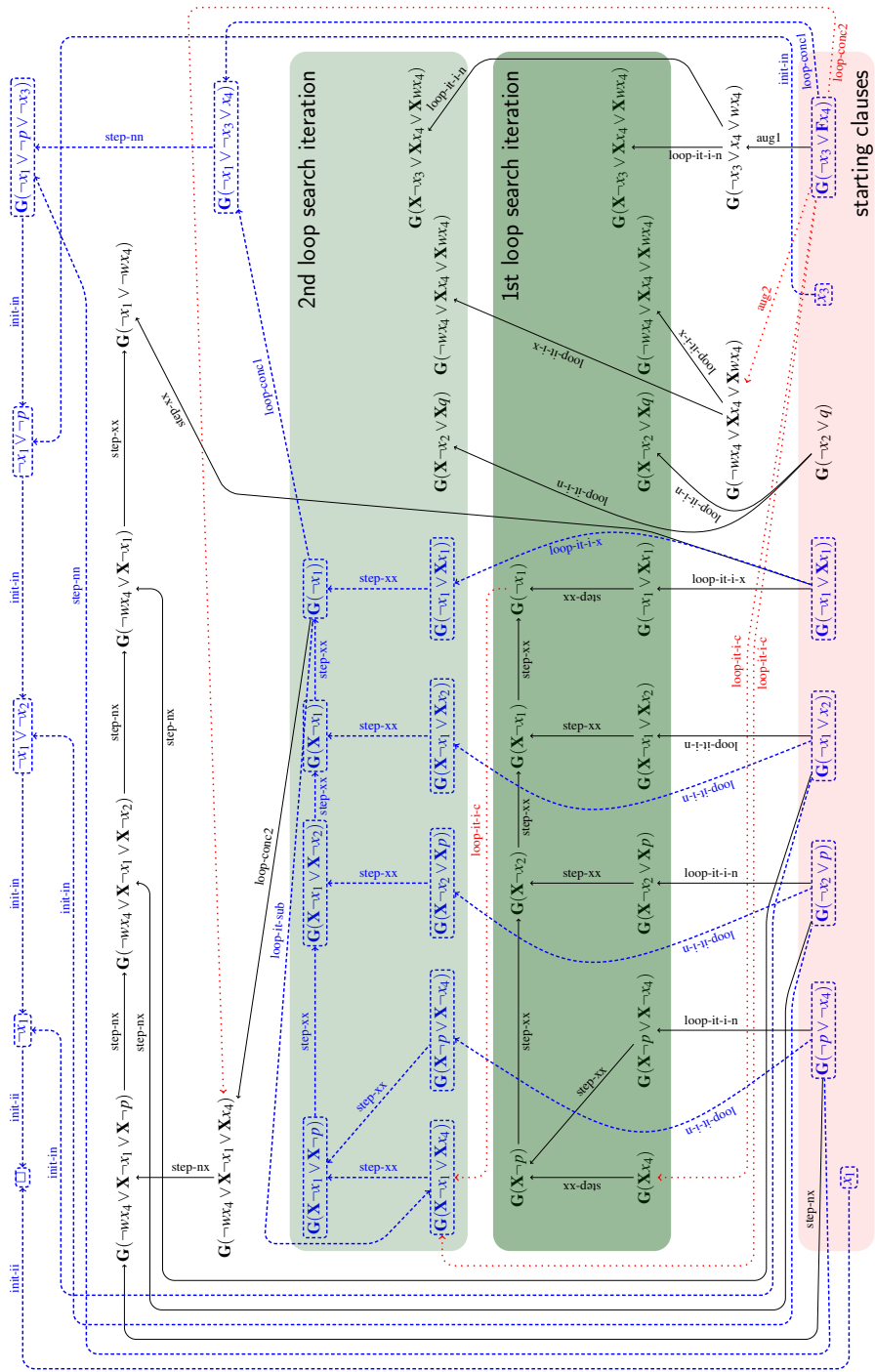}

\end{tikzpicture}
}
\begin{minipage}{0.63\linewidth}\caption{\label{fig:completeexresgraphucexopt} Example of an optimized resolution graph and \uc extraction in SNF from an optimized resolution graph}\end{minipage}
\end{sidewaysfigure}
}

\subsection{Extracting a \UC in LTL}
\label{sec:extractingaucinltl}

In Def.~\ref{def:ltlcoreextraction} we describe how to map a \uc in SNF back to a \uc in LTL.
The correctness of the construction is then proved in Thm.~\ref{thm:ltlcoreisunsat}.
The main idea in the proof is to compare the SNF of $\inp$ and of its \uc in LTL by partitioning the SNF clauses into three sets: one that is shared by the two SNFs, one that replaces some occurrences of propositions in $\fdCNF{\inp}$ with $\true$ or $\false$, and one whose clauses are only in $\fdCNF{\inp}$. Then one can show that the \uc of $\inp$ in SNF must be contained in the first partition.

\mydefinition{\UC in LTL}\label{def:ucinltl}
(cf.~Def.~10 of \cite{VSchuppan-SCP-2012})
Let $\inp$ be an unsatisfiable LTL formula.
Let $\inpuc$
\begin{inparaenum}[(i)]
\item be obtained from $\inp$ by replacing a set of positive polarity occurrences of subformulas of $\inp$ with $\true$ and a set of negative polarity occurrences of subformulas of $\inp$ with $\false$ and \label{enum:defucinltl:1}
\item be unsatisfiable.
\end{inparaenum}
Then $\inpuc$ is a \emph{\uc of $\inp$ in LTL}.
\end{definition}

\mydefinition{\UC in LTL from SNF}\label{def:ltlcoreextraction}
Let $\inp$ be an unsatisfiable LTL formula, let $\fdCNF{\inp}$ be its SNF, and let $\setclausesuc$ be a \uc of $\fdCNF{\inp}$ in SNF. Then the \emph{\uc of $\inp$ in LTL from SNF}, $\inpuc$, is obtained as follows. For each positive (resp.~negative) polarity occurrence of a proper subformula $\prt$ of $\inp$ with proposition $\fdCNFvar{\prt}$ according to Tab.~\ref{tab:ltltosnf}, replace $\prt$ in $\inp$ with $\true$ (resp.~$\false$) iff $\setclausesuc$ contains no clause with an occurrence of proposition $\fdCNFvar{\prt}$ that is marked {\color{blue}\setlength\fboxsep{1pt}\fbox{blue boxed}} in Tab.~\ref{tab:ltltosnf}. (We are sloppy in that we ``replace'' subformulas of replaced subformulas, while in effect they simply vanish.)
\end{definition}

\mytheorem%
{\thmltlcoreisunsat}%
{thm:ltlcoreisunsat}%
{Unsatisfiability of \UC in LTL from SNF}%
{Let $\inp$ be an unsatisfiable LTL formula, and let $\inpuc$ be a \uc of $\inp$ in LTL from SNF. Then $\inpuc$ is unsatisfiable.}
\thmltlcoreisunsat{true}

\begin{proof}
Let $\fdCNF{\inp}$ be the SNF of $\inp$, and let $\setclausesuc$ be a \uc of $\fdCNF{\inp}$ in SNF. Let $\setclausesucp$ be the \uc of $\setclausesuc$ in SNF via \tr from an optimized resolution graph.

First, consider the trivial case that $\inp$ is $\false$. Here, Def.~\ref{def:ltlcoreextraction} results in the \uc of $\inp$ in LTL being $\inpuc \definedas \false$ as desired.

Now assume that $\inp$ is not $\false$, i.e., the size of the syntax tree of $\inp$ is greater than 1. Let $\fdCNF{\inpuc}$ be the SNF of $\inpuc$. In order to prove that $\inpuc$ is unsatisfiable we show that the clauses of $\setclausesucp$ (which is unsatisfiable by Def.~\ref{def:ucinsnf} and Thm.~\ref{thm:coreisunsatopt}) are a subset of the SNF of $\inpuc$: $\setclausesucp \subseteq \fdCNF{\inpuc}$.

By comparing the clauses of $\fdCNF{\inp}$ with those of $\fdCNF{\inpuc}$ we can partition the clauses of $\fdCNF{\inp}$ into three sets:\footnote{We disregard the issue of the indices of the variables $\dCNFvar, \dCNFvarp, \ldots$.}
\begin{inparaenum}[(i)]
\item Some clauses are present in both $\fdCNF{\inp}$ and $\fdCNF{\inpuc}$: $\setclausesp_1 \definedas \fdCNF{\inp} \cap \fdCNF{\inpuc}$.
\item Some clauses are present in $\fdCNF{\inp}$ and are present in $\fdCNF{\inpuc}$ with one or more occurrences of some propositions $\dCNFvar, \dCNFvarp, \ldots$ that are marked {\color{blue}\setlength\fboxsep{1pt}\fbox{blue boxed}} in Tab.~\ref{tab:ltltosnf} replaced with $\true$ or $\false$. Call that set $\setclausesp_2$.
\item Some clauses are present in $\fdCNF{\inp}$ but not in $\fdCNF{\inpuc}$: $\setclausesp_3 \definedas \fdCNF{\inp} \setminus (\fdCNF{\inpuc} \cup \setclausesp_2)$.
\end{inparaenum}

By Def.~\ref{def:ucinsnf} and Cor.~\ref{thm:coreisunsatoptcor} $\setclausesucp$ is a subset of $\fdCNF{\inp}$: $\setclausesucp \subseteq \setclausesuc \subseteq \fdCNF{\inp}$.

\begin{sloppypar}
By Def.~\ref{def:ltlcoreextraction} $\setclausesuc$ and, therefore, also $\setclausesucp$ contains no member of $\setclausesp_2$; otherwise, there could not be one or more occurrences of some propositions $\dCNFvar, \dCNFvarp, \ldots$ that are marked {\color{blue}\setlength\fboxsep{1pt}\fbox{blue boxed}} in Tab.~\ref{tab:ltltosnf} replaced with $\true$ or $\false$ in the clauses of $\setclausesp_2$: $\setclausesucp \cap \setclausesp_2 = \emptyset$.
\end{sloppypar}

Now we argue that $\setclausesucp$ also contains no member of $\setclausesp_3$. First, let $\clause \in \setclausesp_3$ be an initial or a global clause. $\clause$ cannot be a member of $\setclausesucp$ as, in order to be part of a proof that derives the empty clause, all literals of $\clause$ need to be ``resolved away''. However, this is not possible for $\clause$ as for the literal $(\fbnotname) \fdCNFvar{\prt}$ on the left side of the implication in Tab.~\ref{tab:ltltosnf} there is no clause with an opposite literal in $\setclausesucp$. This follows by induction on the nesting depth of the subformula $\prt$ to which $(\fbnotname) \fdCNFvar{\prt}$ belongs from the occurrence of the superformula of $\prt$ that has been replaced with $\true$ or $\false$ in $\inpuc$. Now let $\clause \in \setclausesp_3$ be an eventuality clause. By Def.~\ref{def:coreextractionopt} for such $\clause$ to be part of $\setclausesucp$ there would have to be a clause $\clausep$ in the resolution graph $\graph$ according to Def.~\ref{def:resolutiongraphtype}, \ref{def:resolutiongraphinitialization}, \ref{def:optimizedresolutiongraphupdate} that was generated by production rules \refaugone or \refloopconclusionone and that is backward reachable in $\graph$ from the vertex labeled with the empty clause $\emptyclause$ in the main partition $\mainpartition$, $\vertex_\emptyclause$. Again, for the latter to happen, all literals of $\clausep$ would have to be ``resolved away'', which is impossible by a similar inductive argument as before.

Hence, we have shown that all clauses in $\setclausesucp$ come from $\setclausesp_1$, which is a subset of $\fdCNF{\inpuc}$. This concludes the proof.
\qed
\end{proof}

As a \uc in LTL from SNF fulfills requirement (\ref{enum:defucinltl:1}) in Def.~\ref{def:ucinltl}, we obtain Cor.~\ref{thm:ltlcoreisunsatcor}.

\mycorollary%
{\thmltlcoreisunsatcor}%
{thm:ltlcoreisunsatcor}%
{\UC in LTL from SNF is \UC in LTL}%
{Let $\inp$ be an unsatisfiable LTL formula, and let $\inpuc$ be a \uc of $\inp$ in LTL from SNF. Then $\inpuc$ is a \uc of $\inp$ in LTL.}
\thmltlcoreisunsatcor{true}

\myremark%
{\thmltlcoredet}%
{thm:ltlcoredet}%
{Determinants of \UC in LTL from SNF}%
{Note that an unsatisfiable LTL formula $\inp$ may have several \ucs in LTL.
If the \uc of $\inp$ is obtained via Def.~\ref{def:ltlcoreextraction}, \ref{def:coreextractionopt}, then Rem.~\ref{thm:coredet} on the determinants of a \uc in SNF applies here, too.
Moreover, using a different translation from LTL into SNF instead of Def.~\ref{def:ltltosnf} might also lead to a different \uc in LTL from SNF.
}
\thmltlcoredet{true}

We now complete the running example from Sec.~\ref{sec:preliminaries} and the previous subsection.
We show how to obtain a \uc of $\inp$ in LTL from SNF from the \uc of $\inp$ in SNF.
Remember that the formula $\inp$, which we would like to obtain a \uc of, is $\fband{(\fglobally{(\fband{p}{q})})}{\ffinally{\fbnot{p}}}$ (cf.~\eqref{ex:ltltosnf}).
Moreover, the SNF of $\inp$, $\setclauses$, on which in Fig.~\ref{fig:completeexresgraphucexopt} we executed \refalgltlsattrp, is $\{x_{1}$, $\fglobally{(\fbor{\fbnot{x_{1}}}{\fnext{x_{1}}})}$, $\fglobally{(\fbor{\fbnot{x_{1}}}{x_{2}})}$, $\fglobally{(\fbor{\fbnot{x_{2}}}{p})}$, $\fglobally{(\fbor{\fbnot{x_{2}}}{q})}$, $x_{3}$, $\fglobally{(\fbor{\fbnot{x_{3}}}{\ffinally{x_{4}}})}$, $\fglobally{(\fbor{\fbnot{p}}{\fbnot{x_{4}}})}\}$ (cf.~\eqref{ex:complete:snf}).
Finally, the \uc of $\inp$ in SNF, $\setclausesuc$, is $\setclauses \setminus \{\fglobally{(\fbor{\fbnot{x_{2}}}{q})}\}$ (cf.~\eqref{ex:complete:snfcore}).
Using the correspondence of $x_{1}$ to $\fdCNFvar{\fglobally{(\fband{p}{q})}}$, of $x_2$ to $\fdCNFvar{\fband{p}{q}}$, of $x_3$ to $\fdCNFvar{\ffinally{\fbnot{p}}}$, and of $x_4$ to $\fdCNFvar{\fbnot{p}}$ as well as the definition of implication we rewrite\footnote{Notice that this is done purely for the convenience of the reader and does not correspond to a step in our method for \uc extraction.} the \uc of $\inp$ in SNF, $\setclauses$ to \eqref{ex:complete:snfcoreexpanded}.
\begin{equation}\label{ex:complete:snfcoreexpanded}
\begin{array}{l}
\hspace{-0.5em}\{(\fdCNFvar{\fglobally{(\fband{p}{q})}}),\\
(\fglobally{(\fbimplies{\fdCNFvar{\fglobally{(\fband{p}{q})}}}{\fnext{\fdCNFvar{\fglobally{(\fband{p}{q})}}}})}),\\
(\fglobally{(\fbimplies{\fdCNFvar{\fglobally{(\fband{p}{q})}}}{\fdCNFvar{\fband{p}{q}}})}),\\
(\fglobally{(\fbimplies{\fdCNFvar{\fband{p}{q}}}{p})}),\\
(\fdCNFvar{\ffinally{\fbnot{p}}}),\\
(\fglobally{(\fbimplies{\fdCNFvar{\ffinally{\fbnot{p}}}}{\ffinally{\fdCNFvar{\fbnot{p}}}})}),\\
(\fglobally{(\fbimplies{\fdCNFvar{\fbnot{p}}}{\fbnot{p}})})\}
\end{array}
\end{equation}
By careful inspection of \eqref{ex:complete:snfcoreexpanded} we see that $q$ is the only subformula of $\inp$ whose proposition according to column 2 of Tab.~\ref{tab:ltltosnf}, which is $q$ itself, does not occur in any clause of \eqref{ex:complete:snfcoreexpanded} in a position that is marked {\color{blue}\setlength\fboxsep{1pt}\fbox{blue boxed}} in Tab.~\ref{tab:ltltosnf}.
Hence, $q$ is the only subformula to be replaced by $\true$ or $\false$ in $\inp$, yielding the \uc of $\inp$ in LTL from SNF, $\inpuc$, in \eqref{ex:complete:core}.
\begin{equation}\label{ex:complete:core}
\inpuc \definedas \fband{(\fglobally{(\fband{p}{\true})})}{\ffinally{\fbnot{p}}}
\end{equation}

\section{Post-Processing \UCs}
\label{sec:postprocessingucs}

In this section we adapt two techniques to our setting of \ucs for LTL that can be used to make the \ucs obtained so far more useful.
Both techniques will typically be applied after an initial \uc has been obtained; hence, we term this section post-processing of \ucs.
First, we discuss minimality of \ucs.
Then we show how to partition occurrences of propositions in a \uc according to whether they interact in a \tr proof of unsatisfiability, leading to a more fine-grained notion of \uc.
For a more advanced method of post-processing, which extracts information on the time points at which occurrences of subformulas are relevant to unsatisfiability, see \cite{VSchuppan-QAPL-2013}.

\subsection{Minimal \UCs}
\label{sec:minimalucs}

In this subsection we introduce notions of and algorithms to obtain minimal \ucs. The results are either straightforward (Prop.~\ref{thm:snfminucnotltlminuc}) or well known (Def.~\ref{def:snfltlmincore}, Rem.~\ref{thm:minucextraction}). Still, the material is needed in the experimental evaluation, and within the flow of the article this seems to be the appropriate place.

\mydefinition{Minimal \uc in SNF and LTL}\label{def:snfltlmincore}
(See, e.g., \cite{VSchuppan-SCP-2012}: irreducible \uc)
A \uc $\setclausesuc$ in SNF is \emph{minimal} iff $\forall \clause \in \setclausesuc \;.\; \setclausesuc \setminus \{\clause\}$ is satisfiable.
A \uc $\inpuc$ in LTL is \emph{minimal} iff there is no positive polarity occurrence of a subformula that can be replaced with $\true$ and no negative polarity occurrence of a subformula that can be replaced with $\false$ without making $\inpuc$ satisfiable.
\end{definition}

\myproposition%
{\thmsnfminucnotltlminuc}%
{thm:snfminucnotltlminuc}%
{Minimal \uc in SNF No Guarantee for Minimal \uc in LTL}%
{Let $\inp$ be an unsatisfiable LTL formula, let $\fdCNF{\inp}$ be its SNF, let $\setclausesuc$ be a minimal \uc of $\fdCNF{\inp}$ in SNF, and let $\inpuc$ be the \uc of $\inp$ in LTL from SNF. Then $\inpuc$ is not necessarily minimal.}%
\thmsnfminucnotltlminuc{true}

\begin{proof}
Let $\inp \definedas \fband{(\fbnot{p})}{(\fband{(\fglobally{\fbnot{q}})}{(\funtil{p}{q})})}$.
Then
\[
\fdCNF{\inp} \definedas
\begin{array}{l}
\hspace{-0.5em}\{(\fdCNFvar{\inp}),\\
\fgnclause{\fbimplies{\fdCNFvar{\inp}}{\fdCNFvar{\fbnot{\ap}}}},\\
\fgnclause{\fbimplies{\fdCNFvar{\fbnot{\ap}}}{\fbnot{\ap}}},\\
\fgnclause{\fbimplies{\fdCNFvar{\inp}}{\fdCNFvar{\fband{(\fglobally{\fbnot{\app}})}{(\funtil{\ap}{\app})}}}},\\
\fgnclause{\fbimplies{\fdCNFvar{\fband{(\fglobally{\fbnot{\app}})}{(\funtil{\ap}{\app})}}}{\fdCNFvar{\fglobally{\fbnot{\app}}}}},\\
\fgnclause{\fbimplies{\fdCNFvar{\fglobally{\fbnot{\app}}}}{\fnext{\fdCNFvar{\fglobally{\fbnot{\app}}}}}},\\
\fgnclause{\fbimplies{\fdCNFvar{\fglobally{\fbnot{\app}}}}{\fdCNFvar{\fbnot{\app}}}},\\
\fgnclause{\fbimplies{\fdCNFvar{\fbnot{\app}}}{\fbnot{\app}}},\\
\fgnclause{\fbimplies{\fdCNFvar{\fband{(\fglobally{\fbnot{\app}})}{(\funtil{\ap}{\app})}}}{\fdCNFvar{\funtil{\ap}{\app}}}},\\
\fgnclause{\fbimplies{\fdCNFvar{\funtil{\ap}{\app}}}{(\fbor{\app}{\ap})}},\\
\fgnclause{\fbimplies{\fdCNFvar{\funtil{\ap}{\app}}}{(\fbor{\app}{\fnext{\fdCNFvar{\fdCNFvar{\funtil{\ap}{\app}}}}})}},\\
\fgnclause{\fbimplies{\fdCNFvar{\funtil{\ap}{\app}}}{\ffinally{\app}}}\}
\end{array}
\]
is its SNF according to Def.~\ref{def:ltltosnf}.
A minimal \uc of $\fdCNF{\inp}$ in SNF is
\[
\setclausesuc \definedas
\begin{array}{l}
\hspace{-0.5em}\{(\fdCNFvar{\inp}),\\
\fgnclause{\fbimplies{\fdCNFvar{\inp}}{\fdCNFvar{\fbnot{\ap}}}},\\
\fgnclause{\fbimplies{\fdCNFvar{\fbnot{\ap}}}{\fbnot{\ap}}},\\
\fgnclause{\fbimplies{\fdCNFvar{\inp}}{\fdCNFvar{\fband{(\fglobally{\fbnot{\app}})}{(\funtil{\ap}{\app})}}}},\\
\fgnclause{\fbimplies{\fdCNFvar{\fband{(\fglobally{\fbnot{\app}})}{(\funtil{\ap}{\app})}}}{\fdCNFvar{\fglobally{\fbnot{\app}}}}},\\
\fgnclause{\fbimplies{\fdCNFvar{\fglobally{\fbnot{\app}}}}{\fdCNFvar{\fbnot{\app}}}},\\
\fgnclause{\fbimplies{\fdCNFvar{\fbnot{\app}}}{\fbnot{\app}}},\\
\fgnclause{\fbimplies{\fdCNFvar{\fband{(\fglobally{\fbnot{\app}})}{(\funtil{\ap}{\app})}}}{\fdCNFvar{\funtil{\ap}{\app}}}},\\
\fgnclause{\fbimplies{\fdCNFvar{\funtil{\ap}{\app}}}{(\fbor{\app}{\ap})}}\}.
\end{array}
\]
Mapping $\setclausesuc$ back to a \uc in LTL from SNF via Def.~\ref{def:ltlcoreextraction} yields $\inp$.
$\inp$ is not a minimal \uc, as the first conjunct $\fbnot{\ap}$ can be replaced with $\true$ while retaining unsatisfiability.
This concludes the proof.
\qed
\end{proof}

Note that given $\inp$ from the proof of Prop.~\ref{thm:snfminucnotltlminuc} our implementation actually produces $\inp$ as a \uc in LTL.
This is due to the fact that the \uc in SNF, $\setclausesuc$, is found during the first execution of saturation in line \ref{line:fig:ltlsattrp:sat1} of \refalgltlsattrp, while the contradiction between $\fglobally{\fbnot{\app}}$ and the eventuality part of $\funtil{\ap}{\app}$ requires loop search, which is only performed at a later stage.

Note also that the result just proved depends on the notion of \uc for LTL: the proof above obviously does not hold, if the notion of \uc allows to not only replace $\fglobally{\fbnot{\app}}$ with $\true$ but also, alternatively, with $\fbnot{\app}$ and $\funtil{\ap}{\app}$ not only with $\true$ but also, alternatively, with $\fbor{\ap}{\app}$.

\myremark%
{\thmminucextraction}%
{thm:minucextraction}%
{Extraction of Minimal \UCs}%
{A common way to obtain minimal \ucs works by repeatedly attempting to remove parts of a \uc (e.g., \ifnoappendix\cite{JChinneckEDravnieks-ORSAJournalOnComputing-1991,RBakkerFDikkerFTempelmanPWognum-IJCAI-1993,LZhang-PhDThesis-2003,JMarquesSilva-JournalOnMultipleValuedLogicAndSoftComputing-2012,AAwadRGoreZHouJThomsonMWeidlich-InformationSystems-2012,ETorlakFShengHoChangDJackson-FM-2008}\else\cite{JChinneckEDravnieks-ORSAJournalOnComputing-1991,RBakkerFDikkerFTempelmanPWognum-IJCAI-1993,LZhangSMalik-SAT-2003,LZhang-PhDThesis-2003,JMarquesSilva-JournalOnMultipleValuedLogicAndSoftComputing-2012,AAwadRGoreZHouJThomsonMWeidlich-InformationSystems-2012,ETorlakFShengHoChangDJackson-FM-2008}\fi). If the modified formula is still unsatisfiable, then the removal is made permanent; otherwise the removal is undone. The procedure continues until all parts of the \uc have been considered for removal. This is called \emph{deletion-based extraction of minimal \ucs} (e.g., \cite{JChinneckEDravnieks-ORSAJournalOnComputing-1991,JMarquesSilva-JournalOnMultipleValuedLogicAndSoftComputing-2012}).

In the case of LTL the algorithm attempts to replace positive polarity occurrences of subformulas with $\true$ and negative polarity ones with $\false$. It terminates, if no more replacements can be performed without making the resulting formula satisfiable.

Naturally, this method may be expensive due to the number of satisfiability tests to be performed. It is therefore often used to minimize a \uc that has been obtained by other means such as those described in Sec.~\ref{sec:ucextraction} (see, e.g., \ifnoappendix\cite{JChinneckEDravnieks-ORSAJournalOnComputing-1991,RBakkerFDikkerFTempelmanPWognum-IJCAI-1993,LZhang-PhDThesis-2003,JMarquesSilva-JournalOnMultipleValuedLogicAndSoftComputing-2012,AAwadRGoreZHouJThomsonMWeidlich-InformationSystems-2012,ETorlakFShengHoChangDJackson-FM-2008}\else\cite{JChinneckEDravnieks-ORSAJournalOnComputing-1991,RBakkerFDikkerFTempelmanPWognum-IJCAI-1993,LZhangSMalik-SAT-2003,LZhang-PhDThesis-2003,JMarquesSilva-JournalOnMultipleValuedLogicAndSoftComputing-2012,AAwadRGoreZHouJThomsonMWeidlich-InformationSystems-2012,ETorlakFShengHoChangDJackson-FM-2008}\fi).
Potential optimizations of the minimization algorithm include binary search (e.g., \cite{AZellerRHildebrandt-TrSE-2002,UJunker-CONS-2001,JMarquesSilva-JournalOnMultipleValuedLogicAndSoftComputing-2012,RKoenighoferGHofferekRBloem-FMCAD-2009}) and reusing intermediate results (e.g., \cite{ETorlakFShengHoChangDJackson-FM-2008,RKoenighoferGHofferekRBloem-FMCAD-2009}).
}%
\thmminucextraction{true}

\subsection{Grouping Propositions in \UCs}
\label{sec:groupedpropositionsucs}

In this subsection we show how to extract information from a \tr proof of unsatisfiability on which occurrences of same propositions in a \uc actually need to be the same propositions to retain unsatisfiability and which occurrences of same propositions might be substituted with different propositions without losing unsatisfiability.
This can provide helpful information on the interaction of parts of a formula to a user who is debugging a specification.

\subsubsection{Intuition}

As an example consider
\begin{equation}\label{ex:partitioned}
\fbor{(\fband{\ap}{\fbnot{\ap}})}{\fnext{(\fband{(\fglobally{\ap})}{\ffinally{\fbnot{\ap}}})}}.
\end{equation}
Note that \eqref{ex:partitioned} is a minimal \uc in LTL.
However, from the point of view of unsatisfiability of \eqref{ex:partitioned}, the occurrences of $\ap$ in $\fband{\ap}{\fbnot{\ap}}$ ``have nothing to do'' with the occurrences of $\ap$ in $\fnext{(\fband{(\fglobally{\ap})}{\ffinally{\fbnot{\ap}}})}$.
On the other hand, the first occurrence of $\ap$ in $\fband{\ap}{\fbnot{\ap}}$ must be an occurrence of the same proposition as the second occurrence of $\ap$ in $\fband{\ap}{\fbnot{\ap}}$ to obtain unsatisfiability.
The same is true for the two occurrences of $\ap$ in $\fnext{(\fband{(\fglobally{\ap})}{\ffinally{\fbnot{\ap}}})}$.
Hence, the first two occurrences of $\ap$ in \eqref{ex:partitioned} could be substituted with, e.g., $\ap_0$ and the second two occurrences with, e.g., $\ap_1$, obtaining \eqref{ex:partitioned:c}.
\eqref{ex:partitioned:c} retains unsatisfiability and the structure of \eqref{ex:partitioned}.
\begin{equation}\label{ex:partitioned:c}
\fbor{(\fband{\ap_0}{\fbnot{\ap_0}})}{\fnext{(\fband{(\fglobally{\ap_1})}{\ffinally{\fbnot{\ap_1}}})}}
\end{equation}

Note that, depending on the intended application of the \ucs, \eqref{ex:partitioned:c} may or may not be considered to be a valid notion of \uc of \eqref{ex:partitioned}.
For the purpose of debugging specifications we believe that this more fine-grained notion of \uc can provide helpful additional information to the user.

We call \eqref{ex:partitioned:c} a \uc of \eqref{ex:partitioned} in LTL \emph{with grouped propositions}.
We use the term ``group'' as a synonym to ``partition'' in order to avoid confusion with the notion of ``partition'' of vertices in the resolution graph in Sec.~\ref{sec:extractingaucinsnf}.

The information required to construct a \uc with grouped propositions is obtained by observing the interaction of occurrences of propositions in a \tr proof.
Essentially, if two occurrences of the same proposition $\ap$ in an LTL formula $\inp$ are found not to be interacting in a \tr proof of the unsatisfiability of $\inp$, then these occurrences of $\ap$ can be substituted with different propositions in a \uc of $\inp$ without losing unsatisfiability.

Let $\occurrence$ and $\occurrencep$ be two occurrences of a proposition $\ap$ in two SNF clauses $\clause$ and $\clausep$.
The occurrences $\occurrence$ and $\occurrencep$ in $\clause$ and $\clausep$ can interact in four different ways:
\begin{enumerate}
\item $\clause$ might be a premise and $\clausep$ might be a conclusion that is obtained from $\clause$ (and possibly some other clauses) by application of a production rule from Tab.~\ref{tab:productionrules}.
Then, some propositions may be transferred from $\clause$ to $\clausep$.
Occurrences of propositions that are subject to transfer from $\clause$ to $\clausep$ are said to interact.
For example, assume that $\ficlause{\fbor{\ap}{\app}}$ and $\fgnclause{\fbor{(\fbnot{\app})}{\appp}}$ are resolved to $\ficlause{\fbor{\ap}{\appp}}$ using \refinitin.
Then the occurrence of $\ap$ in the premise $\ficlause{\fbor{\ap}{\app}}$ is transferred to the occurrence of $\ap$ in the conclusion $\ficlause{\fbor{\ap}{\appp}}$, as is the occurrence of $\appp$ in the premise $\fgnclause{\fbor{(\fbnot{\app})}{\appp}}$ to the occurrence of $\appp$ in the conclusion $\ficlause{\fbor{\ap}{\appp}}$.
\item $\clause$ and $\clausep$ might be resolved with each other using one of the saturation rules in Tab.~\ref{tab:productionrules}.
Then the two occurrences of the proposition which is resolved upon (which have different polarity) are also said to interact.
In the example above the occurrences of $\app$ in $\ficlause{\fbor{\ap}{\app}}$ and in $\fgnclause{\fbor{(\fbnot{\app})}{\appp}}$ interact.
\item $\clause$ and $\clausep$ might be eventuality clauses with the same eventuality literal.
Then the two occurrences of the eventuality literal are said to interact.
\item There might be a finite sequence of interactions of the former three kinds linking the occurrences $\occurrence$ and $\occurrencep$ in $\clause$ and $\clausep$, i.e., we form the transitive closure of the interaction relation.
\end{enumerate}

\subsubsection{Formalization --- SNF}

In Def.~\ref{def:groupingofpropositionsinresolutiongraph}--\ref{def:coreextractiongroupedpropositions} we formalize the idea as follows.\footnote{Remember that in this section we only consider optimized resolution graphs and \ucs in SNF via \tr from optimized resolution graphs, and we drop the designators ``optimized'' and ``from an optimized resolution graph''. Moreover, we drop general definitions and statements in the style of Def.~\ref{def:ucinsnf} and Cor.~\ref{thm:coreisunsatoptcor} and restrict ourselves to the specific cases \'a la Def.~\ref{def:coreextractionopt} and Thm.~\ref{thm:coreisunsatopt}.}
Each occurrence of a proposition in the resolution graph is mapped to some group (here groups are arbitrarily represented by natural numbers).
If two occurrences of a proposition interact, then they are forced to be mapped to the same group.
In our example the occurrences of $\ap$ in $\ficlause{\fbor{\ap}{\app}}$ and in $\ficlause{\fbor{\ap}{\appp}}$ might be mapped to $\pos \in \allnats$, the occurrences of $\app$ in $\ficlause{\fbor{\ap}{\app}}$ and in $\fgnclause{\fbor{(\fbnot{\app})}{\appp}}$ to $\posp \in \allnats$ different from $\pos$, and the occurrences of $\appp$ in $\fgnclause{\fbor{(\fbnot{\app})}{\appp}}$ and in $\ficlause{\fbor{\ap}{\appp}}$ to $\pospp \in \allnats$ different from both $\pos$ and $\posp$.
This leads to a partitioning of the occurrences of propositions in clauses labeling vertices in the resolution graph.
Occurrences in each partition are then replaced with a different proposition.

\mydefinition{Grouping of Propositions in Resolution Graph}\label{def:groupingofpropositionsinresolutiongraph}
Let $\setclauses \subseteq \allclauses$ be a set of SNF clauses, and let $\graph$ be a resolution graph with set of vertices $\setvertices$ and labeling of vertices with SNF clauses $\vertexlabelingname$.
Let $\alloccurrences$ be the set of occurrences of propositions in clauses in $\allclauses$.
Let $\fgroupname$ be a mapping from pairs of vertices in $\setvertices$ and occurrences of propositions in clauses in $\allclauses$ to groups in $\allnats$, $\fgroupname : \setvertices \times \alloccurrences \rightarrow \allnats$, fulfilling the following four conditions.
\begin{enumerate}
\item If
\begin{inparaenum}[(i)]
\item $\vertexzero$ is a vertex in $\setvertices$ labeled with a clause $\clausezero$: $\vertexzero \in \setvertices$ and $\fvertexlabeling{\vertexzero} = \clausezero$,
\item $\vertexone$ is a vertex in $\setvertices$ labeled with a clause $\clauseone$: $\vertexone \in \setvertices$ and $\fvertexlabeling{\vertexone} = \clauseone$,
\item $\clauseone$ is a conclusion obtained from premise $\clausezero$ (and possibly other premises) by an application of a production rule \refrule from Tab.~\ref{tab:productionrules} in $\graph$,
\item $\occurrencezero$ is an occurrence of a proposition $\ap$ in $\clausezero$,
\item $\occurrenceone$ is an occurrence of a proposition $\app$ in $\clauseone$, and
\item \refrule constrains $\ap$ in $\occurrencezero$ and $\app$ in $\occurrenceone$ to be the same proposition,
\end{inparaenum}
then $\mkpair{\vertexzero}{\occurrencezero}$ and $\mkpair{\vertexone}{\occurrenceone}$ are mapped to the same group: $\fgroup{\vertexzero}{\occurrencezero} = \fgroup{\vertexone}{\occurrenceone}$.
\item If
\begin{inparaenum}[(i)]
\item $\vertextwo$ is a vertex in $\setvertices$ labeled with a clause $\clausetwo$: $\vertextwo \in \setvertices$ and $\fvertexlabeling{\vertextwo} = \clausetwo$,
\item $\vertexthree$ is a vertex in $\setvertices$ labeled with a clause $\clausethree$: $\vertexthree \in \setvertices$ and $\fvertexlabeling{\vertexthree} = \clausethree$,
\item $\clausetwo$ and $\clausethree$ are premises in the application of a production rule $\mbox{\refrule} \in \{\mbox{\refinitii, \refinitin, \refstepnn, \refstepnx, \refstepxx}\}$ from Tab.~\ref{tab:productionrules} in $\graph$,
\item $\occurrencetwo$ is the occurrence of the proposition $\appp$ in $\clausetwo$ that is resolved upon in the application of \refrule, and
\item $\occurrencethree$ is the occurrence of the proposition $\appp$ in $\clausethree$ that is resolved upon in the application of \refrule,
\end{inparaenum}
then $\mkpair{\vertextwo}{\occurrencetwo}$ and $\mkpair{\vertexthree}{\occurrencethree}$ are mapped to the same group: $\fgroup{\vertextwo}{\occurrencetwo} = \fgroup{\vertexthree}{\occurrencethree}$.
\item If
\begin{inparaenum}[(i)]
\item $\vertexfour$ is a vertex in $\setvertices$ labeled with an eventuality clause $\clausefour$: $\vertexfour \in \setvertices$ and $\fvertexlabeling{\vertexfour} = \clausefour$,
\item $\vertexfive$ is a vertex in $\setvertices$ labeled with an eventuality clause $\clausefive$: $\vertexfive \in \setvertices$ and $\fvertexlabeling{\vertexfive} = \clausefive$,
\item $\occurrencefour$ is the occurrence of the eventuality literal $\apres$ in $\clausefour$,
\item $\occurrencefive$ is the occurrence of the eventuality literal $\apres$ in $\clausefive$,
\end{inparaenum}
then $\mkpair{\vertexfour}{\occurrencefour}$ and $\mkpair{\vertexfive}{\occurrencefive}$ are mapped to the same group: $\fgroup{\vertexfour}{\occurrencefour} = \fgroup{\vertexfive}{\occurrencefive}$.
\item If two pairs $\mkpair{\vertexsix}{\occurrencesix}$ and $\mkpair{\vertexseven}{\occurrenceseven}$ are not (transitively) forced to be mapped to the same group by the former three conditions, then they are mapped to different groups.
\end{enumerate}
Then $\fgroupname$ is a \emph{grouping of propositions in a resolution graph}.
\end{definition}

\mydefinition{Grouping-Induced Substitution of Propositions in Resolution Graph}\label{def:groupinginducedsubstitutionofpropositionsinresolutiongraph}
Let $\setclauses \subseteq \allclauses$ be a set of SNF clauses, and let $\graph$ be a resolution graph with set of vertices $\setvertices$ and labeling of vertices with SNF clauses $\vertexlabelingname$.
Let $\alloccurrences$ be the set of occurrences of propositions in clauses in $\allclauses$.
Let $\fgroupname$ be the grouping of propositions in $\graph$.
Let $\fmapname$ be an injective mapping from the image of $\fgroupname$ to propositions $\allaps$: $\fmapname : \{\pos \in \allnats \mid \exists \vertex \in \setvertices \;.\; \exists \occurrence \mbox{ in } \fvertexlabeling{\vertex} \;.\; \pos = \fgroup{\vertex}{\occurrence}\} \rightarrow \allaps$ such that $\fmap{\pos} = \fmap{\posp} \proofimplies \pos = \posp$.
Then the composition of $\fmapname$ and $\fgroupname$, $\fsubstitutionname = \fmapname \circ \fgroupname$, is a \emph{grouping-induced substitution of propositions in a resolution graph}.
\end{definition}

We extend the domain of $\fsubstitutionname$ to clauses in $\graph$ and to $\graph$ itself in the natural way.
If the vertex $\vertex_\clause$ that a clause $\clause$ is labeling is clear from the context, we write $\fsubstitution{\clause}$ instead of $\fsubstitution{\vertex_\clause, \clause}$.

\mydefinition{\UC in SNF via \TR with Grouped Propositions}\label{def:coreextractiongroupedpropositions}
Let $\setclauses$ be an unsatisfiable set of SNF clauses, and let $\setclausesuc$ be a \uc of $\setclauses$ in SNF via \tr.
Let $\fsubstitutionname$ be the grouping-induced substitution of propositions in a resolution graph.
The \emph{\uc of $\setclauses$ in SNF via \tr with grouped propositions}, $\setclausesucgrouped$, is obtained from the \uc in SNF via \tr $\setclausesuc$ by applying $\fsubstitutionname$ to each clause in $\setclausesuc$: $\setclausesucgrouped = \{\fsubstitution{\clause} \mid \clause \in \setclausesuc\}$.
\end{definition}

Assume that Def.~\ref{def:groupingofpropositionsinresolutiongraph}--\ref{def:coreextractiongroupedpropositions} are applied to
\begin{equation}\label{ex:partitioned:2}
\setclauses \definedas \{
\ficlause{\fbor{\ap}{\app}},
\ficlause{\fbnot{\appp}},
(\fglobally{(\fbor{(\fbnot{\ap})}{\fnext{\appp}})}),
(\fglobally{\;\fnext{\fbnot{\appp}}}),
\fgnclause{\fbor{(\fbnot{\app})}{\appp}}
\}.
\end{equation}
Now it is easy to see that there exists a \tr proof of the unsatisfiability of $\setclauses$ such that both occurrences of $\ap$ are mapped to some group $\pos$, both occurrences of $\app$ are mapped to a different group $\posp$, the occurrences of $\appp$ in $\ficlause{\fbnot{\appp}}$ and in $\fgnclause{\fbor{(\fbnot{\app})}{\appp}}$ are mapped to another group $\pospp$, and, finally, the remaining occurrences of $\appp$ in $(\fglobally{(\fbor{(\fbnot{\ap})}{\fnext{\appp}})})$ and in $(\fglobally{\;\fnext{\fbnot{\appp}}})$ are mapped to a last group $\posppp$.
To obtain a \uc of $\setclauses$ in SNF via \tr with grouped propositions we substitute occurrences of propositions according to the groups they are mapped to: occurrences of propositions are substituted with the same proposition if and only if they are mapped to the same group.
In the example \eqref{ex:partitioned:2} we obtain the following \uc \eqref{ex:partitioned:2:c}:
\begin{equation}\label{ex:partitioned:2:c}
\setclausesucgrouped \definedas \{
\ficlause{\fbor{\ap}{\app}},
\ficlause{\fbnot{\appp_0}},
(\fglobally{(\fbor{(\fbnot{\ap})}{\fnext{\appp_1}})}),
(\fglobally{\;\fnext{\fbnot{\appp_1}}}),
\fgnclause{\fbor{(\fbnot{\app})}{\appp_0}}
\}.
\end{equation}

The proof of correctness in Thm.~\ref{thm:coregroupedpropositionsisunsat} is essentially by saying that the \tr proof for the unsatisfiability of a \uc in SNF via \tr is also a \tr proof for the unsatisfiability of a \uc in SNF via \tr with grouped propositions modulo renaming of some occurrences of propositions.
Before stating Thm.~\ref{thm:coregroupedpropositionsisunsat} we mention some properties of $\fgroupname$ and $\fsubstitutionname$ required in its proof of correctness.

\myremark%
{\thmgroupmapsdifferentpropositionstodifferentgroups}%
{thm:groupmapsdifferentpropositionstodifferentgroups}%
{$\fgroupname$ Maps Different Propositions to Different Groups}%
{Let $\setclauses$ be a set of SNF clauses, let $\graph$ be a resolution graph with set of vertices $\setvertices$ and labeling of vertices with SNF clauses $\vertexlabelingname$, and let $\fgroupname$ be the grouping of propositions in a resolution graph.
$\fgroupname$ maps occurrences of different propositions to different groups: if
\begin{inparaenum}[(i)]
\item $\vertex$ is a vertex in $\setvertices$ labeled with a clause $\clause$: $\vertex \in \setvertices$ and $\fvertexlabeling{\vertex} = \clause$,
\item $\vertexp$ is a vertex in $\setvertices$ labeled with a clause $\clausep$: $\vertexp \in \setvertices$ and $\fvertexlabeling{\vertexp} = \clausep$,
\item $\occurrence$ is an occurrence of a proposition $\ap$ in $\clause$,
\item $\occurrencep$ is an occurrence of a proposition $\app$ in $\clausep$, and
\item $\ap \ne \app$,
\end{inparaenum}
then $\fgroup{\vertex}{\occurrence} \ne \fgroup{\vertexp}{\occurrencep}$.
Given the injectivity of $\fmapname$ in Def.~\ref{def:groupinginducedsubstitutionofpropositionsinresolutiongraph} this directly extends to $\fsubstitutionname$.
}%
\thmgroupmapsdifferentpropositionstodifferentgroups{true}

\myremark%
{\thmsubstitutiondoesnotchangepolarityofliteral}%
{thm:substitutiondoesnotchangepolarityofliteral}%
{$\fsubstitutionname$ Does Not Change Polarity of Literal}%
{Let $\setclauses$ be a set of SNF clauses, let $\graph$ be a resolution graph, and let $\fsubstitutionname$ be the grouping-induced substitution of propositions in a resolution graph.
Then $\fsubstitutionname$ does not change the polarity of a literal in a clause.
}%
\thmsubstitutiondoesnotchangepolarityofliteral{true}

\myremark%
{\thmsubstitutiondoesnotchangetimeofliteral}%
{thm:substitutiondoesnotchangetimeofliteral}%
{$\fsubstitutionname$ Does Not Change Time of Literal}%
{Let $\setclauses$ be a set of SNF clauses, let $\graph$ be a resolution graph, and let $\fsubstitutionname$ be the grouping-induced substitution of propositions in a resolution graph.
Then $\fsubstitutionname$ maps an initial literal to an initial literal, a literal in the now part to a literal in the now part, a literal in the $\fnextname$ part to a literal in the $\fnextname$ part, and an eventuality literal to an eventuality literal.
}%
\thmsubstitutiondoesnotchangetimeofliteral{true}

\mytheorem%
{\thmcoregroupedpropositionsisunsat}%
{thm:coregroupedpropositionsisunsat}%
{Unsatisfiability of \UC in SNF via \TR with Grouped Propositions}%
{Let $\setclauses$ be an unsatisfiable set of SNF clauses, and let $\setclausesucgrouped$ be a \uc of $\setclauses$ in SNF via \tr with grouped propositions. Then $\setclausesucgrouped$ is unsatisfiable.}
\thmcoregroupedpropositionsisunsat{true}

\begin{proof}
Let $\setclausesuc$ be the \uc in SNF via \tr.
Let $\graph$ be the resolution graph with set of vertices $\setvertices$ and labeling of vertices with SNF clauses $\vertexlabelingname$.
Let $\fsubstitutionname$ be the grouping-induced substitution of propositions in a resolution graph.
Let $\graphgrouped = \fsubstitution{\graph}$.

Possibly contrary to Def.~\ref{def:groupingofpropositionsinresolutiongraph}, \ref{def:groupinginducedsubstitutionofpropositionsinresolutiongraph}  we assume that $\fsubstitutionname$ maps occurrences of same propositions $\fapwaitfor{\apreszero}$ introduced by augmentation \refaugone, \refaugtwo to same propositions.
Note that propositions $\fapwaitfor{\apreszero}$ do not occur in $\setclausesuc$ and, therefore, their images under $\fsubstitutionname$ do not occur in $\setclausesucgrouped$.
Hence, if we can show unsatisfiability of $\setclausesucgrouped$ under this assumption, then we have shown unsatisfiability of $\setclausesucgrouped$.

Below we show that each application of a production rule \refrule to clauses labeling some vertices in $\graph$ is also an application of \refrule to the clauses labeling these vertices in $\graphgrouped$.
This establishes that $\graphgrouped$ represents a (possibly partial) \tr proof.
Consider the following cases:
\begin{description}
\item[\refinitii]
Let $\vertex, \vertexp, \vertexpp \in \setvertices$ with $\fvertexlabeling{\vertex} = \clause$, $\fvertexlabeling{\vertexp} = \clausep$, and $\fvertexlabeling{\vertexpp} = \clausepp$ where $\clausepp$ is obtained from $\clause$ and $\clausep$ by applying \refinitii.
By Rem.~\ref{thm:substitutiondoesnotchangetimeofliteral} $\fsubstitution{\clause}$, $\fsubstitution{\clausep}$, and $\fsubstitution{\clausepp}$ are initial clauses.
Let $\apreszero$, $\fbnot{\apreszero}$ be the resolved upon literals in $\clause$ and $\clausep$.
By Def.~\ref{def:groupingofpropositionsinresolutiongraph}, \ref{def:groupinginducedsubstitutionofpropositionsinresolutiongraph} $\fsubstitutionname$ maps these occurrences of $\apreszero$ and $\fbnot{\apreszero}$ to opposite polarity occurrences of some proposition $\ap$ in $\fsubstitution{\clause}$ and $\fsubstitution{\clausep}$.
Let there be an occurrence of a not resolved upon literal $\apresone \ne \ap, \fbnot{\ap}$ in $\fsubstitution{\clause}$ (resp.~$\fsubstitution{\clausep}$).
By Def.~\ref{def:groupingofpropositionsinresolutiongraph}, \ref{def:groupinginducedsubstitutionofpropositionsinresolutiongraph} there is an occurrence $\occurrence$ of a literal $\aprestwo$ in $\clause$ (resp.~$\clausep$) such that $\fsubstitution{\vertex,\occurrence} = \apresone$.
By rule \refinitii there is a corresponding occurrence of $\aprestwo$ in $\clausepp$.
By Def.~\ref{def:groupingofpropositionsinresolutiongraph}, \ref{def:groupinginducedsubstitutionofpropositionsinresolutiongraph} $\fsubstitutionname$ maps that occurrence of $\aprestwo$ to an occurrence of $\apresone$ in $\fsubstitution{\clausepp}$.
Let there be an occurrence of a literal $\apresthree$ in $\fsubstitution{\clausepp}$.
By Def.~\ref{def:groupingofpropositionsinresolutiongraph}, \ref{def:groupinginducedsubstitutionofpropositionsinresolutiongraph} there is an occurrence $\occurrencep$ of a literal $\apresfour$ in $\clausepp$ such that $\fsubstitution{\vertexpp,\occurrencep} = \apresthree$.
By rule \refinitii there is a corresponding occurrence of $\apresfour$ in $\clause$ or $\clausep$.
By Def.~\ref{def:groupingofpropositionsinresolutiongraph}, \ref{def:groupinginducedsubstitutionofpropositionsinresolutiongraph} $\fsubstitutionname$ maps that occurrence of $\apresfour$ to an occurrence of $\apresthree$ in $\fsubstitution{\clause}$ or $\fsubstitution{\clausep}$.
This shows that $\fsubstitution{\clausepp}$ can be obtained by applying \refinitii to $\fsubstitution{\clause}$ and $\fsubstitution{\clausep}$.
\item[\refinitin, \refstepnn] Similar to the case of \refinitii.
\item[\refstepnx, \refstepxx] Similar to the case of \refinitii when considering now and $\fnextname$ parts separately as appropriate.
\item[\refaugone]
Let $\vertex, \vertexp \in \setvertices$ with $\fvertexlabeling{\vertex} = \clause$ and $\fvertexlabeling{\vertexp} = \clausep$ where $\clausep$ is obtained from $\clause$ by applying \refaugone, $\apreszero$ is the eventuality literal in $\clause$, and $\fapwaitfor{\apreszero}$ is the fresh literal introduced by \refaugone and \refaugtwo for all eventuality clauses in $\graph$ with eventuality literal $\apreszero$.
By Rem.~\ref{thm:substitutiondoesnotchangetimeofliteral} $\fsubstitution{\clause}$ is an eventuality clause and $\fsubstitution{\clausep}$ is a global clause with empty $\fnextname$ part.
By Def.~\ref{def:groupingofpropositionsinresolutiongraph}, \ref{def:groupinginducedsubstitutionofpropositionsinresolutiongraph} $\fsubstitutionname$ maps the occurrences of $\apreszero$ as the eventuality literal in $\clause$ and in $\clausep$ to two occurrences of some literal $\apresone$ as the eventuality literal in $\fsubstitution{\clause}$ and in $\fsubstitution{\clausep}$.
By Def.~\ref{def:groupingofpropositionsinresolutiongraph}, \ref{def:groupinginducedsubstitutionofpropositionsinresolutiongraph} $\fsubstitutionname$ maps the occurrence of $\fapwaitfor{\apreszero}$ in $\clausep$ to an occurrence of some literal $\fapwaitfor{\apresone}$ in $\fsubstitution{\clausep}$.
By Rem.~\ref{thm:groupmapsdifferentpropositionstodifferentgroups} $\fapwaitfor{\apresone}$ is fresh.
By Def.~\ref{def:groupingofpropositionsinresolutiongraph}, \ref{def:groupinginducedsubstitutionofpropositionsinresolutiongraph} and the assumption above $\fapwaitfor{\apresone}$ is the fresh literal introduced by \refaugone and \refaugtwo for all eventuality clauses in $\graphgrouped$ with eventuality literal $\apresone$.
Let there be an occurrence of a literal $\aprestwo$ in the now part of $\fsubstitution{\clause}$.
By Def.~\ref{def:groupingofpropositionsinresolutiongraph}, \ref{def:groupinginducedsubstitutionofpropositionsinresolutiongraph} there is an occurrence $\occurrence$ of a literal $\apresthree$ in the now part of $\clause$ such that $\fsubstitution{\vertex,\occurrence} = \aprestwo$.
By rule \refaugone there is a corresponding occurrence of $\apresthree$ in $\clausep$.
By Def.~\ref{def:groupingofpropositionsinresolutiongraph}, \ref{def:groupinginducedsubstitutionofpropositionsinresolutiongraph} $\fsubstitutionname$ maps that occurrence of $\apresthree$ to an occurrence of $\aprestwo$ in $\fsubstitution{\clausep}$.
Let there be an occurrence of a literal $\apresfour \ne \apresone, \fapwaitfor{\apresone}$ in $\fsubstitution{\clausep}$.
By Def.~\ref{def:groupingofpropositionsinresolutiongraph}, \ref{def:groupinginducedsubstitutionofpropositionsinresolutiongraph} there is an occurrence $\occurrencep$ of a literal $\apresfive$ in $\clausep$ such that $\fsubstitution{\vertexp,\occurrencep} = \apresfour$.
By rule \refaugone there is a corresponding occurrence of $\apresfive$ in the now part of $\clause$.
By Def.~\ref{def:groupingofpropositionsinresolutiongraph}, \ref{def:groupinginducedsubstitutionofpropositionsinresolutiongraph} $\fsubstitutionname$ maps that occurrence of $\apresfive$ to an occurrence of $\apresfour$ in the now part of $\fsubstitution{\clause}$.
This shows that $\fsubstitution{\clausep}$ can be obtained by applying \refaugone to $\fsubstitution{\clause}$.
\item[\refaugtwo]
Let $\vertex \in \setvertices$ with $\fvertexlabeling{\vertex} = \clause$ where $\clause$ is obtained by applying \refaugtwo to eventuality clauses with eventuality literal $\apreszero$.
By Def.~\ref{def:groupingofpropositionsinresolutiongraph}, \ref{def:groupinginducedsubstitutionofpropositionsinresolutiongraph} and the assumption above $\fsubstitution{\clause}$ is of the form $\fgnxclause{(\fbnot{\fapwaitfor{\apresone}})}{\fbor{\apresone}{\fapwaitfor{\apresone}}}$ for some $\apresone$, $\fapwaitfor{\apresone}$.
By Def.~\ref{def:groupingofpropositionsinresolutiongraph}, \ref{def:groupinginducedsubstitutionofpropositionsinresolutiongraph} $\apresone$ is unique for all occurrences of $\apreszero$ as eventuality literal in $\graph$.
By Rem.~\ref{thm:groupmapsdifferentpropositionstodifferentgroups} $\fapwaitfor{\apresone}$ is fresh.
By Def.~\ref{def:groupingofpropositionsinresolutiongraph}, \ref{def:groupinginducedsubstitutionofpropositionsinresolutiongraph} and the assumption above $\fapwaitfor{\apresone}$ is the fresh literal introduced by \refaugone and \refaugtwo for all eventuality clauses in $\graphgrouped$ with eventuality literal $\apresone$.
This shows that $\fsubstitution{\clause}$ can be obtained by applying \refaugtwo to eventuality clauses with eventuality literal $\apresone$.
\item[\refloopitinitx] 
Let $\vertex, \vertexp \in \setvertices$ with $\fvertexlabeling{\vertex} = \clause$ and $\fvertexlabeling{\vertexp} = \clausep$ where $\clausep$ is obtained from $\clause$ by applying \refloopitinitx.
By Rem.~\ref{thm:substitutiondoesnotchangetimeofliteral} $\fsubstitution{\clause}$ and $\fsubstitution{\clausep}$ are global clauses with non-empty $\fnextname$ part.
Let there be an occurrence of a literal $\apreszero$ in the now (resp.~$\fnextname$) part of $\fsubstitution{\clause}$.
By Def.~\ref{def:groupingofpropositionsinresolutiongraph}, \ref{def:groupinginducedsubstitutionofpropositionsinresolutiongraph} there is an occurrence $\occurrence$ of a literal $\apresone$ in the now (resp.~$\fnextname$) part of $\clause$ such that $\fsubstitution{\vertex,\occurrence} = \apreszero$.
By rule \refloopitinitx there is a corresponding occurrence of $\apresone$ in the now (resp.~$\fnextname$) part of $\clausep$.
By Def.~\ref{def:groupingofpropositionsinresolutiongraph}, \ref{def:groupinginducedsubstitutionofpropositionsinresolutiongraph} $\fsubstitutionname$ maps that occurrence of $\apresone$ to an occurrence of $\apreszero$ in the now (resp.~$\fnextname$) part of $\fsubstitution{\clausep}$.
Let there be an occurrence of a literal $\aprestwo$ in the now (resp.~$\fnextname$) part of $\fsubstitution{\clausep}$.
By Def.~\ref{def:groupingofpropositionsinresolutiongraph}, \ref{def:groupinginducedsubstitutionofpropositionsinresolutiongraph} there is an occurrence $\occurrencep$ of a literal $\apresthree$ in the now (resp.~$\fnextname$) part of $\clausep$ such that $\fsubstitution{\vertexp,\occurrencep} = \aprestwo$.
By rule \refloopitinitx there is a corresponding occurrence of $\apresthree$ in the now (resp.~$\fnextname$) part of $\clause$.
By Def.~\ref{def:groupingofpropositionsinresolutiongraph}, \ref{def:groupinginducedsubstitutionofpropositionsinresolutiongraph} $\fsubstitutionname$ maps that occurrence of $\apresthree$ to an occurrence of $\aprestwo$ in the now (resp.~$\fnextname$) part of $\fsubstitution{\clause}$.
This shows that $\fsubstitution{\clausep}$ can be obtained by applying \refloopitinitx to $\fsubstitution{\clause}$.
\item[\refloopitinitn] Similar to the case of \refloopitinitx.
\item[\refloopitinitc]
First, consider the case that the current BFS loop search iteration is the second or later iteration of a BFS loop search.
Let $\vertex, \vertexp, \vertexpp \in \setvertices$ with $\fvertexlabeling{\vertex} = \clause$, $\fvertexlabeling{\vertexp} = \clausep$, and $\fvertexlabeling{\vertexpp} = \clausepp$ where $\clausepp$ is obtained from global clause with empty $\fnextname$ part $\clause$ and eventuality clause $\clausep$ with eventuality literal $\apreszero$ by applying \refloopitinitc.
By Rem.~\ref{thm:substitutiondoesnotchangetimeofliteral} $\fsubstitution{\clause}$ is a global clause with empty $\fnextname$ part, $\fsubstitution{\clausep}$ is an eventuality clause, and $\fsubstitution{\clausepp}$ is a global clause with empty now part.
By rule \refloopitinitc there is a corresponding occurrence of $\apreszero$ in $\clausepp$.
By Def.~\ref{def:groupingofpropositionsinresolutiongraph}, \ref{def:groupinginducedsubstitutionofpropositionsinresolutiongraph} $\fsubstitutionname$ maps these occurrences of $\apreszero$ to occurrences of some literal $\apresone$ as the eventuality literal in $\fsubstitution{\clausep}$ and as a disjunct in $\fsubstitution{\clausepp}$.
Let there be an occurrence of a literal $\aprestwo$ in $\fsubstitution{\clause}$.
By Def.~\ref{def:groupingofpropositionsinresolutiongraph}, \ref{def:groupinginducedsubstitutionofpropositionsinresolutiongraph} there is an occurrence $\occurrence$ of a literal $\apresthree$ in $\clause$ such that $\fsubstitution{\vertex,\occurrence} = \aprestwo$.
By rule \refloopitinitc there is a corresponding occurrence of $\apresthree$ in $\clausepp$.
By Def.~\ref{def:groupingofpropositionsinresolutiongraph}, \ref{def:groupinginducedsubstitutionofpropositionsinresolutiongraph} $\fsubstitutionname$ maps that occurrence of $\apresthree$ to an occurrence of $\aprestwo$ in $\fsubstitution{\clausepp}$.
Let there be an occurrence of a literal $\apresfour \ne \apresone$ in $\fsubstitution{\clausepp}$.
By Def.~\ref{def:groupingofpropositionsinresolutiongraph}, \ref{def:groupinginducedsubstitutionofpropositionsinresolutiongraph} there is an occurrence $\occurrencep$ of a literal $\apresfive$ in $\clausepp$ such that $\fsubstitution{\vertexpp,\occurrencep} = \apresfour$.
By rule \refloopitinitc there is a corresponding occurrence of $\apresfive$ in $\clause$.
By Def.~\ref{def:groupingofpropositionsinresolutiongraph}, \ref{def:groupinginducedsubstitutionofpropositionsinresolutiongraph} $\fsubstitutionname$ maps that occurrence of $\apresfive$ to an occurrence of $\apresfour$ in $\fsubstitution{\clause}$.
Note that following the reasoning in Thm.~\ref{thm:coreisunsatopt} it is not necessary to show that $\fsubstitution{\clause}$ originated in the previous BFS loop search iteration.
Hence, this shows that $\fsubstitution{\clausepp}$ can be obtained by applying \refloopitinitc to $\fsubstitution{\clause}$ and $\fsubstitution{\clausep}$.

Now consider the case that the current BFS loop search iteration is the first iteration of a BFS loop search.
In that case there essentially is no premise 1 $\clause$.
Hence, this case is a trivial special case of the previous case.
\item[\refloopitsub]
Let $\vertex, \vertexp \in \setvertices$ with $\fvertexlabeling{\vertex} = \clause$ and $\fvertexlabeling{\vertexp} = \clausep$ where $\clausep$ is ``obtained'' from $\clause$ by applying \refloopitsub.
Let $\occurrence$ be the occurrence of the eventuality literal $\apreszero$ in $\clausep$ that the loop search in $\graph$ is for.
Let $\fsubstitution{\vertexp, \occurrence} = \apresone$.
By Rem.~\ref{thm:substitutiondoesnotchangetimeofliteral} $\fsubstitution{\clause}$ is a global clause with empty $\fnextname$ part and $\fsubstitution{\clausep}$ is a global clause with empty now part.
Using an inductive argument on the sequence in which clauses labeling vertices in $\graph$ are generated by \refalgltlsattrp one can see that $\fsubstitution{\clausep}$ was generated by an application of \refloopitinitc.
Let there be an occurrence of a literal $\aprestwo$ in $\fsubstitution{\clause}$.
By Def.~\ref{def:groupingofpropositionsinresolutiongraph}, \ref{def:groupinginducedsubstitutionofpropositionsinresolutiongraph} there is an occurrence $\occurrencep$ of a literal $\apresthree$ in $\clause$ such that $\fsubstitution{\vertex,\occurrencep} = \aprestwo$.
By rule \refloopitsub there is a corresponding occurrence of $\apresthree \ne \apreszero$ in $\clausep$.
By Def.~\ref{def:groupingofpropositionsinresolutiongraph}, \ref{def:groupinginducedsubstitutionofpropositionsinresolutiongraph} $\fsubstitutionname$ maps that occurrence of $\apresthree$ to an occurrence of $\aprestwo \ne \apresone$ in $\fsubstitution{\clausep}$.
This shows that $\fsubstitution{\clausep}$ can be obtained by applying \refloopitsub to $\fsubstitution{\clause}$.
\item[\refloopconclusionone]
Let $\vertex, \vertexp, \vertexpp \in \setvertices$ with $\fvertexlabeling{\vertex} = \clause$, $\fvertexlabeling{\vertexp} = \clausep$, and $\fvertexlabeling{\vertexpp} = \clausepp$ where $\clausepp$ is obtained by applying \refloopconclusionone to a global clause with empty $\fnextname$ part $\clause$ derived in a successful BFS loop search iteration and an eventuality clause $\clausep$ with eventuality literal $\apreszero$.
By Rem.~\ref{thm:substitutiondoesnotchangetimeofliteral} $\fsubstitution{\clausep}$ is an eventuality clause and $\fsubstitution{\clause}$ and $\fsubstitution{\clausepp}$ are global clauses with empty $\fnextname$ part.
By rule \refloopconclusionone there is a corresponding occurrence of $\apreszero$ in $\clausepp$.
By Def.~\ref{def:groupingofpropositionsinresolutiongraph}, \ref{def:groupinginducedsubstitutionofpropositionsinresolutiongraph} $\fsubstitutionname$ maps these occurrences of $\apreszero$ to occurrences of some literal $\apresone$ as the eventuality literal in $\fsubstitution{\clausep}$ and as a disjunct in $\fsubstitution{\clausepp}$.
Let there be an occurrence of a literal $\aprestwo$ in the now part of $\fsubstitution{\clause}$ (resp.~$\fsubstitution{\clausep}$).
By Def.~\ref{def:groupingofpropositionsinresolutiongraph}, \ref{def:groupinginducedsubstitutionofpropositionsinresolutiongraph} there is an occurrence $\occurrence$ of a literal $\apresthree$ in the now part of $\clause$ (resp.~$\clausep$) such that $\fsubstitution{\vertex,\occurrence} = \aprestwo$ (resp.~$\fsubstitution{\vertexp,\occurrence} = \aprestwo$).
By rule \refloopconclusionone there is a corresponding occurrence of $\apresthree$ in $\clausepp$.
By Def.~\ref{def:groupingofpropositionsinresolutiongraph}, \ref{def:groupinginducedsubstitutionofpropositionsinresolutiongraph} $\fsubstitutionname$ maps that occurrence of $\apresthree$ to an occurrence of $\aprestwo$ in $\fsubstitution{\clausepp}$.
Let there be an occurrence of a literal $\apresfour \ne \apresone$ in $\fsubstitution{\clausepp}$.
By Def.~\ref{def:groupingofpropositionsinresolutiongraph}, \ref{def:groupinginducedsubstitutionofpropositionsinresolutiongraph} there is an occurrence $\occurrencep$ of a literal $\apresfive$ in $\clausepp$ such that $\fsubstitution{\vertexpp,\occurrencep} = \apresfour$.
By rule \refloopconclusionone there is a corresponding occurrence of $\apresfive$ in the now part of $\clause$ or $\clausep$.
By Def.~\ref{def:groupingofpropositionsinresolutiongraph}, \ref{def:groupinginducedsubstitutionofpropositionsinresolutiongraph} $\fsubstitutionname$ maps that occurrence of $\apresfive$ to an occurrence of $\apresfour$ in the now part of $\fsubstitution{\clause}$ or $\fsubstitution{\clausep}$.
Except for the requirement that $\fsubstitution{\clause}$ was derived in a successful BFS loop search iteration this shows that $\fsubstitution{\clausepp}$ can be obtained by applying \refloopconclusionone to $\fsubstitution{\clause}$ and $\fsubstitution{\clausep}$.
\item[\refloopconclusiontwo]
Let $\vertex, \vertexp, \vertexpp \in \setvertices$ with $\fvertexlabeling{\vertex} = \clause$, $\fvertexlabeling{\vertexp} = \clausep$, and $\fvertexlabeling{\vertexpp} = \clausepp$ where $\clausepp$ is obtained by applying \refloopconclusiontwo to a global clause with empty $\fnextname$ part $\clause$ derived in a successful BFS loop search iteration and an eventuality clause $\clausep$ with eventuality literal $\apreszero$.
By Rem.~\ref{thm:substitutiondoesnotchangetimeofliteral} $\fsubstitution{\clause}$ is a global clause with empty $\fnextname$ part, $\fsubstitution{\clausep}$ is an eventuality clause, and $\fsubstitution{\clausepp}$ is a global clause.
By rule \refloopconclusiontwo there is a corresponding occurrence of $\apreszero$ in the $\fnextname$ part of $\clausepp$.
By Def.~\ref{def:groupingofpropositionsinresolutiongraph}, \ref{def:groupinginducedsubstitutionofpropositionsinresolutiongraph} $\fsubstitutionname$ maps these occurrences of $\apreszero$ to occurrences of some literal $\apresone$ as the eventuality literal in $\fsubstitution{\clausep}$ and as a disjunct in the $\fnextname$ part of $\fsubstitution{\clausepp}$.
By Def.~\ref{def:groupingofpropositionsinresolutiongraph}, \ref{def:groupinginducedsubstitutionofpropositionsinresolutiongraph} $\fsubstitutionname$ maps the occurrence of $\fbnot{\fapwaitfor{\apreszero}}$ in the now part of $\clausepp$ to an occurrence of some literal $\fbnot{\fapwaitfor{\apresone}}$ in the now part of $\fsubstitution{\clausepp}$.
By Def.~\ref{def:groupingofpropositionsinresolutiongraph}, \ref{def:groupinginducedsubstitutionofpropositionsinresolutiongraph} and the assumption above $\fapwaitfor{\apresone}$ is the fresh literal introduced by \refaugone and \refaugtwo for all eventuality clauses in $\graphgrouped$ with eventuality literal $\apresone$.
Let there be an occurrence of a literal $\aprestwo$ in $\fsubstitution{\clause}$.
By Def.~\ref{def:groupingofpropositionsinresolutiongraph}, \ref{def:groupinginducedsubstitutionofpropositionsinresolutiongraph} there is an occurrence $\occurrence$ of a literal $\apresthree$ in $\clause$ such that $\fsubstitution{\vertex,\occurrence} = \aprestwo$.
By rule \refloopconclusiontwo there is a corresponding occurrence of $\apresthree$ in the $\fnextname$ part of $\clausepp$.
By Def.~\ref{def:groupingofpropositionsinresolutiongraph}, \ref{def:groupinginducedsubstitutionofpropositionsinresolutiongraph} $\fsubstitutionname$ maps that occurrence of $\apresthree$ to an occurrence of $\aprestwo$ in the $\fnextname$ part of $\fsubstitution{\clausepp}$.
Let there be an occurrence of a literal $\apresfour \ne \apresone$ in the $\fnextname$ part of $\fsubstitution{\clausepp}$.
By Def.~\ref{def:groupingofpropositionsinresolutiongraph}, \ref{def:groupinginducedsubstitutionofpropositionsinresolutiongraph} there is an occurrence $\occurrencep$ of a literal $\apresfive$ in the $\fnextname$ part of $\clausepp$ such that $\fsubstitution{\vertexpp,\occurrencep} = \apresfour$.
By rule \refloopconclusiontwo there is a corresponding occurrence of $\apresfive$ in $\clause$.
By Def.~\ref{def:groupingofpropositionsinresolutiongraph}, \ref{def:groupinginducedsubstitutionofpropositionsinresolutiongraph} $\fsubstitutionname$ maps that occurrence of $\apresfive$ to an occurrence of $\apresfour$ in $\fsubstitution{\clause}$.
Except for the requirement that $\fsubstitution{\clause}$ was derived in a successful BFS loop search iteration this shows that $\fsubstitution{\clausepp}$ can be obtained by applying \refloopconclusiontwo to $\fsubstitution{\clause}$ and $\fsubstitution{\clausep}$.
\end{description}
It remains to show that a successful BFS loop search iteration in $\graph$ is also a successful BFS loop search iteration in $\graphgrouped$.
Let $\vertex \in \setvertices$ with $\fvertexlabeling{\vertex} = \clause$ such that $\fsubstitution{\clause}$ was generated by \refloopitinitc.
As shown above $\clause$ was also generated by \refloopitinitc.
If the BFS loop search iteration that $\clause$ is part of was successful, then there exists $\vertexp \in \setvertices$ with $\fvertexlabeling{\vertexp} = \clausep$ such that $\clausep$ is part of the same BFS loop search iteration as $\clause$ and $\clause$ and $\clausep$ are connected by an application of \refloopitsub.
As shown above $\fsubstitution{\clause}$ and $\fsubstitution{\clausep}$ are connected by an application of \refloopitsub.

Clearly, for any clause $\clause$, $\fsubstitution{\clause}$ cannot be larger than $\clause$.
Moreover, with Rem.~\ref{thm:groupmapsdifferentpropositionstodifferentgroups}, $\fsubstitution{\clause}$ cannot be smaller than $\clause$.
Hence, the empty clause (and, therefore, unsatisfiability) is derived in the proof on $\setclausesucgrouped$ at the same place as in the proof on $\setclausesuc$.
This concludes the proof.
\qed
\end{proof}

\subsubsection{Formalization --- LTL}

In order to map a \uc in SNF via \tr with grouped propositions to a \uc in LTL via \tr with grouped propositions we need to take the translation from LTL into SNF and back in Sec.~\ref{sec:translatingltlintosnf}, \ref{sec:extractingaucinltl} into account.
Assume an occurrence of a proposition $\ap$ in a subformula $\prt$.
Depending on $\prt$ that single occurrence of $\ap$ in $\prt$ may be translated into multiple occurrences of $\ap$ in $\fdCNF{\prt}$.
For example, if $\prt \definedas \funtil{\app}{\ap}$ under positive polarity, then $\prt$ is translated into 
$\fgnclause{\fbimplies{\fdCNFvar{\prt}}{(\fbor{\ap}{\app})}}$,
$(\fglobally{(\fbimplies{\fdCNFvar{\prt}}{(\fbor{\ap}{\fnext{\fdCNFvar{\prt}}})})})$, and
$\fgnclause{\fbimplies{\fdCNFvar{\prt}}{\ffinally{\ap}}}$.
Now, even if these occurrences of $\ap$ in the \uc of $\fdCNF{\prt}$ in SNF via \tr with grouped propositions are mapped to different groups (and therefore are substituted with different propositions in the \uc in SNF via \tr with grouped propositions), we must ensure that the occurrence of $\ap$ in the \uc in LTL via \tr with grouped propositions is substituted with a unique proposition such that unsatisfiability of the result is guaranteed.
This can be achieved by merging some groups in the mapping $\fgroupname$: if two occurrences of a proposition $\ap$ in the set of starting clauses $\fdCNF{\inp}$ were obtained from the same occurrence of $\ap$ in the LTL formula $\inp$, then these occurrences of $\ap$ must be mapped to the same group.
Notice that now all occurrences of a proposition $\ap$ in the \uc of $\inp$ in SNF via \tr that were obtained from the same occurrence of $\ap$ in $\inp$ are mapped to the same group.
This allows to transfer the grouping of occurrences of propositions from SNF to LTL.
This is stated more formally in Def.~\ref{def:groupingofpropositionsinresolutiongraphforltl}--\ref{def:ltlcoreextractiongroupedpropositions}.

\mydefinition{Grouping of Propositions in Resolution Graph for LTL}\label{def:groupingofpropositionsinresolutiongraphforltl}
Let $\inp$ be an LTL formula, let $\fdCNF{\inp}$ be the SNF of $\inp$, let $\graph$ be a resolution graph with set of vertices $\setvertices$ and labeling of vertices with SNF clauses $\vertexlabelingname$, and let $\fgroupname$ be the grouping of propositions in a resolution graph.
Let $\fgroupltlname$ be obtained from $\fgroupname$ by merging some groups as follows.
Let $\clause$ and $\clausep$ be two clauses obtained from translating the occurrence of a subformula $\prt$ in $\inp$, let $\vertex$ and $\vertexp$ be the unique vertices in the main partition of $\graph$ labeled with $\clause$ and $\clausep$, and let $\occurrence$ and $\occurrencep$ be two occurrences of the same proposition in $\clause$ and $\clausep$ that are marked {\color{blue}\setlength\fboxsep{1pt}\fbox{blue boxed}} in Tab.~\ref{tab:ltltosnf}.
Let $\fgroup{\vertex}{\occurrence} = \pos$ and $\fgroup{\vertexp}{\occurrencep} = \posp$.
Then all pairs $\mkpair{\vertexpp}{\occurrencepp}$ in the domain of $\fgroupname$ that are mapped to $\pos$ or to $\posp$ in $\fgroupname$ are mapped to the same group $\pospp$ in $\fgroupltlname$.
$\fgroupltlname$ is the \emph{grouping of propositions in a resolution graph for LTL}.
\end{definition}

Inspection of Tab.~\ref{tab:ltltosnf} shows that the only propositions whose groups might be merged because of Def.~\ref{def:groupingofpropositionsinresolutiongraphforltl} are propositions representing the right hand operand of a positive polarity occurrence of a $\funtilname$ formula or of a negative polarity occurrence of a $\freleasesname$ formula.

\mydefinition{Grouping-Induced Substitution of Propositions in Resolution Graph for LTL}\label{def:groupinginducedsubstitutionofpropositionsinresolutiongraphforltl}
Let $\inp$ be an unsatisfiable LTL formula, let $\fdCNF{\inp}$ be the SNF of $\inp$, let $\setclausesuc$ be a \uc of $\fdCNF{\inp}$ in SNF via \tr, let $\graph$ be the resolution graph with set of vertices $\setvertices$ and labeling of vertices with SNF clauses $\vertexlabelingname$, and let $\fgroupltlname$ be the grouping of propositions in a resolution graph for LTL.
Let $\fmapltlocctosnfoccsname$ be a mapping from occurrences of propositions in $\inp$ to sets of occurrences of propositions in $\fdCNF{\inp}$ that maps an occurrence of a proposition $\ap$ in $\inp$ to the set of occurrences of $\ap$ in $\fdCNF{\inp}$ that are induced by the occurrence of $\ap$ in $\inp$. 
Let $\fmapltlname$ be an injective mapping from the image of $\fgroupltlname$ to propositions $\allaps$: $\fmapltlname : \{\pos \in \allnats \mid \exists \vertex \in \setvertices \;.\; \exists \occurrence \mbox{ in } \fvertexlabeling{\vertex} \;.\; \pos = \fgroupltl{\vertex}{\occurrence}\} \rightarrow \allaps$ such that $\fmapltl{\pos} = \fmapltl{\posp} \proofimplies \pos = \posp$.
Extend $\fsubstitutionname$ to LTL as follows.
Use $\fmapltlocctosnfoccsname$ to map an occurrence $\occurrence$ of a proposition in $\inpuc$ to a set of occurrences $\setoccurrences$ in $\fdCNF{\inp}$.
Choose an occurrence $\occurrencep \in \setoccurrences$ that is also an occurrence in $\setclausesuc$.
Use $\fmapltlname \circ \fgroupltlname$ to map $\occurrencep$ to a proposition $\ap \in \allaps$.
The resulting mapping $\fsubstitutionltlname$ is the \emph{grouping-induced substitution of propositions for LTL}.
\end{definition}

We extend the domain of $\fsubstitutionltlname$ from occurrences of propositions in $\inpuc$ to $\inpuc$ in the natural way.

\mydefinition{\UC in LTL via \TR with Grouped Propositions}\label{def:ltlcoreextractiongroupedpropositions}
Let $\inp$ be an unsatisfiable LTL formula, let $\inpuc$ be its \uc in LTL obtained using Def.~\ref{def:ltlcoreextraction}, \ref{def:coreextractionopt}, and let $\fsubstitutionltlname$ be the grouping-induced substitution of propositions in a resolution graph for LTL.
Then $\inpucgrouped = \fsubstitutionltl{\inpuc}$ is the \emph{\uc of $\inp$ in LTL via \tr with grouped propositions}.
\end{definition}

\mytheorem%
{\thmltlcoregroupedpropositionsisunsat}%
{thm:ltlcoregroupedpropositionsisunsat}%
{Unsatisfiability of \UC in LTL via \TR with Grouped Propositions}%
{Let $\inp$ be an unsatisfiable LTL formula, and let $\inpucgrouped$ be a \uc of $\inp$ in LTL via \tr with grouped propositions. Then $\inpucgrouped$ is unsatisfiable.}
\thmltlcoregroupedpropositionsisunsat{true}

\begin{proof}
(Sketch.)
Let $\inpuc$ be the \uc of $\inp$ in LTL obtained using Def.~\ref{def:ltlcoreextraction}, \ref{def:coreextractionopt},
let $\fdCNF{\inpuc}$ be the SNF of $\inpuc$,
let $\setclausesuc$ be the \uc of $\inp$ in SNF via \tr,
let $\fdCNF{\inpucgrouped}$ be the SNF of $\inpucgrouped$,
let $\setclausesucgrouped$ be the \uc of $\inp$ in SNF via \tr with grouped propositions, and
let $\setclausesucgroupedltl$ be as $\setclausesucgrouped$ but using $\fgroupltlname$ instead of $\fgroupname$ in its construction.
By the proof of Thm.~\ref{thm:ltlcoreisunsat} $\fdCNF{\inpuc}$ is a superset of $\setclausesuc$.
By that fact and by the construction of $\inpucgrouped$ $\fdCNF{\inpucgrouped}$ is a superset of $\setclausesucgroupedltl$.
By Thm.~\ref{thm:coregroupedpropositionsisunsat} $\setclausesucgrouped$ and, thus, $\setclausesucgroupedltl$ is unsatisfiable.
Hence, so is $\inpucgrouped$.
This concludes the proof.
\qed
\end{proof}

\section{Relation to Mutual Vacuity}
\label{sec:vacuity}

In this section we explain that, under the frequently legitimate assumption that a system description can be translated into an LTL formula, our results extend to vacuity for LTL \cite{AGurfinkelMChechik-TACAS-2004,IBeerSBenDavidCEisnerYRodeh-FMSD-2001,OKupfermanMVardi-STTT-2003,RArmoniLFixAFlaisherOGrumbergNPitermanATiemeyerMVardi-CAV-2003,JSimmondsJDaviesAGurfinkelMChechik-STTT-2010,DFismanOKupfermanSSheinvaldFaragyMVardi-HVC-2008,OKupferman-CONCUR-2006}.
In \cite{AGurfinkelMChechik-TACAS-2004} Gurfinkel and Chechik introduce the notion of mutual vacuity.
We prove that the problems of finding a \uc of an unsatisfiable formula $\inp$ in LTL and of finding a set of subformula occurrences $\setoccurrences$ of an LTL specification $\spec$ such that $\spec$ is mutually vacuously $\true$ in $\setoccurrences$ in a system $\system$ can be reduced to each other.
%
%
For some more discussion on the relation between \ucs and vacuity see also \cite{VSchuppan-SCP-2012}.

Given a system description $\system$ and a specification $\spec$ formal verification (e.g., \cite{DPeled-2001}) proves or disproves that the system description $\system$ conforms to the specification $\spec$.
If the system description $\system$ conforms to the specification $\spec$, then we say that $\spec$ holds in $\system$ or that $\spec$ is $\true$ in $\system$.
For a specification $\spec$ given as an LTL formula $\spec$ is $\true$ in $\system$ iff every execution of $\system$ satisfies $\spec$.
Note that if the system description itself is given as an LTL formula, then conformance of $\system$ to $\spec$ corresponds to implication: $\spec$ is $\true$ in $\system$ iff $\fbimplies{\system}{\spec}$.

\myremark%
{\thmltlmctoltlsat}%
{thm:ltlmctoltlsat}%
{LTL Formal Verification as LTL Satisfiability}%
{Let $\system$ be a system description in LTL. Let $\spec$ be a specification in LTL. Then $\spec$ is $\true$ in $\system$ iff $\fband{\system}{\fbnot{\spec}}$ is unsatisfiable.}%
\thmltlmctoltlsat{true}

The fact that a system description conforms to a specification does not mean that all is well.
An example is antecedent failure \cite{DBeattyRBryant-DAC-1994}, where some specification $\spec \definedas \fglobally{(\fbimplies{\prt}{\prtp})}$ holds in a system description $\system$, but the antecedent $\prt$ never becomes $\true$ in any execution of $\system$.
In that case the consequent $\prtp$ plays no role in determining that $\spec$ holds in $\system$.
That often indicates presence of a problem in the specification or in the system description \cite{IBeerSBenDavidCEisnerYRodeh-FMSD-2001}.

Vacuity generalizes that idea as follows \cite{IBeerSBenDavidCEisnerYRodeh-FMSD-2001,OKupfermanMVardi-STTT-2003}.
Let $\system$ be a system description, let $\spec$ be an LTL specification such that $\spec$ is $\true$ in $\system$, and let $\prt$ be a positive (resp.~negative) polarity occurrence of a subformula in $\spec$.
Let $\specp$ be obtained from $\spec$ by replacing $\prt$ with $\false$ (resp.~$\true$).
If the modified specification $\specp$ still holds in $\system$, then apparently $\prt$ has no influence on $\spec$ being $\true$ in $\system$ in the following sense: $\prt$ could be replaced with any LTL formula in $\spec$ and the modified specification would still be $\true$ in $\system$.
In that case we say that $\spec$ is vacuously $\true$ in $\system$.
As an example consider a system $\system$ such that $\fband{\ap}{\fband{(\fbnot{\app})}{\fband{\appp}{(\fbnot{\apppp})}}}$ is an invariant of $\system$, i.e., it holds in every reachable state of $\system$.
Let the specification be $\spec \definedas \fglobally{(\fband{(\fbor{\ap}{\app})}{(\fbor{\appp}{\apppp})})}$.
$\spec$ is vacuously $\true$ in $\system$, as $\spec$ is $\true$ in $\system$ and either $\app$ or $\apppp$ could be replaced with $\false$ without falsifying $\spec$ in $\system$.

Mutual vacuity \cite{AGurfinkelMChechik-TACAS-2004} extends that idea to simultaneously replacing several subformulas with $\false$ or $\true$ depending on their polarity.
In our example $\app$ and $\apppp$ could simultaneously be replaced with $\false$ without falsifying $\spec$ in $\system$.
This is not the case for any other pair of propositions in $\spec$.
Definition \ref{def:mutualvacuity} formalizes that notion.
Proposition \ref{thm:mutualvacuityreducibletoucvv} then shows that the problems of determining mutual vacuity and of finding \ucs in LTL can be reduced to each other.

\mydefinition{Mutual Vacuity}\label{def:mutualvacuity}
Let $\system$ be a system description.
Let $\spec$ be a specification in LTL such that $\spec$ is $\true$ in $\system$.
Let $\setoccurrences$ be a set of disjoint subformula occurrences in $\spec$.
Then $\spec$ is \emph{mutually vacuously $\true$} in $\setoccurrences$ in $\system$ iff the modification $\specp$ of $\spec$ that replaces those members of $\setoccurrences$ that have positive (resp.~negative) polarity in $\spec$ with $\false$ (resp.~$\true$) is $\true$ in $\system$.
\end{definition}

\myproposition%
{\thmmutualvacuityreducibletoucvv}%
{thm:mutualvacuityreducibletoucvv}%
{Reducibility between Mutual Vacuity and \UCs in LTL}%
{The problems of finding a \uc of an unsatisfiable formula $\inp$ in LTL and of finding a set of subformula occurrences $\setoccurrences$ of an LTL specification $\spec$ that is $\true$ in an LTL system description $\system$ such that $\spec$ is mutually vacuously $\true$ in $\setoccurrences$ in $\system$ can be reduced to each other.}%
\thmmutualvacuityreducibletoucvv{true}

\begin{proof}
The proof is essentially by the respective definitions.

\begin{sloppypar}
Assume that $\spec$ is mutually vacuously $\true$ in $\setoccurrences$ in $\system$.
Let $\specp$ be $\spec$ with positive (resp.~negative) polarity members of $\setoccurrences$ replaced with $\false$ (resp.~$\true$).
Then
\begin{inparaenum}[(i)]
\item $\sysandnotspec \definedas \fband{\system}{\fbnot{\spec}}$ is unsatisfiable.
\item $\sysandnotspecp \definedas \fband{\system}{\fbnot{\specp}}$ is unsatisfiable.
\end{inparaenum}
$\mu'$ is a \uc of $\mu$ in LTL.
\end{sloppypar}

Assume that $\inp$ is an unsatisfiable LTL formula with \uc $\inpuc$.
Then there exists a non-empty set of subformula occurrences $\setoccurrencesp$ in $\inp$ such that $\inpuc$ is obtained from $\inp$ by replacing positive (resp.~negative) polarity members of $\setoccurrencesp$ with $\true$ (resp.~$\false$).
Now obviously
\begin{inparaenum}[(i)]
\item $\fband{\true}{\fbnot{\fbnot{\inp}}}$ is unsatisfiable and
\item $\fband{\true}{\fbnot{\fbnot{\inpuc}}}$ is unsatisfiable.
\end{inparaenum}
I.e., $\fbnot{\inp}$ is mutually vacuously $\true$ in $\setoccurrencesp$ in the unconstrained system $\true$.
\qed
\end{proof}

\begin{sloppypar}
A limited number of tools have been made available that can determine vacuity.
\tool{Aardvark} \cite{MPurandareTWahlDKroening-DATE-2009} computes --- depending on configuration --- maximal or maximum mutual vacuity for LTL using \tool{VIS} \cite{DBLP:conf/cav/BraytonHSSACEKKPQRSSSV96} as a backend.
\tool{Aardvark} uses binary search and counterexamples to reduce the search space in the lattice of candidate strengthened specifications; it does not use proofs for passing specifications to obtain an initial candidate strengthened specification.
Hence, our method is complementary.
We performed a small set of trials with \tool{Aardvark} on some of our benchmarks, mostly the smaller ones from each family.
Within this, admittedly limited, set of trials we ran into problems with usability (often getting assertion violations or segmentation faults rather than error messages pointing to potential problems in our input) as well as scalability.%
\footnote{%
Note that in vacuity checking it is typically assumed that the system description is more complex than the specification, while in \uc extraction all complexity is in the formula at hand.
When (as we did in our trials) using the reduction from LTL \uc extraction to vacuity checking from Prop.~\ref{thm:mutualvacuityreducibletoucvv} the resulting vacuity checking instance consists of a trivial system description and a complex specification.
In practice, when the vacuity checking procedure is tuned to take advantage of the small specification/complex system description scenario, then complex specifications may lead to problems; it seems that the scalability problems we observed with \tool{Aardvark} are caused to some extent by the translation from the LTL specification into an explicit B\"uchi automaton performed in \tool{VIS}.
}
We therefore opted not to perform a comparison on the full set of benchmarks.
\tool{VaqTree} \cite{JSimmondsJDaviesAGurfinkel-FM-2006} implements the method of \cite{JSimmondsJDaviesAGurfinkelMChechik-STTT-2010} that computes $k$-step vacuity for LTL, i.e., whether an occurrence of an atomic proposition is vacuous when bounded model checking runs only up to some bound $k$ are considered; the problem of removing the bound $k$ is left open.
The \tool{NuSMV Model Advisor}\footnote{\url{http://code.google.com/a/eclipselabs.org/p/nusmv-tools/downloads/detail?name=NuSMVModelAdvisor_20121012.zip}} computes the set of all occurrences of atomic propositions that are vacuous (but not necessarily mutually vacuous) for LTL.
\tool{VaqUoT} \cite{MGheorghiuAGurfinkel-FM-2006}, a simplified implementation of \cite{AGurfinkelMChechik-TACAS-2004}, computes the set of all occurrences of atomic propositions that are vacuous (but not necessarily mutually vacuous) for CTL.
%
%
%
%
%
%
A proof-based formulation of vacuity is suggested by Namjoshi \cite{KNamjoshi-CAV-2004}; no implementation or experiments are reported.
\end{sloppypar}

\section{Experimental Evaluation}
\label{sec:experimentalevaluation}

Our implementation, examples, and log files are available from \url{http://www.schuppan.de/viktor/actainformatica15/}.

\subsection{Implementation}

\begin{sloppypar}
We implemented extraction of \ucs as described in Sec.~\ref{sec:ucextraction} in \trp.
We also implemented deletion-based minimization of \ucs (Sec.~\ref{sec:minimalucs}) and grouping of propositions in \ucs (Sec.~\ref{sec:groupedpropositionsucs}).
\trp provides a translation from LTL into SNF via an external tool.
To facilitate tracing a \uc in SNF back to the input formula in LTL we implemented a translator from LTL into SNF inside \trp.
We used parts of \tspass\footnote{\url{http://www.csc.liv.ac.uk/~michel/software/tspass/}} \cite{MLudwigUHustadt-AICommunications-2010} for our implementation.
For data structures we used C++ Standard Library containers (e.g., \ifnoappendix\cite{NJosuttis-2012}\else\cite{AStepanovMLee-HPLaboratoriesTR-1995,NJosuttis-2012}\fi), for graph operations the Boost Graph Library\footnote{\url{http://www.boost.org/doc/libs/release/libs/graph/}} \cite{JSiekLLeeALumsdaine-2001}.
\end{sloppypar}

Our translator from LTL into SNF reimplements ideas from the external translator, among them
\begin{inparaenum}[(i)]
\item normalizing and simplifying the LTL formula before translation, 
\item sharing the translation of several occurrences of the same subformula, and
\item avoiding translation for SNF clauses in the input formula.
\end{inparaenum}
In addition, we implemented pure literal simplification for LTL \cite{ACimattiMRoveriVSchuppanSTonetta-CAV-2007} for SNF clauses.

On the one hand, these optimizations are crucial for good performance.
On the other hand, most optimizations change the input that is provided by the user before it is passed to our method for \uc extraction, which potentially makes it harder for the user to understand how the \uc she obtains corresponds to the input she provided.
Therefore, for each such optimization we provide a command line option that disables the optimization.
Mapping the optimized formula back to the input formula is left as future work.

Clearly, not solving a given LTL formula due to reaching run time or memory limits is the least useful result for a user.
Hence, we assume the following usage model.
A user will first enable all optimizations for the translation from LTL into SNF in order to solve a given LTL formula.
If she cannot understand how the \uc she obtains corresponds to the input she provided, she will selectively disable some of these optimizations.
Notice that at that stage the user may already be able to exclude some parts of the LTL formula from consideration and, therefore, provide a reduced input to the solver.
Based on these assumptions our implementation enables all optimizations for the translation from LTL into SNF by default.
With one exception discussed next we used these default settings for all experiments reported in this article.

In previous versions of this work (including \cite{VSchuppan-TIME-2013}) we used sharing of same polarity occurrences of a subformula in the translation from LTL into SNF.
This turned out to negatively impact grouping of propositions in a \uc.
Hence, for better comparison all results reported in this article were obtained with sharing of same polarity occurrences of a subformula in the translation from LTL into SNF disabled.
The data available from \url{http://www.schuppan.de/viktor/actainformatica15/} include results with both sharing disabled and sharing enabled.

By default our implementation uses optimized resolution graphs to extract a \uc (Def.~\ref{def:resolutiongraphtype}, \ref{def:resolutiongraphinitialization}, \ref{def:optimizedresolutiongraphupdate}, \ref{def:coreextractionopt}) and employs the optimizations in Rem.~\ref{thm:rgpruning}, \ref{thm:rgpruningopt}.
The results reported in this section were obtained with these optimizations enabled except in Sec.~\ref{sec:experimentalevaluation:optimizations}, where we evaluate the benefit of the optimizations.

\subsection{Benchmarks}

\begin{sloppypar}
Our examples are based on \cite{VSchuppanLDarmawan-ATVA-2011}.
In categories \benchmark{crafted} and \benchmark{random} and in family \benchmark{forobots} we considered all unsatisfiable instances from \cite{VSchuppanLDarmawan-ATVA-2011}.
The version of \benchmark{alaska\_lift} used here contains a small bug fix: in \cite{MDeWulfLDoyenNMaquetJRaskin-TACAS-2008,VSchuppanLDarmawan-ATVA-2011} the subformula $\fnext{u}$ was erroneously written as literal $Xu$.
Combining two variants of \benchmark{alaska\_lift} with three different scenarios we obtain six subfamilies of \benchmark{alaska\_lift}.
For \benchmark{anzu\_genbuf} we invented three scenarios to obtain three subfamilies.
\end{sloppypar}

For benchmark families in categories \benchmark{application} and \benchmark{crafted}, which consist of sequences of instances of increasing difficulty, we stopped after two instances that could not be solved due to time or memory out.
For benchmark families in category \benchmark{random}, which consist of sequences of sets of instances of increasing size, we stopped after no instance was solved in three consecutive size levels.
Some instances were simplified to $\false$ during the translation from LTL into SNF; these instances were discarded.

In Tab.~\ref{tab:benchmarks} we give an overview of the benchmark families.
Columns 1 and 2 give the name and the source of the family.
Columns 3, 5--7 list the numbers of instances that were solved by our implementation without \uc extraction, with \uc extraction, with minimal \uc extraction, and with grouped propositions \uc extraction.
Column 4 indicates the size (as the number of nodes in the syntax tree) of the largest instance solved without \uc extraction.

\begin{table}
\centering
\begin{minipage}{0.39\linewidth}\caption{\label{tab:benchmarks}Overview of benchmark families}\end{minipage}
{
\ifnoappendix
\def\tchs{0.5em}
\else
\def\tchs{0.15em}
\fi
\begin{tabular}{||@{\hspace{\tchs}}l@{\hspace{\tchs}}|@{\hspace{\tchs}}l@{\hspace{\tchs}}|@{\hspace{\tchs}}r@{\hspace{\tchs}}@{\hspace{\tchs}}r@{\hspace{\tchs}}|@{\hspace{\tchs}}r@{\hspace{\tchs}}|@{\hspace{\tchs}}r@{\hspace{\tchs}}|@{\hspace{\tchs}}r@{\hspace{\tchs}}||}
\hline
\hline
family &
source &
\multicolumn{1}{@{\hspace{0em}}l@{\hspace{0em}}}{solved} &
\multirow{4}*{$\left(\begin{array}{l}\mbox{size of}\\\mbox{largest}\\\mbox{solved}\\\mbox{instance}\end{array}\right)$} &
\multicolumn{1}{@{\hspace{0em}}l@{\hspace{\tchs}}|@{\hspace{\tchs}}}{solved} &
\multicolumn{1}{@{\hspace{0em}}l@{\hspace{\tchs}}|@{\hspace{\tchs}}}{solved} &
\multicolumn{1}{@{\hspace{0em}}l@{\hspace{\tchs}}||}{solved}
\\
&
&
\multicolumn{1}{@{\hspace{0em}}l@{\hspace{0em}}}{instances} &
&
\multicolumn{1}{@{\hspace{0em}}l@{\hspace{\tchs}}|@{\hspace{\tchs}}}{instances} &
\multicolumn{1}{@{\hspace{0em}}l@{\hspace{\tchs}}|@{\hspace{\tchs}}}{instances} &
\multicolumn{1}{@{\hspace{0em}}l@{\hspace{\tchs}}||}{instances}
\\
&
&
\multicolumn{1}{@{\hspace{0em}}l@{\hspace{0em}}}{no \uc} &
&
\multicolumn{1}{@{\hspace{0em}}l@{\hspace{\tchs}}|@{\hspace{\tchs}}}{\uc} &
\multicolumn{1}{@{\hspace{0em}}l@{\hspace{\tchs}}|@{\hspace{\tchs}}}{minimal \uc} &
\multicolumn{1}{@{\hspace{0em}}l@{\hspace{\tchs}}||}{grpd.~prop.~\uc}
\\
&
&
\multicolumn{1}{@{\hspace{0em}}l@{\hspace{0em}}}{extraction} &
&
\multicolumn{1}{@{\hspace{0em}}l@{\hspace{\tchs}}|@{\hspace{\tchs}}}{extraction} &
\multicolumn{1}{@{\hspace{0em}}l@{\hspace{\tchs}}|@{\hspace{\tchs}}}{extraction} &
\multicolumn{1}{@{\hspace{0em}}l@{\hspace{\tchs}}||}{extraction}
\\
\hline
\hline
\multicolumn{7}{||c||}{\benchmark{application}} \\
\hline
\benchmark{alaska\_lift} &
\cite{AHarding-PhDThesis-2005,MDeWulfLDoyenNMaquetJRaskin-TACAS-2008} &
70 &
(4605) &
69 &
69 &
69
\\
\hline
\benchmark{anzu\_genbuf} &
\cite{RBloemSGallerBJobstmannNPitermanAPnueliMWeiglhofer-COCV-2007} &
15 &
(1924) &
15 &
15 &
15
\\
\hline
\benchmark{forobots} &
\cite{ABehdennaCDixonMFisher-IntelligentComputingAndCybernetics-2009} &
25 &
(635) &
25 &
25 &
25
\\
\hline
\hline
\multicolumn{7}{||c||}{\benchmark{crafted}} \\
\hline
\benchmark{schuppan\_O1formula} &
\cite{VSchuppanLDarmawan-ATVA-2011} &
27 &
(4006) &
27 &
27 &
27
\\
\hline
\benchmark{schuppan\_O2formula} &
\cite{VSchuppanLDarmawan-ATVA-2011} &
8 &
(91) &
8 &
8 &
8
\\
\hline
\benchmark{schuppan\_phltl} &
\cite{VSchuppanLDarmawan-ATVA-2011} &
4 &
(125) &
4 &
4 &
4
\\
\hline
\hline
\multicolumn{7}{||c||}{\benchmark{random}} \\
\hline
\benchmark{rozier\_random} &
\cite{KRozierMVardi-STTT-2010} &
61 &
(155) &
61 &
61 &
61
\\
\hline
\benchmark{trp} &
\cite{UHustadtRSchmidt-KR-2002} &
397 &
(1421) &
397 &
330 &
397
\\
\hline
\hline
\end{tabular}
}
\end{table}

\subsection{Setup}

The experiments were performed on a laptop with an Intel Core i7 M 620 processor at 2 GHz running Ubuntu 14.04.
Run time and memory usage were measured with \tool{run}\footnote{\url{http://fmv.jku.at/run/}}.
The time and memory limits were 600 seconds and 6 GB.

\subsection{Extraction of \UCs}

In Fig.~\ref{fig:resultsextractionofucs} (a), (b) we show the overhead that is incurred by extracting \ucs as described in Sec.~\ref{sec:ucextraction} over not extracting \ucs.
In Fig.~\ref{fig:resultsextractionofucs} (c) we compare the sizes of the input formulas with the sizes of their \ucs.
For plots by category see \refmoreplots.

Out of the 749 instances of all categories we considered with \uc extraction disabled, 48 were simplified to $\false$ in the translation into SNF, 607 were shown to be unsatisfiable by \tr, and 94 remained unsolved.
Enabling \uc extraction results in one time out out of 607 instances.

The run time overhead for instances that take at least 0.7 seconds to solve without \uc extraction, except for those of the \benchmark{trp\_N12y} subfamily, is at most 65 \%.
Instances of the \benchmark{trp\_N12y} subfamily incur a run time overhead between 150 and 300 \%, which is the maximum overhead on our examples.

The memory overhead for instances that take at least 0.2 seconds to solve without \uc extraction is at most 526 \%.
Note that peak memory usage with \uc extraction was less than 200 MB.

In category \benchmark{application} three subfamilies of family \benchmark{alaska\_lift} and all subfamilies of family \benchmark{anzu\_genbuf} have \ucs of constant size, which are found by \trp.
The \ucs of instances of these two families are at least 76 \% smaller than the input formulas.
Instances of family \benchmark{forobots} are reduced between 47 \% and 97 \% with the median at 77 \%.
In category \benchmark{crafted} family \benchmark{schuppan\_O1formula} also has \ucs of constant size, which are found by \trp.
Instances of family \benchmark{schuppan\_O2formula} are themselves minimal \ucs.
Instances of \benchmark{schuppan\_phltl} have minimal \ucs in which one top-level conjunct is removed from the input formula, which are found by \trp without minimization.
In category \benchmark{random} instances of family \benchmark{rozier\_random} are reduced between 6 \% and 95 \% with the median at 75 \%.
Instances of family \benchmark{trp} exhibit minimum, median, and maximum reductions of 26 \%, 57 \%, and 88 \%, respectively.

Our data show that extraction of \ucs is possible with quite acceptable overhead in run time and memory usage.
The resulting \ucs are often significantly smaller than the input formula.

{
\newcommand{\tabfield}[1]{%
\begin{minipage}{0.304\linewidth}%
\begin{center}%
\includegraphics[type=pdf,ext=.pdf,read=.pdf,scale=0.33,trim=0.5cm 0.0cm 0.9cm 0.5cm,clip]{#1}%
\end{center}%
\end{minipage}%
}%

\begin{figure}
\centering
\begin{tabular}{ccc}
\tabfield{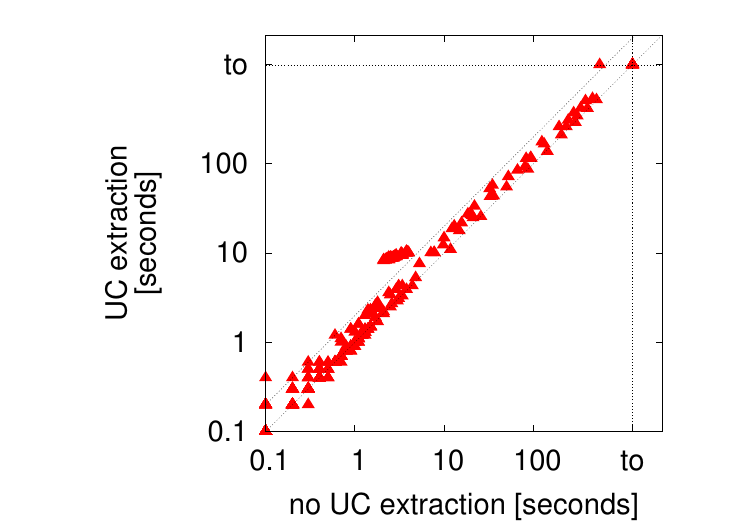} &
\tabfield{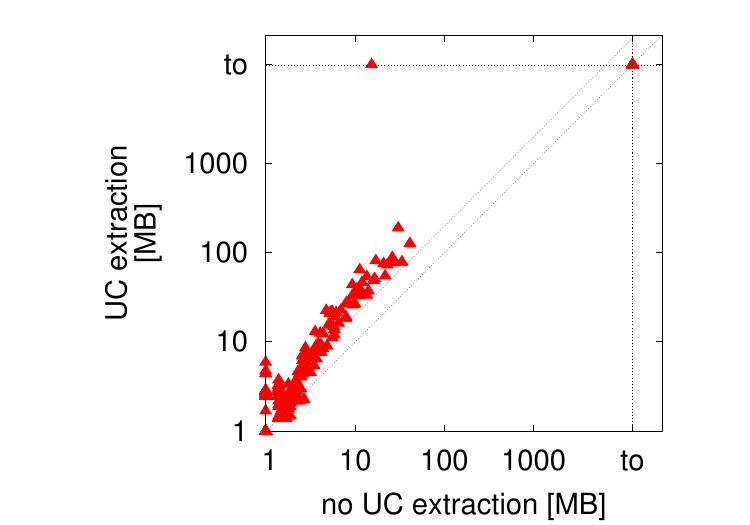} &
\tabfield{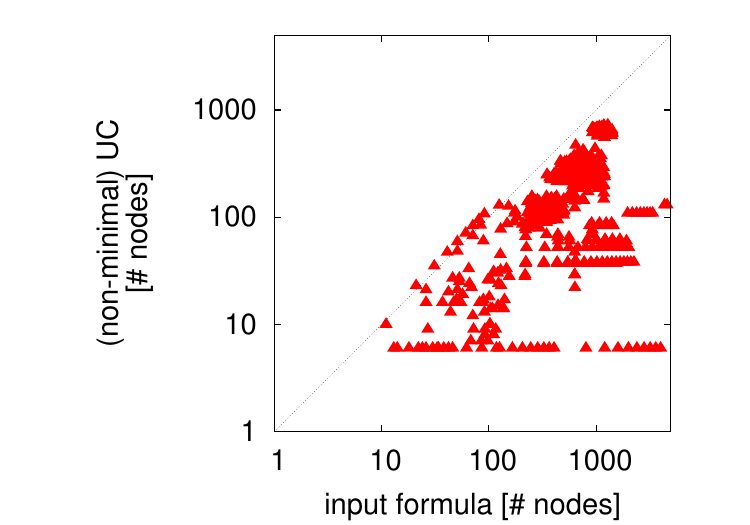}
\\[-1.5ex]
\\
\hspace{4em}(a) & \hspace{4em}(b) & \hspace{4em}(c)
\\[-2.5ex]
\\
\end{tabular}
\caption{\label{fig:resultsextractionofucs}
Comparison of \uc extraction (y-axis) with no \uc extraction (x-axis). (a) and (b) show the overhead incurred in terms of run time (in seconds) and memory (in MB). (c) shows the size reduction obtained, where size is measured as the number of nodes in the syntax trees. The off-center diagonal in (a) and (b) indicates where $y = 2 x$.}
\end{figure}
}

\subsection{Optimizations}
\label{sec:experimentalevaluation:optimizations}

In Fig.~\ref{fig:resultsoptimizations} we show the benefit of the optimizations described in Rem.~\ref{thm:rgpruning} and in Sec.~\ref{sec:extractingaucinsnffromanoptimizedresolutiongraph} when extracting \ucs.
We show the impact on the peak size of the resolution graph rather than on run time or memory, as the former is implementation independent.

The impact of including premise 1 of \refloopitinitc during the construction of the resolution graph and disabling immediate pruning of vertices and edges in partitions of failed loop search iterations from the resolution graph in Fig.~\ref{fig:resultsoptimizations} (b) (the former implies the latter) and of disabling pruning non-reachable vertices from the resolution graph between loop searches in Fig.~\ref{fig:resultsoptimizations} (e) is quite significant.
In Fig.~\ref{fig:resultsoptimizations} (b) the median increase of the peak size of the resolution graph is 39 \%, the maximum is 1339 \%.
In Fig.~\ref{fig:resultsoptimizations} (e) slightly more than half of all instances exhibit no change in the size of the resolution graph; for most of the remaining instances the peak size of the resolution graph increases by 20 \% or more with the maximum at 858 \%.

The impact in the remaining cases (Fig.~\ref{fig:resultsoptimizations} (a), (c), (d)) is negligible: for no instance the peak size of the resolution graph increases by more than 6 \%.
However, in cases (c) and (d) there is an instance where disabling the optimization leads to a larger \uc.
This occurs more often also in case (b).

We note that in each family of categories \benchmark{application} and \benchmark{random} there are instances for which the memory overhead of disabling all optimizations is larger than 1000 \%.
The highest memory overhead is more than 6500 \% for an instance of family \benchmark{alaska\_lift}.
The effect on run time is less pronounced and uniform.
When only instances are taken into account that take at least 0.3 seconds to solve with all optimizations enabled, then the run time overhead of disabling the optimizations is at most 40 \%.

{
\newcommand{\tabfield}[1]{%
\begin{minipage}{0.304\linewidth}%
\begin{center}%
\includegraphics[type=pdf,ext=.pdf,read=.pdf,scale=0.33,trim=0.5cm 0.0cm 0.9cm 0.5cm,clip]{#1}%
\end{center}%
\end{minipage}%
}%

\begin{figure}
\centering
\begin{tabular}{ccc}
\tabfield{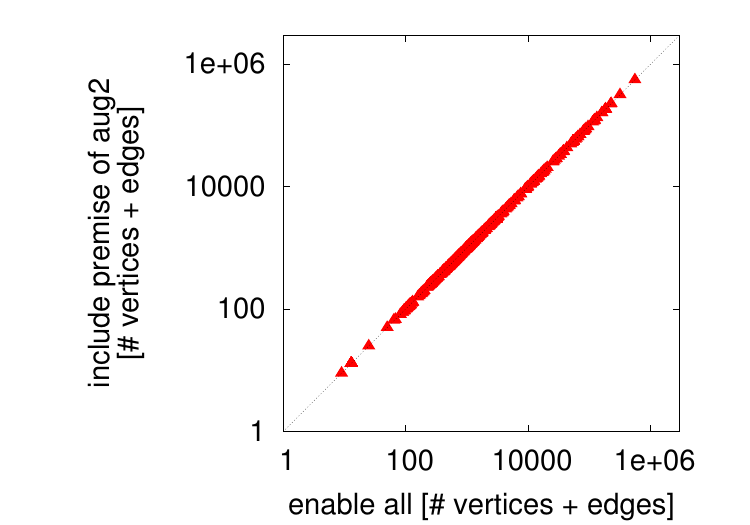} &
\tabfield{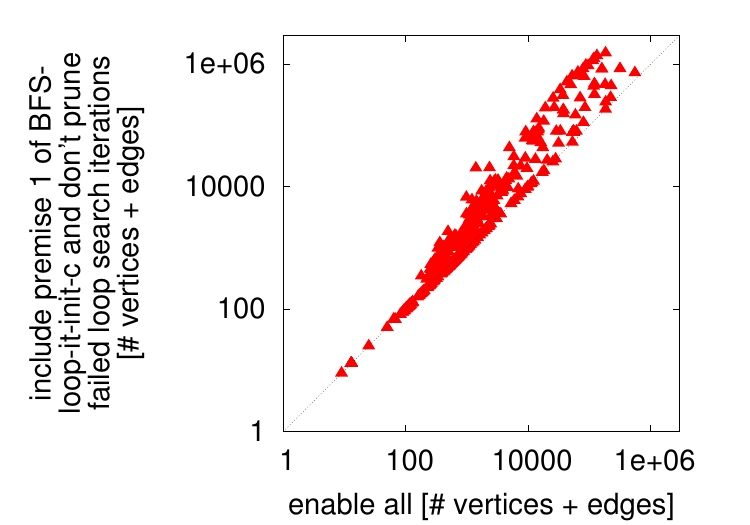} &
\tabfield{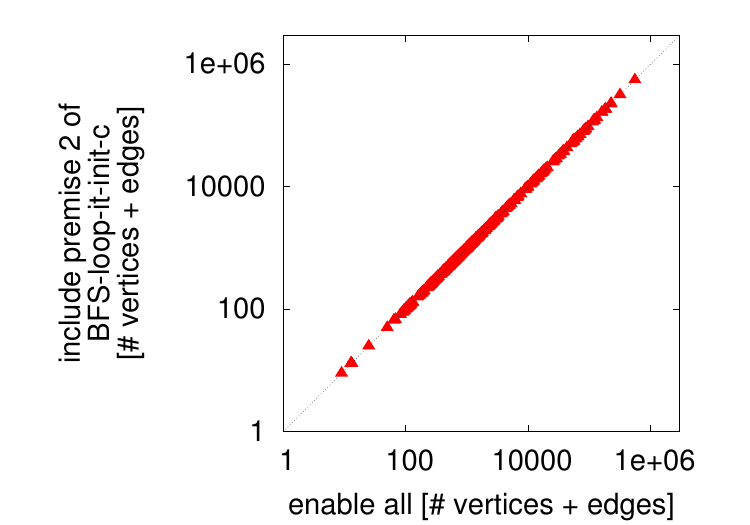}
\\[-1.5ex]
\\
\hspace{4.75em}(a) & \hspace{4.75em}(b) & \hspace{4.75em}(c)
\\[-0.75ex]
\\
\tabfield{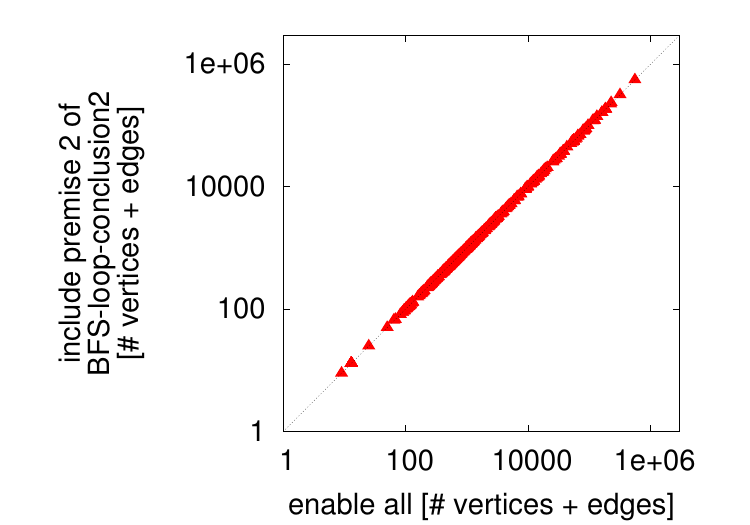} &
\tabfield{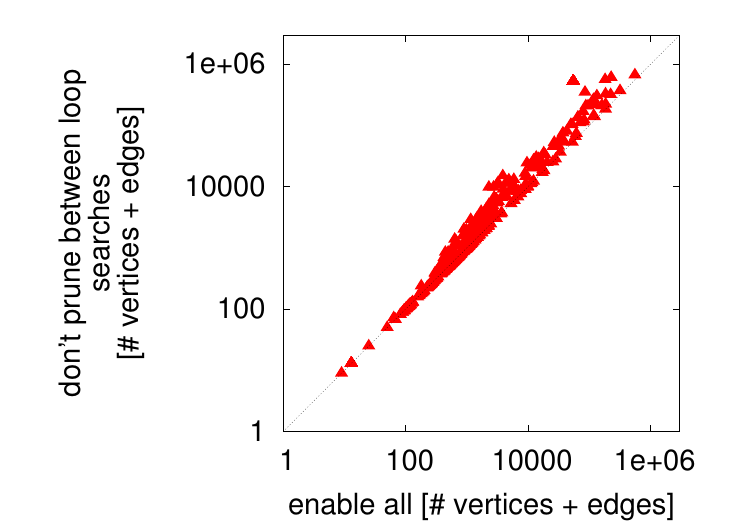} &
\tabfield{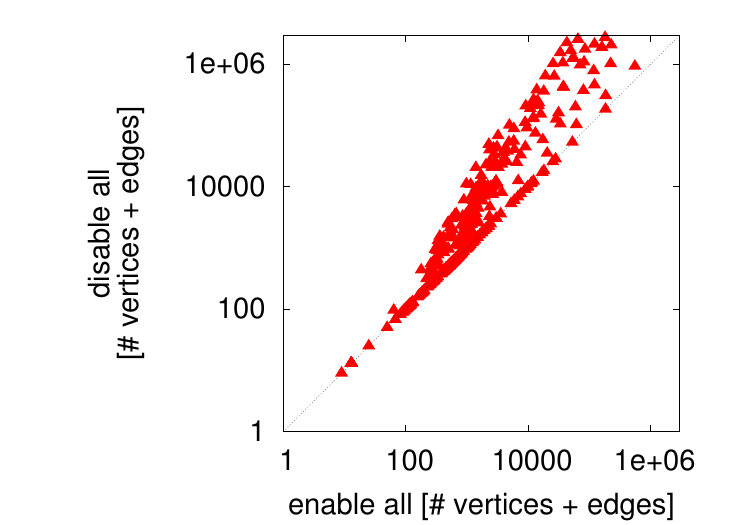} 
\\[-1.5ex]
\\
\hspace{4.75em}(d) & \hspace{4.75em}(e) & \hspace{4.75em}(f)
\\[-2.5ex]
\\
\end{tabular}
\caption{\label{fig:resultsoptimizations}
Benefit of optimizations as reduction in the peak size of the resolution graph (number of vertices + number of edges). The x-axis shows all optimizations enabled. The y-axis of (a)--(e) shows one optimization disabled: (a) include premise of \refaugtwo, (b) include premise 1 of \refloopitinitc and disable immediate pruning of failed loop search iterations, (c) include premise 2 of \refloopitinitc, (d) include premise 2 of \refloopconclusiontwo, (e) disable pruning of the resolution graph between loop searches. The y-axis of (f) shows all optimizations disabled.}
\end{figure}
}

\subsection{Extraction of Minimal \UCs}

Figure \ref{fig:resultsextractionofminimalucs} shows the costs and benefits of applying deletion-based minimization (Sec.~\ref{sec:minimalucs}) to (non-minimal) \ucs obtained as described in Sec.~\ref{sec:ucextraction}.

Costs and benefits are somewhat varied.
Minimal \ucs can be computed for all instances for which (non-minimal) \ucs were obtained except for all 67 instances in family \benchmark{trp\_N12y}.

A closer analysis shows that most instances with reductions of 30 \% or more are in the \benchmark{random} category; some are also in family \benchmark{forobots}.
The largest reduction seen is 62 \% from 469 to 177 nodes in the syntax tree of an instance from family \benchmark{trp\_N12x}.
5 instances of the \benchmark{forobots} family are reduced between 30 \% and 57 \%.
In the \benchmark{application} category several instances in each of the three families exhibit reductions in the range between 20 \% and 28 \%.
On the other hand, for three out of six subfamilies in the \benchmark{alaska\_lift} family, for one out of three subfamilies in the \benchmark{anzu\_genbuf} family, and for all families in the the \benchmark{crafted} category a minimal \uc was already obtained without deletion-based minimization (for the \benchmark{schuppan\_O2formula} family the original instances are minimal \ucs.).

The run time (resp.~memory) overhead to obtain a minimal \uc, except for instances of the \benchmark{trp} family that take 0.5 seconds or less to solve without deletion-based minimization, is at most 428 \% (resp.~87 \%). 

Hence, while based on our data there is no simple conclusion regarding the costs and benefits of applying deletion-based minimization, one should keep in mind that minimizing \ucs trades increased processing time by a computer for reduced analysis time by a human user, and that a minimal \uc has less potential to confuse a user about how parts of a \uc contribute to its unsatisfiability.

{
\newcommand{\tabfield}[1]{%
\begin{minipage}{0.304\linewidth}%
\begin{center}%
\includegraphics[type=pdf,ext=.pdf,read=.pdf,scale=0.33,trim=0.5cm 0.0cm 0.9cm 0.5cm,clip]{#1}%
\end{center}%
\end{minipage}%
}%

\begin{figure}
\centering
\begin{tabular}{ccc}
\tabfield{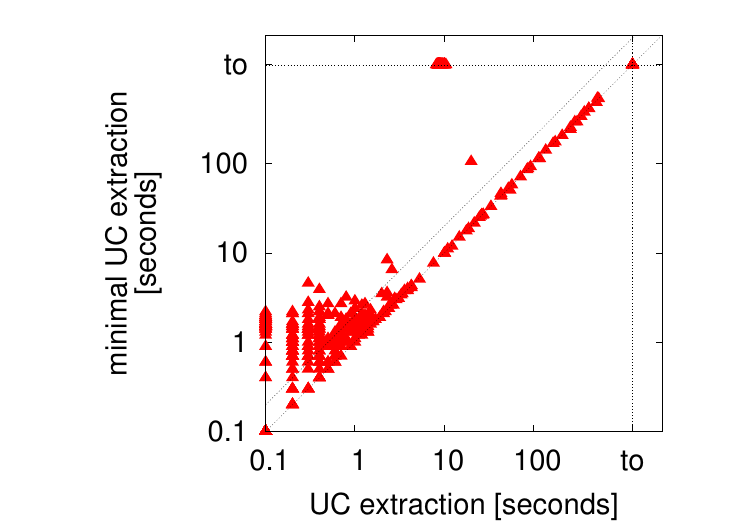} &
\tabfield{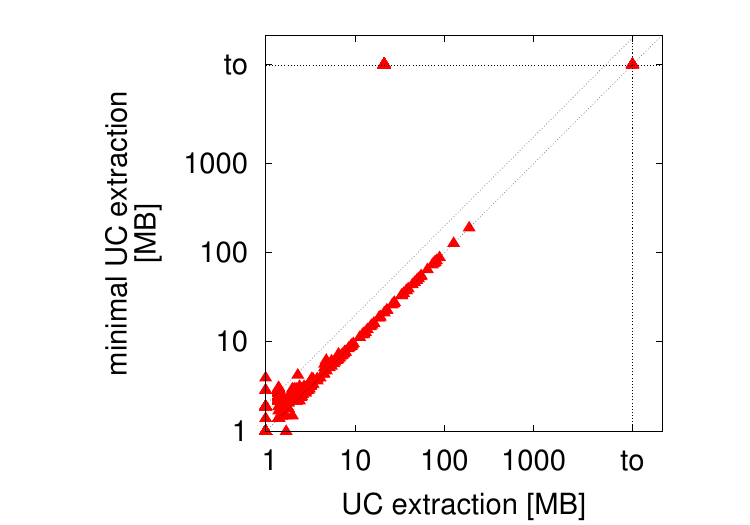} &
\tabfield{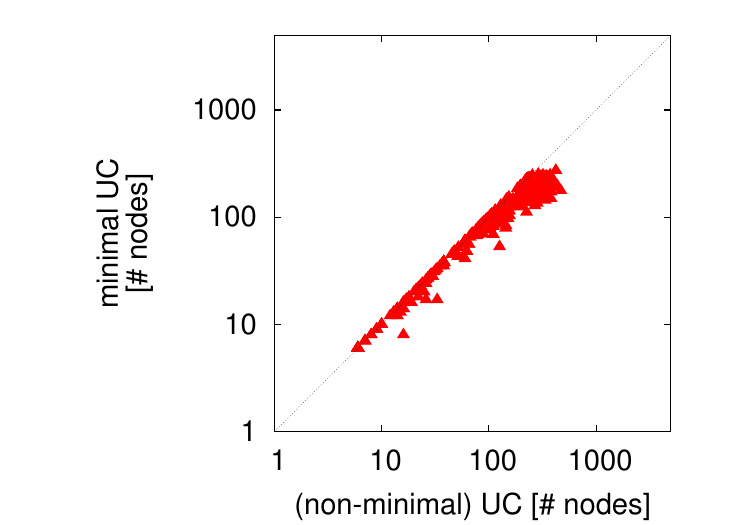}
\\[-1.5ex]
\\
\hspace{4em}(a) & \hspace{4em}(b) & \hspace{4em}(c)
\\[-2.5ex]
\\
\end{tabular}
\caption{\label{fig:resultsextractionofminimalucs}
Comparison of minimal \uc extraction (y-axis) with \uc extraction (x-axis). (a) and (b) show the overhead incurred in terms of run time (in seconds) and memory (in MB). (c) shows the size reduction obtained, where size is measured as the number of nodes in the syntax trees. The off-center diagonal in (a) and (b) indicates where $y = 2 x$.}
\end{figure}
}

\subsection{\UCs with Grouped Propositions}

In Tab.~\ref{tab:effectallgroupedpropositions} we show the effect of grouping propositions in \ucs.
The first column gives the name of the family.
The second column shows the maximum number of groups that the set of occurrences of some proposition in the \uc of a member of this benchmark family was partitioned into.
The number in parentheses shows the number of propositions (not occurrences of propositions!) in that \uc.
The third column indicates the maximum number of propositions whose occurrences in the \uc were partitioned into two or more groups for a member of this benchmark family.
The number in parentheses is as for the previous column.

Our data show that instances from the \benchmark{application} and \benchmark{random} categories exhibit significant potential for grouping propositions in \ucs.
The \benchmark{schuppan\_O1formula} and \benchmark{schuppan\_phltl} families in the \benchmark{crafted} category provide no opportunity for grouping propositions.
Manual inspection of the groupings obtained in the \benchmark{crafted} category proves them to be optimal.

\begin{table}
\centering
\begin{minipage}{0.44\linewidth}\caption{\label{tab:effectallgroupedpropositions}Effect of grouping propositions in \ucs}\end{minipage}
\begin{tabular}{||@{\hspace{0.6em}}l@{\hspace{0.6em}}|@{\hspace{0.6em}}r@{\hspace{0.6em}}|@{\hspace{0.6em}}r@{\hspace{0.6em}}||}
\hline
\hline
family &
max.~\# groups per prop.~in \uc &
max.~\# prop.~w.~$\ge 2$ groups per prop.~in \uc
\\
&
(\# prop.~in \uc) &
(\# prop.~in \uc)
\\
\hline
\hline
\multicolumn{3}{||c||}{\benchmark{application}} \\
\hline
\benchmark{alaska\_lift} &
3 (9) &
6 (18)
\\
\hline
\benchmark{anzu\_genbuf} &
2 (4) &
1 (4)
\\
\hline
\benchmark{forobots} &
2 (4) &
4 (18)
\\
\hline
\hline
\multicolumn{3}{||c||}{\benchmark{crafted}} \\
\hline
\benchmark{schuppan\_O1formula} &
1 (1) &
0 (1)
\\
\hline
\benchmark{schuppan\_O2formula} &
2 (3) &
8 (17)
\\
\hline
\benchmark{schuppan\_phltl} &
1 (2) &
0 (2)
\\
\hline
\hline
\multicolumn{3}{||c||}{\benchmark{random}} \\
\hline
\benchmark{rozier\_random} &
4 (4) &
2 (6)
\\
\hline
\benchmark{trp} &
7 (27) &
10 (27)
\\
\hline
\hline
\end{tabular}
\end{table}

\eqref{ex:lift_b_l3_3:c} shows the \uc with grouped propositions obtained from instance \benchmark{lift\_b\_l3\_3} in family \benchmark{alaska\_lift}.\footnote{The duplicate occurrence of $\fbnot{bf_1^1}$ in line 2 of \eqref{ex:lift_b_l3_3:c} is an artifact resulting from the simplification of \ucs by removing conjunction with $\true$ and disjunction with $\false$.}
We use superscripts ${}^0,{}^1,{}^2$ to distinguish groups of occurrences of the same proposition.
The example illustrates clearly that grouped propositions make it easier to see which occurrences of propositions interact.
\begin{equation}\label{ex:lift_b_l3_3:c}
\begin{array}{l}
\fband{(\fbnot{u^0})}{\fband{(\fbnot{up})}{\fband{(\fbnot{bf_0^0})}{\fband{(\fbnot{bf_1^0})}{\hspace{0em}}}}}\\
\fband{(\fglobally{(\fbor{bf_0^0}{\fbor{bf_1^0}{\fnext{(\fbor{(\fbnot{bf_1^1})}{\fbnot{bf_1^1}})}}})})}{\hspace{0em}}\\
\fband{(\fglobally{(\fbor{bf_0^1}{\fbor{(\fbnot{u^1})}{\fbnot{\fnext{bf_0^2}}}})})}{\hspace{0em}}\\
\fband{(\fglobally{(\fbor{bf_1^1}{\fbor{(\fbnot{u^1})}{\fbnot{\fnext{bf_1^2}}}})})}{\hspace{0em}}\\
\fband{(\fglobally{(\fbor{up}{\fbor{bf_0^0}{\fbor{bf_1^0}{\fbnot{\fnext{(\fband{bf_0^1}{\fbnot{bf_1^1}})}}}}})})}{\hspace{0em}}\\
\fband{(\fglobally{(\fbimplies{(\fbnot{\fnext{u^1}})}{u^0})})}{\hspace{0em}}\\
\fnext{\fnext{(\fbor{bf_0^2}{bf_1^2})}}
\end{array}
\end{equation}

In Fig.~\ref{fig:resultsallgroupedpropositions} we show the overhead that is incurred by the extraction of \ucs with grouped propositions over the extraction of \ucs without grouped propositions.

The run time (resp.\ memory) overhead for instances that take at least 0.4 seconds to solve without grouped propositions is at most 24 \% (resp.~111 \%).
The memory overhead is less than 30 \% for more than 93 \% of all instances.

The moderate run time and memory overhead seems justified given the potential benefits of grouped propositions illustrated by the example above.

{
\newcommand{\tabfield}[1]{%
\begin{minipage}{0.304\linewidth}%
\begin{center}%
\includegraphics[type=pdf,ext=.pdf,read=.pdf,scale=0.33,trim=0.5cm 0.0cm 0.9cm 0.5cm,clip]{#1}%
\end{center}%
\end{minipage}%
}%

\begin{figure}[b]
\centering
\begin{tabular}{cc}
\tabfield{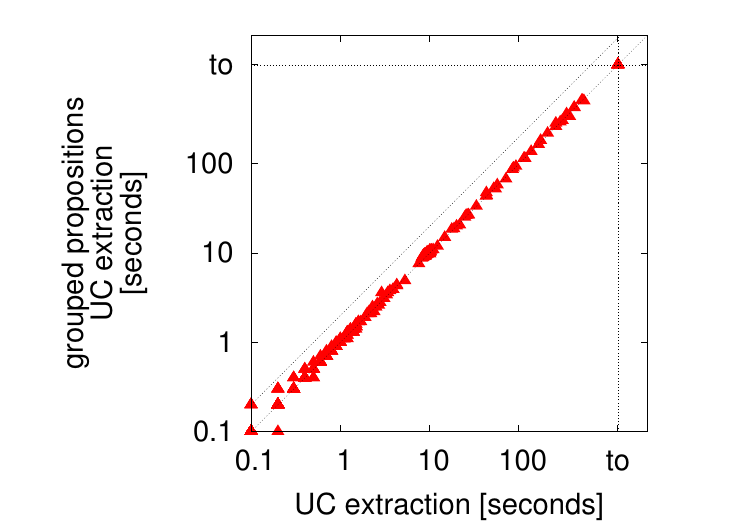} &
\tabfield{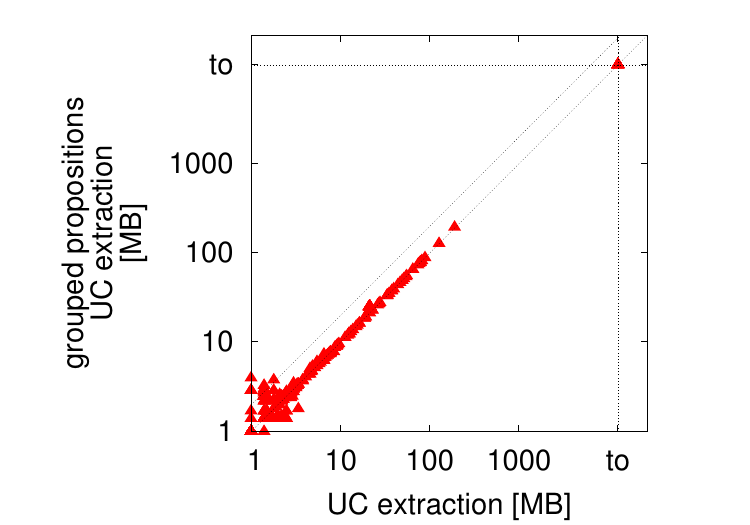}
\\[-1.5ex]
\\
\hspace{3.5em}(a) & \hspace{3.5em}(b)
\\[-2.5ex]
\\
\end{tabular}
\caption{\label{fig:resultsallgroupedpropositions}
Overhead of \uc extraction with grouped propositions (y-axis) over \uc extraction (x-axis) in terms of run time (in seconds, (a)) and memory (in MB, (b)). The off-center diagonal indicates where $y = 2 x$.}
\end{figure}
}

\subsection{Comparison with Related Approaches}

We now discuss \trp in comparison to \pltlmup \cite{RGoreJHuangTSergeantJThomson-Report-2013} and \procmine \cite{AAwadRGoreZHouJThomsonMWeidlich-InformationSystems-2012}.
Remember that both \pltlmup and \procmine compute an unsatisfiable subset of a set of LTL formulas rather than a \uc in LTL as a syntax tree according to Def.~\ref{def:snfltlmincore}.
I.e., compared to our notion of \uc \xspace \pltlmup and \procmine likely explore a smaller search space during minimization (if enabled) and possibly produce a less fine-grained result.

\trp, \pltlmup, and \procmine differ in important respects such as preprocessing of the input formula, algorithm used to obtain a \uc, algorithm used to minimize a \uc, programming language used for the implementation, and --- last but not least --- the notion of \uc (syntax tree vs.~subset of a set of formulas).
Therefore, we regard the value of a detailed discussion of the performance of \trp compared to \pltlmup and \procmine as somewhat questionable.
We only state that, when using the options for \trp as explained above and default options for \pltlmup and \procmine, and minimization in \trp enabled iff it was enabled in its competitor
\begin{inparaenum}[(i)]
\item \trp mostly outperforms \pltlmup and \procmine in category \benchmark{application} and in family \benchmark{trp} of category \benchmark{random},
\item in category \benchmark{crafted} \trp is outperformed by \pltlmup in one out of three families and by \procmine in most instances, and
\item in family \benchmark{rozier\_random} of category \benchmark{random} neither tool clearly dominates an other.
\end{inparaenum}
For plots please refer to \refmoreplots, for data to \url{http://www.schuppan.de/viktor/actainformatica15/}.

We finally illustrate the difference between our notion of \uc and the notion of \uc used by \pltlmup and \procmine on an example.
\eqref{ex:ucprocminevstrp:c} shows the \uc produced by \procmine for instance \benchmark{liftf\_l1\_2} from family \benchmark{alaska\_lift}.
The \uc contains 6 out of 15 top level conjuncts.
The parts that are additionally replaced with $\true$ by \trp due to our more fine-grained notion of \uc are marked {\color{blue}\setlength\fboxsep{1pt}\fbox{blue boxed}}.
For \pltlmup or \procmine to obtain a similar result, the user would have to rewrite the input.
More such examples can be found in families \benchmark{alaska\_lift}, \benchmark{forobots}, and \benchmark{rozier\_random} (\pltlmup and \procmine solve no instance in the remaining family in category \benchmark{application}, \benchmark{anzu\_genbuf}).
{
\newcommand{\rmuc}[1]{{\color{blue}\setlength\fboxsep{1pt}\fbox{{#1}}}}
\begin{equation}\label{ex:ucprocminevstrp:c}
\begin{array}{l}
f_0,\\
\fbnot{up},\\
\fglobally{(\fband{\rmuc{(\fbimplies{u}{(\fband{(\fbimplies{f_0}{\fnext{f_0}})}{\fband{(\fbimplies{(\fnext{f_0})}{f_0})}{\fband{(\fbimplies{f_1}{\fnext{f_1}})}{(\fbimplies{(\fnext{f_1})}{f_1})}}})})}}{\hspace{0em}}}\\
\phantom{\fgloballyname(}\fband{(\fbimplies{f_0}{\fnext{(\fbor{f_0}{f_1})}})}{\hspace{0em}}\\
\phantom{\fgloballyname(}\rmuc{(\fbimplies{f_1}{\fnext{(\fbor{f_0}{f_1})}})}),\\
\fglobally{(\fband{(\fbimplies{(\fband{f_0}{\fnext{f_0}})}{(\fband{\rmuc{(\fbimplies{up}{\fnext{up}})}}{(\fbimplies{(\fnext{up})}{up})})})}{\hspace{0em}}}\\
\phantom{\fgloballyname(}\fband{\rmuc{(\fbimplies{(\fband{f_1}{\fnext{f_1}})}{(\fband{(\fbimplies{up}{\fnext{up}})}{(\fbimplies{(\fnext{up})}{up})})})}}{\hspace{0em}}\\
\phantom{\fgloballyname(}\fband{(\fbimplies{(\fband{f_0}{\fnext{f_1}})}{up})}{\hspace{0em}}\\
\phantom{\fgloballyname(}\rmuc{(\fbimplies{(\fband{f_1}{\fnext{f_0}})}{\fbnot{up}})}),\\
\fglobally{(\fband{\rmuc{(\fbimplies{b_0}{\ffinally{f_0}})}}{(\fbimplies{b_1}{\ffinally{f_1}})})},\\
\fglobally{\ffinally{b_1}}\\
\end{array}
\end{equation}
}

All in all \trp produces more fine-grained \ucs than \pltlmup and \procmine while remaining at least competitive in terms of run time and memory usage.

\section{Conclusions}
\label{sec:conclusions}

In this article we showed how to obtain \ucs for LTL via temporal resolution, and we demonstrated with an implementation in \trp that \uc extraction can be performed efficiently.
The resulting \ucs are significantly smaller than the corresponding input formulas and more fine-grained than those produced by existing tools.
In parallel work \cite{VSchuppan-QAPL-2013} this article has been used as a basis to suggest enhancing \ucs for LTL with information on when subformulas of a \uc are relevant for unsatisfiability.
The similarity of temporal resolution and some BDD-based algorithms at a high level and work on resolution with BDDs (\cite{TJussilaCSinzABiere-SAT-2006}) suggests to explore whether computation of \ucs is feasible for BDD-based algorithms.
Another direction for transfer of our results is resolution-based computation of unrealizable cores \cite{PNoel-FAC-1995}.

\begin{acknowledgements}
I am grateful to Boris Konev and Michel Ludwig for making \trp and \tspass (which are the basis of this paper) including their LTL translators available and for answering my questions.
I thank Rajeev Gor\'e, Zhe Hou, Timothy Sergeant, and Jimmy Thomson for availability of and discussion about \pltlmup and \procmine as well as Daniel Kroening, Mitra Purandare, and Thomas Wahl for availability of and discussion about \tool{Aardvark}.
I thank Alessandro Cimatti for bringing up the subject of temporal resolution.
I also thank the reviewers of the current and previous iterations of this article for their helpful feedback.
Initial parts of the work were performed while working under a grant by the Provincia Autonoma di Trento (project EMTELOS).
\end{acknowledgements}


\begin{sloppypar}
\printbibliography
\end{sloppypar}
\ifnoappendix
\else
\clearpage
\appendix
\normalsize
\section{Proofs: \ref{sec:extractingaucinsnf} \UC Extraction}
\label{sec:formal-coreextraction}

\begin{sloppypar}
\thmminimalitypremises{false}
\end{sloppypar}

\begin{proof}
Below we provide the triples that were omitted in Sec.~\ref{sec:extractingaucinsnf}.
For the remainder of the proof see Sec.~\ref{sec:extractingaucinsnf}.

{
\newcommand{\tablineone}[2]{
\begin{samepage}
\noindent
\begin{tabular}{lc}
\hline
\\
\begin{minipage}[c]{0.65\linewidth}
{#1}
\end{minipage}
&
\begin{minipage}[c]{0.25\linewidth}
\begin{center}
{#2}
\end{center}
\end{minipage}
\end{tabular}
}

\newcommand{\tablinethree}{
\\
\begin{center}
}

\newcommand{\tablinefour}{
\end{center}
\vspace{2ex}
\end{samepage}

\hspace{0em}

}

\vspace{4ex}

\tablineone
{\refaugone}
{$
\begin{array}{l}
\hspace{-0.5em}\{\ficlause{a},\\
(\fglobally{(\fbor{(\fbnot{b})}{\fnext{\fbnot{c}}})}),\\
(\fglobally{(\fbor{(\fbnot{c})}{\fnext{c}})}),\\
\fgneclause{(\fbnot{a})}{\fbnot{c}},\\
\fgneclause{c}{b},\\
(\fglobally{\;\ffinally{\;c}})\}
\end{array}
$}
{}
\tablinethree
See Fig.~\ref{fig:rgaugone}.
\tablinefour
\begin{sidewaysfigure}
\vspace{98ex}
\centering
\adjustbox{scale=0.325,keepaspectratio=true}{
\iftrue
\input{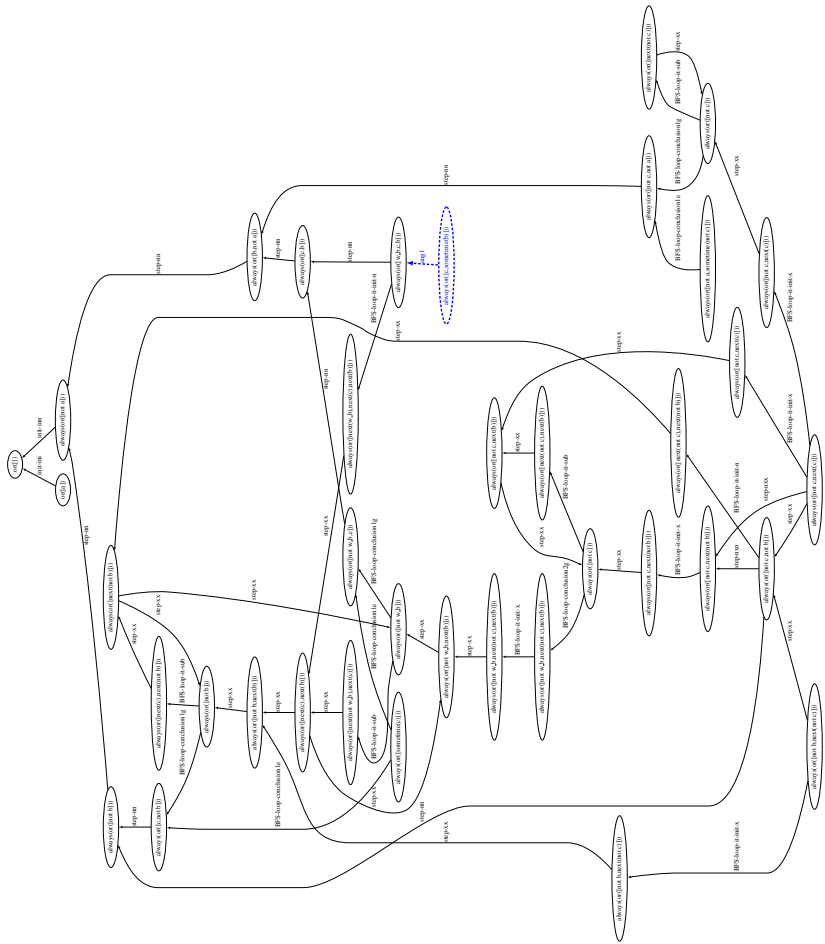}
\else
\begin{dot2tex}[options={--cache --graphstyle "scale=0.625"}]
digraph G {
graph [rankdir="BT"];
0[label="or([a])"];
1[label="always(or([not c,next(c)]))"];
2[label="always(or([not b,next(not c)]))"];
3[label="always(or([not a,sometime(not c)]))"];
4[label="always(or([c,sometime(b)]))",color=blue,style="ultra thick, densely dashed, text=blue"];
5[label="always(or([sometime(c)]))"];
6[label="always(or([not c,not b]))"];
7[label="always(or([not c,next(not b)]))"];
12[label="always(or([w_b,c,b]))"];
14[label="always(or([not c,next(c)]))"];
15[label="always(or([next(not c)]))"];
16[label="always(or([not c]))"];
18[label="always(or([not c,not a]))"];
21[label="always(or([not c,next(not b)]))"];
22[label="always(or([not c,next(c)]))"];
23[label="always(or([next(not c),next(b)]))"];
24[label="always(or([not c,next(b)]))"];
25[label="always(or([not c]))"];
26[label="always(or([not w_b,next(not c),next(b)]))"];
34[label="always(or([not b,next(not c)]))"];
36[label="always(or([not w_b,next(not c),next(b)]))"];
38[label="always(or([next(not c),next(not b)]))"];
40[label="always(or([next(w_b),next(c),next(b)]))"];
44[label="always(or([next(c),next(not b)]))"];
45[label="always(or([next(not w_b),next(c)]))"];
47[label="always(or([next(not b)]))"];
49[label="always(or([next(c),next(b)]))"];
50[label="always(or([not b,next(b)]))"];
53[label="always(or([not w_b,next(b)]))"];
54[label="always(or([not b]))"];
57[label="always(or([not w_b]))"];
59[label="always(or([c,not b]))"];
61[label="always(or([not w_b,c]))"];
64[label="always(or([not b]))"];
65[label="always(or([c,b]))"];
67[label="always(or([b,not a]))"];
70[label="always(or([not a]))"];
71[label="or([])"];
2->6 [label="step-xx"];
1->6 [label="step-xx"];
6->7 [label="step-nxn"];
1->7 [label="step-nxx"];
4->12 [label="aug1",color=blue,lblstyle="blue",style="ultra thick, densely dashed"];
1->14 [label="BFS-loop-it-init-x"];
15->16 [label="step-xx"];
14->16 [label="step-xx"];
16->15 [label="BFS-loop-it-sub"];
16->18 [label="BFS-loop-conclusion1g"];
3->18 [label="BFS-loop-conclusion1e"];
7->21 [label="BFS-loop-it-init-x"];
1->22 [label="BFS-loop-it-init-x"];
23->24 [label="step-xx"];
22->24 [label="step-xx"];
24->25 [label="step-xx"];
21->25 [label="step-xx"];
25->23 [label="BFS-loop-it-sub"];
25->26 [label="BFS-loop-conclusion2g"];
2->34 [label="BFS-loop-it-init-x"];
26->36 [label="BFS-loop-it-init-x"];
6->38 [label="BFS-loop-it-init-n"];
12->40 [label="BFS-loop-it-init-n"];
44->47 [label="step-xx"];
38->47 [label="step-xx"];
45->49 [label="step-xx"];
40->49 [label="step-xx"];
49->50 [label="step-xx"];
34->50 [label="step-xx"];
49->53 [label="step-xx"];
36->53 [label="step-xx"];
50->54 [label="step-xx"];
47->54 [label="step-xx"];
53->57 [label="step-xx"];
47->57 [label="step-xx"];
54->44 [label="BFS-loop-it-sub"];
57->45 [label="BFS-loop-it-sub"];
54->59 [label="BFS-loop-conclusion1g"];
5->59 [label="BFS-loop-conclusion1e"];
57->61 [label="BFS-loop-conclusion1g"];
5->61 [label="BFS-loop-conclusion1e"];
59->64 [label="step-nn"];
6->64 [label="step-nn"];
61->65 [label="step-nn"];
12->65 [label="step-nn"];
65->67 [label="step-nn"];
18->67 [label="step-nn"];
67->70 [label="step-nn"];
64->70 [label="step-nn"];
70->71 [label="init-inn"];
0->71 [label="init-ini"];
}
\end{dot2tex}
\fi
}
\begin{minipage}{0.52\linewidth}\caption{\label{fig:rgaugone} Subgraph $\graphp$ of resolution graph for the case of \refaugone in the proof of Proposition \ref{thm:minimalitypremises}}\end{minipage}
\end{sidewaysfigure}

\tablineone
{\refloopitinitn}
{$
\begin{array}{l}
\hspace{-0.5em}\{\ficlause{a},\\
\fgnclause{\fbor{(\fbnot{b})}{a}},\\
\fgnclause{\fbor{(\fbnot{b})}{(\fbnot{a})}},\\
\fgneclause{(\fbnot{a})}{b}\}
\end{array}
$}
\tablinethree
See Fig.~\ref{fig:rgloopitinitn}.
\tablinefour
\begin{sidewaysfigure}
\vspace{100ex}
\centering
\adjustbox{scale=0.67,keepaspectratio=true}{
\iftrue
\begin{tikzpicture}[>=latex,line join=bevel,scale=0.625]
\node (11) at (640bp,458bp) [draw,ellipse] {always(or([next(a)]))};
  \node (10) at (869bp,370bp) [draw,ellipse] {always(or([next(b)]))};
  \node (13) at (456bp,282bp) [draw,ellipse] {or([])};
  \node (12) at (456bp,194bp) [draw,ellipse] {always(or([next(not a)]))};
  \node (20) at (155bp,546bp) [draw,ellipse] {or([])};
  \node (19) at (237bp,458bp) [draw,ellipse] {always(or([not a]))};
  \node (15) at (166bp,370bp) [draw,ellipse] {always(or([b,not a]))};
  \node (1) at (640bp,282bp) [draw=blue,ellipse,ultra thick, densely dashed, text=blue] {always(or([not b,a]))};
  \node (0) at (108bp,458bp) [draw,ellipse] {or([a])};
  \node (3) at (125bp,282bp) [draw,ellipse] {always(or([not a,sometime(b)]))};
  \node (2) at (328bp,18bp) [draw,ellipse] {always(or([not b,not a]))};
  \node (7) at (640bp,370bp) [draw,ellipse] {always(or([next(not b),next(a)]))};
  \node (8) at (448bp,106bp) [draw,ellipse] {always(or([next(not b),next(not a)]))};
  \draw [->] (2) ..controls (296.57bp,52.016bp) and (278bp,78.058bp)  .. (278bp,105bp) .. controls (278bp,105bp) and (278bp,105bp)  .. (278bp,371bp) .. controls (278bp,392.91bp) and (267.02bp,415.17bp)  .. (19);
  \definecolor{strokecol}{rgb}{0.0,0.0,0.0};
  \pgfsetstrokecolor{strokecol}
  \draw (298bp,238bp) node {step-nn};
  \draw [->] (12) ..controls (456bp,224.25bp) and (456bp,240.18bp)  .. (13);
  \draw (476bp,238bp) node {step-xx};
  \draw [->] (10) ..controls (827.87bp,328.86bp) and (780.79bp,287.13bp)  .. (733bp,264bp) .. controls (672.72bp,234.82bp) and (598.89bp,217.22bp)  .. (12);
  \draw (810bp,282bp) node {step-xx};
  \draw [->] (13) ..controls (399.25bp,295.06bp) and (352.02bp,305.9bp)  .. (312bp,318bp) .. controls (292.08bp,324.02bp) and (287.62bp,327.07bp)  .. (268bp,334bp) .. controls (252.07bp,339.63bp) and (234.69bp,345.64bp)  .. (15);
  \draw (377bp,326bp) node {BFS-loop-conclusion1g};
  \draw [->] (15) ..controls (190.43bp,400.59bp) and (204.78bp,417.97bp)  .. (19);
  \draw (228bp,414bp) node {step-nn};
  \draw [->] (3) ..controls (127.3bp,310.5bp) and (129.8bp,323.49bp)  .. (135bp,334bp) .. controls (136.77bp,337.58bp) and (139.01bp,341.08bp)  .. (15);
  \draw (199.5bp,326bp) node {BFS-loop-conclusion1e};
  \draw [->] (13) ..controls (492.16bp,308.42bp) and (520.55bp,325.73bp)  .. (548bp,334bp) .. controls (644.9bp,363.2bp) and (674.64bp,338.98bp)  .. (775bp,352bp) .. controls (782.75bp,353.01bp) and (790.83bp,354.22bp)  .. (10);
  \draw (592bp,326bp) node {BFS-loop-it-sub};
  \draw [->] (11) ..controls (533.04bp,438.5bp) and (489.78bp,421.26bp)  .. (465bp,388bp) .. controls (448.49bp,365.84bp) and (448.33bp,333.26bp)  .. (13);
  \draw (485bp,370bp) node {step-xx};
  \draw [blue,->,ultra thick,densely dashed] (1) ..controls (640bp,312.25bp) and (640bp,328.18bp)  .. (7);
  \draw (689.5bp,326bp) node[blue] {BFS-loop-it-init-n};
  \draw [->] (2) ..controls (369.61bp,48.821bp) and (395.7bp,67.516bp)  .. (8);
  \draw (447.5bp,62bp) node {BFS-loop-it-init-n};
  \draw [->] (10) ..controls (790.24bp,400.58bp) and (731.04bp,422.81bp)  .. (11);
  \draw (795bp,414bp) node {step-xx};
  \draw [->] (8) ..controls (450.72bp,136.25bp) and (452.2bp,152.18bp)  .. (12);
  \draw (472bp,150bp) node {step-xx};
  \draw [->] (7) ..controls (640bp,400.25bp) and (640bp,416.18bp)  .. (11);
  \draw (660bp,414bp) node {step-xx};
  \draw [->] (0) ..controls (123.93bp,488.15bp) and (133.49bp,505.65bp)  .. (20);
  \draw (154.5bp,502bp) node {init-ini};
  \draw [->] (19) ..controls (208.12bp,489.29bp) and (190.37bp,507.91bp)  .. (20);
  \draw (224bp,502bp) node {init-inn};
\end{tikzpicture}
\else
\begin{dot2tex}[options={--cache --graphstyle "scale=0.625"}]
digraph G {
graph [rankdir="BT"];
0[label="or([a])"];
1[label="always(or([not b,a]))",color=blue,style="ultra thick, densely dashed, text=blue"];
2[label="always(or([not b,not a]))"];
3[label="always(or([not a,sometime(b)]))"];
7[label="always(or([next(not b),next(a)]))"];
8[label="always(or([next(not b),next(not a)]))"];
10[label="always(or([next(b)]))"];
11[label="always(or([next(a)]))"];
12[label="always(or([next(not a)]))"];
13[label="or([])"];
15[label="always(or([b,not a]))"];
19[label="always(or([not a]))"];
20[label="or([])"];
1->7 [label="BFS-loop-it-init-n",color=blue,lblstyle="blue",style="ultra thick, densely dashed"];
2->8 [label="BFS-loop-it-init-n"];
10->11 [label="step-xx"];
7->11 [label="step-xx"];
10->12 [label="step-xx"];
8->12 [label="step-xx"];
12->13 [label="step-xx"];
11->13 [label="step-xx"];
13->10 [label="BFS-loop-it-sub"];
13->15 [label="BFS-loop-conclusion1g"];
3->15 [label="BFS-loop-conclusion1e"];
15->19 [label="step-nn"];
2->19 [label="step-nn"];
19->20 [label="init-inn"];
0->20 [label="init-ini"];
}
\end{dot2tex}
\fi
}
\begin{minipage}{0.59\linewidth}\caption{\label{fig:rgloopitinitn} Subgraph $\graphp$ of resolution graph for the case of \refloopitinitn in the proof of Proposition \ref{thm:minimalitypremises}}\end{minipage}
\end{sidewaysfigure}

\tablineone
{\refloopitsub}
{$
\begin{array}{l}
\hspace{-0.5em}\{\ficlause{a},\\
(\fglobally{(\fbor{(\fbnot{a})}{\fnext{b}})}),\\
(\fglobally{(\fbor{(\fbnot{b})}{\fnext{a}})}),\\
(\fglobally{\;\ffinally{\;c}}),\\
\fgnclause{\fbor{(\fbnot{c})}{(\fbnot{a})}},\\
(\fglobally{(\fbor{(\fbnot{a})}{\fnext{\fbnot{c}}})})\}
\end{array}
$}
\tablinethree
See Fig.~\ref{fig:rgloopitsub}.
\tablinefour
\begin{sidewaysfigure}
\vspace{98ex}
\centering
\adjustbox{scale=0.6,keepaspectratio=true}{
\iftrue
\begin{tikzpicture}[>=latex,line join=bevel,scale=0.625]
\node (24) at (80bp,810bp) [draw,ellipse] {or([])};
  \node (11) at (224bp,106bp) [draw,ellipse] {always(or([next(not c),next(not a)]))};
  \node (13) at (217bp,634bp) [draw,ellipse] {always(or([next(c),next(not a)]))};
  \node (20) at (445bp,634bp) [draw,ellipse] {always(or([c,not a]))};
  \node (14) at (371bp,370bp) [draw,ellipse] {always(or([next(c),next(not b)]))};
  \node (17) at (391bp,458bp) [draw,ellipse] {always(or([not a,next(not b)]))};
  \node (23) at (162bp,722bp) [draw,ellipse] {always(or([not a]))};
  \node (18) at (391bp,546bp) [draw,ellipse] {always(or([not a]))};
  \node (16) at (371bp,282bp) [draw,ellipse] {always(or([not b]))};
  \node (15) at (224bp,194bp) [draw,ellipse] {always(or([next(not a)]))};
  \node (1) at (106bp,18bp) [draw,ellipse] {always(or([not c,not a]))};
  \node (0) at (33bp,722bp) [draw,ellipse] {or([a])};
  \node (3) at (875bp,370bp) [draw,ellipse] {always(or([not a,next(b)]))};
  \node (2) at (487bp,106bp) [draw=blue,ellipse,ultra thick, densely dashed, text=blue] {always(or([not b,next(a)]))};
  \node (5) at (591bp,546bp) [draw,ellipse] {always(or([sometime(c)]))};
  \node (4) at (633bp,282bp) [draw,ellipse] {always(or([not a,next(not c)]))};
  \node (9) at (754bp,458bp) [draw,ellipse] {always(or([not a,next(b)]))};
  \node (8) at (467bp,194bp) [draw,ellipse] {always(or([not b,next(a)]))};
  \node (10) at (633bp,370bp) [draw,ellipse] {always(or([not a,next(not c)]))};
  \draw [->] (9) ..controls (628.37bp,488.76bp) and (520.85bp,514.24bp)  .. (18);
  \definecolor{strokecol}{rgb}{0.0,0.0,0.0};
  \pgfsetstrokecolor{strokecol}
  \draw (625bp,502bp) node {step-xx};
  \draw [->] (15) ..controls (275.54bp,225.15bp) and (309.36bp,244.94bp)  .. (16);
  \draw (330bp,238bp) node {step-xx};
  \draw [->] (5) ..controls (576.88bp,575.46bp) and (567.74bp,589.31bp)  .. (556bp,598bp) .. controls (544.13bp,606.79bp) and (529.98bp,613.42bp)  .. (20);
  \draw (635.5bp,590bp) node {BFS-loop-conclusion1e};
  \draw [blue,->,ultra thick,densely dashed] (16) ..controls (371bp,312.25bp) and (371bp,328.18bp)  .. (14);
  \draw (415bp,326bp) node[blue] {BFS-loop-it-sub};
  \draw [->] (14) ..controls (377.83bp,400.37bp) and (381.58bp,416.5bp)  .. (17);
  \draw (402bp,414bp) node {step-xx};
  \draw [->] (4) ..controls (633bp,312.25bp) and (633bp,328.18bp)  .. (10);
  \draw (682.5bp,326bp) node {BFS-loop-it-init-x};
  \draw [->] (0) ..controls (48.933bp,752.15bp) and (58.492bp,769.65bp)  .. (24);
  \draw (79.5bp,766bp) node {init-ini};
  \draw [->] (11) ..controls (224bp,136.25bp) and (224bp,152.18bp)  .. (15);
  \draw (244bp,150bp) node {step-xx};
  \draw [->] (17) ..controls (391bp,488.25bp) and (391bp,504.18bp)  .. (18);
  \draw (411bp,502bp) node {step-xx};
  \draw [->] (8) ..controls (433.48bp,225.03bp) and (413.21bp,243.19bp)  .. (16);
  \draw (447bp,238bp) node {step-xx};
  \draw [->] (18) ..controls (331.53bp,576.39bp) and (290.46bp,596.69bp)  .. (13);
  \draw (363bp,590bp) node {BFS-loop-it-sub};
  \draw [->] (2) ..controls (480.27bp,135.93bp) and (476.5bp,152.15bp)  .. (8);
  \draw (527.5bp,150bp) node {BFS-loop-it-init-x};
  \draw [->] (10) ..controls (548.15bp,401.15bp) and (488.26bp,422.44bp)  .. (17);
  \draw (553bp,414bp) node {step-xx};
  \draw [->] (18) ..controls (409.43bp,576.36bp) and (420.09bp,593.32bp)  .. (20);
  \draw (487bp,590bp) node {BFS-loop-conclusion1g};
  \draw [->] (3) ..controls (832.96bp,400.88bp) and (806.5bp,419.68bp)  .. (9);
  \draw (874.5bp,414bp) node {BFS-loop-it-init-x};
  \draw [->] (23) ..controls (133.12bp,753.29bp) and (115.37bp,771.91bp)  .. (24);
  \draw (149bp,766bp) node {init-inn};
  \draw [->] (1) ..controls (146.92bp,48.821bp) and (172.57bp,67.516bp)  .. (11);
  \draw (225.5bp,62bp) node {BFS-loop-it-init-n};
  \draw [->] (20) ..controls (349.06bp,664.15bp) and (269.67bp,688.28bp)  .. (23);
  \draw (348bp,678bp) node {step-nn};
  \draw [->] (13) ..controls (217bp,598.51bp) and (217bp,570.89bp)  .. (217bp,547bp) .. controls (217bp,281bp) and (217bp,281bp)  .. (217bp,281bp) .. controls (217bp,261.34bp) and (218.84bp,239.23bp)  .. (15);
  \draw (237bp,414bp) node {step-xx};
  \draw [->] (1) ..controls (74.567bp,52.016bp) and (56bp,78.058bp)  .. (56bp,105bp) .. controls (56bp,105bp) and (56bp,105bp)  .. (56bp,635bp) .. controls (56bp,666.4bp) and (84.404bp,688.34bp)  .. (23);
  \draw (76bp,370bp) node {step-nn};
\end{tikzpicture}
\else
\begin{dot2tex}[options={--cache --graphstyle "scale=0.625"}]
digraph G {
graph [rankdir="BT"];
0[label="or([a])"];
1[label="always(or([not c,not a]))"];
2[label="always(or([not b,next(a)]))",color=blue,style="ultra thick, densely dashed, text=blue"];
3[label="always(or([not a,next(b)]))"];
4[label="always(or([not a,next(not c)]))"];
5[label="always(or([sometime(c)]))"];
8[label="always(or([not b,next(a)]))"];
9[label="always(or([not a,next(b)]))"];
10[label="always(or([not a,next(not c)]))"];
11[label="always(or([next(not c),next(not a)]))"];
13[label="always(or([next(c),next(not a)]))"];
14[label="always(or([next(c),next(not b)]))"];
15[label="always(or([next(not a)]))"];
16[label="always(or([not b]))"];
17[label="always(or([not a,next(not b)]))"];
18[label="always(or([not a]))"];
20[label="always(or([c,not a]))"];
23[label="always(or([not a]))"];
24[label="or([])"];
2->8 [label="BFS-loop-it-init-x"];
3->9 [label="BFS-loop-it-init-x"];
4->10 [label="BFS-loop-it-init-x"];
1->11 [label="BFS-loop-it-init-n"];
13->15 [label="step-xx"];
11->15 [label="step-xx"];
15->16 [label="step-xx"];
8->16 [label="step-xx"];
14->17 [label="step-xx"];
10->17 [label="step-xx"];
17->18 [label="step-xx"];
9->18 [label="step-xx"];
18->13 [label="BFS-loop-it-sub"];
16->14 [label="BFS-loop-it-sub",color=blue,lblstyle="blue",style="ultra thick, densely dashed"];
18->20 [label="BFS-loop-conclusion1g"];
5->20 [label="BFS-loop-conclusion1e"];
20->23 [label="step-nn"];
1->23 [label="step-nn"];
23->24 [label="init-inn"];
0->24 [label="init-ini"];
}
\end{dot2tex}
\fi
}
\begin{minipage}{0.58\linewidth}\caption{\label{fig:rgloopitsub} Subgraph $\graphp$ of resolution graph for the case of \refloopitsub in the proof of Proposition \ref{thm:minimalitypremises}}\end{minipage}
\end{sidewaysfigure}

\tablineone
{\refloopconclusiontwo, premise 1}
{$
\begin{array}{l}
\hspace{-0.5em}\{\ficlause{\fbor{a}{b}},\\
\fgneclause{(\fbnot{a})}{f},\\
\fgneclause{(\fbnot{b})}{f},\\
\fgneclause{(\fbnot{a})}{\fbnot{f}},\\
\fgneclause{(\fbnot{b})}{\fbnot{f}},\\
(\fglobally{(\fbor{(\fbnot{a})}{\fbor{f}{\fnext{c}}})}),\\
(\fglobally{(\fbor{(\fbnot{b})}{\fbor{f}{\fnext{d}}})}),\\
(\fglobally{(\fbor{(\fbnot{c})}{\fnext{c}})}),\\
(\fglobally{(\fbor{(\fbnot{d})}{\fnext{d}})}),\\
\fgnclause{\fbor{(\fbnot{c})}{(\fbnot{f})}},\\
\fgnclause{\fbor{(\fbnot{d})}{(\fbnot{f})}},\\
\ficlause{\fbor{(\fbnot{f})}{g}},\\
(\fglobally{(\fbor{(\fbnot{g})}{\fnext{g}})}),\\
\fgnclause{\fbor{(\fbnot{g})}{f}}\}
\end{array}
$}
\tablinethree
See Fig.~\ref{fig:rgloopconclusiontwo}.
\tablinefour
\begin{sidewaysfigure}
\vspace{98ex}
\centering
\adjustbox{scale=0.30,keepaspectratio=true}{
\iftrue
\input{trpuc-paper-dot2tex-fig6.tex}
\else
\begin{dot2tex}[options={--cache --graphstyle "scale=0.6"}]
digraph G {
graph [rankdir="BT"];
0[label="or([b,a])"];
1[label="or([g,not f])"];
2[label="always(or([not c,not f]))"];
3[label="always(or([not d,not f]))"];
4[label="always(or([not g,f]))"];
5[label="always(or([not c,next(c)]))"];
6[label="always(or([not d,next(d)]))",color=blue,style="ultra thick, densely dashed, text=blue"];
7[label="always(or([not g,next(g)]))"];
8[label="always(or([f,not a,next(c)]))"];
9[label="always(or([f,not b,next(d)]))"];
10[label="always(or([not a,sometime(f)]))"];
11[label="always(or([not b,sometime(f)]))"];
12[label="always(or([not a,sometime(not f)]))"];
13[label="always(or([not b,sometime(not f)]))"];
14[label="always(or([not c,next(not f)]))"];
15[label="always(or([not d,next(not f)]))"];
16[label="always(or([not g,next(f)]))"];
21[label="always(or([f,not a,next(not f)]))"];
22[label="always(or([f,not b,next(not f)]))"];
26[label="always(or([w_f,f,not b]))"];
28[label="always(or([w_not_f,not f,not b]))"];
29[label="always(or([not c,next(not f)]))"];
30[label="always(or([not d,next(not f)]))"];
31[label="always(or([f,not a,next(not f)]))"];
33[label="always(or([not c,next(c)]))"];
34[label="always(or([f,not a,next(c)]))"];
35[label="always(or([not d,next(d)]))"];
37[label="always(or([next(not c),next(f)]))"];
38[label="always(or([next(not d),next(f)]))"];
39[label="always(or([not c,next(f)]))"];
40[label="always(or([f,not a,next(f)]))"];
41[label="always(or([not c]))"];
42[label="always(or([not d,next(f)]))"];
44[label="always(or([not d]))"];
45[label="always(or([f,not a]))"];
48[label="always(or([f,not a]))"];
53[label="always(or([not w_f,next(not d),next(f)]))"];
67[label="always(or([not w_f,f,not b,next(f)]))"];
73[label="always(or([not w_f,f,not b]))"];
74[label="always(or([f,not b]))"];
107[label="always(or([not g,next(f)]))"];
124[label="always(or([not g,next(g)]))"];
138[label="always(or([next(not g),next(not f)]))"];
139[label="always(or([not g,next(not f)]))"];
140[label="always(or([not g]))"];
143[label="always(or([not w_not_f,next(not g),next(not f)]))"];
144[label="always(or([not g,not f,not a]))"];
145[label="always(or([not w_not_f,not g,next(not f)]))"];
146[label="always(or([not w_not_f,not g]))"];
147[label="always(or([not g,not f,not b]))"];
148[label="or([not f,not a])"];
150[label="or([not a])"];
152[label="or([not f,not b])"];
154[label="or([not b])"];
155[label="or([a])"];
156[label="or([])"];
2->14 [label="step-nxn"];
5->14 [label="step-nxx"];
3->15 [label="step-nxn"];
6->15 [label="step-nxx"];
4->16 [label="step-nxn"];
7->16 [label="step-nxx"];
8->21 [label="step-nxx"];
2->21 [label="step-nxn"];
9->22 [label="step-nxx"];
3->22 [label="step-nxn"];
11->26 [label="aug1"];
13->28 [label="aug1"];
14->29 [label="BFS-loop-it-init-x"];
15->30 [label="BFS-loop-it-init-x"];
21->31 [label="BFS-loop-it-init-x"];
5->33 [label="BFS-loop-it-init-x"];
8->34 [label="BFS-loop-it-init-x"];
6->35 [label="BFS-loop-it-init-x"];
37->39 [label="step-xx"];
33->39 [label="step-xx"];
37->40 [label="step-xx"];
34->40 [label="step-xx"];
39->41 [label="step-xx"];
29->41 [label="step-xx"];
38->42 [label="step-xx"];
35->42 [label="step-xx"];
42->44 [label="step-xx"];
30->44 [label="step-xx"];
40->45 [label="step-xx"];
31->45 [label="step-xx"];
41->37 [label="BFS-loop-it-sub"];
44->38 [label="BFS-loop-it-sub"];
45->48 [label="BFS-loop-conclusion1-woemptyg"];
10->48 [label="BFS-loop-conclusion1-woemptye"];
44->53 [label="BFS-loop-conclusion2-woemptyg",color=blue,lblstyle="blue",style="ultra thick, densely dashed"];
53->67 [label="step-xx"];
9->67 [label="step-xx"];
67->73 [label="step-xx"];
22->73 [label="step-xx"];
73->74 [label="step-nn"];
26->74 [label="step-nn"];
16->107 [label="BFS-loop-it-init-x"];
7->124 [label="BFS-loop-it-init-x"];
138->139 [label="step-xx"];
124->139 [label="step-xx"];
139->140 [label="step-xx"];
107->140 [label="step-xx"];
140->138 [label="BFS-loop-it-sub"];
140->143 [label="BFS-loop-conclusion2-woemptyg"];
140->144 [label="BFS-loop-conclusion1-woemptyg"];
12->144 [label="BFS-loop-conclusion1-woemptye"];
143->145 [label="step-xx"];
7->145 [label="step-xx"];
145->146 [label="step-xx"];
16->146 [label="step-xx"];
146->147 [label="step-nn"];
28->147 [label="step-nn"];
144->148 [label="init-inn"];
1->148 [label="init-ini"];
148->150 [label="init-ini"];
48->150 [label="init-inn"];
147->152 [label="init-inn"];
1->152 [label="init-ini"];
152->154 [label="init-ini"];
74->154 [label="init-inn"];
154->155 [label="init-ii"];
0->155 [label="init-ii"];
155->156 [label="init-ii"];
150->156 [label="init-ii"];
48->7 [style=invis];
74->7 [style=invis];
}
\end{dot2tex}
\fi
}
\begin{minipage}{0.61\linewidth}\caption{\label{fig:rgloopconclusiontwo} Subgraph $\graphp$ of resolution graph for the case of \refloopconclusiontwo in the proof of Proposition \ref{thm:minimalitypremises}}\end{minipage}
\end{sidewaysfigure}

\noindent
\begin{tabular}{p{0.65\linewidth}p{0.25\linewidth}}
&\\
\hline
&\\
\end{tabular}
}

\noindent
This concludes the proof.
\qed
\end{proof}

\clearpage
\section{Additional Plots}
\label{sec:moreplots}

Figures \ref{fig:overheadbycategoryucvsnouc} and \ref{fig:sizereductionbycategoryucvsnouc} show the overhead that is incurred and the size reduction that is obtained by extracting \ucs compared to not extracting \ucs split by category.
Figures \ref{fig:optimizationsbycategoryonetothree} and \ref{fig:optimizationsbycategoryfourtosix} show the benefit of optimizations split by category.
Figures \ref{fig:overheadbycategoryminimalucvsuc} and \ref{fig:sizereductionbycategoryminimalucvsuc} show the overhead that is incurred and the size reduction that is obtained by extracting minimal \ucs compared to extracting (non-minimal) \ucs split by category.
Figure \ref{fig:overheadbycategorypartitioneducvsuc} shows the overhead that is incurred by extracting \ucs with grouped propositions compared to extracting \ucs (without grouped propositions) split by category.
Figures \ref{fig:proofdeletionvspltlmup}--\ref{fig:proofvsprocminenomus} compare \UC extraction with \trp and with \pltlmup and \procmine in terms of run time and memory split by category.


\clearpage
{
\newcommand{\tabfield}[1]{%
\begin{minipage}{0.29\linewidth}%
\begin{center}%
\includegraphics[type=pdf,ext=.pdf,read=.pdf,scale=0.36,trim=1.7cm 0.0cm 1.3cm 0.5cm,clip]{#1}%
\end{center}%
\end{minipage}%
}%

\newcommand{\tablineheader}[1]{%
\begin{minipage}{0.02\linewidth}%
\begin{sideways}%
\hspace*{2em}{#1}%
\end{sideways}%
\end{minipage}%
}

\begin{figure}[h]
\hspace*{-0.5em}\begin{tabular}{cccc}
\tablineheader{run time} &
\tabfield{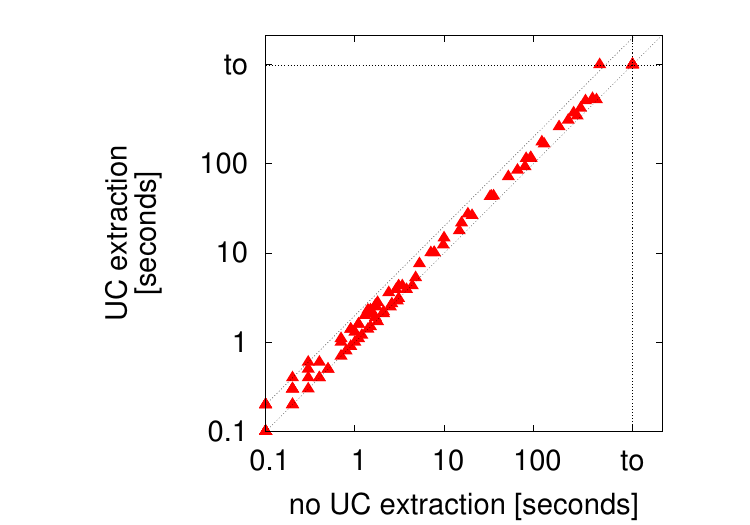} &
\tabfield{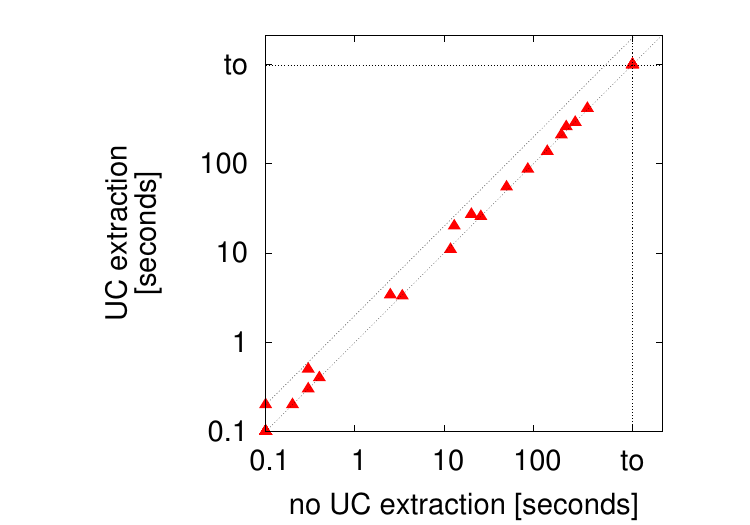} &
\tabfield{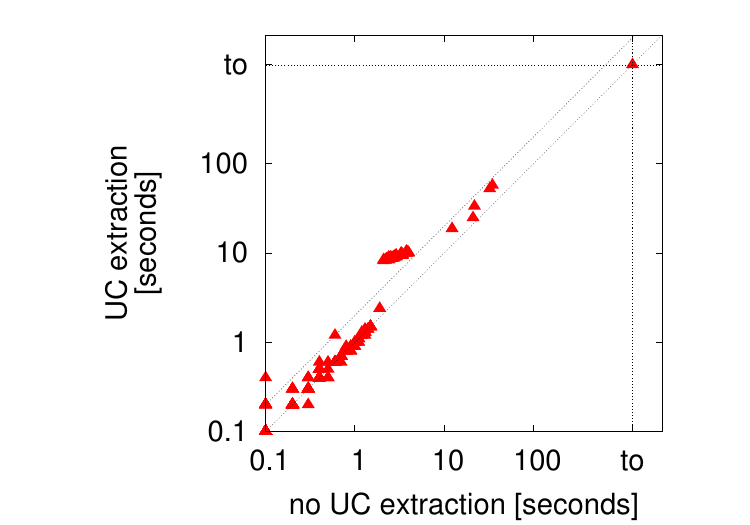} \\
\\
\tablineheader{memory} &
\tabfield{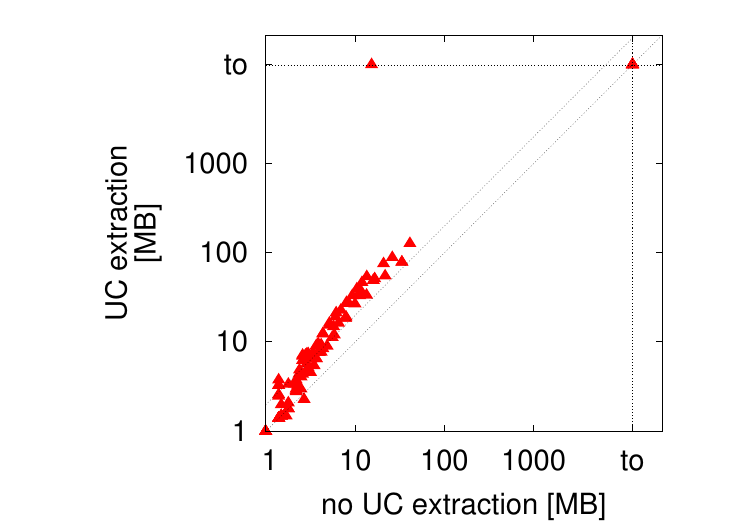} &
\tabfield{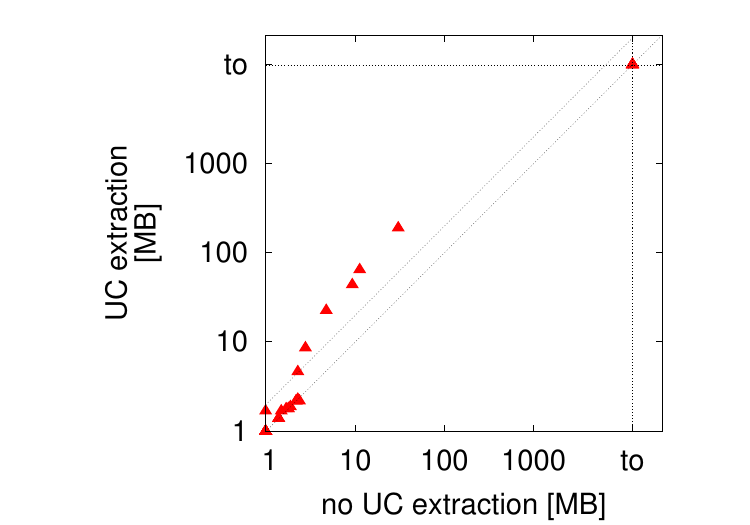} &
\tabfield{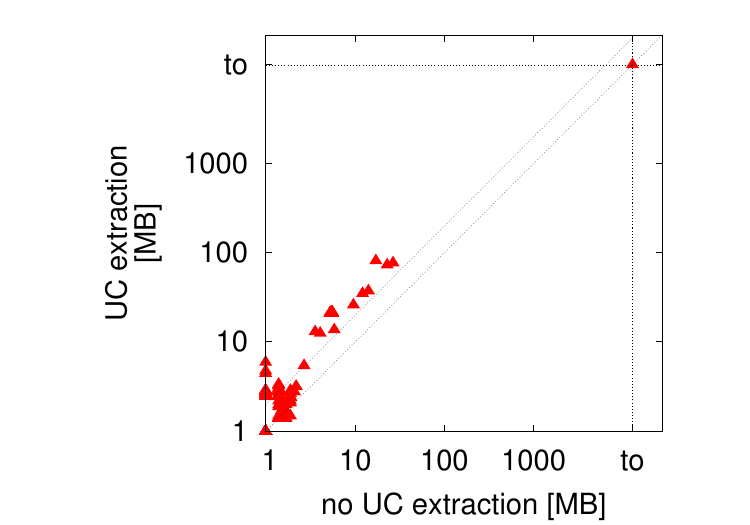} \\
\\
& \hspace*{3.25em}application & \hspace*{3.25em}crafted & \hspace*{3.25em}random \\
\\
\end{tabular}
\caption{\label{fig:overheadbycategoryucvsnouc} Overhead incurred by \uc extraction compared to not extracting \ucs in terms of run time (in seconds) and memory (in MB) separated by categories \benchmark{application}, \benchmark{crafted}, and \benchmark{random}. In each graph extraction of \ucs is on the y-axis and no \uc extraction on the x-axis. The off-center diagonal indicates where $y = 2 x$.}
\end{figure}
}

{
\newcommand{\tabfield}[1]{%
\begin{minipage}{0.29\linewidth}%
\begin{center}%
\includegraphics[type=pdf,ext=.pdf,read=.pdf,scale=0.36,trim=1.6cm 0.0cm 1.3cm 0.5cm,clip]{#1}%
\end{center}%
\end{minipage}%
}%

\newcommand{\tablineheader}[1]{%
\begin{minipage}{0.01\linewidth}%
\begin{sideways}%
\hspace*{2em}{#1}%
\end{sideways}%
\end{minipage}%
}

\begin{figure}[h]
\hspace*{-0.5em}\begin{tabular}{cccc}
\tablineheader{size} &
\tabfield{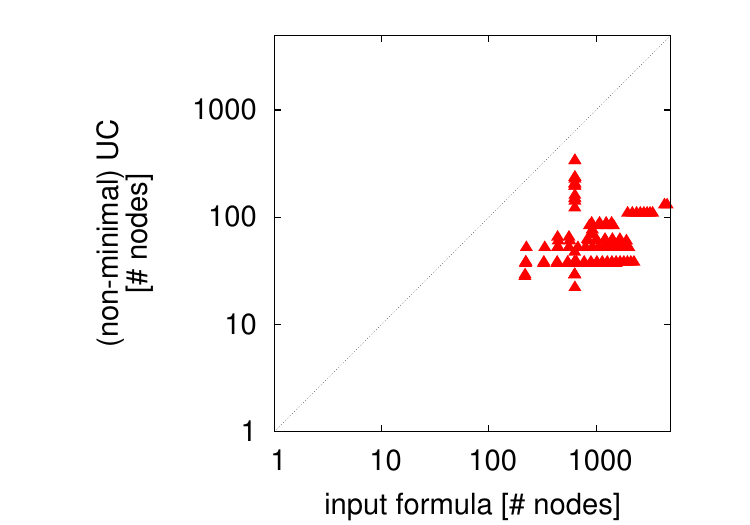} &
\tabfield{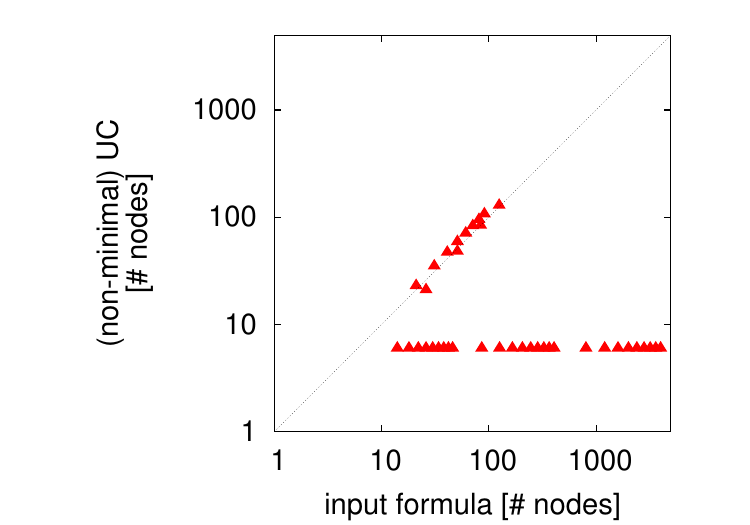} &
\tabfield{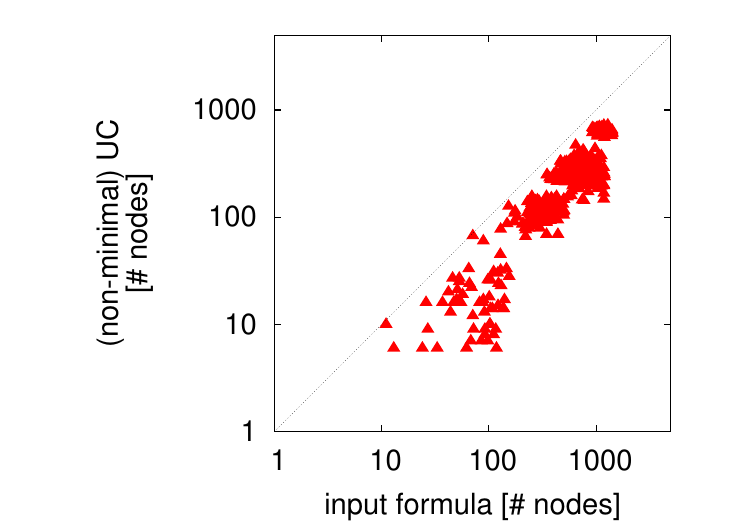} \\
\\
& \hspace*{3.75em}application & \hspace*{3.75em}crafted & \hspace*{3.75em}random \\
\\
\end{tabular}
\caption{\label{fig:sizereductionbycategoryucvsnouc} Size reduction obtained by \uc extraction compared to not extracting \ucs separated by categories \benchmark{application}, \benchmark{crafted}, and \benchmark{random}. The y-axes show the sizes of the \ucs, the x-axes show the sizes of the input formulas. Size is measured as the number of nodes in the syntax trees.}
\end{figure}
}


{
\newcommand{\tabfield}[1]{%
\begin{minipage}{0.31\linewidth}%
\begin{center}%
\includegraphics[type=pdf,ext=.pdf,read=.pdf,scale=0.33,trim=0.4cm 0.0cm 0.9cm 0.5cm,clip]{#1}%
\end{center}%
\end{minipage}%
}%

\clearpage
\begin{figure}[h]
\hspace{-0.75em}\begin{tabular}{ccc}
\tabfield{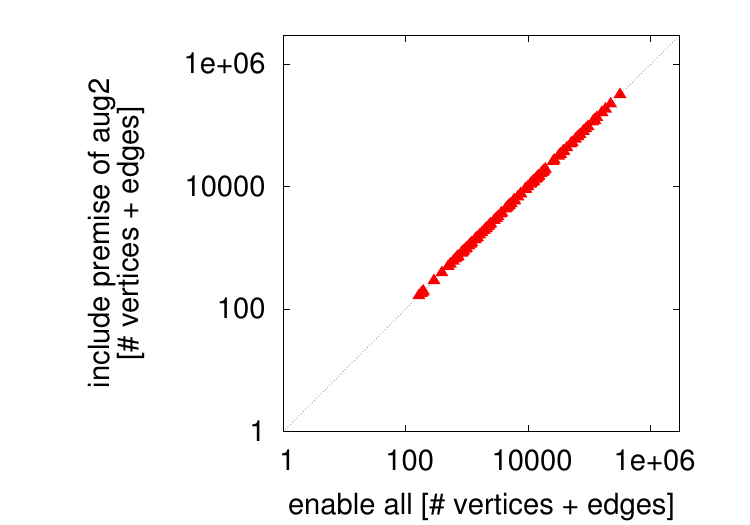} &
\tabfield{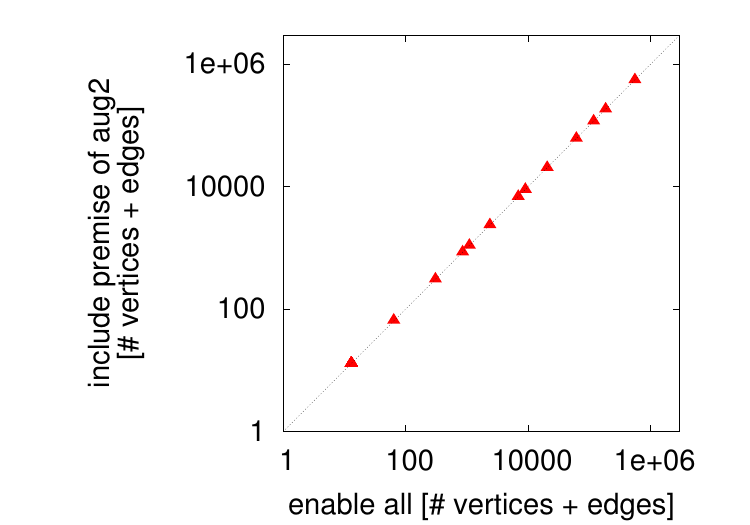} &
\tabfield{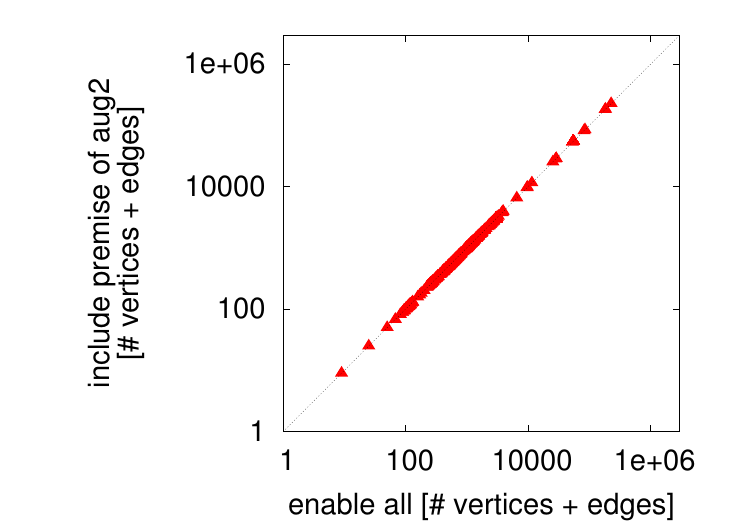} \\
\\
\tabfield{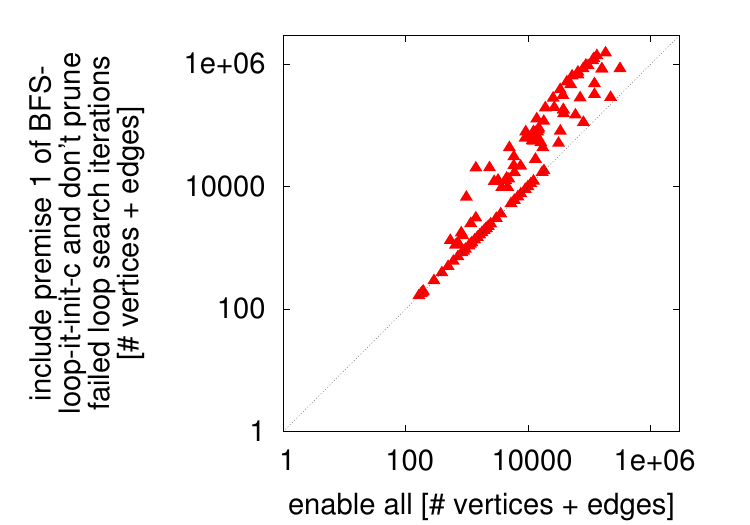} &
\tabfield{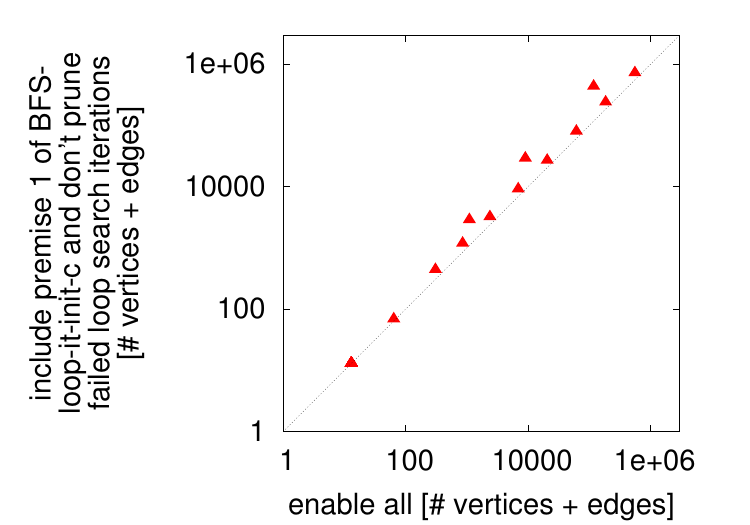} &
\tabfield{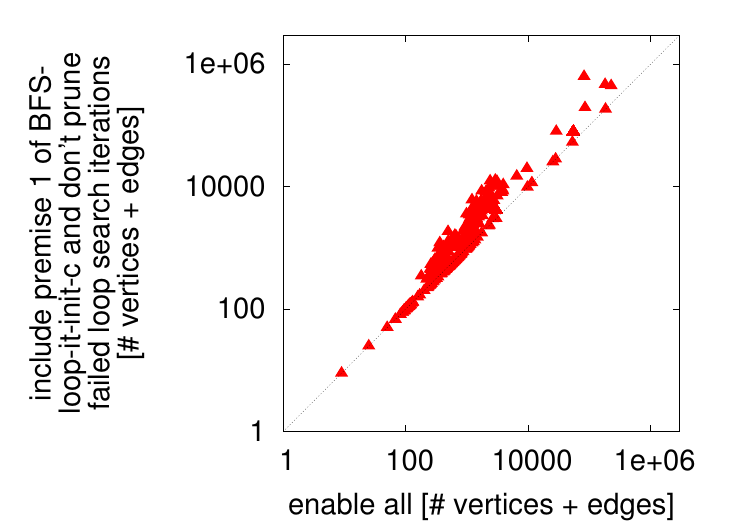} \\
\\
\tabfield{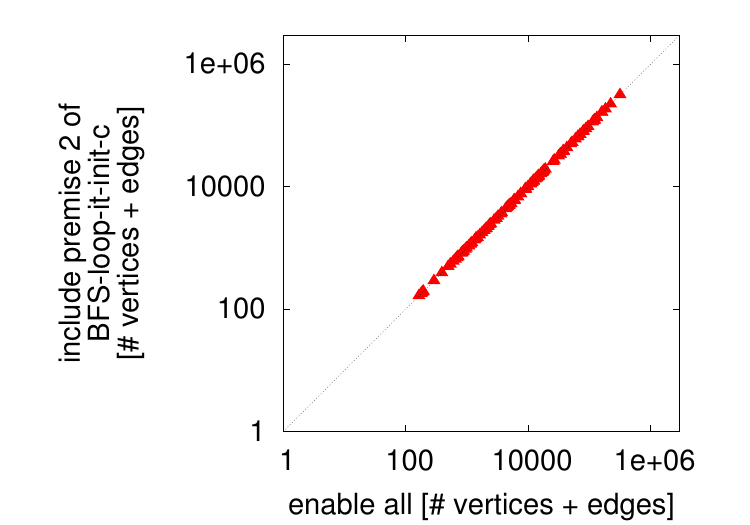} &
\tabfield{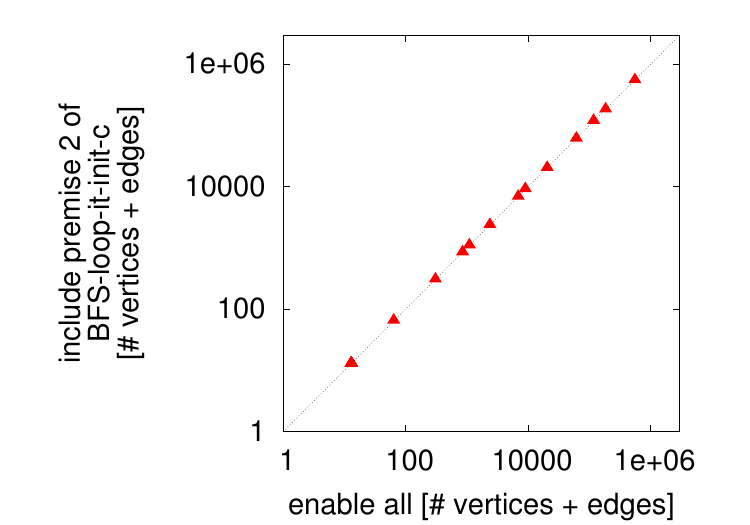} &
\tabfield{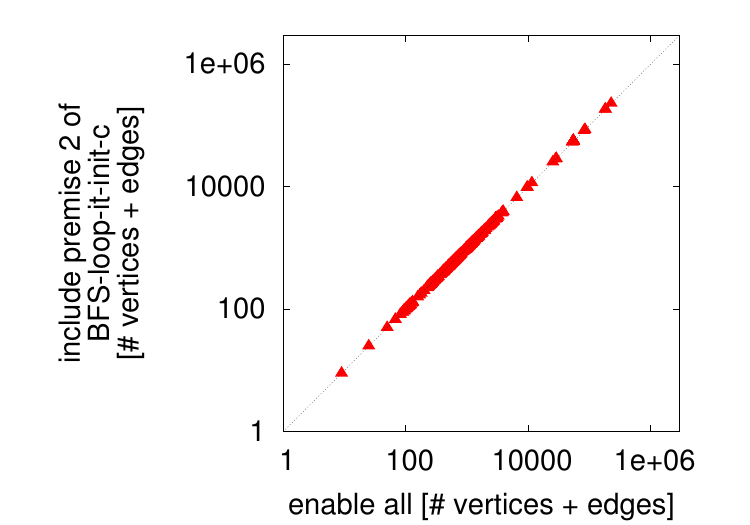} \\
\\
\hspace*{4.75em}application & \hspace*{4.75em}crafted & \hspace*{4.75em}random \\
\\
\end{tabular}
\caption{\label{fig:optimizationsbycategoryonetothree} Benefit of optimizations as reduction in the peak size of the resolution graph (number of vertices + number of edges) separated by categories \benchmark{application}, \benchmark{crafted}, and \benchmark{random}. The x-axis shows all optimizations enabled. The y-axis of rows 1--3 shows one optimization disabled: (row 1) include premise of \refaugtwo, (row 2) include premise 1 of \refloopitinitc and disable immediate pruning of failed loop search iterations, and (row 3) include premise 2 of \refloopitinitc.}
\end{figure}

\clearpage
\begin{figure}[h]
\hspace{-0,75em}\begin{tabular}{ccc}
\tabfield{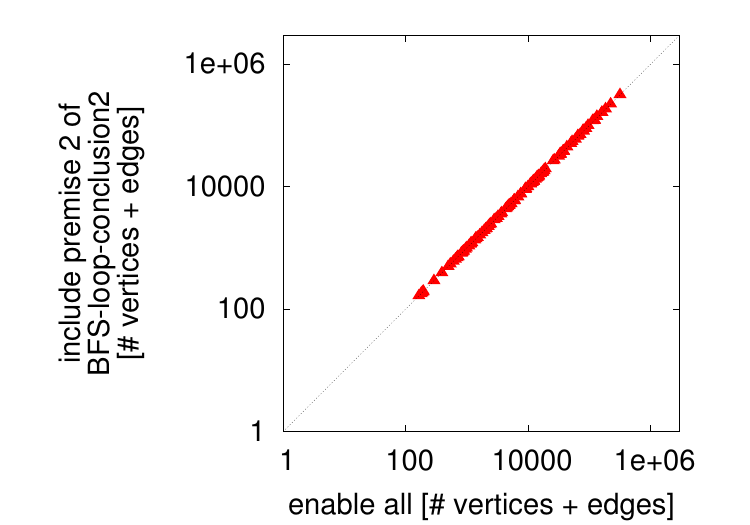} &
\tabfield{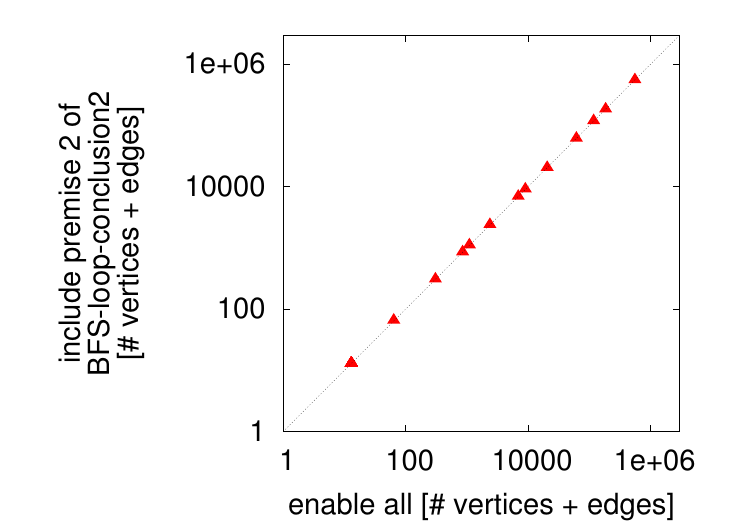} &
\tabfield{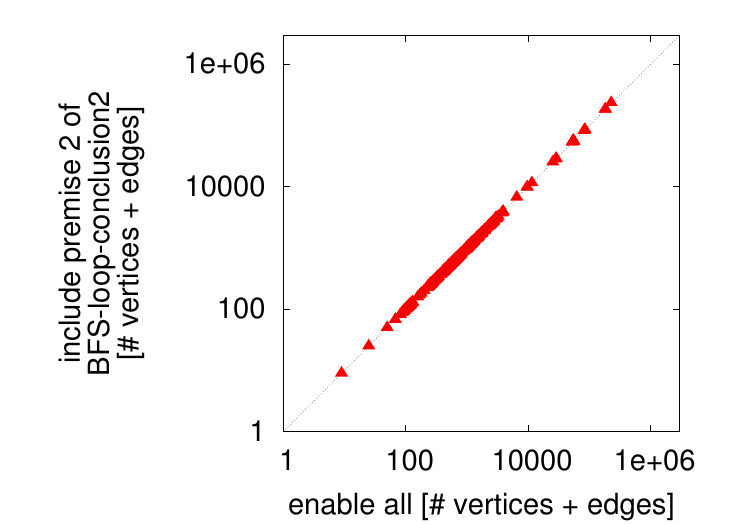} \\
\\
\tabfield{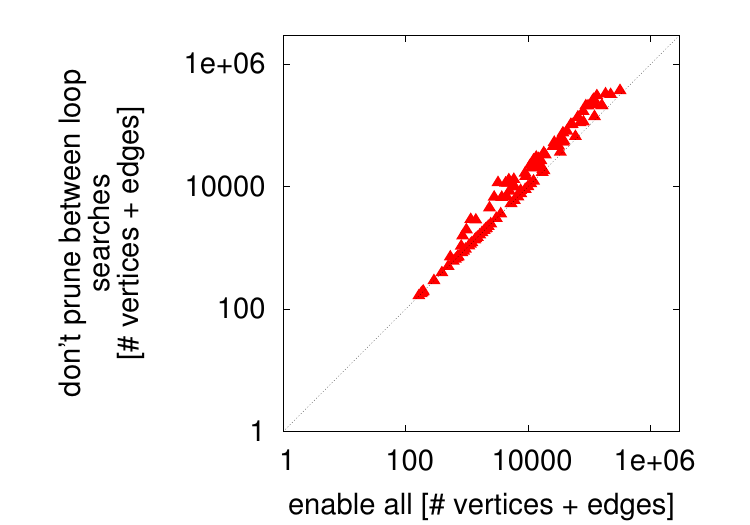} &
\tabfield{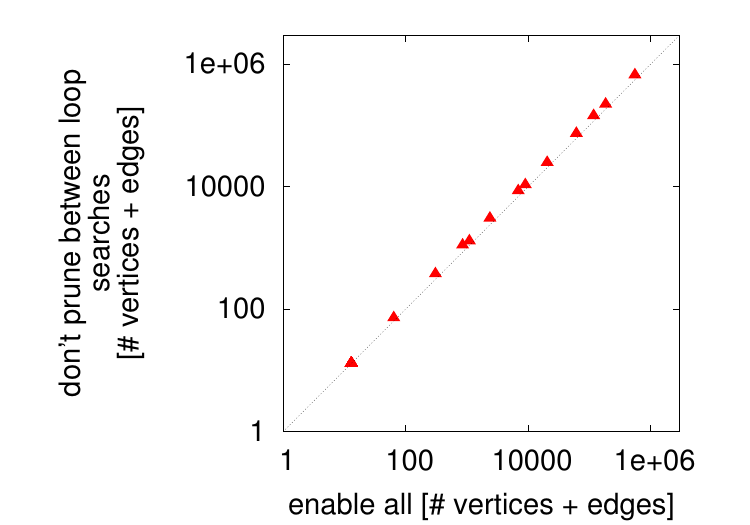} &
\tabfield{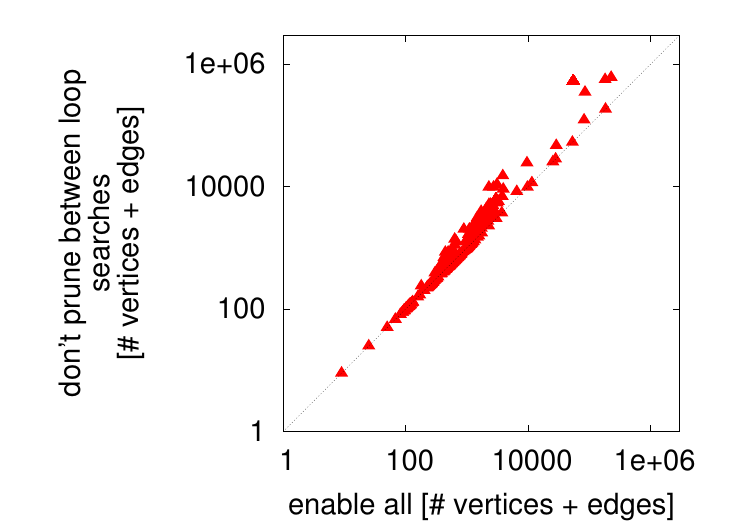} \\
\\
\tabfield{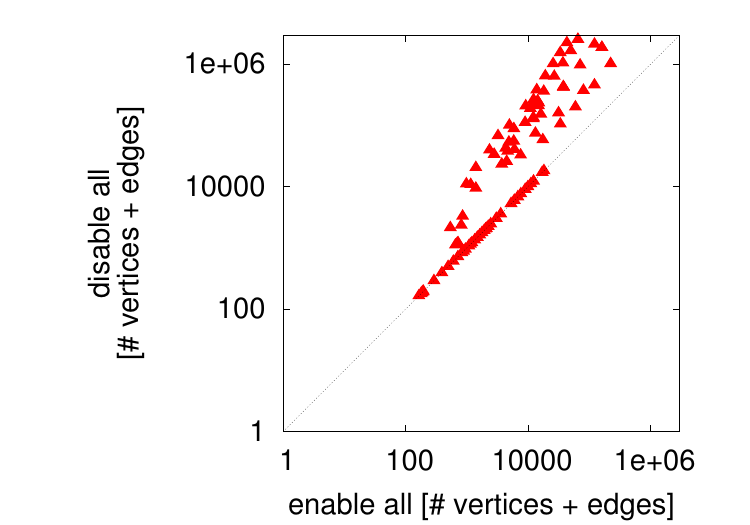} &
\tabfield{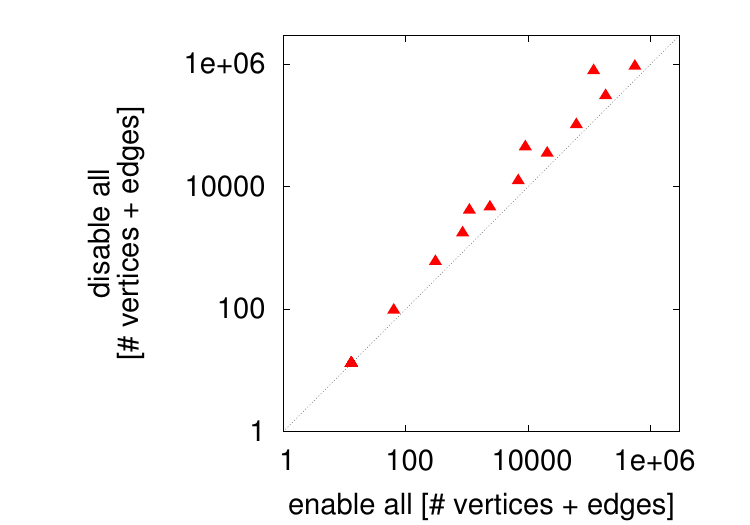} &
\tabfield{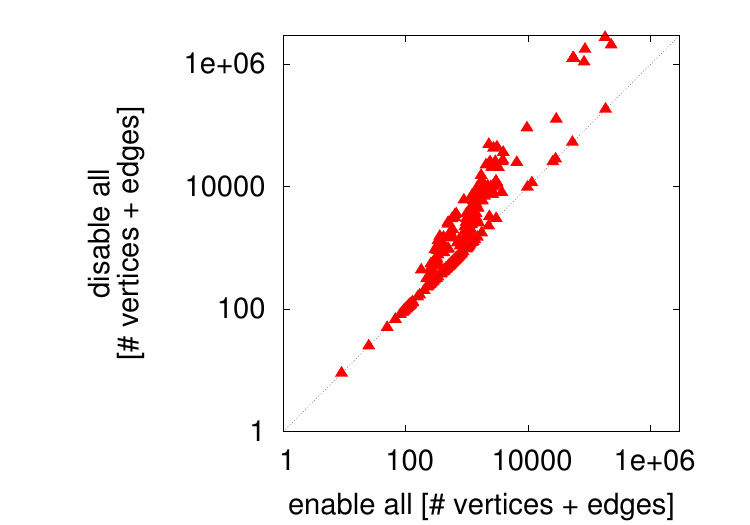} \\
\\
\hspace*{4.75em}application & \hspace*{4.75em}crafted & \hspace*{4.75em}random \\
\\
\end{tabular}
\caption{\label{fig:optimizationsbycategoryfourtosix} Benefit of optimizations as reduction in the peak size of the resolution graph (number of vertices + number of edges) separated by categories \benchmark{application}, \benchmark{crafted}, and \benchmark{random}. The x-axis shows all optimizations enabled. The y-axis of rows 1 and 2 shows one optimization disabled: (row 1) include premise 2 of \refloopconclusiontwo and (row 2) disable pruning of the resolution graph between loop searches. The y-axis of row 3 shows all optimizations disabled.}
\end{figure}
}


\clearpage
{
\newcommand{\tabfield}[1]{%
\begin{minipage}{0.29\linewidth}%
\begin{center}%
\includegraphics[type=pdf,ext=.pdf,read=.pdf,scale=0.36,trim=1.7cm 0.0cm 1.3cm 0.5cm,clip]{#1}%
\end{center}%
\end{minipage}%
}%

\newcommand{\tablineheader}[1]{%
\begin{minipage}{0.02\linewidth}%
\begin{sideways}%
\hspace*{2em}{#1}%
\end{sideways}%
\end{minipage}%
}

\begin{figure}[h]
\hspace*{-0.5em}\begin{tabular}{cccc}
\tablineheader{run time} &
\tabfield{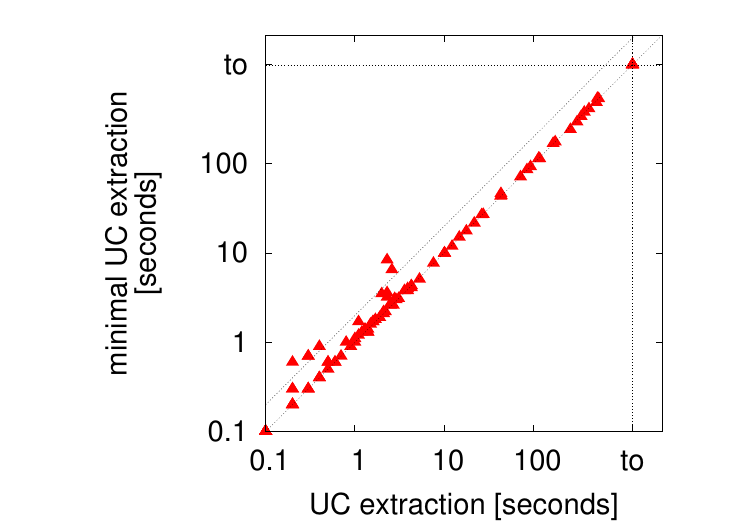} &
\tabfield{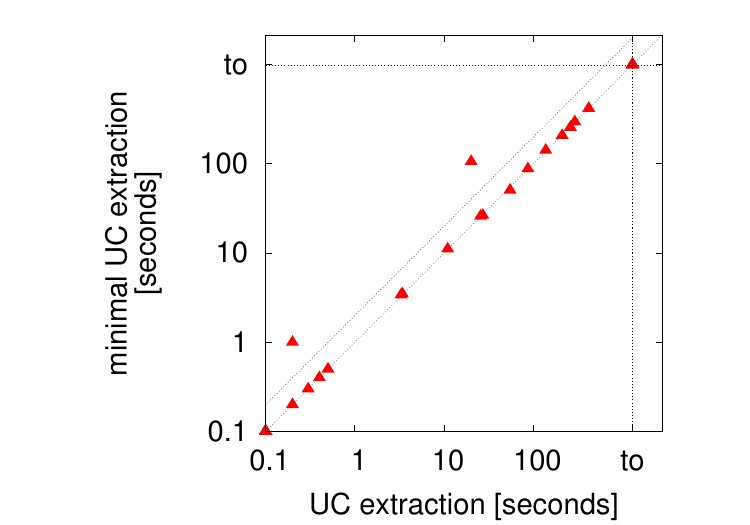} &
\tabfield{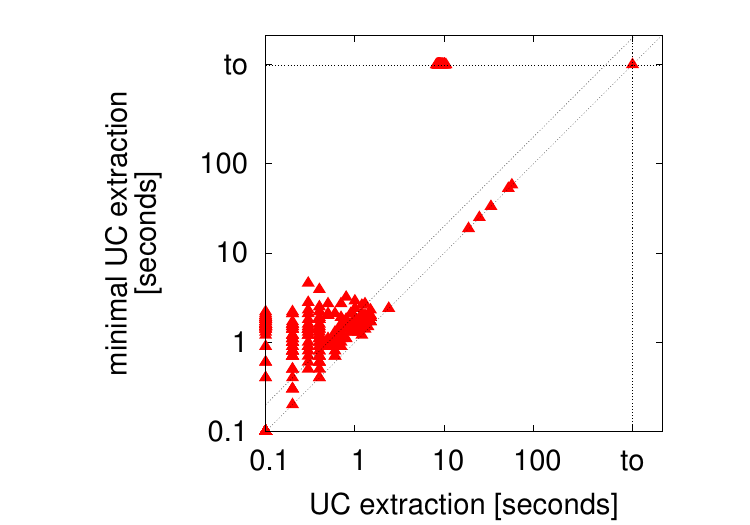} \\
\\
\tablineheader{memory} &
\tabfield{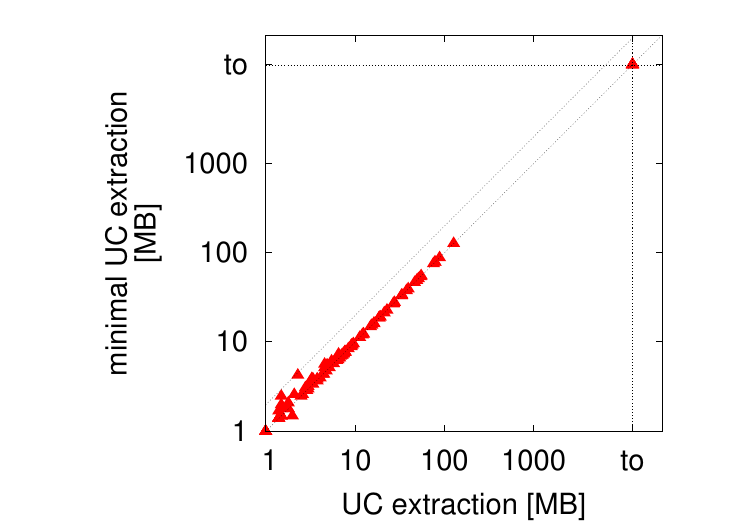} &
\tabfield{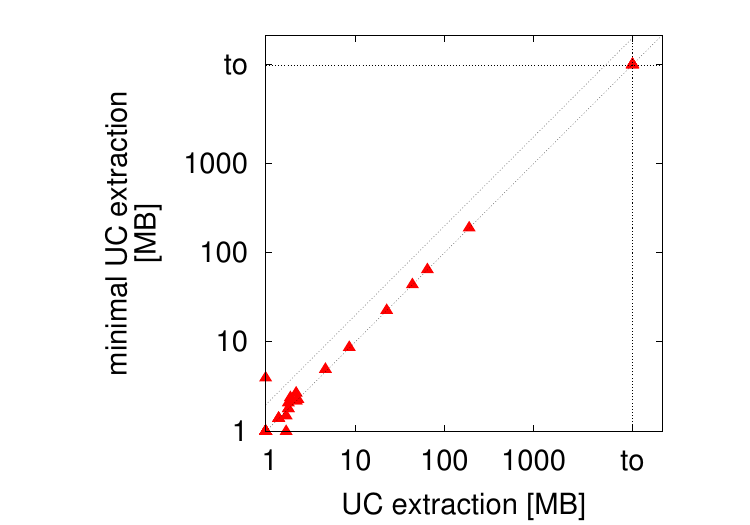} &
\tabfield{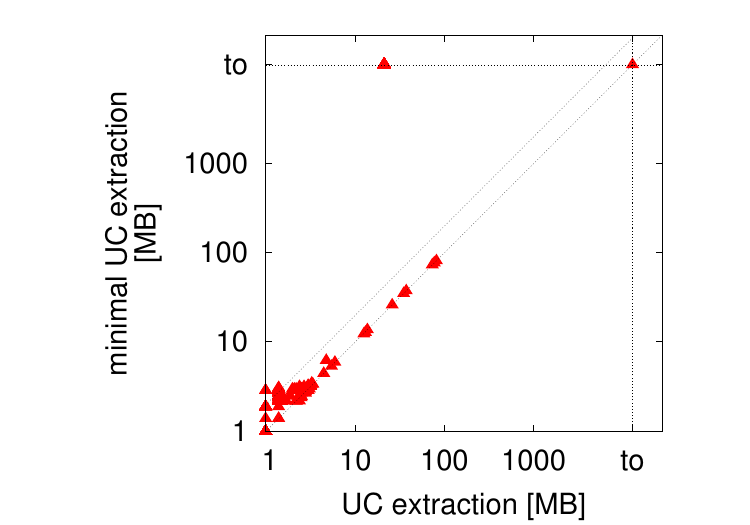} \\
\\
& \hspace*{3.25em}application & \hspace*{3.25em}crafted & \hspace*{3.25em}random \\
\\
\end{tabular}
\caption{\label{fig:overheadbycategoryminimalucvsuc} Overhead incurred by minimal \uc extraction compared to (non-minimal) \uc extraction in terms of run time (in seconds) and memory (in MB) separated by categories \benchmark{application}, \benchmark{crafted}, and \benchmark{random}. In each graph extraction of minimal \ucs is on the y-axis and (non-minimal) \uc extraction on the x-axis. The off-center diagonal indicates where $y = 2 x$.}
\end{figure}
}

{
\newcommand{\tabfield}[1]{%
\begin{minipage}{0.29\linewidth}%
\begin{center}%
\includegraphics[type=pdf,ext=.pdf,read=.pdf,scale=0.36,trim=1.6cm 0.0cm 1.3cm 0.5cm,clip]{#1}%
\end{center}%
\end{minipage}%
}%

\newcommand{\tablineheader}[1]{%
\begin{minipage}{0.01\linewidth}%
\begin{sideways}%
\hspace*{2em}{#1}%
\end{sideways}%
\end{minipage}%
}

\begin{figure}[h]
\hspace*{-0.5em}\begin{tabular}{cccc}
\tablineheader{size} &
\tabfield{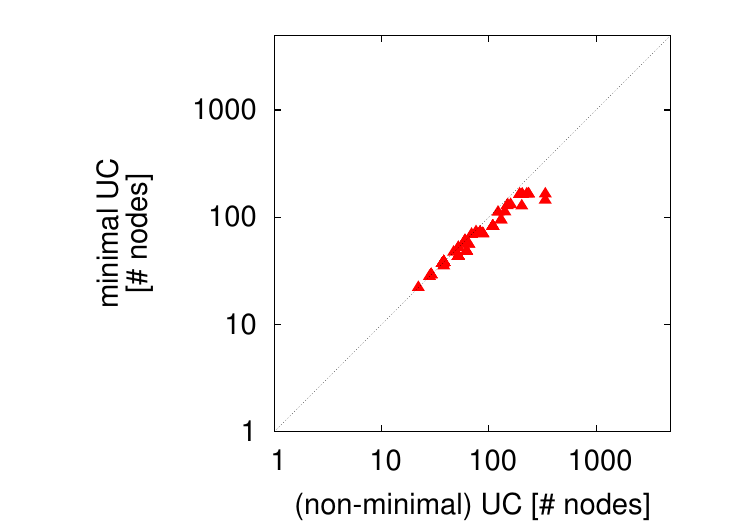} &
\tabfield{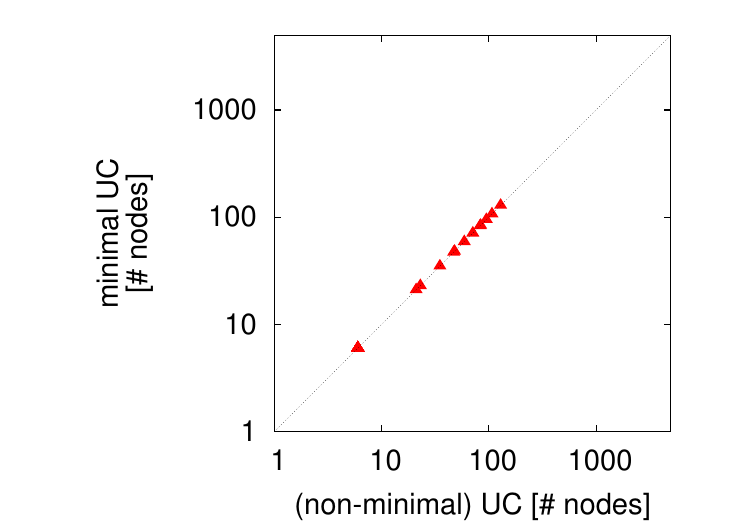} &
\tabfield{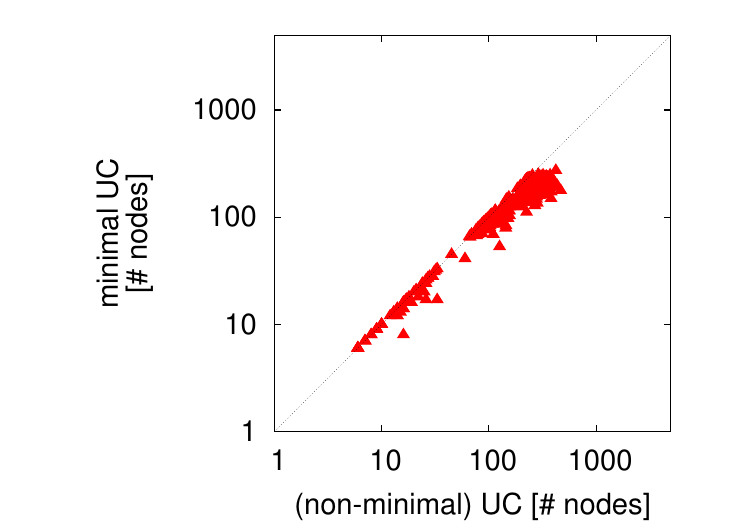} \\
\\
& \hspace*{3.75em}application & \hspace*{3.75em}crafted & \hspace*{3.75em}random \\
\\
\end{tabular}
\caption{\label{fig:sizereductionbycategoryminimalucvsuc} Size reduction obtained by minimal \uc extraction compared to (non-minimal) \uc extraction separated by categories \benchmark{application}, \benchmark{crafted}, and \benchmark{random}. The y-axes show the sizes of the minimal \ucs, the x-axes show the sizes of the (non-minimal) \ucs. Size is measured as the number of nodes in the syntax trees.}
\end{figure}
}


\clearpage
{
\newcommand{\tabfield}[1]{%
\begin{minipage}{0.29\linewidth}%
\begin{center}%
\includegraphics[type=pdf,ext=.pdf,read=.pdf,scale=0.34,trim=1.5cm 0.0cm 1.1cm 0.5cm,clip]{#1}%
\end{center}%
\end{minipage}%
}%

\newcommand{\tablineheader}[1]{%
\begin{minipage}{0.02\linewidth}%
\begin{sideways}%
\hspace*{1.75em}{#1}%
\end{sideways}%
\end{minipage}%
}

\begin{figure}[h]
\hspace*{-0.5em}\begin{tabular}{cccc}
\tablineheader{run time} &
\tabfield{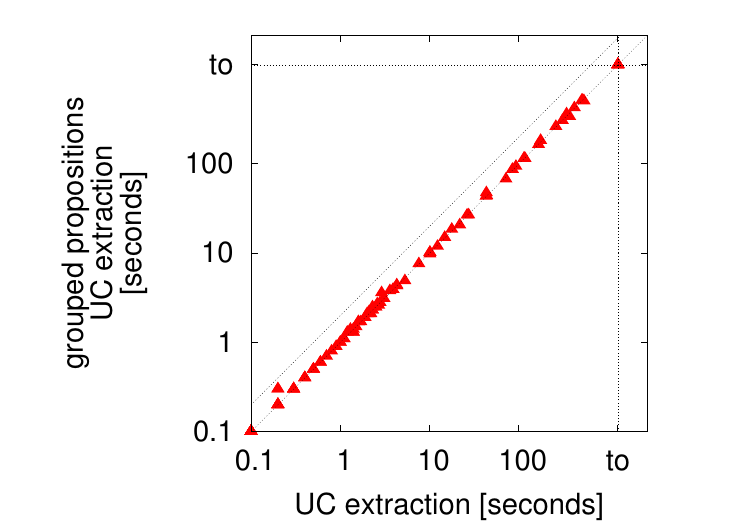} &
\tabfield{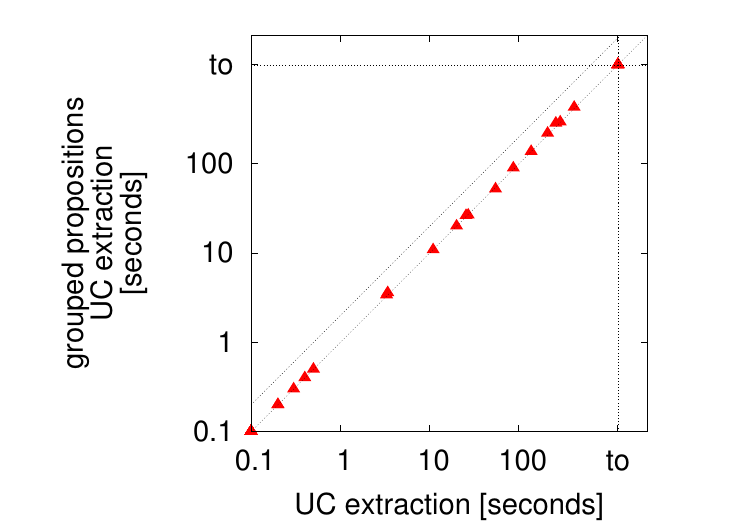} &
\tabfield{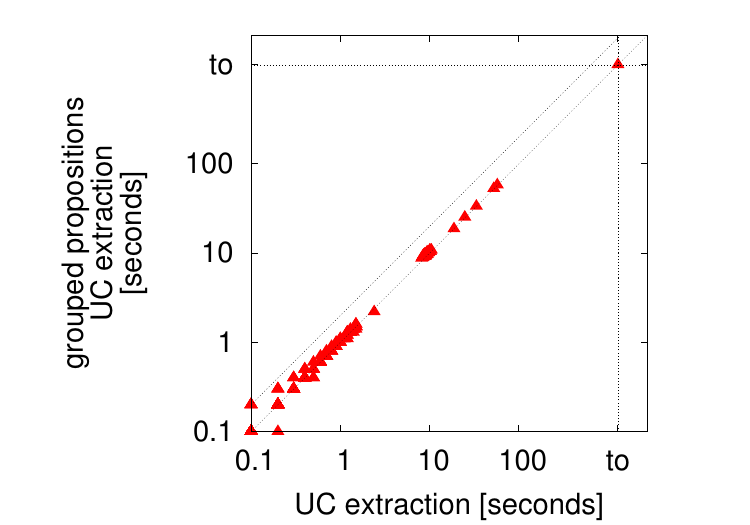} \\
\\
\tablineheader{memory} &
\tabfield{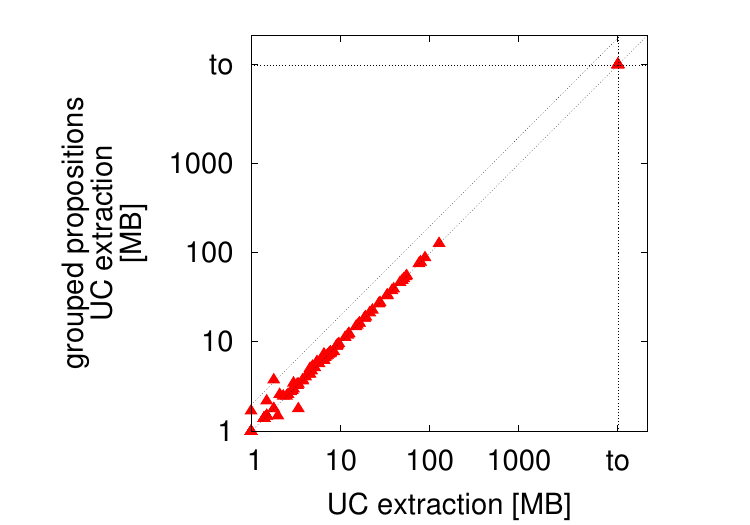} &
\tabfield{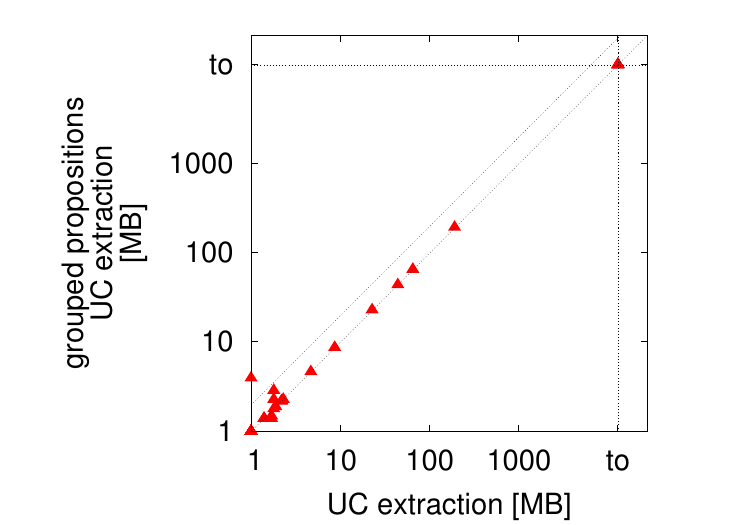} &
\tabfield{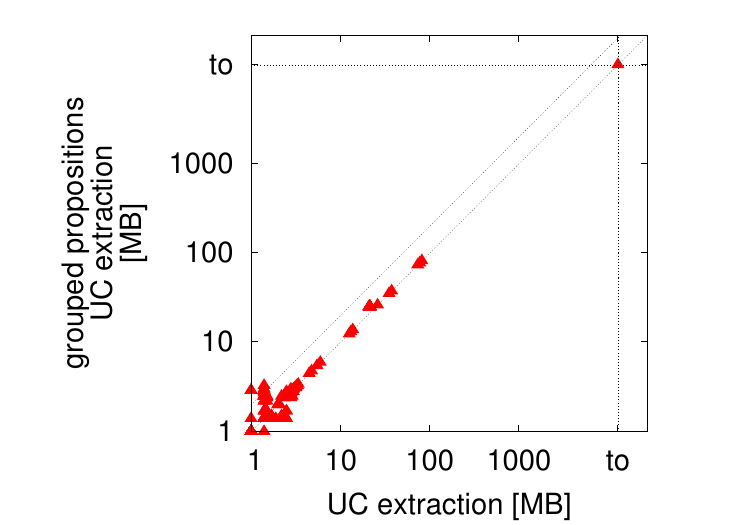} \\
\\
& \hspace*{2.5em}application & \hspace*{2.5em}crafted & \hspace*{2.5em}random \\
\\
\end{tabular}
\caption{\label{fig:overheadbycategorypartitioneducvsuc} Overhead of \uc extraction with grouped propositions compared to \uc extraction (without grouped propositions) in terms of run time (in seconds) and memory (in MB) separated by categories \benchmark{application}, \benchmark{crafted}, and \benchmark{random}. In each graph extraction of \ucs with grouped propositions is on the y-axis and \uc extraction (without grouped propositions) on the x-axis. The off-center diagonal indicates where $y = 2 x$.}
\end{figure}
}


\clearpage
{
\newcommand{\tabfield}[1]{%
\begin{minipage}{0.29\linewidth}%
\begin{center}%
\includegraphics[type=pdf,ext=.pdf,read=.pdf,scale=0.36,trim=1.1cm 0.0cm 1.3cm 0.5cm,clip]{#1}%
\end{center}%
\end{minipage}%
}%

\newcommand{\tablineheader}[1]{%
\begin{minipage}{0.02\linewidth}%
\begin{sideways}%
\hspace*{1.75em}{#1}%
\end{sideways}%
\end{minipage}%
}

\begin{figure}[h]
\hspace*{-0.5em}\begin{tabular}{cccc}
\tablineheader{run time} &
\tabfield{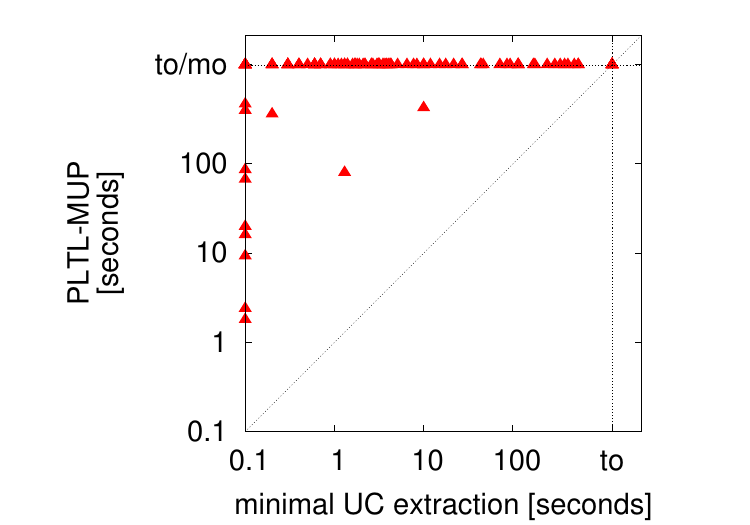} &
\tabfield{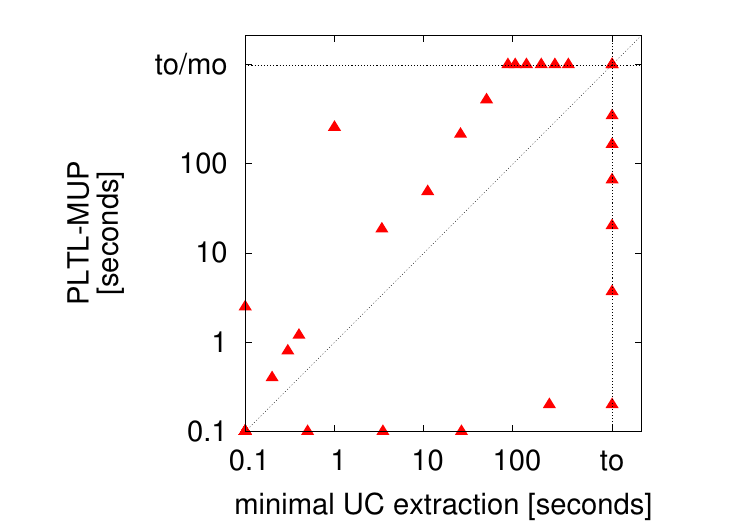} &
\tabfield{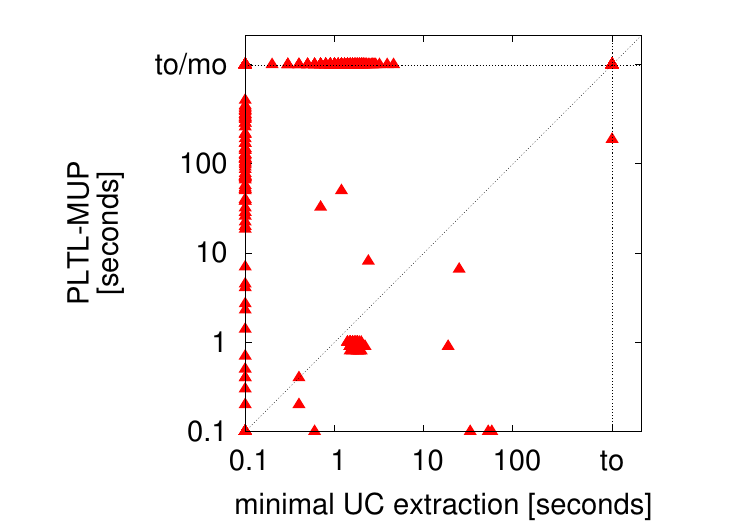} \\
\\
\tablineheader{memory} &
\tabfield{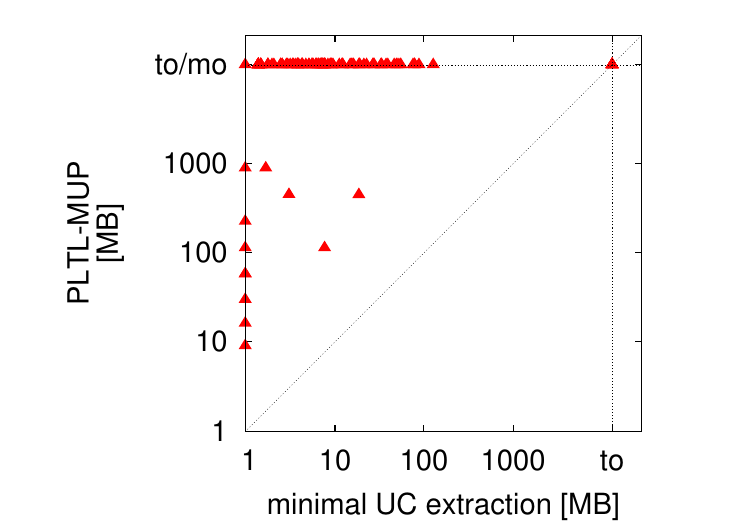} &
\tabfield{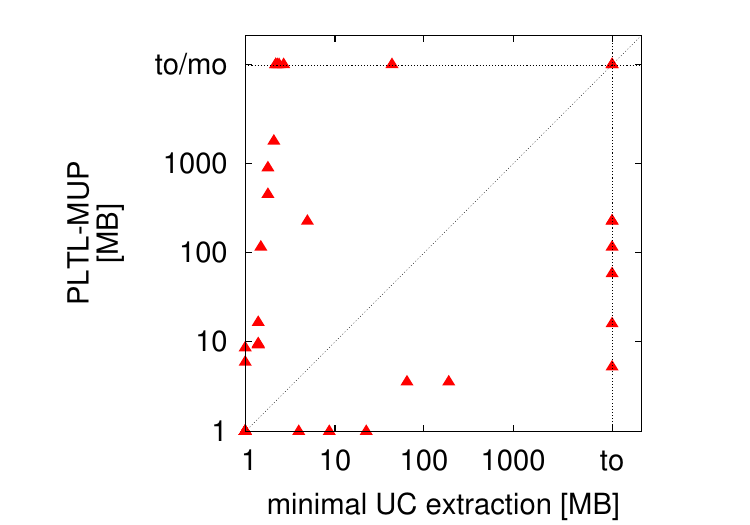} &
\tabfield{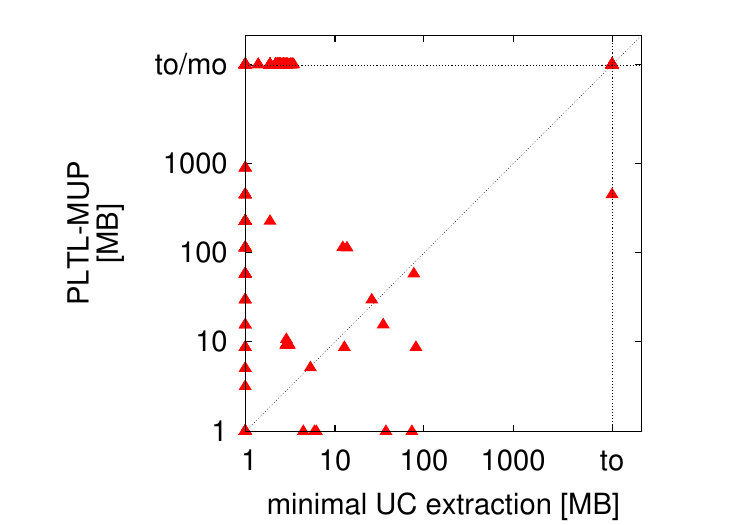} \\
\\[-1.5ex]
& \hspace*{3.75em}application & \hspace*{3.75em}crafted & \hspace*{3.75em}random \\
\\[-1.5ex]
\end{tabular}
\caption{\label{fig:proofdeletionvspltlmup} Comparison of minimal \uc extraction with \trp (x-axis) and with \pltlmup (y-axis) in terms of run time (in seconds) and memory (in MB) separated by categories \benchmark{application}, \benchmark{crafted}, and \benchmark{random}. \trp produced only time outs but no memory outs. \pltlmup produced both. \trp likely explores a larger search space during minimization and possibly produces more fine-grained \ucs than \pltlmup in that it computes a minimal \uc in LTL as a syntax tree according to Def.~\ref{def:snfltlmincore}, whereas \pltlmup computes a minimal set of unsatisfiable top-level conjuncts.}
\vspace{-3ex}
\end{figure}
}

{
\newcommand{\tabfield}[1]{%
\begin{minipage}{0.29\linewidth}%
\begin{center}%
\includegraphics[type=pdf,ext=.pdf,read=.pdf,scale=0.36,trim=1.1cm 0.0cm 1.3cm 0.5cm,clip]{#1}%
\end{center}%
\end{minipage}%
}%

\newcommand{\tablineheader}[1]{%
\begin{minipage}{0.02\linewidth}%
\begin{sideways}%
\hspace*{2em}{#1}%
\end{sideways}%
\end{minipage}%
}

\begin{figure}[h]
\hspace*{-0.5em}\begin{tabular}{cccc}
\tablineheader{run time} &
\tabfield{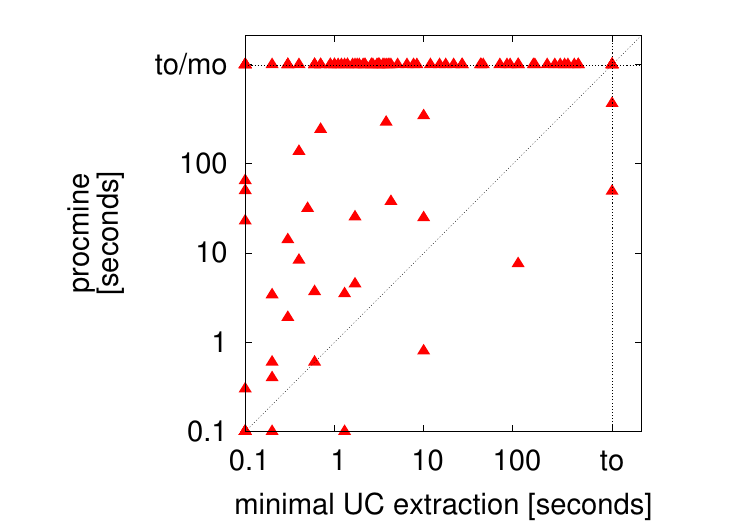} &
\tabfield{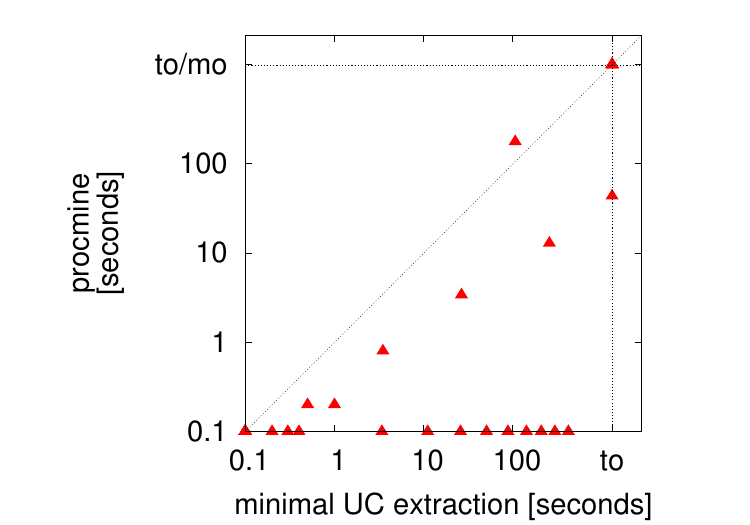} &
\tabfield{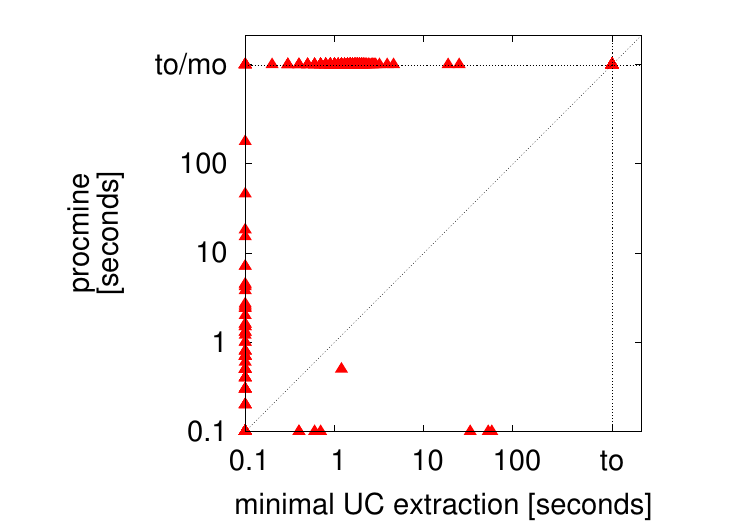} \\
\\
\tablineheader{memory} &
\tabfield{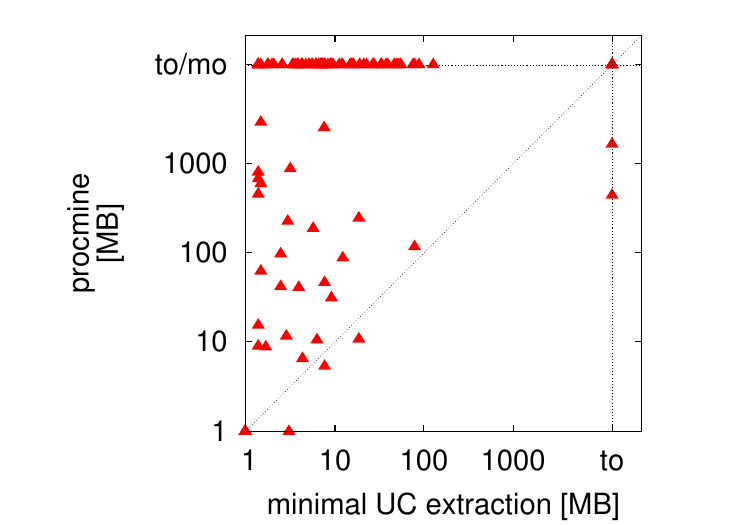} &
\tabfield{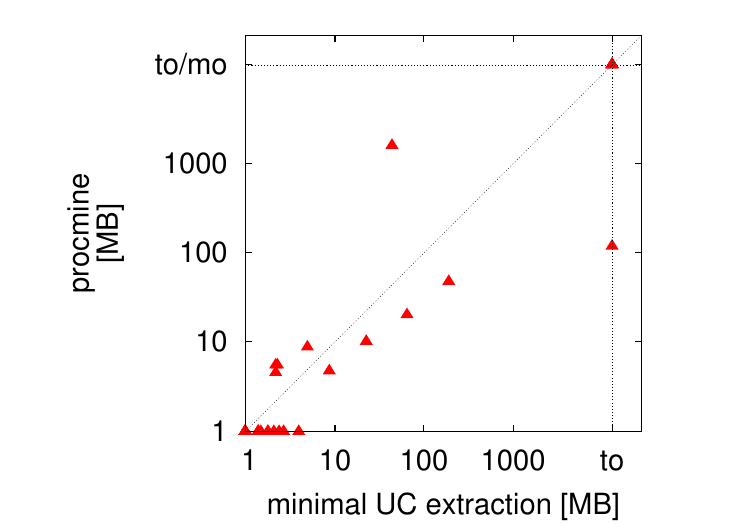} &
\tabfield{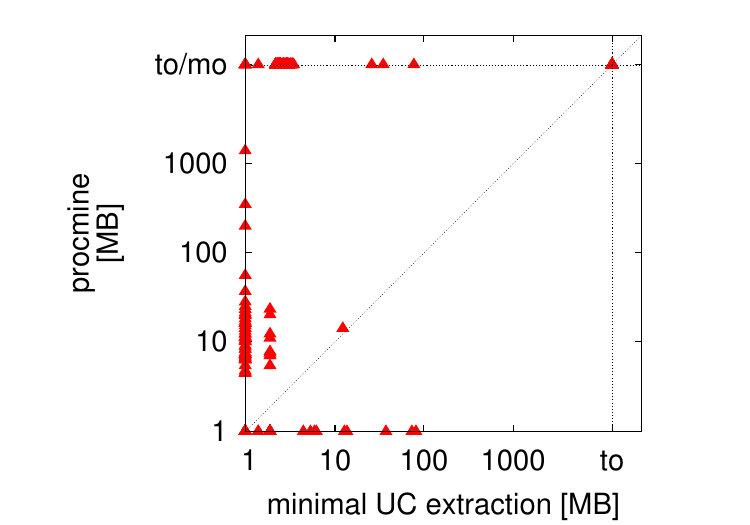} \\
\\[-1.5ex]
& \hspace*{3.75em}application & \hspace*{3.75em}crafted & \hspace*{3.75em}random \\
\\[-1.5ex]
\end{tabular}
\caption{\label{fig:proofdeletionvsprocmine} Comparison of minimal \uc extraction with \trp (x-axis) and with \procmine (y-axis) in terms of run time (in seconds) and memory (in MB) separated by categories \benchmark{application}, \benchmark{crafted}, and \benchmark{random}. \trp produced only time outs but no memory outs. \procmine produced both. \trp likely explores a larger search space during minimization and possibly produces more fine-grained \ucs than \procmine in that it computes a minimal \uc in LTL as a syntax tree according to Def.~\ref{def:snfltlmincore}, whereas \procmine computes a minimal set of unsatisfiable top-level conjuncts.}
\vspace{-5ex}
\end{figure}
}

{
\newcommand{\tabfield}[1]{%
\begin{minipage}{0.29\linewidth}%
\begin{center}%
\includegraphics[type=pdf,ext=.pdf,read=.pdf,scale=0.36,trim=1.1cm 0.0cm 1.3cm 0.5cm,clip]{#1}%
\end{center}%
\end{minipage}%
}%

\newcommand{\tablineheader}[1]{%
\begin{minipage}{0.02\linewidth}%
\begin{sideways}%
\hspace*{2em}{#1}%
\end{sideways}%
\end{minipage}%
}

\begin{figure}[h]
\hspace*{-0.5em}\begin{tabular}{cccc}
\tablineheader{run time} &
\tabfield{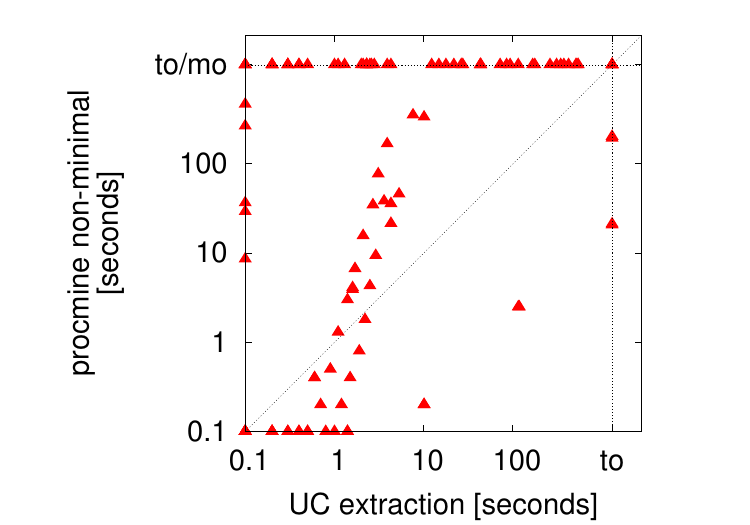} &
\tabfield{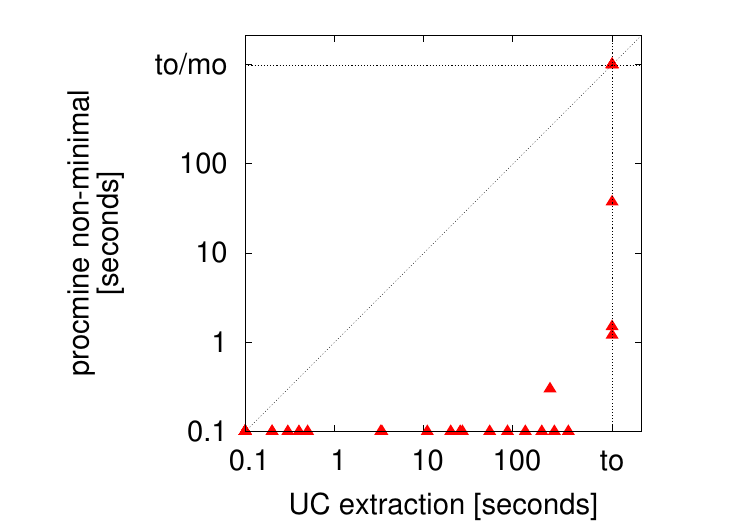} &
\tabfield{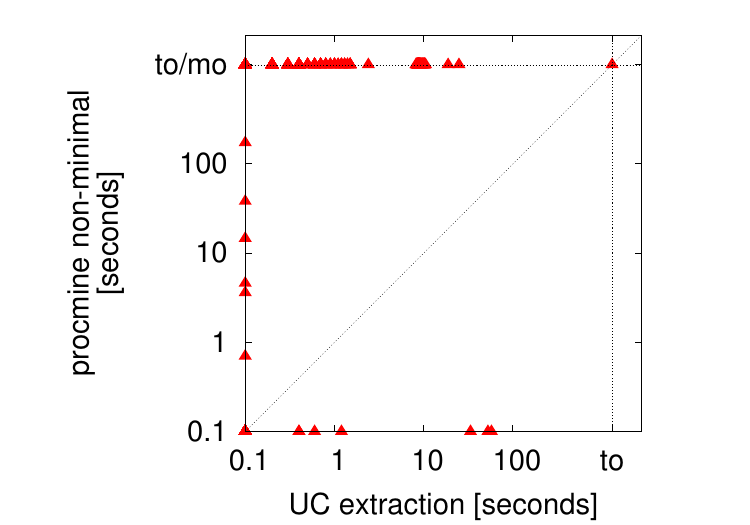} \\
\\
\tablineheader{memory} &
\tabfield{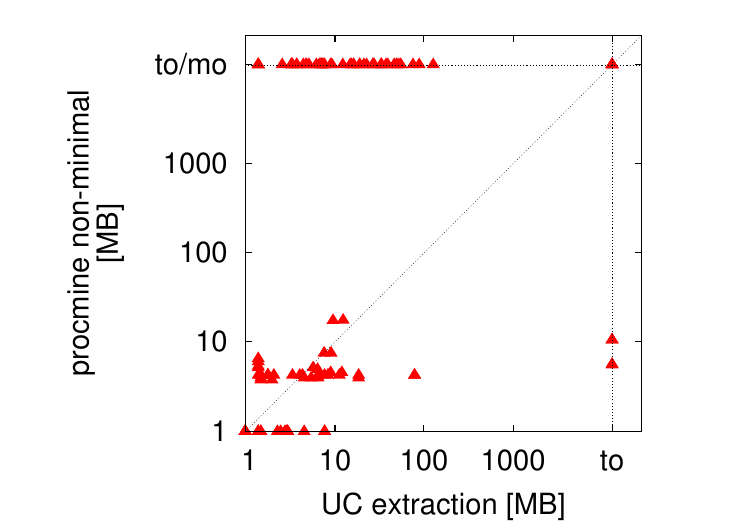} &
\tabfield{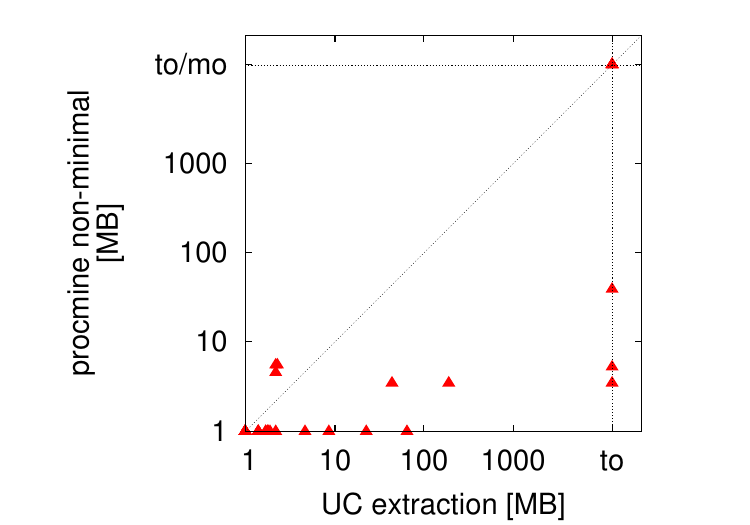} &
\tabfield{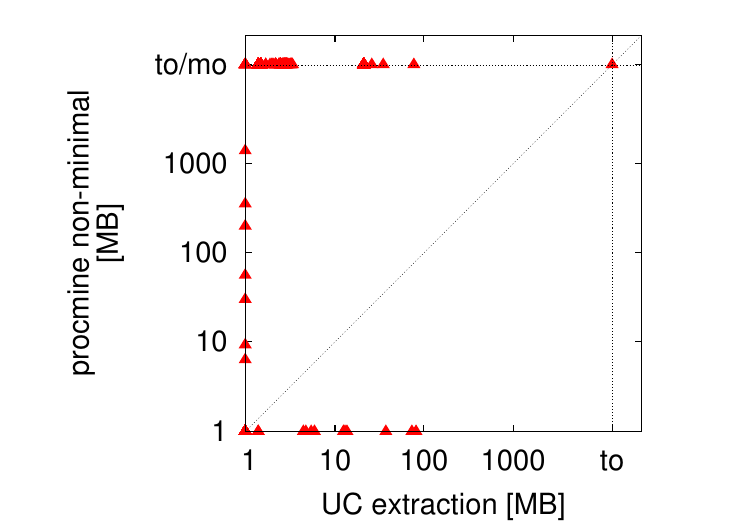} \\
\\
& \hspace*{3.75em}application & \hspace*{3.75em}crafted & \hspace*{3.75em}random \\
\\
\end{tabular}
\caption{\label{fig:proofvsprocminenomus} Comparison of (non-minimal) \uc extraction with \trp (x-axis) and with \procmine (y-axis) in terms of run time (in seconds) and memory (in MB) separated by categories \benchmark{application}, \benchmark{crafted}, and \benchmark{random}. \trp produced only time outs but no memory outs. \procmine produced both. \trp possibly produces more fine-grained \ucs than \procmine in that it computes a \uc in LTL as a syntax tree according to Def.~\ref{def:snfltlmincore}, whereas \procmine computes a set of unsatisfiable top-level conjuncts.}
\end{figure}
}

\fi
\end{document}